\documentclass[a5paper,11pt,russian]{article}
\usepackage[cp1251]{inputenc}
\usepackage[T2A]{fontenc}
\usepackage[russian]{babel}
\usepackage[dvips]{graphicx}
\usepackage{caption}
\usepackage{amsmath}
\usepackage{amssymb}
\usepackage{amsxtra}
\usepackage{amsthm,amscd}
\usepackage{color}
\usepackage[unicode,colorlinks]{hyperref}
\usepackage{bookmark}

\usepackage[papersize={175mm,250mm},left=2cm,right=1.5cm,
    top=2cm,bottom=2cm,bindingoffset=0cm]{geometry}

\newcommand{\ds}{\displaystyle}
\newcommand{\Bif}{\textit{Bif}}

\newcommand{\mbs}[1]{ {\boldsymbol #1} }

\newcommand{\id}[1]{\mathop{\rm id}\nolimits_{#1}}

\newcommand {\bbR}{\mathbb{R}}
\newcommand {\bbI}{\bbR^3(h,g,k)}

\newcommand {\tm}{T\mathcal{M}}
\newcommand {\bbRR}{\bbR^3}
\newcommand {\bbRRR}{\bbR_0^3}
\newcommand {\Id}{\mathop{\rm Id}\nolimits}
\newcommand {\Ass}{\mathop{\rm Ass}\nolimits}
\newcommand {\grad}{\mathop{\rm grad}\nolimits}
\newcommand {\ddiv}{\mathop{\rm div}\nolimits}
\newcommand {\rang}{\mathop{\rm rank}\nolimits}
\newcommand {\is}{\mathop{\rm is}\nolimits}
\newcommand{\iid}{\mathop{\rm id}\nolimits}

\newcommand {\astup}[1]{\mathop{\mbs{#1}}\limits^{\ast} \vphantom{#1}}

\newcommand{\Osc}{\mathop{\rm Osc}\nolimits}
\newcommand{\Con}{\mathop{\rm Con}\nolimits}
\newcommand{\Asp}{\mathop{\rm Asp}\nolimits}
\newcommand{\sn}{\mathop{\rm sn}\nolimits}
\newcommand{\cn}{\mathop{\rm cn}\nolimits}
\newcommand{\dn}{\mathop{\rm dn}\nolimits}

\newcommand{\cons}{\mathop{\rm const}\nolimits}

\newcommand{\gs}{\geqslant}
\newcommand{\ls}{\leqslant}
\newcommand{\sgn}{\mathop{\rm sgn}\nolimits}

\numberwithin{equation}{section}
\theoremstyle{plain}
\newtheorem{propos}{Предложение}[section]
\newtheorem{theorem}{Теорема}[section]
\newtheorem{lemma}{Лемма}[section]
\theoremstyle{definition}
\newtheorem{defin}{Определение}[section]
\theoremstyle{remark}
\newtheorem{remark}{Замечание}[section]
\newtheorem{example}{Пример}[section]
\newtheorem{soglash}{Соглашение}[section]
\numberwithin{figure}{section}

\newcounter{myta}
\numberwithin{myta}{section}
\newcommand{\myt}{\refstepcounter{myta}\themyta}
\newcommand {\ts}[1] {\textsl{#1}}

\begin{document}

\begin{titlepage}

{\footnotesize

\noindent ЛЕНИНГРАДСКИЙ ОРДЕНА ЛЕНИНА\\
И ОРДЕНА ТРУДОВОГО КРАСНОГО ЗНАМЕНИ\\
ГОСУДАРСТВЕННЫЙ УНИВЕРСИТЕТ\\
имени А.А. ЖДАНОВА

}

\vspace{3cm}

\noindent {\bf {\Huge М. П. Х А Р Л А М О В}}

\vspace{3cm}

{\bf \huge

\noindent ТОПОЛОГИЧЕСКИЙ АНАЛИЗ

\vspace{2mm}

\noindent ИНТЕГРИРУЕМЫХ ЗАДАЧ

\vspace{2mm}

\noindent ДИНАМИКИ ТВЕРДОГО ТЕЛА

}

\vspace{7cm}

{\footnotesize

\noindent ЛЕНИНГРАД\\
ИЗДАТЕЛЬСТВО\\
ЛЕНИНГРАДСКОГО УНИВЕРСИТЕТА\\
1988

}

\end{titlepage}
\setcounter{page}{2}

\tableofcontents

\clearpage

\section*{Предисловие к первому изданию}\label{sec0}

Классическая задача о движении твёрдого тела вокруг неподвижной точки в потенциальном силовом поле занимает исключительно важное место в теоретической механике. Значение ее сравнимо с тем, какое задача трех тел имеет в динамике систем, моделируемых совокупностью точечных масс. Впрочем, в последние десятилетия эти две основные области аналитической механики испытывают взаимотяготение и своеобразное <<врастание>> друг в друга. С другой стороны, задачу небесной механики теперь нередко ставят уже не как задачу точечных масс, а как задачу твердых тел (Г.Н.\,Дубошин, В.Т.\,Кондурарь, В.В.\,Белецкий и др.). С другой стороны, вследствие восходящей к Г.В.\,Колосову аналогии редуцированных уравнений динамики твердого тела с уравнениями движения материальных точек, оказалось возможным распространить в динамику твердого тела методы небесной механики (работы В.Г.\,Демина и его учеников). Таким образом, новые результаты, полученные в одной из указанных областей, достаточно быстро находят отклик и в другой из указанных областей, достаточно быстро находят отклик и в другой, что существенно повышает их ценность для теоретической механики в целом.

Традиционный подход к решению проблем механики, состоявший в отыскании случаев интегрируемости и указания способов сведения задачи к квадратурам, в динамике твердого тела привел, по существу, к тупиковой ситуации. Это обусловлено, во-первых, принципиальной неинтегрируемостью общей задачи (теорема Пуанкаре\,--\,Гюссона о не\-существовании алгебраических интегралов, теорема В.В.\,Козлова о несуществовании интеграла, аналитического по малому параметру в окрестности случая Эйлера), во-вторых, значительными техническими препятствиями при построении решений. Исчерпывает себя метод алгебраических инвариантных соотношений: уже для соотношений четвертой степени полное исследование требует огромной работы, а для более высоких степеней аналитические вычисления практически невыполнимы\footnote[1]{В 80-х годах XX в. представлялось невероятным, что аналитические вычисления можно полностью возложить на вычислительные машины.}. В связи с этим наиболее перспективны методы глобального анализа механических систем, привлекающие аппарат дифференциальной геометрии, теория гладких многообразий и гладких отображений, КАМ-теория, теория Морса. Их появлением мы обязаны в первую очередь задачам небесной механики. Намеченная С.\,Смейлом \cite{bib41} программа топологического исследования классических механических систем, указанные им пути ее реализации в натуральных системах с симметрией дали толчок и к исследованиям фазовой топологии в динамике твердого тела. К настоящему времени изучены бифуркации совместных уровней интеграла энергии и одного (в случае осесимметричного потенциала) или двух (в случае Лагранжа) интегралов, линейных относительно компонент угловой скорости (А.\,Якоб \cite{bib67}, Я.В.\,Татаринов \cite{bib51,bib52} и др.).

Так называемые общие случаи интегрируемости задачи о движении твердого тела вокруг неподвижной точки (в которых нет ограничений на начальные условия, за исключением, быть может, фиксации постоянной площадей) и их обобщения на задачу о движении гиростата с топологической точки зрения не рассматривались. Связано это в определенной мере с большим числом свободных параметров, а в части обобщений - с отсутствием удобной, действительно рабочей конкретизации гамильтоновой теории для механических систем с гироскопическими силами. Принципиальная же трудность состоит в том, что дополнительные интегралы уравнений Эйлера\,--\,Пуассона в решениях Эйлера\,--\,Жуковского, Ковалевской и Чаплыгина\,--\,Сретенского нелинейны по компонентам угловой скорости в подвижных осях и изучение соответствующих интегральных многообразий не укладываются в схему Смейла. В то же время роль точных решений в динамике твердого тела весьма велика. Подобно тому, как особые точки кривой определяют ее свойства в целом, так и интегрируемые случаи, порождая разделяющие поверхности в пространстве параметров изучаемого объекта, задают основные тенденции движения и должны послужить фундаментом для изучения существенно более сложных задач.

В предлагаемой книге изложен метод исследования фазовой топологии механических систем с нелинейными по скоростям первыми интегралами и выполнен топологический анализ указанных общих случаев интегрируемости задачи о движении гиростата.

Согласно известной теореме Лиувилля\,--\,Арнольда регулярная поверхность уровня первых интегралов вполне интегрируемой гамильтоновой системы представляет собой объединение торов, заполненных условно-периодическими траекториями. Более сложным и более интересным с практической точки зрения оказывается вопрос об устройстве критических интегральных поверхностей, разделяющих области независимых интегралов. Эти поверхности несут на себе все особые движения и полностью определяют характер перестроек, происходящих с интегральными многообразиями при изменении констант первых интегралов. Изучение критических поверхностей и связанных с ними бифуркаций интегральных многообразий в классических задачах твердого тела - основная цель настоящего исследования. Наличие явных закономерностей в многочисленных, иногда весьма нетривиальных, примерах бифуркаций позволяет надеяться на возникновение в ближайшем будущем общей <<теории Морса интегрируемых систем>>\footnote[1]{К концу 80-х такая теория действительно была создана в работах А.Т.\,Фоменко и его школы.}.

От читателя требуется знакомство с динамикой твердого тела в объеме стандартного курса теоретической механики и с основными понятиями геометрии многообразий. Используемые обозначения общеприняты в теории множеств и анализе.

В книге изложены результаты, полученные автором в 1975~-- 1982~гг. Оглядываясь назад, я понимаю, сколь многим обязан моему учителю Владимиру Михайловичу Алексееву, светлой памяти которого посвящаю эту работу.

\clearpage

\section{Гироскопические системы и симметрия}\label{sec1}
Современные задачи аналитической механики требуют в первую очередь изучения свойств, присущих самим системам, а не являющихся следствием выбора той или иной системы обобщенных координат. Существенную роль начинает играть язык описания механических систем. Уровень формализации изначально накладывает ограничения на широту рассматриваемых проблем. На сегодняшний день наиболее перспективным в плане приложений к механике представляется активно развивающийся аппарат нескольких отраслей математики, которые можно объединить под названием <<дифференциальная топология>>. Кроме очевидных соображений удобства и компактности изложения необходимость использования этого аппарата диктуется желанием вскрыть геометрическую сущность явлений, возможностью увидеть результаты, ускользающие от восприятия при классическом <<координатном>> подходе.

\subsection{Формализм Лагранжа}\label{ssec11}
В этом параграфе излагаются основные понятия, связанные с определением ла\-г\-ранжевых систем на гладких многообразиях. Подробности и доказательства неявно фигурирующих утверждений можно найти в \cite{bib03, bib65}.

Симплектической структурой на гладком многообразии $\mathfrak{M}$ называется замкнутая невырожденная 2-форма на $\mathfrak{M}$. Многообразие с заданной на нем симплектической структурой называется симплектическим. Симплектическое многообразие всегда имеет четную размерность.

Векторному полю $X$ на $\mathfrak{M}$ и $k$-форме $\alpha$ на $\mathfrak{M}$ сопоставим $(k - 1)$-форму, называемую внутренним произведением поля $X$ и формы $\alpha$ и обозначаемую $i_X \alpha$, такую, что для любых векторных полей $X_1 , …, X_{k-1}$ на $\mathfrak{M}$
$$(i_X\alpha)( X_1 , …, X_{k-1}) = \alpha(X, X_1 , …, X_{k-1}).$$

Пусть $(\mathfrak{M}, \sigma)$ -- симплектическое многообразие. Для любой 1-формы $\lambda$ на $\mathfrak{M}$ существует единственное векторное поле $X$ на $\mathfrak{M}$, удовлетворяющее равенству
$$
i_X\sigma = \lambda.
$$
Если форма $\lambda$ точна, то поле $X$ называется гамильтоновым, а функция $H$, такая, что
$$
i_X\sigma = -dH
$$
-- гамильтонианом поля $X$. При заданной симплектической структуре гамильтониан определен с точностью до аддитивной постоянной.

Пусть ${F_1, F_2}$ -- функции на симплектическом многообразии $(\mathfrak{M}, \sigma)$. Их скобкой Пуассона называется величина
$$
\{F_1, F_2\}=\sigma(X_1, X_2),
$$
где $X_i$ -- гамильтоново поле, порожденное гамильтонианом $F_i \: (i = 1, 2)$. Очевидны равенства
\begin{equation}\label{aeq1_1}
    \{F_1, F_2\} = X_2F_1 = dF_1(X_2) = -X_1F_2 = -dF_2(X_1).
\end{equation}
Говорят, что функции $F_1, F_2$ находятся в инволюции, если
$$
\{F_1, F_2\} \equiv 0.
$$

Пусть $\mathcal{M}$ - гладкое конечномерное многообразие. Положим $\mathfrak{M} = T^*\mathcal{M}$ и обозначим через
$$
q_\mathcal{M} : T^*\mathcal{M} \rightarrow \mathcal{M}
$$
стандартную проекцию кокасательного расслоения. Касательное к ней отображение
$$
Tq_\mathcal{M} : T\mathfrak{M} \rightarrow T\mathcal{M}
$$
переводит слой над элементом $\lambda_x \in \mathfrak{M}$ в слой над точкой $x = q_M(\lambda_x) \in \mathcal{M}$. Соответствие
$$
X \in T_{\lambda_x}\mathfrak{M} \mapsto \lambda_x(Tq_\mathcal{M}(X))
$$
является 1-формой на многообразии $\mathfrak{M}$. Обозначим эту форму через $\theta_*$. Ее внешний дифференциал $\sigma_* = d\theta_*$ определяет на $\mathfrak{M}$ симплектическую структуру. Формы $\theta_*$ и $\sigma_*$ называются каноническими формами на $T^*\mathcal{M}$.

Обозначим через
$$
p_\mathcal{M}:T\mathcal{M} \rightarrow \mathcal{M}
$$
стандартную проекцию касательного расслоения.

Векторное поле $X$ на многообразии $T\mathcal{M}$ (т.е. сечение второго касательного расслоения $X : T\mathcal{M} \rightarrow TT\mathcal{M}$) называется уравнением второго порядка (на $\mathcal{M}$) или, в терминологии \cite{bib10}, специальным векторным полем, если
\begin{equation}\label{aeq1_2}
    T_{p_\mathcal{M}}{\circ} X = \id{T\mathcal{M}}
\end{equation}
(тождественное отображение). Каждая интегральная кривая $w: (-\varepsilon, \varepsilon) \rightarrow T\mathcal{M}$ уравнения второго порядка является производной своей проекции
\begin{equation*}
    w(t)=\frac{d}{dt}(p_\mathcal{M}{\circ} w(t)).
\end{equation*}
Проекции $p_\mathcal{M}{\circ} w$ интегральных кривых принято называть решениями уравнения второго порядка.

Пусть $L$ - гладкая функция на многообразии $T\mathcal{M}$. Обозначим через $L_x$ ограничение $L$ на слой $T_x\mathcal{M}$ над точкой $x \in \mathcal{M}$.

\begin{defin}\label{def111}
Функция $L$ называется лагранжианом, если для всех $x \in \mathcal{M}$ функция $L_x$ имеет всюду отличный от нуля гессиан.
\end{defin}

Каждому лагранжиану $L$ сопоставим преобразование Лежандра
\begin{equation}\label{aeq1_3}
    \Phi_L : T\mathcal{M} \rightarrow T^*\mathcal{M},
\end{equation}
ставящее в соответствие вектору $w \in T_x\mathcal{M}$ элемент $dL_x(w)$ сопряженного пространства $T_x^*\mathcal{M}$. Отображение $\Phi_L$ сохраняет слои и всюду регулярно.

Энергией лагранжиана $L$ называется функция $H$ на $T\mathcal{M}$, определяемая соотношением
\begin{equation}\label{aeq1_4}
    H(w)=\Phi_L(w)\cdot w -L(w), \qquad w \in T\mathcal{M}.
\end{equation}

Если $f:A \to B$ -- отображение многообразий, то через $f^*:\Lambda(B) \to \Lambda(A)$ будем обозначать порожденное им антиотображение дифференциальных форм.

Лагранжиану $L$ на $T\mathcal{M}$ сопоставим дифференциальные формы
\begin{equation}\label{aeq1_5}
    \theta_L=\Phi_L^*\theta_*, \qquad \sigma_L=\Phi_L^*\sigma_*.
\end{equation}
Поскольку отображение $\Phi_L$ регулярно, пара $(T\mathcal{M}, \sigma_L)$ есть симплектическое многообразие.

\begin{defin}\label{def112}
Лагранжева система с лагранжианом $L$ -- это векторное поле $X$ на многообразии $T\mathcal{M}$, такое, что
\begin{equation}\label{aeq1_6}
    i_X \sigma_L = -dH.
\end{equation}
\end{defin}

Таким образом, лагранжева система есть гамильтоново векторное поле на симплектическом многообразии $(T\mathcal{M}, \sigma_L)$ с гамильтонианом, равным энергии лагранжиана $L$.

Пусть $q^1, \ldots, q^n$ -- локальные координаты в некоторой области на $\mathcal{M}$. Дифференциалы координатных функций $dq^1, \ldots, dq^n$ образуют базис в каждом кокасательном пространстве $T_x^*\mathcal{M}$. Следовательно, величины
$$
\dot{q}^1=dq^1(v_x), \; \ldots, \;\dot{q}^n=dq^n(v_x)
$$
могут быть выбраны в качестве координат касательного вектора $v_x \in T_x\mathcal{M}$, а набор функций $q^1, \ldots, q^n, \dot{q}^1, \ldots, \dot{q}^n$ -- в качестве системы координат в соответствующей области $T\mathcal{M}$. Такие системы координат на $T\mathcal{M}$ будем в дальнейшем называть естественными.

Записывая в естественных координат $L=L({\bf q},\dot{{\bf q}})$ и обозначая через ${d}/{dt}$ производную вдоль векторного поля $X$, получим, что лагранжева система, определенная согласно \eqref{aeq1_6}, равносильна заданию системы дифференциальных уравнений
\begin{equation}\label{aeq1_7}
\begin{array}{c}
    \ds \frac{dq^i}{dt}=\dot{q}^i, \quad \frac{d}{dt}\frac{\partial}{\partial \dot{q}^i}\left({\bf q},\frac{d{\bf q}}{dt}\right)-\frac{\partial L}{\partial q^i}\left({\bf q},\frac{d{\bf q}}{dt}\right)=0 \\[2mm]
    (i=1, \ldots, n).
\end{array}
\end{equation}
В частности, лагранжева система представляет собой уравнение второго порядка на $\mathcal{M}$.

Выделим особо класс лагранжианов на $T\mathcal{M}$, у которых ограничение на каждый слой $T_x\mathcal{M}$ имеет структуру полинома второй степени. Такие лагранжианы будем называть квадратичными. Квадратично-однородную часть лагранжиана $L$ обозначим через $Q(L)$. Очевидно, $Q(L)$ определяет на каждом слое $T_x\mathcal{M}$ невырожденную квадратичную форму.
Любой квадратичный лагранжиан на $T\mathcal{M}$ однозначно записывается в виде
\begin{equation}\label{aeq1_8}
    L=Q(L)+\lambda-\Pi \circ p_\mathcal{M},
\end{equation}
где $\lambda$ есть 1-форма на $\mathcal{M}$, рассматриваемая как функция на $T\mathcal{M}$, а $\Pi$ -- функция на $\mathcal{M}$.
Выражение энергии для лагранжиана \eqref{aeq1_8} получим из \eqref{aeq1_4} и теоремы Эйлера об однородных функциях:
\begin{equation}\label{aeq1_9}
    H=Q(L)+\Pi \circ p_\mathcal{M}.
\end{equation}

\begin{propos}\label{pro113}
 Пусть лагранжиан $L$ имеет вид \eqref{aeq1_8}. Тогда
\begin{equation}\label{aeq1_10}
    \sigma_L=\sigma_{Q(L)}+p_\mathcal{M}^*d\lambda.
\end{equation}
\end{propos}

\begin{proof} Из выражения \eqref{aeq1_8}, определения отображения \eqref{aeq1_3} и линейности функции $\lambda$ на слоях $T\mathcal{M}$ найдем
\begin{equation}\label{aeq1_11}
    \Phi_L=\Phi_Q(L)+\lambda \circ p_\mathcal{M}.
\end{equation}

Пусть $w \in T_x\mathcal{M}, X \in T_wT\mathcal{M}$. По определению \eqref{aeq1_5}
\begin{equation}\label{aeq1_12}
    \theta_L(w) \cdot X = \theta_*(\Phi_L(w)) \cdot T_w\Phi_L(X) = \Phi_L(w) \cdot T_w (q_\mathcal{M} \circ \Phi_L)(X).
\end{equation}
Подставляя в это выражение \eqref{aeq1_11}, учтем, что $\Phi_L$ сохраняет слои над $\mathcal{M}$, т.е. $q_\mathcal{M} \circ \Phi_L = p_\mathcal{M}$. Поэтому
\begin{equation}\label{aeq1_13}
    \theta_L(w) \cdot X = \theta_*(\Phi_L(w)) \cdot T_{p_\mathcal{M}}(X) + \lambda(x) \cdot T_{p_\mathcal{M}}(X).
\end{equation}
С другой стороны, заменив в \eqref{aeq1_12} $L$ на $Q(L)$, получим
\begin{equation}\label{aeq1_14}
    \theta_{Q(L)}(w) \cdot X = \Phi_{Q(L)}(w) \cdot T_{p_\mathcal{M}}(X).
\end{equation}
Сравнение соотношений \eqref{aeq1_13}, \eqref{aeq1_14} дает $\theta_L = \theta_{Q(L)} + p_\mathcal{M}^*\lambda$. Дифференцируя последнее равенство внешним образом, приходим к выражению \eqref{aeq1_10}, что и требовалось доказать.\end{proof}

Согласно формулам \eqref{aeq1_6}, \eqref{aeq1_9}, \eqref{aeq1_10} лагранжева система зависит не от входящей в лагранжиан линейной формы, а от ее внешнего дифференциала -- некоторой точной \mbox{2-формы} на конфигурационном пространстве $\mathcal{M}$. Соответствующие слагаемые в уравнениях \eqref{aeq1_7} обычно трактуются как действующие на систему гироскопические силы. Требование лагранжевости таких сил иногда становится слишком ограничительным. В самом деле, гироскопические силы являются обычно следствием <<скрытых движений>>, т.е. понижения порядка натуральных систем. При этом могут возникать системы, не обладающие глобальным лагранжианом. Естественно желание снять подобные ограничения, что приводит к определению понятия гироскопической системы, обобщающего понятие системы с квадратичным лагранжианом.

\subsection{Механические системы с гироскопическими силами}\label{ssec12}
Механической системой с гироскопическими силами, или, короче, гироскопической системой, назовем четверку
\begin{equation}\label{aeq2_1}
    (\mathcal{M}, s, \Pi, \varkappa),
\end{equation}
в которой $\mathcal{M}$ -- гладкое многообразие (конфигурационное пространство системы), $s$ -- риманова метрика на $\mathcal{M}$ (скалярное произведение в касательных пространствах $T_x\mathcal{M}$, гладко зависящее от точки $x \in \mathcal{M}$), $\Pi$ -- дифференцируемая функция на $\mathcal{M}$ (потенциальная энергия системы или потенциал), $\varkappa$ -- замкнутая 2-форма на $\mathcal{M}$ (форма гироскопических сил).

Скалярное произведение в метрике $s$ обозначим для удобства угловыми скобками. Пусть, как обычно, $\|w\|=\langle w,w\rangle ^{1/2}, \: w \in T\mathcal{M}$. Функция
 \begin{equation}\label{aeq2_2}
    w \mapsto \frac{1}{2}\|w\|^2, \quad w \in T\mathcal{M}
 \end{equation}
называется кинетической энергией системы \eqref{aeq2_1}. Полная энергия определяется равенством
\begin{equation}\label{aeq2_3}
    H(w)=\frac{1}{2}\|w\|^2+\Pi\left(p_\mathcal{M}(w)\right).
\end{equation}

Для описания системы необходимо ввести соответствующее векторное поле на $T\mathcal{M}$. Начнем с определения форм, аналогичных \eqref{aeq1_5}.

Пусть $w \in T\mathcal{M}$. В касательном пространстве $T_wT\mathcal{M}$ определим линейную форму $\theta(w)$, положив для всякого $Y \in T_w T\mathcal{M}$
\begin{equation}\label{aeq2_4}
    \theta(w) \cdot Y = \langle w, T_{p_\mathcal{M}}(Y) \rangle .
\end{equation}
\begin{propos}\label{pro121}
Соответствие
\begin{equation}\label{aeq2_5}
    \theta:w \mapsto \theta(w)
\end{equation}
задает на $T\mathcal{M}$ дифференциальную 1-форму. Ее внешний дифференциал
\begin{equation}\label{aeq2_6}
    \sigma=d\theta,
\end{equation}
называемый лагранжевой формой метрики $s$, превращает $T\mathcal{M}$ в симплектическое многообразие.
\end{propos}

Нетрудно увязать это утверждение со свойствами форм \eqref{aeq1_5}. Проведем, однако, независимое наглядное доказательство.

Пусть $\left( q^1,\ldots, q^m, \dot{q}^1, \ldots, \dot{q}^m \right)$ -- естественные координаты на $T\mathcal{M}$ в окрестности некоторой точки $w$. Линейные формы $dq^1, \ldots, dq^m, d\dot{q}^1, \ldots, d\dot{q}^m$ образуют базис в слоях кокасательного пространства $T^*T\mathcal{M}$. Если $A = \left(a_{ij}({\bf q})\right)$ -- положительно-определенная симметрическая матрица римановой метрики $s$, т.е.
\begin{equation*}
    s_q\left(\dot{\bf q}_{(1)}, \dot{\bf q}_{(2)}\right) = \sum_{i,j=1}^m a_{ij}\left({\bf q}\right) \dot{q}_{(1)}^i \dot{q}_{(2)}^j,
\end{equation*}
то соответствие \eqref{aeq2_5} согласно \eqref{aeq2_4} принимает вид
\begin{equation}\label{aeq2_7}
     \theta \left({\bf q}, \dot{\bf q}\right) = \sum_{i,j=1}^m a_{ij}\left({\bf q}\right) \dot{q}^i dq^j
\end{equation}
и является гладким. Следовательно, \eqref{aeq2_5} представляет собой дифференцируемое сечение $T^*T\mathcal{M}$, т.е. дифференциальную 1-форму на $T\mathcal{M}$.

Для доказательства второго утверждения необходимо проверить лишь невырожденность формы \eqref{aeq2_6}. Дифференцируя выражение \eqref{aeq2_7} внешним образом, находим
\begin{equation*}
     \sigma \left({\bf q}, \dot{\bf q}\right) = \sum_{i,j=1}^m a_{ij}\left({\bf q}\right) d\dot{q}^i \wedge dq^j + \sum_{i,j,k=1}^m \dot{q}^i\frac{\partial a_{ij}\left({\bf q}\right)}{\partial q^k} dq^k \wedge dq^j.
\end{equation*}
Следовательно, матрица формы $\sigma$
\begin{equation}\label{aeq2_8}
    S=\left[ \begin{array}{cc}* & -A \\ A & 0 \end{array}\right]
\end{equation}
имеет отличный от нуля определитель $\det S = (\det A)^2$, что и требовалось.
\begin{theorem}\label{the122}
Существует единственное векторное поле $X$ на многообразии $T\mathcal{M}$, удовлетворяющее соотношению
\begin{equation}\label{aeq2_9}
    i_X(\sigma + p_\mathcal{M}^*\varkappa)=-dH.
\end{equation}
Это поле является уравнением второго порядка.
\end{theorem}

\begin{proof}
Заметим, что образ формы $\varkappa$ при отображении $p_\mathcal{M}^*$ не содержит в локальной записи форм $d {\dot q}^i$, поэтому матрица 2-формы $\sigma + p_\mathcal{M}^* \varkappa$ отличается от матрицы \eqref{aeq2_8} лишь левым верхним углом. В частности, ее детерминант совпадает с $\det S$ и отличен от нуля. Таким образом, форма $\sigma + p_\mathcal{M}^* \varkappa$ невырождена, откуда следует существование и единственность поля $X$. Второе утверждение элементарно проверяется с помощью естественных координат на $T\mathcal{M}$. Условие \eqref{aeq1_2} сводится при этом к виду $\ds \dot{\bf q} = d{\bf q}/{dt}$, где $d/{dt}$ -- производная вдоль поля $X$, и для его проверки достаточно приравнять коэффициенты при $d \dot{q}^i$ в обеих частях соотношения \eqref{aeq2_9}.\end{proof}

\begin{defin}\label{def123}
Векторное поле $X$ на многообразии $T\mathcal{M}$, введенное согласно \eqref{aeq2_9}, называем динамической системой, соответствующей механической системе \eqref{aeq2_1}. Интегральные кривые поля $X$ называем фазовыми траекториями, а его решения (проекции интегральных кривых на $\mathcal{M})$ -- движениями в системе \eqref{aeq2_1}.
\end{defin}

Интегралом механической системы \eqref{aeq2_1} назовем первый интеграл соответствующей динамической системы.

\begin{propos}\label{pro124}
Полная энергия $H$ является интегралом системы \eqref{aeq2_1}.
\end{propos}

\noindent Действительно, $H$ есть гамильтониан поля $X$ на симплектическом многообразии
$$
\left(T\mathcal{M}, \sigma+p_\mathcal{M}^*\varkappa\right).
$$
Другой важный класс интегралов укажем в следующем параграфе.

В заключение укажем любопытную геометрическую интерпретацию движений в гироскопических системах, обобщающую принцип Мопертюи.

Зафиксируем постоянную $h$ интеграла энергии. Из его представления \eqref{aeq2_3} следует, что движения происходят в области $\mathcal{M}_h=\{x \in \mathcal{M}: \Pi(x)\ls h\}$. Определим в этой области риманову метрику (вырождающуюся на границе) $s_h$, положив для любых $w_1,w_2 \in T_x\mathcal{M}_h$
\begin{equation*}
    s_h(x)(w_1, w_2) = 2[h - \Pi(x)]\langle w_1, w_2 \rangle_x.
\end{equation*}
В частности, $\|w\|_{s_h}=\sqrt{2[h - \Pi(x)]}\|w\|_s$, если $w \in T_x\mathcal{M}_h$. Пусть $\Gamma_h:T_x\mathcal{M}_h\rightarrow T_x^*\mathcal{M}_h$ -- оператор <<опускания индексов>> в  метрике $s_h$, так что
\begin{equation*}
    \Gamma_h(v) \cdot w=s_h(x)(v, w), \qquad v, w \in T_x\mathcal{M}_h.
\end{equation*}

\begin{theorem}\label{the125}
Пусть $x(t)$ -- движение в системе \eqref{aeq2_1}, при котором выполнен закон сохранения
\begin{equation}\label{aeq2_10}
    H\left(x'(t)\right) \equiv h.
\end{equation}
Если через $\overline{x}(\tau)$ обозначить ту же кривую, но параметризированную длиной дуги метрики $s_h$, то справедливо равенство
\begin{equation}\label{aeq2_11}
    \frac{D}{d\tau}\frac{d\overline{x}}{d\tau} = -\Gamma_h^{-1}\left(i_{\frac{d\overline{x}}{d\tau}} \varkappa\right),
\end{equation}
где ковариантная производная берется в метрике $s_h$.

Обратно, пусть $\overline{x}(\tau)$ -- кривая в области $\mathcal{M}_h$, параметризованная дугой метрики $s_h$ и удовлетворяющая равенству \eqref{aeq2_11}. Тогда существует такая монотонная замена параметра $\tau = \tau(t)$, что кривая $x(t) = \overline{x}(\tau(t))$ является движением в системе \eqref{aeq2_1}, при котором выполнен закон \eqref{aeq2_10}.
\end{theorem}
Утверждение носит локальный характер, и его доказательство в координатах сводится к элементарным выкладкам. Вариационное доказательство первой части для случая двумерного конфигурационного пространства $\mathcal{M}$ имеется в работе [2].

Потоки на многообразиях постоянной энергии, определенные уравнением вида \eqref{aeq2_11}, естественно называть потоками заданной кривизны. В случае $\varkappa=0$ получаем обычный геодезический поток метрики $s_h$ (одна из формулировок принципа Мопертюи).

Отметим случай ${\rm dim}\mathcal{M}=2$. При этом геодезическая кривизна траекторий \eqref{aeq2_11} зависит лишь от точки многообразия, а не от направления самих траекторий. А именно, пусть $\omega_h$~-- форма объема в области ${\rm int}\mathcal{M}=\{x \in \mathcal{M}:\Pi(x)<h\}$, соответствующая метрике $s_h$. Тогда существует функция $\varkappa = k_h \omega_h$. Для любого вектора $v \in T_x\mathcal{M}$, единичного в метрике $s_h$, вектор $w = -\Gamma_h^{-1}(i_v\varkappa)$ ортогонален $v$ в метрике $s_h$:
\begin{equation*}
    s_h(v, w) = \Gamma_h(w) \cdot v = -\varkappa(v, v) = 0.
\end{equation*}
Длина вектора $w$ (геодезическая кривизна траектории, если положить $\ds v = \frac{d\overline{x}}{d\tau}$) поэтому равна $|\omega_h(v, w)|$. С другой стороны,
\begin{equation*}
    \begin{array}{rcl}
        \|w\|_{s_h}^2 & = & \|\Gamma_h^{-1}(i_v\varkappa)\|_{s_h}^2 = i_v\varkappa \cdot \Gamma_h^{-1}(i_v\varkappa) = \\
        & = & \varkappa \left(v,\Gamma_h^{-1}(i_v\varkappa) \right) = -\varkappa(v, w) = -k_h\omega_h(v, w).
    \end{array}
\end{equation*}

Отсюда вытекает, во-первых, что $\|w\|_{s_h} = |k_h|$. Во-вторых, $\varkappa(v, w) > 0$, если $k_h \neq 0$, так что базис $\{w,v \}$ в $T_x\mathcal{M}_h$ определяет ту же ориентацию, что и 2-форма $\varkappa$. Итак, доказано следующее утверждение.

\begin{propos}\label{pro126}
Пусть в системе \eqref{aeq2_1} ${\rm dim}\mathcal{M}=2$. Кривая $x(t)$, удовлетворяющая закону \eqref{aeq2_10}, является движением в этой системе тогда и только тогда, когда, будучи параметризована длиной дуги $s_h$, имеет геодезическую кривизну, равную $\left| {\varkappa}/{\omega_x} \right|$, а направление искривления таково, что базис {\rm \{}вектор кривизны, касательный вектор{\rm \}} ориентирует касательное к $\mathcal{M}$ пространство в точке $x(t)$ так же, как и форма гироскопических сил.
\end{propos}

\begin{remark}\label{rem121}
В утверждениях этого параграфа не использована замкнутость формы гироскопических сил. Однако это требование естественно. Уже упоминавшееся понижение порядка всегда приводит к замкнутым формам \cite{bib59}. Кроме того, замкнутость необходима, чтобы 2-форма $\sigma + p_\mathcal{M}^*\varkappa$ определяла на фазовом пространстве $T\mathcal{M}$ симплектическую структуру и система \eqref{aeq2_1} обладала всеми свойствами систем Гамильтона.
\end{remark}

\subsection{Симметрия в гироскопических системах}\label{ssec13}
Пусть на гладком многообразии $\mathcal{M}$ действует однопараметрическая группа диффеоморфизмов $\Psi = \left\{\psi^\tau\right\}$:
\begin{equation}\label{aeq3_1}
    \psi ^ \tau  : \mathcal{M} \to \mathcal{M}, \qquad \tau \in \bbR.
\end{equation}
Действие $\Psi$ распространяется на $T\mathcal{M}$ с помощью касательных отображений:
\begin{equation}\label{aeq3_2}
    T {\psi ^ \tau} : T\mathcal{M} \to T\mathcal{M}, \qquad \psi^\tau \in \ \Psi.
\end{equation}
Векторные поля, порождающие действия \eqref{aeq3_1}, \eqref{aeq3_2}, обозначим
\begin{eqnarray}
&     \ds v(x) & = \left. \frac{d}{d\tau}\right|_{\tau = 0} \psi^\tau(x), \quad x \in \mathcal{M}, \label{aeq3_3}\\
&     \ds v_T(x) & = \left. \frac{d}{d\tau}\right|_{\tau = 0} T\psi^\tau(x), \quad x \in T\mathcal{M}. \label{aeq3_4}
\end{eqnarray}

Согласно известной теореме Нётер (см., например, \cite[\S 20]{bib05}), лагранжева система на $T\mathcal{M}$ с $\Psi$-сохранным лагранжианом $L : T\mathcal{M} \to \bbR$
$$
 \forall \tau \in \bbR, w \in T\mathcal{M} : L \left(T \psi^ \tau (w) \right) = L(w)
$$
обладает первым интегралом $G : T\mathcal{M} \to \bbR$, имеющим в локальных координатах вид
\begin{equation}\label{aeq3_5}
    \ds G \left( {\bf q}, \dot{{\bf q}} \right) = \frac{\partial L}{\partial {\dot {\bf q}}} \cdot \left. \frac{d}{d\tau}\right|_{\tau = 0} \psi^\tau({\bf q}).
\end{equation}
В случае, когда лагранжиан $L$ квадратичен и его квадратично-однородная часть порождена римановой метрикой $ \langle,\rangle$ на $\mathcal{M}$, интеграл \eqref{aeq3_5} можно записать в инвариантном виде
\begin{equation}\label{aeq3_6}
    G(w) = \langle v, w \rangle + f \circ p_\mathcal{M}(w),
\end{equation}
где $f : \mathcal{M} \to \bbR$ -- гладкая функция, $v$ -- поле \eqref{aeq3_3}.

Условие $\Psi$-сохранности лагранжиана является достаточным, но отнюдь не необходимым условием для существования интеграла \eqref{aeq3_6}. В самом деле, если к $L$ прибавить произвольную замкнутую форму (в локальных координатах~-- полную производную функции от ${\bf q}$), то, вообще говоря, нарушится свойство сохранности. Однако, как показано в $\S$ 1.1, соответствующее векторное поле на $T\mathcal{M}$ не изменяется (в координатах не изменяются уравнения \eqref{aeq1_7}~-- факт хорошо известный \cite{bib30, bib54}), а значит, остается и интеграл \eqref{aeq3_6}. Поэтому имеет смысл перейти к инвариантным формулировкам предыдущего параграфа и поставить вопрос о нахождении критериев существования у гироскопической системы \eqref{aeq2_1} интегралов вида \eqref{aeq3_6}.

Предварительно отметим одно интересное характеристическое свойство ура\-внений второго порядка.
\begin{lemma}\label{lem131}
Векторное поле $X$ на $T\mathcal{M}$ является уравнением второго порядка тогда и только тогда, когда для всякой функции $f$ на многообразии $\mathcal{M}$ выполнено равенство
\begin{equation}\label{aeq3_7}
    \ds \frac{d}{dt}(f \circ p_\mathcal{M})=df.
\end{equation}
\end{lemma}
\noindent Здесь ${d}/{dt}$ -- производная вдоль поля $X$, а форма $df$ на $\mathcal{M}$ рассматривается как функция на $T\mathcal{M}$.

\begin{proof}
Для любого векторного поля $X$ на $T\mathcal{M}$ имеем
\begin{equation}\label{aeq3_8}
    \ds \frac{d}{dt}(f \circ p_\mathcal{M}) = d(f \circ p_\mathcal{M}) \cdot X = df \cdot T_{p_\mathcal{M}}(X).
\end{equation}
Согласно определению \eqref{aeq1_2} необходимость \eqref{aeq3_7} очевидна.

Пусть \eqref{aeq3_7} выполнено для любой функции $f$. Допустим, что при этом
$$
T_{p_\mathcal{M}}(X(w_1)) = w_2 \neq w_1, \quad w_1, w_2 \in T_x\mathcal{M}.
$$
Всегда найдется функция $f$, для которой $df(x) (w_1 - w_2) \neq 0$, что противоречит равенствам \eqref{aeq3_7}, \eqref{aeq3_8}. Этим доказана достаточность условия \eqref{aeq3_7}.\end{proof}

Рассмотрим теперь гироскопическую систему \eqref{aeq2_1}, на конфигурационном пространстве которой действует группа \eqref{aeq3_1}.
\begin{defin}\label{def132}
Группу $\Psi$ называем группой симметрий системы
$$(\mathcal{M}, s, \Pi, \varkappa),$$
если ее преобразования являются изометриями $s$ и сохраняют потенциал $\Pi$ и форму $\varkappa$.
\end{defin}
При $\varkappa \equiv 0$ получаем определение симметрии в натуральной механической системе \cite{bib41}.

Данное определение можно записать в терминах порождающих полей \eqref{aeq3_3}, \eqref{aeq3_4}. В этом случае получим: если $\Psi$ -- группа симметрий системы \eqref{aeq2_1}, то выполнены равенства
\begin{gather}
    v\Pi \equiv 0, \label{aeq3_9}\\
    v_T H \equiv 0,\label{aeq3_10}\\
    \mathbf{L}_v\varkappa \equiv 0,\label{aeq3_11}
\end{gather}
где $\mathbf{L}_v$ -- производная Ли вдоль векторного поля $v$.

\begin{lemma}\label{lem133}
Если $\Psi$ -- группа симметрий, то она сохраняет дифференциальные формы \eqref{aeq2_5}, \eqref{aeq2_6}.
\end{lemma}
С учетом сохранности метрики $s$ это очевидно.

\begin{propos}\label{pro134}
Если $\Psi$ -- группа симметрий, то ее преобразования являются симплектическими диффеоморфизмами $(T\mathcal{M}, \sigma + p_\mathcal{M}^* \varkappa).$
\end{propos}
\begin{proof}
Заметим, что для всех $w \in T_x\mathcal{M}$
\begin{equation*}
\begin{array}{rcl}
    \ds{T_{p_\mathcal{M}}(v_T(w))} & = & \ds{T_{p_\mathcal{M}}\left(\left. \frac{d}{dt}\right|_{\tau=0}T\psi^\tau(w)\right) =} \\[3mm]
    & = & \ds{\left. \frac{d}{dt}\right|_{\tau=0}p_\mathcal{M} \circ T\psi^\tau(w) = \left. \frac{d}{dt}\right|_{\tau=0}\psi^\tau \left(p_\mathcal{M}(w) \right) = v(x).}
\end{array}
\end{equation*}
Следовательно,
\begin{equation}\label{aeq3_12}
    T_{p_\mathcal{M}}(v_T) = v.
\end{equation}
Тогда
\begin{equation*}
    {\bf L}_{v_T}(\sigma + p_\mathcal{M}^* \varkappa)={\bf L}_{v_T}\sigma +p_\mathcal{M}^*{\bf L}_{T_{p_M}(v_T)}\varkappa = 0
\end{equation*}
в силу \eqref{aeq3_11} и леммы \ref{lem133}. Поэтому преобразования группы $\Psi$ сохраняют симплектическую структуру, что и утверждалось.\end{proof}

Далее $X$ означает векторное поле, определяющее динамику системы \eqref{aeq2_1} с группой симметрий \eqref{aeq3_1}, $d/dt$ -- производную вдоль $X$ функций на
${\tm}$. Положим
\begin{equation}\label{aeq3_13}
    \widetilde{G}(w) = \langle v(x), w \rangle, \qquad w \in T_x\mathcal{M}.
\end{equation}

\begin{propos}\label{pro135}
Справедливо равенство
\begin{equation}\label{aeq3_14}
    \frac{d\widetilde{G}}{dt} = i_v(\varkappa),
\end{equation}
в котором форма $i_v \varkappa$ на многообразии $\mathcal{M}$ рассматривается как функция на ${\tm}$.
\end{propos}
\begin{proof}
С учетом соотношения \eqref{aeq3_12} представим функцию \eqref{aeq3_13} в виде
\begin{equation}\label{aeq3_15}
    \widetilde{G}(w) = \left\langle w, T_{p_\mathcal{M}}(v_T)\right\rangle = \left( i_{v_T}\theta\right)(w).
\end{equation}
Следовательно,
\begin{equation*}
    \ds \frac{d\widetilde{G}}{dt}=Xi_{v_T}\theta = d \left(i_{v_T}\theta \right) \cdot X.
\end{equation*}
Используя формулу для производной Ли \cite{bib13}
\begin{equation}\label{aeq3_16}
    {\bf L}_y = d\:i_y + i_y d,
\end{equation}
с учетом леммы \ref{lem133} получаем
\begin{equation}\label{aeq3_17}
    di_{v_T}\theta = i_{v_T}\sigma.
\end{equation}
Согласно формулам \eqref{aeq1_2}, \eqref{aeq2_9}, \eqref{aeq3_10}, \eqref{aeq3_12} находим
\begin{equation*}
\begin{array}{rcl}
    \ds \frac{d\widetilde{G}}{dt} & = & - \left(i_{v_T}\sigma \right) \cdot X = \left(i_X\sigma \right) \cdot v_T = \\
    {}& = & -dH \cdot v_T - \left(i_X p_\mathcal{M}^* \varkappa \right) \cdot v_T = \varkappa \left(T_{p_\mathcal{M}}(v_T), T_{p_\mathcal{M}}(X) \right) = i_v \varkappa,
\end{array}
\end{equation*}
что и требовалось доказать.\end{proof}

\begin{theorem}\label{the136}
Пусть на конфигурационном пространстве механической системы с гироскопическими силами $(\mathcal{M}, s, \Pi, \varkappa)$ действует однопараметрическая группа, сохраняющая риманову метрику и потенциал. Тогда для существования интеграла вида \eqref{aeq3_6}, где $v$ -- поле, порождающее действие группы, необходимо и достаточно, чтобы 1-форма $i_v \varkappa$ была точной. При этом ее связь с интегралом \eqref{aeq3_6} такова:
\begin{equation}\label{aeq3_18}
    i_v \varkappa = -df.
\end{equation}
\end{theorem}
\begin{proof}
Согласно лемме \ref{lem131} и равенству \eqref{aeq3_14}, условие \eqref{aeq3_18} равносильно тому, что $\ds \frac{d}{dt}(\widetilde{G} + f \circ p_\mathcal{M})\equiv 0$, т.е. существованию интеграла \eqref{aeq3_6}.\end{proof}

Условие \eqref{aeq3_18} по форме совпадает с определением <<гамильтоновой системы>>  $v$ на <<симплектическом многообразии>>   $(\mathcal{M}, \varkappa)$ с <<гамильтонианом>> $f$. Его точный смысл устанавливается следующим образом.
\begin{propos}\label{pro137}
Если выполнено условие \eqref{aeq3_18}, то поток, определяемый группой $\Psi$ на симплектическом многообразии $({\tm}, \sigma + p_\mathcal{M}^* \varkappa)$, гамильтонов, и его гамильтониан есть функция \eqref{aeq3_6}.
\end{propos}
\begin{proof}
Так как поток $\Psi$ на ${\tm}$ определяется полем \eqref{aeq3_4}, то необходимо доказать, что
\begin{equation}\label{aeq3_19}
    i_{v_T}(\sigma + p_\mathcal{M}^* \varkappa) = -dG.
\end{equation}
Учитывая соотношение \eqref{aeq3_12}, имеем $i_{v_T}(\sigma + p_\mathcal{M}^* \varkappa) = i_{v_T} \sigma + p_\mathcal{M}^*i_v \varkappa$.
Сравнивая выражения \eqref{aeq3_15} и \eqref{aeq3_17}, находим $i_{v_T}\sigma = -d\widetilde{G}$, а из условия \eqref{aeq3_18} получаем
$$
p_\mathcal{M}^*i_v \varkappa = -p_\mathcal{M}^*df = -d(f \circ p_\mathcal{M}).
$$
Таким образом, $i_{v_T}(\sigma + p_\mathcal{M}^* \varkappa) = -d(f \circ p_\mathcal{M})$, и в силу \eqref{aeq3_6} получаем выражение \eqref{aeq3_19}. Предложение доказано.\end{proof}

Из предложения \ref{pro134} следует, что поток $\Psi$ на ${\tm}$ локально гамильтонов. Тот же факт вытекает и непосредственно из условия \eqref{aeq3_11}. В самом деле, форма $\varkappa$ замкнута $(d\varkappa = 0)$, так что по общей формуле \eqref{aeq3_16} имеем ${\bf L}_v \varkappa = di_v \varkappa = 0$, т.е. локально условие \eqref{aeq3_18} выполнено для любой группы симметрий. Таким образом, группа симметрий порождает набор локальных интегралов вида \eqref{aeq3_6}, из которых <<склеить>> глобальный интеграл, вообще говоря, нельзя.

Отметим также очевидную связь теоремы \ref{the136} с обобщенной теоремой Нётер \cite[\S~40]{bib05}, в которой глобальная гамильтоновость потока группы симплектических диффеоморфизмов предполагается {\it a priori}.

\begin{remark}\label{rem131}
Интеграл \eqref{aeq3_6} является обобщением интеграла момента в натуральных системах с симметрией \cite{bib05, bib69, bib70,bib71}. Чтобы не возникало терминологической путаницы (в реальных системах определен вектор кинетического момента, не обязательно постоянный), будем, следуя традиции небесной механики и динамики твердого тела, называть интеграл \eqref{aeq3_6} интегралом площадей.
\end{remark}
\begin{remark}\label{rem132} Из точности формы $i_v \varkappa$ не следует точность формы $\varkappa$. Поэтому интеграл площадей может существовать в системах, не допускающих глобального лагранжиана.
\end{remark}
\begin{remark}\label{rem133} Из точности форм $\varkappa$ и $i_v \varkappa$ следует существование $\Psi$-сохранной 1-формы $\lambda$, такой, что $\varkappa = d \lambda$. Таким образом, если система имеет глобальный лагранжиан и обладает интегралом площадей, то в целом ее можно записать с помощью $\Psi$-сохранного лагранжиана. Здесь, следовательно, применима теорема Нётер.\end{remark}

Отметим еще одно, почти очевидное, свойство, связанное с вопросом об отыскании дополнительных интегралов системы \eqref{aeq2_1}. Поскольку поле $X$, задающее ее динамику, в конечном счете гамильтоново, желательно, чтобы подобные интегралы были в инволюции с интегралом \eqref{aeq3_6}.

\begin{propos}\label{pro138}
Функция $K$ на ${\tm}$ находится в инволюции с интегралом площадей тогда и только тогда, когда $K$ сохраняется группой $\Psi$.
\end{propos}
\begin{proof}
По предложению \ref{pro137} функция $G$ -- гамильтониан поля \eqref{aeq3_4}. Поэтому из равенств \eqref{aeq1_1} следует, что $\{K, G\} = dK(v_T) = v_TK$, поэтому условие $\{K, G\} = 0$ равносильно равенству $v_TK = 0$.
\end{proof}
Отметим, что гамильтониан всегда сохраняется своим гамильтоновым полем. Следовательно,
\begin{equation}\label{aeq3_20}
    v_TG=0.
\end{equation}
В частности, $\Psi$-сохранна и входящая в интеграл площадей функция $f$ (что следует непосредственно и из условия \eqref{aeq3_18}):
\begin{equation}\label{aeq3_21}
    vf = 0.
\end{equation}
Эти факты будут использованы в дальнейшем.

\subsection{Пример с локальными интегралами}\label{ssec14}
Рассмотрим задачу о движении материальной точки с единичной массой по цилиндру
\begin{equation}\label{aeq4_1}
    \mathcal{M} = \{x = (\varphi \mathop{\rm mod}\nolimits 2\pi, z)\}
\end{equation}
в поле потенциальных сил, не зависящих от $z$ (так что потенциал $\Pi = \Pi (\varphi)$ -- функция с периодом $2\pi$), и гироскопических сил, определяемых 2-формой
\begin{equation}\label{aeq4_2}
    \varkappa = dz \wedge d\varphi
\end{equation}
Последняя есть, очевидно, форма объема на $\mathcal{M}$.

Несмотря на то что координата $\varphi$ -- многозначная функция точки $\mathcal{M}$, ее дифференциал $d\varphi$ -- корректно определенная замкнутая (но не точная) 1-форма на $\mathcal{M}$. Поэтому форма \eqref{aeq4_2} точна: $\varkappa = d\lambda$, где
\begin{equation}\label{aeq4_3}
    \lambda = zd\varphi.
\end{equation}
В частности, рассматриваемая система имеет глобальную функцию Лагранжа вида \eqref{aeq1_8}
\begin{equation}\label{aeq4_4}
    \ds L=\frac{1}{2}\left(\dot{\varphi}^2 + \dot{z}^2\right) + z\dot{\varphi}-\Pi(\varphi).
\end{equation}
Его энергия \eqref{aeq1_9}
\begin{equation}\label{aeq4_5}
    \ds H=\frac{1}{2}\left(\dot{\varphi}^2 + \dot{z}^2\right) + \Pi(\varphi)
\end{equation}
сохраняется группой преобразований сдвига
\begin{equation}\label{aeq4_6}
    \Psi=\left\{\psi^\tau\right\}, \qquad \psi^\tau(\varphi, z)=(\varphi,z+\tau).
\end{equation}
Соответствующие преобразования ${\tm}=\mathcal{M} \times \bbR^2(\dot{\varphi}, \dot{z})$ имеют вид
\begin{equation}\label{aeq4_7}
    T\varphi^\tau(\varphi, z, \dot{\varphi}, \dot{z})
\end{equation}
и сохраняют форму \eqref{aeq4_2}.

Итак, выполнены все условия определения \ref{def132} и $\Psi$ -- группа симметрий. Ее порождающее поле $v=(0, 1)$, поэтому из \eqref{aeq4_2}
\begin{equation}\label{aeq4_8}
    i_v\varkappa=d\varphi.
\end{equation}
По теореме \ref{the136} глобальный интеграл площадей не существует -- форма \eqref{aeq4_8} не точна. Однако во всякой односвязной подобласти $U \subset \mathcal{M}$ может быть выбрана однозначная функция $f(\varphi, z) = -\varphi$, чтобы выполнялось условие \eqref{aeq3_18}. Соответствующий локальный интеграл вида \eqref{aeq3_6} есть
\begin{equation}\label{aeq4_9}
    G=\dot{z}-\varphi.
\end{equation}
Его существование, конечно, легко увидеть и из уравнений Лагранжа
\begin{equation}\label{aeq4_10}
    \ddot{z} - \dot{\varphi} = 0, \qquad \ddot{\varphi} + \dot{z} + \Pi'(\varphi) = 0.
\end{equation}
В односвязных областях интеграл \eqref{aeq4_9} является циклическим -- замена в формуле \eqref{aeq4_4} слагаемого $z\dot{\varphi}$ на $-\varphi \dot{z}$ не меняет уравнений \eqref{aeq4_10}. Однако в целом на многообразии \eqref{aeq4_1} определить эти уравнения $\Psi$-сохранным лагранжианом (не зависящим от $z$) нельзя -- не существует 1-формы, сохраняемой преобразованиями \eqref{aeq4_6} и лежащей в одном классе когомологий с формой \eqref{aeq4_3}.

Отметим еще интересный факт в связи с двойственностью Лагранжа\,--\,Га\-миль\-то\-на. Произведем преобразование Лежандра:
\begin{equation}\label{aeq4_11}
\begin{array}{c}
    \Phi_L:(\varphi, z,\dot{\varphi},\dot{z})\mapsto(\varphi,z,p_\varphi,p_z),\\
    p_\varphi=\dot{\varphi}+z,\quad p_z=\dot{z}.
\end{array}
\end{equation}
Уравнения \eqref{aeq4_10} перейдут в каноническую систему на $T^*\mathcal{M}$
\begin{equation*}
\begin{array}{ll}
    \dot{z} = p_z, & \dot{p}_z = p_\varphi - z,\\
    \dot{\varphi} = p_\varphi - z, & \dot{p}_\varphi = -\Pi'(\varphi)
\end{array}
\end{equation*}
с гамильтонианом
\begin{equation}\label{aeq4_12}
    \ds H_* = \frac{1}{2} \left(p_{\varphi}^2 + p_z^2\right) - zp_\varphi + \frac{z^2}{2} + \Pi(\varphi).
\end{equation}
Эта функция явно зависит от $z$ и $\varphi$, поэтому <<видимых>> симметрий, т.е. преобразований, кокасательных к диффеоморфизмам конфигурационного пространства $\mathcal{M}$, в терминологии \cite{bib11} у нее нет. Преобразования \eqref{aeq4_7} переходят при отображениях \eqref{aeq4_11} в диффеоморфизмы
\begin{equation*}
    \psi_*^\tau(\varphi, z, p_\varphi, p_z) = (\varphi, z + \tau, p_\varphi + \tau, p_z),
\end{equation*}
сохраняющие гамильтониан \eqref{aeq4_12}. Таким образом, симметрия, <<видимая>> в ла\-г\-ранжевой постановке (функция \eqref{aeq4_5} не зависит от $z$), становится <<скрытой>> в гамильтоновой.

\subsection{Понижение порядка в гироскопических системах\\  с симметрией}\label{ssec15}

Обратимся вновь к ситуации, рассмотренной в  $\S$~\ref{ssec13}, сохраняя введенные обозначения. Пусть
\begin{equation}\label{aeq5_1}
    (\mathcal{M}, s, \Pi, \varkappa),
\end{equation}
-- гироскопическая система, допускающая однопараметрическую группу симметрий $\Psi$, которой отвечает глобальный интеграл площадей $G$, а $X$ -- динамическая система, определенная равенством \eqref{aeq2_9}.

Будем предполагать действие $\Psi$ на $\mathcal{M}$ таким, что существует фактор-мно\-го\-обра\-зие $\widetilde{\mathcal{M}} = \mathcal{M} / \Psi$ и соответствующее главное $\Psi$-расслоение \cite{bib08}
\begin{equation}\label{aeq5_2}
    p:\mathcal{M} \to \widetilde{\mathcal{M}}.
\end{equation}

Покажем, что при этих условиях можно перейти от системы \eqref{aeq5_1} к гироскопической системе
\begin{equation}\label{aeq5_3}
    (\widetilde{\mathcal{M}}, \widetilde{s}, \Pi_g, \varkappa_g),
\end{equation}
в которой метрика $\widetilde{s}$ зависит только от действия группы $\Psi$ на римановом многообразии $(\mathcal{M}, s)$, а приведенный потенциал $\Pi_g$ и форма гироскопических сил $\varkappa_g$ определяются функцией $\Pi$, формой $\varkappa$, значением $g$ интеграла площадей и свойствами расслоения $p$. Локально такой переход будет отвечать понижению порядка по Раусу.

Тем самым будет установлено, что определенный выше класс механических систем с гироскопическими силами замкнут относительно процедуры редукции, чего нельзя сказать о классах натуральных систем и систем с квадратичными лагранжианами.

Согласно общей теории, гамильтоновой системе на фазовом пространстве размерности $2n$, допускающей однопараметрическую группу симметрий, соответствует приведенная система на фазовом пространстве размерности $2n-2$, также гамильтонова в некоторой симплектической структуре. Алгоритм построения приведенных объектов -- фазового пространства, симплектической структуры и гамильтониана -- известен \cite{bib05}. В нашей задаче необходимо конкретизировать все эти объекты для системы \eqref{aeq5_1}.

Отметим некоторые факты, непосредственно вытекающие из сделанных допущений.

В силу наличия главного расслоения \eqref{aeq5_2} группа $\Psi$ действует на $\mathcal{M}$ без неподвижных точек. Следовательно, порождающее поле \eqref{aeq3_3} нигде не обращается в нуль. Поэтому интеграл площадей \eqref{aeq3_6}, существующий по предположению, является всюду регулярной функцией на ${\tm}$. Зафиксируем его значение $g \in \mathbb{R}$. Тогда подмножество
\begin{equation}\label{aeq5_4}
    G_g = G^{-1}(g)\subset{\tm}
\end{equation}
есть гладкое многообразие в фазовом пространстве ${\tm}$ коразмерности 1. Оно, очевидно, инвариантно относительно фазового потока динамической системы $X$, а согласно \eqref{aeq3_20} -- и относительно преобразований группы $\Psi$.

Обозначим через
\begin{equation}\label{aeq5_5}
    \tau_g : G_g \to T\widetilde{\mathcal{M}}
\end{equation}
ограничение на $G_g$ отображения $T_p$, касательного к проекции \eqref{aeq5_2}.

\begin{propos}\label{pro151} Отображение $\tau_g$ есть главное расслоение со структурной группой $\Psi$.
\end{propos}

Доказательству предпошлем некоторые построения.

Пусть $x \in \mathcal{M}$. Обозначим через $V_x$ прямую в касательном пространстве $T_x\mathcal{M}$, натянутую на вектор $v(x)$, а через $N_x$ -- гиперплоскость, ортогональную $V_x$ в метрике $s$. Множество $N = \bigcup \limits_{x \in \mathcal{M}} N_x$ есть гладкое подмногообразие в ${\tm}$ коразмерности 1. Положим
\begin{equation}\label{aeq5_6}
    \tau = \left. Tp \right|_N : N \to T\widetilde{\mathcal{M}}.
\end{equation}
Ввиду инвариантности метрики $s$ совокупность подпространств $N_x$ определяет связность в расслоении \eqref{aeq5_2}. Пусть $\xi$ -- дифференциальная 1-форма данной связности
\begin{equation}\label{aeq5_7}
    \xi_x(v(x)) \equiv 1, \qquad \xi_x(w) = 0 \quad \forall w \in N_x,
\end{equation}
а $\eta = d\xi$ -- ее форма кривизны.

Из структуры интеграла \eqref{aeq3_6} видно, что множество $G_g(x) = G_g \bigcap T_x\mathcal{M}$ -- это гиперплоскость в $T_x\mathcal{M}$, параллельная $N_x$. Ее пересечение с $V_x$ состоит из единственного вектора
\begin{equation}\label{aeq5_8}
    v_g(x)=\frac{g-f(x)}{\|v(x)\|^2}v(x).
\end{equation}
Векторное поле $v_g$ гладкое и сохраняется преобразованиями группы $\Psi$.

Рассмотрим диффеоморфизм
\begin{equation}\label{aeq5_9}
    \gamma_g: N \to G_g,
\end{equation}
определенный формулой $\gamma_g(w) = w + v_g, \quad w \in N$. Очевидно, что для любого $\psi^t \in \Psi$
\begin{equation}\label{aeq5_10}
    T\psi^t \circ \gamma_g = \gamma_g \circ T\psi^t.
\end{equation}
Кроме того, поскольку $Tp$ линейно на касательных пространствах и $Tp(v_g) = 0$,
\begin{equation}\label{aeq5_11}
    \tau_g \circ \gamma_g = \tau.
\end{equation}

Покажем, что отображение \eqref{aeq5_6} является главным $\Psi$-расслоением. Предложение \ref{pro151} будет следовать тогда из свойств \eqref{aeq5_10}, \eqref{aeq5_11} диффеоморфизма \eqref{aeq5_9}.

Пусть $\{U_\alpha,\varphi_\alpha\}$ -- атлас расслоения \eqref{aeq5_2}, т.е. набор открытых множеств $U_\alpha \subset \widetilde{\mathcal{M}}$ и отображений $\varphi_\alpha : p^{-1}(U_\alpha) \to \Psi$, таких, что

1) $\{U_\alpha\}$ -- покрытие $\widetilde{\mathcal{M}}$;

2) все $\varphi_\alpha$ коммутируют с заданным действием $\Psi$ на $\mathcal{M}$ и естественным действием $\Psi$ на самой себе;

3) все отображения $p{\times} \varphi_\alpha : p^{-1}(U_\alpha) \to U_\alpha {\times} \Psi$ являются диффеоморфизмами.

Построим атлас расслоения \eqref{aeq5_6}. Множества $TU_\alpha$ открыты и покрывают $T\widetilde{\mathcal{M}}$. Обозначим $\mathcal{M}_\alpha=p^{-1}(U_\alpha)$. Отображение
\begin{equation*}
    Tp {\times} (\varphi_\alpha \circ p_\mathcal{M}) {\times} \xi : {\tm}_\alpha \to TU_\alpha {\times} \Psi {\times} \mathbb{R}
\end{equation*}
является диффеоморфизмом. Здесь форма связности $\xi$ рассматривается как функция на ${\tm}$. Зафиксировав значение $\xi = 0$, получим согласно \eqref{aeq5_7} диффеоморфизм $N \bigcap {\tm}_\alpha \to TU_\alpha {\times} \Psi$, композиция которого с проекцией $TU_\alpha {\times} \Psi \to \Psi$ дает отображение $\Phi_\alpha : \tau^{-1} (TU_\alpha) \to \Psi$, коммутирующее с действием $\Psi$ на прообразе и образе. Набор $\{TU_\alpha, \Phi_\alpha\}$ и есть искомый атлас расслоения $\tau$.

\begin{lemma}\label{lem152}
Ограничение симплектической формы $\sigma + p_\mathcal{M}^*\varkappa$ на подмногообразие $G_g$ является горизонтальной 2-формой в смысле расслоения \eqref{aeq5_5}.
\end{lemma}

\begin{proof} Напомним, что дифференциальная форма на тотальном пространстве главного расслоения называется горизонтальной, если она обращается в нуль, как только какой-нибудь из ее аргументов есть касательный вектор к слою.

В нашем случае касательное пространство к слою в точке $x \in G_g$ -- это прямая, натянутая на вектор $v_T(x)$. Поэтому необходимо доказать, что для любого $w \in T_x G_q$
\begin{equation}\label{aeq5_12}
    (\sigma + p_\mathcal{M}^*\varkappa)(v_T(x), w) = 0.
\end{equation}

Равенство \eqref{aeq5_12} известно в общей теории как косоортогональность группы симметрий и многообразия постоянного момента \cite{bib05}. Проверим его непосредственно. По формулам \eqref{aeq3_15}, \eqref{aeq3_17}
\begin{equation*}
\begin{array}{rcl}
    \sigma(v_T, w) & = & -d(i_{v_t}\theta)(w) = d(f \circ p_\mathcal{M} - G)(w) = \\
    & = & df(T_{p_\mathcal{M}}(w))-wG.
\end{array}
\end{equation*}

Но, поскольку, вектор $w$ касается уровня функций $G$, то $wG = 0$. Из равенств \eqref{aeq3_12}, \eqref{aeq3_18} получим
\begin{equation*}
\begin{array}{rcl}
    df(T_{p_\mathcal{M}}(w))& = & -\varkappa(T_{p_\mathcal{M}}(v_T), T_{p_\mathcal{M}}(w)) = \\
       &    = & -(p_\mathcal{M}^* \varkappa)(v_T, w),
\end{array}
\end{equation*}
т.е. $\sigma(v_T, w)= -(p_\mathcal{M}^* \varkappa)(v_T, w)$, что равносильно равенству \eqref{aeq5_12}. \end{proof}

{\bf Следствие.} Существует единственная 2-форма $\sigma_g$ на пространстве кокасательного расслоения $T\widetilde{\mathcal{M}}$, удовлетворяющая равенству
\begin{equation}\label{aeq5_13}
    \tau_g^* \sigma_g = \left. (\sigma + p_{\mathcal{M}}^* \varkappa) \right| _{G_g}.
\end{equation}
Доказательство очевидно.

Построим теперь приведенную систему, соответствующую постоянной площадей $g$. Исходное векторное поле $X$ касается подмногообразия $G_g$ и сохраняется преобразованиями группы $\Psi$ (точнее, преобразованиями, касательными к диффеоморфизмам \eqref{aeq3_2}). Поэтому существует единственное векторное поле $X_g$ на $T\widetilde{\mathcal{M}}$, удовлетворяющее равенству
\begin{equation}\label{aeq5_14}
    T\tau_g \circ X = X_g \circ \tau_g.
\end{equation}
Поле $X_g$ и будем называть приведенной динамической системой. Интегральные кривые $X_g$ -- это образы при отображении $\tau_g$ интегральных кривых поля $X$, лежащих на заданном уровне интеграла площадей.

Вспоминая $\Psi$-сохранность функции $H$ (равенство \eqref{aeq3_10}), определим приведенный гамильтониан $H_g$ соотношением
\begin{equation}\label{aeq5_15}
    H_g \circ \tau_g = \left. H \right|_{G_g}.
\end{equation}
Из равенств \eqref{aeq5_13} - \eqref{aeq5_15} совместно с определением \eqref{aeq2_9} следует соотношение
\begin{equation}\label{aeq5_16}
    i_{X_g} \sigma_g = -dH_g.
\end{equation}

Найдем явное выражение для функции $H_g$. Для этого превратим $\widetilde{\mathcal{M}}$ в риманово многообразие.
Пусть $z \in \widetilde{\mathcal{M}}$, $\widetilde{w}_1, \widetilde{w}_2 \in T_z\widetilde{\mathcal{M}}$. Выберем произвольно точку $x \in p^{-1}(z) \subset \mathcal{M}$, и пусть векторы $w_1, w_2 \in N_x$ таковы, что $Tp(w_i) = \widetilde{w}_i$ $(i = 1, 2)$. Положим
\begin{equation}\label{aeq5_17}
    \widetilde{s}_x(\widetilde{w}_1, \widetilde{w}_2) = s_x(w_1, w_2).
\end{equation}
В силу инвариантности метрики $s$ и подпространств $N_x$ относительно действия группы симметрий это равенство корректно определяет скалярное произведение в слоях $T\widetilde{\mathcal{M}}$.

Возьмем произвольный вектор $\widetilde{w} \in T_z\widetilde{\mathcal{M}}$. Найдутся точка $x \in \mathcal{M}$, $p(x)=z$ и вектор $w \in G_g(x)$, такие, что $\tau_g(w) = \widetilde{w}.$ В этом случае
\begin{equation*}
    w = w_0 + v_g(x), \quad w_0 \in N_x, \quad w_0 \bot v_g(x),
\end{equation*}
и, следовательно,
\begin{equation*}
\begin{array}{rcl}
    H(w) & = & \ds {\frac{1}{2} \left( \|w_0\|^2 + \|v_g(x)\|^2 \right) + \Pi(x) =} \\
    &  = & \ds {\frac{1}{2}s_x(w_0, w_0)+\Pi(x)+\frac{[g-f(x)]^2}{\|v(x)\|^2}.}
\end{array}
\end{equation*}
Здесь использовано равенство \eqref{aeq5_8}. Выражение в правой части, не зависящее от $w_0$, вследствие \eqref{aeq3_9}, \eqref{aeq3_21} является функцией на многообразии $\mathcal{M}$, инвариантной относительно действия $\Psi$. Она индуцирует функцию $\Pi_g$ на $\widetilde{\mathcal{M}}$, называемую эффективным потенциалом:
\begin{equation}\label{aeq5_18}
    \Pi_g(p(x))=\Pi(x)+\frac{[g-f(x)]^2}{\|v(x)\|^2}.
\end{equation}
Сохраняя для нормы вектора в метрике $\widetilde{s}$ стандартное обозначение, из соотношений \eqref{aeq5_15}, \eqref{aeq5_17}, \eqref{aeq5_18} получаем
\begin{equation}\label{aeq5_19}
    \ds H_g(\widetilde{w}) = \frac{1}{2}\|\widetilde{w}\|^2 + \Pi_g \left( p_\mathcal{M} (\widetilde{w}) \right), \quad \widetilde{w} \in T\widetilde{\mathcal{M}}.
\end{equation}

Установим теперь связь формы $\sigma_g$ на $T\widetilde{\mathcal{M}}$ с элементами исходной гироскопической системы.

\begin{lemma}\label{lem153}
Пусть $\theta$ -- 1-форма на ${\tm}$, определенная согласно \eqref{aeq2_4}. Тогда
\begin{equation}\label{aeq5_20}
    \gamma_g^* \left( \left. \theta \right|_{G_g}\right) = \left. \theta \right|_N + (g - f \circ p_\mathcal{M}) p_\mathcal{M}^* \xi.
\end{equation}
\end{lemma}

\begin{proof}
Из соотношений \eqref{aeq5_7} найдем выражение формы связности на любом векторе $w \in T_x\mathcal{M}$:
\begin{equation}\label{aeq5_21}
    \xi_x(w) = \frac{\langle v(x), w\rangle}{\|v(x)\|^2}.
\end{equation}

Возьмем произвольно $x \in \mathcal{M}, w \in N_x, Y \in T_wN$. По определению имеем
\begin{equation}\label{aeq5_22}
    \theta_w(Y) = \langle w, T_{p_\mathcal{M}}(Y) \rangle.
\end{equation}

С другой стороны,
\begin{equation*}
    (\gamma_g^* \theta)_w (Y) = \theta_{\gamma_g (w)}(T\gamma_g(Y)) = \langle w + v_g(x), T(p_\mathcal{M} \circ \gamma_g)(Y)\rangle.
\end{equation*}
Но по построению диффеоморфизма \eqref{aeq5_9} $p_\mathcal{M} \circ \gamma_g = p_\mathcal{M}$. Следовательно, с учетом формул \eqref{aeq5_8}, \eqref{aeq5_21}, \eqref{aeq5_22}
\begin{equation*}
\begin{array}{rcl}
    (\gamma_g^* \theta)_w (Y) &=& \langle w, Tp_\mathcal{M}(Y) \rangle + \langle v_g(x), Tp_\mathcal{M}(Y) \rangle = \\
    &=& \theta_w(Y) + [g-f(x)](p_\mathcal{M}^* \xi)(Y),
\end{array}
\end{equation*}
что и доказывает наше утверждение.\end{proof}

Применим антиотображение $\gamma_g^*$ к обеим частям равенства \eqref{aeq5_13}. По формуле \eqref{aeq5_11} имеем $\gamma_g^* \circ \tau_g^* = \tau^*$. Кроме того, $\gamma_g^* \circ p_\mathcal{M}^* \varkappa = (p_\mathcal{M} \circ \gamma_g)^* \varkappa = p_\mathcal{M}^* \varkappa$. Принимая во внимание \eqref{aeq2_6}, \eqref{aeq5_20} и перестановочность внешнего дифференциала с антиотображениями дифференциальных форм, приходим к соотношению
\begin{equation}\label{aeq5_23}
    \tau^* \sigma_g = \sigma + p_\mathcal{M}^*[(g-f)\eta + \varkappa -df\wedge \xi]
\end{equation}
в точках многообразия $N$.

Пусть $\widetilde{\sigma}$ -- лагранжева форма на $T\widetilde{\mathcal{M}}$, отвечающая приведенной метрике (см. предложение \ref{pro121}). Непосредственно проверяется, что $\left. \sigma \right|_N = \tau^*\widetilde{\sigma}$.

То, что форма кривизны $\eta$ горизонтальна в смысле расслоения \eqref{aeq5_2}, -- известный факт из теории связностей \cite{bib08}. Поэтому она индуцирована некоторой 2-формой $\widetilde{\eta}$ на $\widetilde{\mathcal{M}}: \eta = p^* \widetilde{\eta}$ (иногда формой кривизны называют именно профакторизованную форму $\widetilde{\eta}$ \cite{bib53}). Горизонтальность формы $\varkappa - df \wedge \xi$ следует из предположения \eqref{aeq3_18}, равенства \eqref{aeq3_21} и определения формы связности $\xi$. Следовательно, существует форма $\widetilde{\varkappa}$ на $\widetilde{\mathcal{M}}$, такая, что $p^*\widetilde{\varkappa} = \varkappa - df \wedge \xi$. Функция $f$ сохраняется группой $\Psi$, поэтому функция $\widetilde{f} = f / \Psi$ корректно определена на $\widetilde{\mathcal{M}}$. Замечая, наконец, что в точках многообразия $N$ выполнено равенство $p \circ p_\mathcal{M} = p_{\widetilde{\mathcal{M}}} \circ \tau$, из \eqref{aeq5_23} получим
\begin{equation}\label{aeq5_24}
    \sigma_g = \widetilde{\sigma} + p_{\widetilde{\mathcal{M}}}^*[g - \widetilde{f}\widetilde{\eta} + \widetilde{\varkappa}].
\end{equation}

Сравнивая теперь соотношения \eqref{aeq5_16}, \eqref{aeq5_19}, \eqref{aeq5_24} с общим определением динамики в гироскопической системе, приходим к следующему утверждению.

\begin{theorem}\label{the154}
При фиксированной постоянной интеграла площадей $g$ проекции интегральных кривых на приведенное фазовое пространство являются решениями уравнения второго порядка, соответствующего гироскопической системе \eqref{aeq5_3} с эффективным потенциалом \eqref{aeq5_18} и формой гироскопических сил $(g - \widetilde{f})\widetilde{\eta} + \widetilde{\varkappa}$.
\end{theorem}

Отметим случай, когда исходная система \eqref{aeq5_1} натуральна. При этом $\varkappa \equiv 0, f \equiv 0$, и форма гироскопических сил в приведенной системе есть произведение постоянной интеграла площадей на форму кривизны связности, индуцированной инвариантной метрикой в расслоении конфигурационного пространства на орбиты группы симметрий. При $g = 0$ приведенная система натуральна. При $g \neq 0$ она обладает глобальной функцией Лагранжа лишь в случае, когда форма $\widetilde{\varkappa}$ точна, т.е. порождает нулевой элемент в группе $H^2(\widetilde{\mathcal{M}}, \mathbb{R})$ когомологий де Рама. Пусть $\Psi = S^1$. Классы изоморфных $S^1$-расслоений над $\widetilde{\mathcal{M}}$ находятся во взаимно однозначном соответствии с элементами группы $H^2(\widetilde{\mathcal{M}}, \mathbb{Z})$. Интерпретация любого целочисленного коцикла как вещественного порождает гомоморфизм $H^2(\widetilde{\mathcal{M}}, \mathbb{Z}) \to H^2(\widetilde{\mathcal{M}}, \mathbb{R})$, причем характеристический класс расслоения переходит в класс когомологий форм кривизны допустимых связностей в этом расслоении (подробности см. в \cite{bib53,bib60}). Таким образом, если группа $H^2(\widetilde{\mathcal{M}}, \mathbb{Z})$ не имеет кручений, то приведенная система при отличной от нуля постоянной площадей обладает глобальной функцией Лагранжа тогда и только тогда, когда $S^1$-расслоение \eqref{aeq5_2} тривиально. С примерами нетривиальных расслоений мы встретимся далее при изучении задач динамики твердого тела.

\begin{remark}\label{rem151}
Переход к приведенной системе можно осуществить по другому. Поскольку $\Psi$ действует на ${\tm}$ без неподвижных точек, перейдем к фактор-многообразию
\begin{equation}\label{aeq5_25}
    \mathfrak{M} = {\tm}/\Psi.
\end{equation}
\end{remark}
Это пространство векторного расслоения над $\widetilde{\mathcal{M}}$, слои которого изоморфны $T_x\mathcal{M}$. Очевидно, что поле $X$ и его первые интегралы $H$ и $G$ переносятся на $\mathfrak{M}$ с помощью отображения факторизации. Нетрудно показать, что функция $G/\Psi : \mathfrak{M} \to \mathbb{R}$ остается всюду регулярной, поэтому $\{G/\Psi = g\}$ -- гладкое интегральное подмногообразие в $\mathfrak{M}$ коразмерности 1. Оно, очевидно, диффеоморфно $T\widetilde{\mathcal{M}}$.  Соответствующий диффеоморфизм осуществляется субмерсией
\begin{equation*}
    Tp/\Psi:\mathfrak{M} \to T\widetilde{\mathcal{M}}.
\end{equation*}
Действительно, для любого $z \in T\widetilde{\mathcal{M}}$ множество $(Tp)^{-1}(z) \bigcap \{G = g\}$ есть единственная орбита группы $\Psi$ (слой расслоения $\tau_g$ над точкой $z$). Следовательно, пересечение
\begin{equation*}
    (Tp / \Psi)^{-1}(z) \bigcap \{G / \Psi = g\}
\end{equation*}
состоит из единственной точки.

Таким образом, с точностью до диффеоморфизма следующие пути преобразования \mbox{$\Psi$-со}\-хран\-ных объектов на ${\tm}$ эквивалентны: ограничение на $\{G = g\}$ с последующей факторизацией по $\Psi$ и факторизация по $\Psi$ с последующим ограничением на $\{G/\Psi = g\}$.

Многообразие \eqref{aeq5_25} играет основную роль в динамике твердого тела, так как служит фазовым пространством классических уравнений Эйлера\,--\,Пуассона.

\subsection{Комментарий к главе 1}\label{ssec16}

Широкое проникновение в механику идей и методов дифференциальной геометрии и топологии, восходящее, по сути, еще к А.Пуанкаре \cite{bib38} и Дж.Биркгофу \cite{bib07} и в значительной мере стимулированное фундаментальной работой С.Смейла \cite{bib41}, потребовало свободного от неинвариантности локальных координат описания основного объекта теории -- механической системы.

Значительные результаты получены на основе так называемого канонического формализма -- инвариантного описания систем Гамильтона (его изложение имеется, например, в \cite{bib05,bib65}). Сюда следует отнести вопросы интегрируемости гамильтоновых систем \cite{bib26,bib27,bib28}, КАМ-теорию \cite{bib04} и многие другие.

Однако гамильтоновы системы, обладая общей природой, не учитывают многих специфических свойств задач механики. В этом смысле <<ближе к действительности>> оказывается подход Лагранжа. Существующие весьма разнообразные по форме его инвариантные изложения (см., например, \cite{bib10, bib48,bib51,bib65}), по существу, конкретизируют понятие гамильтоновой системы для случая, когда переменные имеют смысл координат и скоростей.

Предложенное в этой главе описание часто встречающегося в механике объекта -- гироскопической системы -- представляет собой в определенном смысле конгломерат двух упомянутых подходов. Иную формализацию гироскопических систем можно получить, используя понятие интегрального инварианта и обобщая некоторые идеи \cite{bib66} по поводу задачи о движении заряженной частицы в магнитно-электрическом поле.

Понятие симметрии в гироскопической системе естественным образом обобщает симметрию натуральных систем \cite{bib41}. С ним тесно связан вопрос о факторизации соответствующего векторного поля. Хорошо известно, что если гамильтониан канонической системы уравнений не зависит от $m$ обобщенных координат, то соответствующие им импульсы являются первыми интегралами и, зафиксировав их значения, можно понизить порядок системы на $2m$ единиц. В лагранжевой постановке такой ситуации отвечает система с циклическими координатами, и процедура редукции называется <<понижением порядка по Раусу>> или <<игнорированием циклических координат>>. Эти результаты используют координатное представление и, следовательно, носят локальный характер. Глобальная теория понижения порядка в гамильтоновых системах с группой симметрий разработана Дж.Марсденом и А. Вейнстейном \cite{bib69}. Ее упрощенный вариант для коммутативных групп приведен в \cite{bib70}. Доступное изложение вопросов, связанных с канонической редукцией, имеется в \cite{bib05}.

Глобальный аналог понижения порядка по Раусу дан в работе \cite{bib57}, где, по-видимому, впервые отмечен тот факт, что приведенная система не всегда обладает глобальной функцией Лагранжа. В \cite{bib60} установлена связь этого явления со свойствами расслоения конфигурационного пространства на орбиты группы симметрий. Основные результаты работ \cite{bib32,bib57,bib58,bib59} нашли отражение в настоящей главе. Несколько иной формализм понижения порядка в натуральных системах с симметрией, близкой к методам римановой геометрии, изложен в \cite{bib48}.

Несуществование глобальной функции Лагранжа у некоторых механических систем ведет к тому, что их траектории становятся экстремалями многозначных функционалов. Для таких систем С.П.\,Новиковым \cite{bib34,bib35} разработано обобщение теории Люстерника~-- Шнирельмана~-- Морса.

\clearpage

\section{Формализация задачи о движении твердого тела\\  под действием потенциальных и гироскопических сил}\label{sec2}

Твердое тело -- это материальная область в трехмерном евклидовом пространстве, движущаяся относительно наблюдателя таким образом, что расстояние между любыми двумя ее точками (частицами) остается постоянным. Если вдобавок некоторая частица занимает в пространстве неизменное положение, то говорят, что тело имеет неподвижную точку.

В задачах динамики наряду с неподвижными (относительно наблюдателя) системами отсчета широко используются системы, жестко связанные с объектом изучения. Уравнения вращательного движения в последних записываются несравненно проще. Однако конечной целью полагается все же исследование движения относительно наблюдателя, вследствие чего отказ от неподвижных систем отсчета невозможен.

Таким образом, каждой точке реального пространства необходимо сопоставить два набора координат (в подвижной и неподвижной системах отсчета) и каждому вектору -- два набора компонент (в подвижном и неподвижном базисах). В связи с этим при формализации задачи возможны два подхода. Первый \cite{bib49,bib50} состоит в том, чтобы каждому реальному вектору сопоставить два вектора одного и того же арифметического пространства. Тогда движение тела есть семейство вращений этого пространства, переводящих некоторое раз и навсегда фиксированное положение тела в его положение в заданный момент времени. Такая точка зрения естественна, если иметь в виду обобщение на механические системы, конфигурационным пространством которых служит некоторая группа Ли. Левые и правые сдвиги в единицу группы сопоставляют одному касательному вектору два, вообще говоря, различных элемента алгебры Ли, что применительно к твердому телу дает угловую скорость <<в теле>> и <<в пространстве>>.

Подобная <<раздвоенность>> векторов, однако, не всегда удобна при определении ориентации тела в пространстве. Трудными для восприятия становятся такие важные понятия, как абсолютное и относительное дифференцирование, распределение скоростей в теле и т.п. Поэтому далее мы предпочтем иметь дело с двумя различными арифметическим пространствами, координаты в которых -- проекции векторов на неподвижные оси и на оси, связанные с телом, и их вложениями друг в друга.

Указанные подходы в конечном счете равносильны. Выбор того или иного диктуется лишь целью исследования.

\subsection{Конфигурация, движение, классические формулы}\label{ssec21}

Рассмотрим два трехмерных арифметических пространства. Первое обозначим $\bbRRR$ и назовем неподвижным, второе -- $\bbRR$ -- подвижным.
Евклидову структуру на $\bbRRR$ и $\bbRR$ определяет скалярное произведение $( \cdot )$ векторов, а структуру алгебры Ли -- векторное $({\times})$ произведение.
Элементы введенных пространств обозначаем жирными строчными буквами, снабжая их в случае $\bbRRR$ индексом <<ноль>>. Этот индекс будет иметь и дополнительный смысл.
Стандартные координаты точки $\mbs{\xi} \in \bbRR \left(\mbs{\xi}^0 \in \bbRRR\right) $ будут $\xi_i$ (соответственно $\xi_i^0$), где ${i = 1, 2, 3}$. Элементы с нулевыми координатами обозначаем, как обычно, нулем.

Пусть
\begin{equation}\label{beq1_1}
    \mbs{e}, \mbs{e'}, \mbs{e''} \in \bbRR
\end{equation}
и
\begin{equation}\label{beq1_2}
    \mbs{\nu}_0, \mbs{\nu}_0', \mbs{\nu}_0'' \in \bbRRR
\end{equation}
-- канонические единичные базисы.

\begin{defin}\label{def211}
Конфигурацией твердого тела назовем отображение
\begin{equation}\label{beq1_3}
    Q: \bbRR \to \bbRRR,
\end{equation}
сохраняющее евклидову структуру (скалярное произведение).
\end{defin}

В частности, $Q$ есть изоморфизм евклидовых пространств, и $Q(0) = 0$. Это значит, что при любой конфигурации начало координат подвижного пространства совмещается с одной и той же точкой (а именно началом координат) неподвижного пространства. Таким образом, мы ограничились задачей о движении твердого тела с неподвижной точкой.
Условимся теперь векторы, соответствующие друг другу при изоморфизме \eqref{beq1_3}, обозначать одной и той же буквой, так что
\begin{equation}\label{beq1_4}
\begin{array}{ll}
    \forall \mbs{\xi} \in \bbRR : & Q \mbs{\xi} = \mbs{\xi}^0 \in \bbRRR,\\
    \forall \mbs{\xi}^0 \in \bbRRR : & Q^{-1} \mbs{\xi}^0 = \mbs{\xi} \in \bbRR.
\end{array}
\end{equation}
Кроме того, отождествим изоморфизм $Q$ с его матрицей
\begin{equation}\label{beq1_5}
    Q = \left(
    \begin{array}{ccc}
     \nu_1 & \nu_2 & \nu_3 \\
    \nu_1' & \nu_2' & \nu_3' \\
    \nu_1'' & \nu_2'' & \nu_3''
    \end{array}
    \right)
\end{equation}
в базисах \eqref{beq1_1}, \eqref{beq1_2}.

Без ограничения общности можно считать изоморфизм \eqref{beq1_3} сохраняющим ориентацию, поскольку матрица \eqref{beq1_5} ортогональна и ее определитель равен~1. В частности, выполняется соотношение
\begin{equation}\label{beq1_6}
    QQ^T = Q^TQ = \Id.
\end{equation}
Таким образом, конфигурационное пространство твердого тела с неподвижной точкой является группой $SO(3)$.

Поясним обозначение компонент \eqref{beq1_5}. Согласно общепринятому, столбцы
\begin{equation}\label{beq1_7}
    \mbs{\nu}, \mbs{\nu}', \mbs{\nu}''
\end{equation}
матрицы $Q^T = Q^{-1}$ должны быть разложениями по базису \eqref{beq1_1} образов векторов \eqref{beq1_2} в $\bbRR$, что соответствует соотношению \eqref{beq1_4}.

Для удобства всюду в этой главе обозначаем группу $SO(3)$ через $\mathcal{M}$. Многообразие $\mathcal{M}$ вложено в $\mathbb{R}^9(\nu_1, \ldots, \nu_3'')$  и определяется там соотношениями
\begin{equation}\label{beq1_8}
    \nu_i \nu_j + \nu_i' \nu_j' + \nu_i'' \nu_j'' = \left\{
    \begin{array}{c}
    1, \qquad i = j \\
    0, \qquad i \neq j
    \end{array}
    \right. , \qquad i,j = 1, 2, 3,
\end{equation}
либо им равносильными
\begin{equation}\label{beq1_9}
    \nu_1^{(i)} \nu_1^{(j)} + \nu_2^{(i)} \nu_2^{(j)} + \nu_3^{(i)} \nu_3^{(j)} = \left\{
    \begin{array}{c}
    1, \qquad i = j \\
    0, \qquad i \neq j
    \end{array}
    \right. , \qquad i,j = 0, 1, 2.
\end{equation}

\begin{defin}\label{def212}
Движением твердого тела назовем семейство конфигураций, гладко зависящих от параметра $t \in \bbR$, или, что то же самое, дифференцируемую кривую $\bbR \to \mathcal{M}$ $(t \mapsto Q(t))$.
\end{defin}

Фазовое пространство рассматриваемой задачи -- касательное расслоение $\tm$ многообразия $\mathcal{M}$. Изучим его простейшие свойства.

Пусть $\Omega_Q$ -- касательный вектор к многообразию $\mathcal{M}$ в точке $Q$. Тогда существует кривая $Q(t)$ на $\mathcal{M}$, такая что
\begin{equation}\label{beq1_10}
    Q(t_0)=Q, \qquad \ds \left. \frac{dQ(t)}{dt} \right|_{t=t_0} = \Omega_Q.
\end{equation}

Ясно, что $\Omega_Q$ -- некоторая $3 {\times} 3$-матрица. Если тело движется по закону $Q(t)$, будем называть матрицу \eqref{beq1_10} мгновенным вращением тела в момент $t_0$.
Дифференцируя по $t$ тождество \eqref{beq1_6}, найдем, что $Q^{-1} \Omega_Q$ и $\Omega_QQ^{-1}$ принадлежат $\Ass(3)$ -- трехмерному линейному пространству кососимметрических матриц. Таким образом,
\begin{equation}\label{beq1_11}
    T_Q\mathcal{M} = Q \cdot \Ass(3) = \Ass(3) \cdot Q.
\end{equation}
В частности, алгебра Ли группы $\mathcal{M}$ (касательное пространство в единице) есть $\Ass(3)$, и коммутатор в ней -- обычный коммутатор матриц $[A,B] = AB - BA$.

Определим изоморфизмы
\begin{equation}\label{beq1_12}
    \mbs{\iota}, \mbs{\iota}_0 : \Ass(3) \to \bbRR, \bbRRR,
\end{equation}
действующие по закону
\begin{equation*}
    \left(
    \begin{array}{ccc}
    0 & -c & b \\
    c & 0 & -a \\
    -b & a & 0
    \end{array}
    \right) \mapsto
    \left(
    \begin{array}{c}
    a \\
    b \\
    c
    \end{array}
    \right).
\end{equation*}
Для вектора $\mbs{\xi} \in \bbRR$ матрицу $\mbs{\iota}^{-1}(\mbs{\xi})$ обозначим для краткости через $/\mbs{\xi}/$. Отметим известные соотношения \cite{bib49}.

Пусть $\mbs{\xi}, \mbs{\eta} \in \bbRR$, $Q \in \mathcal{M}$. Тогда
\begin{equation}\label{beq1_13}
\begin{array}{l}
    [/\mbs{\xi}/, /\mbs{\eta}/] = /\mbs{\xi} {\times} \mbs{\eta}/, \qquad
    /\mbs{\xi}/\mbs{\eta} = \mbs{\xi} {\times} \mbs{\eta}, \qquad
    /Q\mbs{\xi}/ = Q /\mbs{\xi}/ Q^{-1}.
\end{array}
\end{equation}
То же самое справедливо, конечно, и для $\bbRRR$.

Теперь с помощью изоморфизмов \eqref{beq1_12} устанавливаем две различные тривиализации расслоения $\tm$:
\begin{eqnarray}
& &    {\rm triv}_0: \tm \to \mathcal{M} {\times} \bbRRR, \qquad {\rm triv}_0(\Omega_Q) = (Q, \mbs{\iota}_0(\Omega_QQ^{-1})),\label{beq1_14}\\
& &    {\rm triv}: \tm \to \mathcal{M} {\times} \bbRR, \qquad \; {\rm triv}(\Omega_Q) = (Q, \mbs{\iota}(Q^{-1}\Omega_Q)). \label{beq1_15}
\end{eqnarray}
Корректность этих определений следует из \eqref{beq1_11}.

Векторы
\begin{equation}\label{beq1_16}
    \mbs{\omega}^0 = \mbs{\iota}_0(\Omega_QQ^{-1})
\end{equation}
и
\begin{equation}\label{beq1_17}
    \mbs{\omega} = \mbs{\iota}(Q^{-1}\Omega_Q)
\end{equation}
назовем соответственно внешней и внутренней угловой скоростью мгновенного вращения $\Omega_Q$.

Из соотношений \eqref{beq1_13}, \eqref{beq1_16}, \eqref{beq1_17} находим
\begin{equation}\label{beq1_18}
    Q\mbs{\omega} = \mbs{\omega}^0,
\end{equation}
т.е. обозначения согласованы с \eqref{beq1_4}. Физически равенство \eqref{beq1_18} означает, что векторы $\mbs{\omega}$ и $\mbs{\omega}^0$ соответствуют одному и тому же реальному вектору, фигурирующему в \eqref{beq1_10}, -- мгновенной угловой скорости тела в момент $t_0$. Подтвердим сказанное выводом привычных соотношений.

Пусть некоторая точка движется в пространстве $\bbRRR$ по закону $\mbs{\xi}^0 = \mbs{\xi}^0(t)$. Вектор
\begin{equation}\label{beq1_19}
    \dot{\mbs{\xi}}^0(t) = \ds \frac{d{\mbs{\xi}}^0(t)}{dt}
\end{equation}
назовем абсолютной скоростью точки (абсолютной производной вектора $\mbs{\xi}^0(t)$).

Если при этом тело движется по закону $Q(t)$, то вектор
\begin{equation}\label{beq1_20}
    \astup{\xi}^0(t) = \ds Q(t) \frac{d}{dt}(Q^{-1}(t) \mbs{\xi}^0(t))
\end{equation}
будем называть относительной скоростью точки (относительной производной вектора $\mbs{\xi}^0(t)$).

Используя соотношения \eqref{beq1_10}, \eqref{beq1_13}, \eqref{beq1_16}, имеем цепочку равенств
\begin{equation*}
\begin{array}{c}
    \astup{\xi}^0 = \ds Q(t) \frac{d}{dt}(Q^T \mbs{\xi}^0) = Q\Omega_Q^T\mbs{\xi}^0 + \dot{\mbs{\xi}}^0 = (\Omega_QQ^{-1})^T\mbs{\xi}^0 + \dot{\mbs{\xi}}^0 = \\
    = -\Omega_QQ^{-1}\mbs{\xi}^0 + \dot{\mbs{\xi}}^0 = -/\mbs{\omega}/\mbs{\xi}^0 + \dot{\mbs{\xi}}^0 = \dot{\mbs{\xi}}^0 - \mbs{\omega}^0 {\times} \mbs{\xi}^0,
\end{array}
\end{equation*}
откуда
\begin{equation}\label{beq1_21}
    \dot{\mbs{\xi}}^0 =\astup{\xi}^0 + \mbs{\omega}^0 {\times} \mbs{\xi}^0.
\end{equation}

Опишем теперь движение точки радиус-вектором $\mbs{\xi}(t)$ в подвижном пространстве. Тогда естественно ввести абсолютную и относительную скорости точки (абсолютную и относительную производные вектора $\mbs{\xi}(t)$), полагая соответственно
\begin{equation}\label{beq1_22}
    \dot{\mbs{\xi}}(t) = \ds Q^{-1}(t) \frac{d}{dt}(Q(t)\mbs{\xi}(t)), \qquad \astup{\xi} = \ds \frac{d\mbs{\xi}(t)}{dt}.
\end{equation}
Аналогично \eqref{beq1_21} имеем
\begin{equation}\label{beq1_23}
    \dot{\mbs{\xi}}(t) = \astup{\xi}(t) + \mbs{\omega} {\times} \mbs{\xi}.
\end{equation}

Сравнивая \eqref{beq1_21}, \eqref{beq1_23} с известными соотношениями \cite{bib47,bib61}, убеждаемся, что $\mbs{\omega}$ и $\mbs{\omega}^0$ суть векторы, изображающие реальный вектор угловой скорости.

Отметим некоторые следствия формул \eqref{beq1_19} - \eqref{beq1_23}.

\begin{propos}\label{pro213}
$1$. Производные векторов $\mbs{\xi}(t)$ и $\mbs{\xi}^0(t)$ связаны соотношениями
\begin{equation}\label{beq1_24}
    \dot{\mbs{\xi}}^0(t) = Q(t)\dot{\mbs{\xi}}(t), \qquad \astup{\xi}^0(t) = Q(t)\astup{\xi}(t).
\end{equation}

$2$. Если вектор $\mbs{\xi}$ коллинеарен $\mbs{\omega}$ {\rm (}$\mbs{\xi}^0$ коллинеарен $\mbs{\omega}^0${\rm )}, то
\begin{equation*}
    \dot{\mbs{\xi}} = \astup{\xi}, \qquad \dot{\mbs{\xi}}^0 = \astup{\xi}^0.
\end{equation*}

$3$. Если вектор $\mbs{\xi}$ движется вместе с телом $(\astup{\xi} = 0)$, то
\begin{equation*}
    \dot{\mbs{\xi}}(t) = \mbs{\omega} {\times} \mbs{\xi}, \qquad \dot{\mbs{\xi}}^0(t) = \mbs{\omega}^0 {\times} \mbs{\xi}^0
\end{equation*}
{\rm (}закон распределения скоростей в твердом теле{\rm )}.

$4$. Если вектор $\mbs{\xi}^0$ неподвижен в пространстве $( \dot{\mbs{\xi}}^0 = 0 )$, то
\begin{equation}\label{beq1_25}
    \astup{\xi}^0 = \mbs{\xi}^0 {\times} \mbs{\omega}^0, \qquad \astup{\xi} = \mbs{\xi} {\times} \mbs{\omega}
\end{equation}
{\rm (}уравнения Пуассона{\rm )}.
\end{propos}

\begin{proof}
Утверждение 1 следует непосредственно из соглашения \eqref{beq1_4} и определений \eqref{beq1_19}, \eqref{beq1_20}, \eqref{beq1_22}. Остальные вытекают из \eqref{beq1_21}, \eqref{beq1_23}, \eqref{beq1_24} и того факта, что ортогональные отображения сохраняют векторные произведения (отображение \eqref{beq1_3} -- изоморфизм алгебр Ли). \end{proof}

Отметим, что свойство \eqref{beq1_24} совместно с \eqref{beq1_4} означает, что операции абсолютного и относительного дифференцирования коммутируют с <<приписыванием нолика>>, т.е. с переходом из $\bbRR$ в $\bbRRR$.

\begin{remark}\label{rem211}
Относительные и абсолютные производные скалярных функций совпадают между собой. Пусть, например, $\mbs{\xi}(t), \mbs{\eta}(t) \in \bbR$. Тогда
\begin{equation*}
    (\mbs{\xi} \cdot \mbs{\eta})^\mbs{\cdot} = (\astup{\xi} + \mbs{\omega} {\times} \mbs{\xi}) \cdot \mbs{\eta} + \mbs{\xi} \cdot (\astup{\eta} + \mbs{\omega} {\times} \mbs{\eta}) = \astup{\xi} \cdot \mbs{\eta} + \mbs{\xi} \cdot \astup{\eta} = (\mbs{\xi} \cdot \mbs{\eta})^*,
\end{equation*}
так как выражение $\mbs{\xi} \cdot (\mbs{\eta} {\times} \mbs{\zeta})$ есть кососимметрическая форма. Таким образом, если $\mbs{\xi}(t) \in \bbRR$, то
\begin{equation*}
    \astup{\xi}(t) = \left( \dot{\xi}_1(t), \dot{\xi}_2(t), \dot{\xi}_3(t)\right)^T.
\end{equation*}
\end{remark}

\subsection{Некоторые структуры на группе вращений}\label{ssec22}

Введем на многообразии $\mathcal{M} = {SO}(3)$ векторные поля
\begin{equation}\label{beq2_1}
    \Omega_i:\mathcal{M} \to \tm,
\end{equation}
индуцированные каноническим базисом $\bbRR$ при отображении \eqref{beq1_15}
\begin{equation}\label{beq2_2}
    \Omega_1(Q) = {\rm triv}^{-1}(Q, {\bf e}), \quad \Omega_2(Q) = {\rm triv}^{-1}(Q, {\bf e}'), \quad \Omega_3(Q) = {\rm triv}^{-1}(Q, {\bf e}'').
\end{equation}
В явном виде
\begin{equation*}
    \Omega_1(Q)=Q\left/ \left( \begin{array}{c} 1 \\ 0 \\ 0 \end{array} \right) \right/, \quad \Omega_2(Q)=Q\left/ \left( \begin{array}{c} 0 \\ 1 \\ 0 \end{array} \right) \right/, \quad \Omega_3(Q)=Q\left/ \left( \begin{array}{c} 0 \\ 0 \\ 1 \end{array} \right) \right/,
\end{equation*}
поэтому поля \eqref{beq2_1} левоинвариантны. Их выражение через базисные поля в $\bbR^9$ таково:
\begin{equation}\label{beq2_3}
\begin{array}{c}
    \ds \Omega_1(Q) = \nu_3 \frac{\partial}{\partial \nu_2} - \nu_2 \frac{\partial}{\partial \nu_3} + \nu_3' \frac{\partial}{\partial \nu_2'} - \nu_2' \frac{\partial}{\partial \nu_3'} + \nu_3'' \frac{\partial}{\partial \nu_2''} - \nu_2'' \frac{\partial}{\partial \nu_3''}, \\
    \ds \Omega_2(Q) = \nu_1 \frac{\partial}{\partial \nu_3} - \nu_3 \frac{\partial}{\partial \nu_1} + \nu_1' \frac{\partial}{\partial \nu_3'} - \nu_3' \frac{\partial}{\partial \nu_1'} + \nu_1'' \frac{\partial}{\partial \nu_3''} - \nu_3'' \frac{\partial}{\partial \nu_1''}, \\
    \ds \Omega_3(Q) = \nu_2 \frac{\partial}{\partial \nu_1} - \nu_1 \frac{\partial}{\partial \nu_2} + \nu_2' \frac{\partial}{\partial \nu_1'} - \nu_1' \frac{\partial}{\partial \nu_2'} + \nu_2'' \frac{\partial}{\partial \nu_1''} - \nu_1'' \frac{\partial}{\partial \nu_2''}.
\end{array}
\end{equation}
Соответствующие однопараметрические группы вращают тело вокруг фиксированных в нем осей, определяемых триэдром \eqref{beq1_1}. Базис \eqref{beq2_3} в $T_Q\mathcal{M}$ неголономен. Из первого свойства \eqref{beq1_13} имеем
\begin{equation}\label{beq2_4}
    [\Omega_1, \Omega_2] = \Omega_3, \quad [\Omega_2, \Omega_3] = \Omega_1, \quad [\Omega_3, \Omega_1] = \Omega_2.
\end{equation}

В силу соотношений \eqref{beq1_8}, \eqref{beq1_9} следующие равенства корректно определяют дифференциальные 1-формы на $\mathcal{M}$:
\begin{equation}\label{beq2_5}
\begin{array}{c}
    \Lambda_1(Q) = \nu_3 d\nu_2 + \nu_3' d\nu_2' + \nu_3'' d\nu_2'' = -(\nu_2 d\nu_3 + \nu_2' d\nu_3' + \nu_2'' d\nu_3''), \\
    \Lambda_2(Q) = \nu_1 d\nu_3 + \nu_1' d\nu_3' + \nu_1'' d\nu_3'' = -(\nu_3 d\nu_1 + \nu_3' d\nu_1' + \nu_3'' d\nu_1''), \\
    \Lambda_3(Q) = \nu_2 d\nu_1 + \nu_2' d\nu_1' + \nu_2'' d\nu_1'' = -(\nu_1 d\nu_2 + \nu_1' d\nu_2' + \nu_1'' d\nu_2'').
\end{array}
\end{equation}
Непосредственно проверяется, что эти формы образуют базис в $T_Q^*\mathcal{M}$, дуальный к \eqref{beq2_3}.
Отметим полезное в дальнейшем обращение зависимостей \eqref{beq2_5}:
\begin{equation}\label{beq2_6}
\begin{array}{c}
    d\nu_1^{(i)} = \nu_2^{(i)} \Lambda_3 - \nu_3^{(i)} \Lambda_2, \\
    d\nu_2^{(i)} = \nu_3^{(i)} \Lambda_1 - \nu_1^{(i)} \Lambda_3, \quad
    d\nu_3^{(i)} = \nu_1^{(i)} \Lambda_2 - \nu_2^{(i)} \Lambda_1 \\
    (i = 0, 1, 2).
\end{array}
\end{equation}

\begin{lemma}\label{lem221}
Форма $\Lambda_i$ сопоставляет мгновенному вращению {\rm (}касательному вектору к $\mathcal{M}${\rm )} $i$-ю компоненту его внутренней угловой скорости. Формы \eqref{beq2_5} левоинвариантны и удовлетворяют равенствам
\begin{equation}\label{beq2_7}
    d\Lambda_1 = \Lambda_3 \wedge \Lambda_2, \quad d\Lambda_2 = \Lambda_1 \wedge \Lambda_3, \quad d\Lambda_3 = \Lambda_2 \wedge \Lambda_1.
\end{equation}
\end{lemma}

\begin{proof}
Пусть $\Omega_Q \in T_Q\mathcal{M}$ и $\mbs{\omega}$ -- вектор \eqref{beq1_17} внутренней угловой скорости. Тогда в силу \eqref{beq1_15} ${\rm triv}(\Omega_Q) = (Q, \omega_1{\bf e} + \omega_2{\bf e}' + \omega_3{\bf e}'')$. Применяя к этому равенству отображение ${\rm triv}^{-1}$, учтем определение \eqref{beq2_2}. В результате получим
\begin{equation}\label{beq2_8}
    \Omega_Q = \omega_1 \Omega_1(Q) + \omega_2 \Omega_2(Q) + \omega_3 \Omega_3(Q).
\end{equation}
Утверждения леммы вытекают теперь из дуальности базисов \eqref{beq2_3}, \eqref{beq2_5} и из равенств \eqref{beq2_4}, \eqref{beq2_8}.
\end{proof}

Произвольные поля и 1-формы на $\mathcal{M}$, в отличие от выделенных базисов, будем обозначать строчными буквами. Кроме того, отступая от формальной строгости, компоненты полей и форм в этих базисах обозначим теми же буквами, но с индексом 1, 2, 3. Так, векторное поле $\omega: \mathcal{M} \to \tm$ имеет вид
\begin{equation}\label{beq2_9}
    \omega = \omega_1 \Omega_1 + \omega_2 \Omega_2 + \omega_3 \Omega_3.
\end{equation}

Подчеркнем различие между выражениями \eqref{beq2_8} и \eqref{beq2_9}. Вектор $\Omega_Q$ -- <<одинокий>>, и  поэтому $\omega_i$ в \eqref{beq2_8} -- числа. Напротив, $\omega$ -- векторное поле и $\omega_i : \mathcal{M} \to \bbR$ -- гладкие функции. Но и при этом для каждого $Q \in \mathcal{M}$ вектор
$$\left( \begin{array}{c} \omega_1(Q) \\ \omega_2(Q) \\ \omega_3(Q) \end{array}\right) \in \bbRR
$$
есть вектор внутренней угловой скорости мгновенного вращения $\omega(Q) \in T_Q\mathcal{M}$. Функции $\omega_i(Q)$ в формуле \eqref{beq2_9} постоянны тогда и только тогда, когда поле $\omega$ левоинвариантно.

Дифференциальная 1-форма $\lambda: \mathcal{M} \to T^*\mathcal{M}$ записывается в виде
\begin{equation}\label{beq2_10}
    \lambda = \lambda_1 \Lambda_1 + \lambda_2 \Lambda_2 + \lambda_3 \Lambda_3,
\end{equation}
где $\lambda_i : \mathcal{M} \to \bbR$ -- гладкие функции, а всякую 2-форму $\varkappa$ на $\mathcal{M}$, учитывая \eqref{beq2_7}, можно записать так:
\begin{equation}\label{beq2_11}
\begin{array}{rcl}
    \varkappa &=& \varkappa_1 \Lambda_3 \wedge \Lambda_2 + \varkappa_2 \Lambda_1 \wedge \Lambda_3 + \varkappa_3 \Lambda_2 \wedge \Lambda_1 =\\
     &=& \varkappa_1 d\Lambda_1 + \varkappa_2 d\Lambda_2 + \varkappa_3 d\Lambda_3.
\end{array}
\end{equation}
И здесь $\varkappa_i : \mathcal{M} \to \bbR$ -- гладкие функции.

Пространство 3-форм на $\mathcal{M}$ одномерно как модуль над кольцом функций на $\mathcal{M}$ и порождается формой объема $\Lambda_1 \wedge \Lambda_2 \wedge \Lambda_3$.

Пусть $F$ -- функция на $\mathcal{M}$. Гладко продолжая ее на некоторую окрестность $\mathcal{M}$ в $\bbR^9$ можно считать $F = F(\nu_1, \ldots, \nu_3'')$. Используя \eqref{beq2_6}, получим формулы дифференциала $F$:
\begin{equation}\label{beq2_12}
\begin{array}{c}
    dF = (\Omega_1F)\Lambda_1 + (\Omega_2F)\Lambda_2 + (\Omega_3F)\Lambda_3 = \\
    = \ds \sum\limits_{i=0}^2 \left[\left(\nu_2^{(i)} \Lambda_3 - \nu_3^{(i)} \Lambda_2 \right)\frac{\partial F}{\partial \nu_1^{(i)}} +
    \left(\nu_3^{(i)} \Lambda_1 - \nu_1^{(i)} \Lambda_3 \right)\frac{\partial F}{\partial \nu_2^{(i)}} + \right. \\
    \ds \left. + \left(\nu_1^{(i)} \Lambda_2 - \nu_2^{(i)} \Lambda_1 \right)\frac{\partial F}{\partial \nu_3^{(i)}}
    \right].
\end{array}
\end{equation}
После этого нетрудно записать и правила внешнего дифференцирования форм степени 1 и 2 (дифференциал 3-формы всегда равен нулю, так как ${\rm dim}\mathcal{M} = 3$), но в таком общем виде они в дальнейшем не понадобятся.

Рассмотрим второе касательное расслоение $T\tm$. Тривиализация \eqref{beq1_15} превращает каждый слой $T_{\Omega_Q}\tm$ в произведение $T_Q\mathcal{M} {\times} T_\mbs{\omega}\bbRR$. Последний сомножитель естественно отождествляется с $\bbRR$. Теперь любой элемент $X \in T\tm$ имеет вид
\begin{equation}\label{beq2_13}
    X = ((Q, \mbs{\omega}), (\Delta, \mbs{\varepsilon})),
\end{equation}
где матрица $\Delta$ принадлежит $Q{\cdot}\Ass(3)$, а $\mbs{\varepsilon}$ --некоторый вектор из $\bbRR$.

Напомним, что $p_\mathcal{A}:T\mathcal{A} \to \mathcal{A}$ для любого многообразия $\mathcal{A}$ означает каноническую проекцию касательного расслоения. В нашем случае
\begin{equation}\label{beq2_14}
\begin{array}{c}
    p_\mathcal{M}(Q, \mbs{\omega}) = Q, \qquad  p_{\tm}((Q, \mbs{\omega}), (\Delta, \mbs{\varepsilon})) = (Q, \mbs{\omega}), \\
    T_{p_\mathcal{M}}((Q, \mbs{\omega}), (\Delta, \mbs{\varepsilon})) = (Q, \mbs{\iota}(Q^{-1}\Delta)).
\end{array}
\end{equation}

Условимся векторные поля \eqref{beq2_3} рассматривать и на $\tm$, считая $\Delta(Q, \mbs{\omega}) = \Omega_i(Q), \mbs{\varepsilon} = 0$. В пространстве $T_{\Omega_Q}\tm$ имеем базис
\begin{equation}\label{beq2_15}
    \ds \Omega_1(Q), \Omega_2(Q), \Omega_3(Q), \frac{\partial}{\partial \omega_1}, \frac{\partial}{\partial \omega_2}, \frac{\partial}{\partial \omega_3}.
\end{equation}

Формы \eqref{beq2_5} <<поднимаются>> в кокасательное расслоение $T^*\tm$ с помощью антиотображения $p_\mathcal{M}^*$. Полученные в результате 1-формы на $\tm$ условимся для упрощения записи обозначать так же, как исходные. Таким образом, если $X \in T_{\Omega_Q}\tm$, то по определению полагаем
\begin{equation}\label{beq2_16}
    \Lambda_i(Q, \omega) \cdot X = \Lambda_i(Q) \cdot T_{p_\mathcal{M}}(X).
\end{equation}
Поэтому базис в $T_{\Omega_Q}^*\tm$, сопряженный к \eqref{beq2_15}, таков:
\begin{equation}\label{beq2_17}
    \Lambda_1(Q), \Lambda_2(Q), \Lambda_3(Q), d\omega_1, d\omega_2, d\omega_3.
\end{equation}

Отметим, что эти базисы, как и \eqref{beq2_3}, \eqref{beq2_5} сохраняются левыми сдвигами группы $\mathcal{M} = {SO}(3)$, а точнее, отображениями $\tm \to \tm$, касательными к левым сдвигам.
Нетрудно построить аналогичные правоинвариантные объекты, исходя из тривиализации \eqref{beq1_14}, однако в дальнейшем они использованы не будут.

\subsection{Уравнения движения твердого тела\\  в поле потенциальных и гироскопических сил}\label{ssec23}

Введем на многообразии $\mathcal{M} = {SO}(3)$ риманову метрику $s$, полагая с учетом \eqref{beq1_17}
\begin{equation}\label{beq3_1}
    s_Q(\Omega_q^1, \Omega_q^2) = \mbs{A}\mbs{\omega}^1 \cdot \mbs{\omega}^2,
\end{equation}
где $\mbs{A}:\bbRR \to \bbRR$ -- симметричный постоянный оператор (тензор инерции), а точка, как обычно, означает стандартное скалярное произведение в $\bbRR$.

Определим теперь твердое тело с неподвижной точкой как механическую систему с гироскопическими силами
\begin{equation}\label{beq3_2}
    (\mathcal{M}, s, \Pi, \varkappa),
\end{equation}
где
\begin{equation}\label{beq3_3}
    \Pi = \Pi(Q)
\end{equation}
-- функция на $\mathcal{M}$, а $\varkappa$ -- замкнутая 2-форма на $\mathcal{M}$. Будем считать, что последняя задана в виде \eqref{beq2_11}. Условие замкнутости $(d\varkappa = 0)$ с учетом \eqref{beq2_12} запишется так:
\begin{equation}\label{beq3_4}
    \Omega_1\varkappa_1 + \Omega_2\varkappa_2 + \Omega_3\varkappa_3 =0.
\end{equation}

Найдем выражения дифференциальных форм \eqref{aeq2_5}, \eqref{aeq2_6} в базисе \eqref{beq2_17}. Согласно формулам \eqref{aeq2_4}, \eqref{beq2_13}, \eqref{beq2_14}, \eqref{beq3_1} имеем
\begin{equation*}
    \theta(\Omega_Q) \cdot X = s_Q(\Omega_Q, \Delta) = \mbs{A} \mbs{\omega} \cdot \mbs{\iota}(Q^{-1}\Delta),
\end{equation*}
откуда по лемме \ref{lem221}
\begin{equation}\label{beq3_5}
    \theta(Q, \omega) = \sum\limits_{i, j = 1}^3 A_{ij}\omega_i \Lambda_j(Q).
\end{equation}
Здесь $A_{ij}$ -- компоненты тензора $\mbs{A}$ в каноническом базисе $\bbRR$.
Дифференцируя \eqref{beq3_5}, находим
\begin{equation}\label{beq3_6}
    \sigma(Q, \omega) = \sum\limits_{i, j = 1}^3 A_{ij}(d\omega_i \wedge \Lambda_j(Q) + \omega_i d\Lambda_j(Q))
\end{equation}
Ограничимся случаем, когда тензор инерции диагонален (что достигается подходящим выбором подвижных осей). Обозначая $\mbs{A} = {\rm diag} \{A_1, A_2, A_3\}$ и учитывая \eqref{beq2_7}, перепишем уравнение \eqref{beq3_6}:
\begin{equation}\label{beq3_7}
\begin{array}{l}
    \sigma = A_1 d\omega_1 \wedge  \Lambda_1 + A_2 d\omega_2 \wedge  \Lambda_2 + A_3 d\omega_3 \wedge  \Lambda_3 + \\
    \qquad  + A_1 \omega_1 \Lambda_3 \wedge  \Lambda_2 + A_2 \omega_2 \Lambda_1 \wedge  \Lambda_3 + A_3 \omega_3 \Lambda_2 \wedge  \Lambda_1.
\end{array}
\end{equation}

Пусть
\begin{equation}\label{beq3_8}
    X(Q, \omega) = \left((Q, \mbs{\omega}),(\Delta(Q, \mbs{\omega}), \mbs{\varepsilon}(Q, \mbs{\omega})) \right)
\end{equation}
-- векторное поле, определяющее динамику системы \eqref{beq3_2} и $d/dt$ -- оператор дифференцирования вдоль $X$:
\begin{equation}\label{beq3_9}
    \ds \frac{dQ}{dt} = \Delta(Q, \mbs{\omega}), \qquad \frac{d\mbs{\omega}}{dt} = \mbs{\varepsilon}(Q, \mbs{\omega}).
\end{equation}
По теореме \ref{the122} поле $X$ есть уравнение второго порядка. Условие \eqref{aeq1_2} в силу \eqref{beq2_14} приводит к соотношению
\begin{equation}\label{beq3_10}
    \mbs{\omega} = \mbs{\iota}(Q^{-1} \Delta),
\end{equation}
откуда $\Delta = Q \mbs{\iota}^{-1}(\mbs{\omega}) = Q/\mbs{\omega}/$. Поэтому, транспонируя первое (матричное) уравнение \eqref{beq3_9}, имеем
\begin{equation*}
    \ds \frac{dQ^T}{dt} = -/\mbs{\omega}/Q^T.
\end{equation*}
Учитывая, что столбцы $Q^T$ -- это векторы \eqref{beq1_7}, и используя второе свойство \eqref{beq1_13}, получаем уравнения Пуассона \eqref{beq1_25}:
\begin{equation}\label{beq3_11}
    \astup{\nu} = \mbs{\nu} {\times} \mbs{\omega}, \quad \astup{\nu}' = \mbs{\nu}' {\times} \mbs{\omega}, \quad  \astup{\nu}'' = \mbs{\nu}'' {\times} \mbs{\omega}.
\end{equation}

Вычислим элементы общего уравнения \eqref{aeq2_9}. Кинетическая энергия \eqref{aeq2_2} для системы \eqref{beq3_2} в силу \eqref{beq3_1} имеет обычный вид: $\ds \frac{1}{2}\mbs{A} \mbs{\omega} \cdot \mbs{\omega}$. В этом случае полная энергия \eqref{aeq2_3}
\begin{equation}\label{beq3_12}
    H(Q, \mbs{\omega}) = \ds \frac{1}{2}\mbs{A} \mbs{\omega} \cdot \mbs{\omega} + \Pi(Q).
\end{equation}
Следовательно,
\begin{equation}\label{beq3_13}
    dH(Q, \mbs{\omega}) = A_1\omega_1d\omega_1 + A_2\omega_2d\omega_2 + A_3\omega_3d\omega_3 + d\Pi(Q).
\end{equation}
Отметим, что по определению $d\omega_i(Q, \mbs{\omega}) (X) = X\omega_i = \ds \frac{d\omega_i}{dt}$, а в силу \eqref{beq3_10} и соглашения \eqref{beq2_16} $\Lambda_i (Q, \mbs{\omega}) (X) = \omega_i$. Поэтому
\begin{equation}\label{beq3_14}
\begin{array}{rcl}
    i_X\sigma &=& \ds \Bigl[A_1 \frac{d\omega_1}{dt} + (A_3 - A_2)\omega_2\omega_3  \Bigr] \Lambda_1 + \\[2mm]
     & + & \ds \Bigl[A_2 \frac{d\omega_2}{dt} + (A_1 - A_3)\omega_1\omega_3  \Bigr] \Lambda_2 + \\[2mm]
     & + & \ds \Bigl[A_3 \frac{d\omega_3}{dt} + (A_2 - A_1)\omega_1\omega_2  \Bigr] \Lambda_3 - \\[2mm]
     & - & A_1\omega_1d\omega_1 - A_2\omega_2d\omega_2 - A_3\omega_3 d\omega_3 ,\\[2mm]
    i_X p_\mathcal{M}^*\varkappa &=& (\omega_2 \varkappa_3 - \omega_3 \varkappa_2)\Lambda_1 + (\omega_3 \varkappa_1 - \omega_1 \varkappa_3)\Lambda_2 + (\omega_1 \varkappa_2 - \omega_2 \varkappa_1)\Lambda_3.
\end{array}
\end{equation}
Подставляя \eqref{beq3_13}, \eqref{beq3_14} в \eqref{aeq2_9} и приравнивая коэффициенты при независимых формах $\Lambda_i$ $(i = 1, 2, 3)$, учтем правило \eqref{beq2_12}. В результате преобразований приходим к уравнениям Эйлера (см. замечание в конце $\S$~\ref{ssec21}):
\begin{equation}\label{beq3_15}
\begin{array}{c}
    \ds A_1\dot{\omega}_1 + (A_3 - A_2)\omega_2 \omega_3 + \omega_2 \varkappa_3 - \omega_3 \varkappa_2 = \sum\limits_{i=0}^2 \Bigl[ \nu_2^{(i)}\frac{\partial \Pi}{\partial \nu_3^{(i)}} - \nu_3^{(i)}\frac{\partial \Pi}{\partial \nu_2^{(i)}}\Bigr], \\
    \ds A_2\dot{\omega}_2 + (A_1 - A_3)\omega_3 \omega_1 + \omega_3 \varkappa_1 - \omega_1 \varkappa_3 = \sum\limits_{i=0}^2 \Bigl[ \nu_3^{(i)}\frac{\partial \Pi}{\partial \nu_1^{(i)}} - \nu_1^{(i)}\frac{\partial \Pi}{\partial \nu_3^{(i)}}\Bigr], \\
    \ds A_3\dot{\omega}_3 + (A_2 - A_1)\omega_1 \omega_2 + \omega_1 \varkappa_2 - \omega_2 \varkappa_1 = \sum\limits_{i=0}^2 \Bigl[ \nu_1^{(i)}\frac{\partial \Pi}{\partial \nu_2^{(i)}} - \nu_2^{(i)}\frac{\partial \Pi}{\partial \nu_1^{(i)}}\Bigr].
\end{array}
\end{equation}
Эта система замыкается уравнениями \eqref{beq3_11}.

Согласно общему предложению \ref{pro124} функция \eqref{beq3_12} является первым интегралом для векторного поля \eqref{beq3_8}, а значит, и для системы \eqref{beq3_11}, \eqref{beq3_15}.

Подчеркнем, что здесь $\varkappa_1, \varkappa_2, \varkappa_3$ -- любые функции на многообразии $\mathcal{M}$. При выводе не использовано даже условие \eqref{beq3_4}. Таким образом, получены самые общие уравнения, описывающие движение твердого тела вокруг неподвижной точки в поле потенциальных и гироскопических сил.

\subsection{Существование интеграла площадей}\label{ssec24}

Потенциальные силовые поля, встречающиеся в наиболее интересных задачах динамики твердого тела, обладают осью симметрии. Так, осесимметричными являются поле силы тяжести и центральное ньютоновское поле. С точки зрения, изложенной в предыдущих параграфах, это означает, что потенциал \eqref{beq3_3} сохраняется группой вращений вокруг некоторого направления, неизменного в неподвижном пространстве. По аналогии с основным случаем поля силы тяжести такое направление будем называть вертикалью.

Без ограничения общности можно считать, что в триэдре \eqref{beq1_2} вертикален первый вектор. Образом вертикали в подвижном пространстве будет, следовательно, первый вектор триэдра \eqref{beq1_7}. Тогда, если конфигурация $Q'$ твердого тела получена из конфигурации $Q$ поворотом на угол $\tau \in \bbR$ вокруг оси $[{\boldsymbol \nu}]$, то справедливо равенство $Q' = Q^\tau Q$, где
\begin{equation}\label{beq4_1}
    Q^\tau = \left(
    \begin{array}{ccc}
    1 & 0 & 0 \\
    0 & \cos(\tau) & -\sin(\tau) \\
    0 & \sin(\tau) & \cos(\tau)
    \end{array}
    \right).
\end{equation}
Рассмотрим однопараметрическую группу диффеоморфизмов $\Psi = \{\psi^\tau: \tau \in \bbR\}$ многообразия~$\mathcal{M}$:
\begin{equation}\label{beq4_2}
    \psi^\tau(Q) = Q^\tau Q.
\end{equation}
Очевидно, $\Psi$ изоморфна $S^1$ и, в частности, компактна.
Порождающее поле
\begin{equation}\label{beq4_3}
    \ds v(Q) = \left. \frac{d}{d\tau} \right|_{\tau=0}\psi^\tau(Q) = \left. \frac{dQ^\tau}{d\tau} \right|_{\tau=0}Q
\end{equation}
правоинвариантно и, поскольку траектория группы есть вращение $\bbRR$ в $\bbRRR$ с постоянной угловой скоростью вокруг первой оси $[{\boldsymbol \nu}_0]$, то вектор внутренней угловой скорости мгновенного вращения \eqref{beq4_3} равен ${\boldsymbol \nu}$. Вычислим его по определению \eqref{beq1_17}, воспользовавшись при этом последним свойством \eqref{beq1_13} и выражением \eqref{beq4_1}:
\begin{equation}\label{beq4_4}
    \mbs{\iota}(Q^{-1}v(Q)) = \mbs{\iota}\left(Q^{-1}\left/ \left(
    \begin{array}{c} 1 \\ 0 \\ 0 \end{array} \right) \right/ Q\right) = \mbs{\iota}\left(\left/Q^{-1} \left(\begin{array}{c} 1 \\ 0 \\ 0 \end{array} \right) \right/ \right) = \mbs{\nu}.
\end{equation}
Отображение $p:Q \mapsto \mbs{\nu}$, определенное в соответствии с \eqref{beq4_4}, переводит многообразие $\mathcal{M}$ на единичную сферу в $\bbRR$
\begin{equation}\label{beq4_5}
    \nu_1^2 + \nu_2^2 + \nu_3^2 = 1,
\end{equation}
называемую сферой Пуассона. Прообраз любой точки сферы при этом является в точности орбитой группы $\Psi$. Следовательно,
\begin{equation}\label{beq4_6}
    p:\mathcal{M} \to S^2
\end{equation}
есть отображение факторизации и, в силу компактности группы, главное $\Psi$-расслоение \cite{bib08,bib09}.

Известно, что в этом случае всякая $\Psi$-сохранная гладкая функция $F$ на $\mathcal{M}$ представляет собой композицию $F = \widetilde{F} \circ p$, где $\widetilde{F}$ -- гладкая функция на $\mathcal{M}/ \Psi$. Таким образом, будучи записана в переменных $\nu_i^{(j)} (i = 1,2,3; j = 0,1,2)$, $\Psi$-сохранная функция на $SO(3)$ имеет вид
\begin{equation}\label{beq4_7}
    F = F(\nu_1, \nu_2, \nu_3).
\end{equation}
Зависимость \eqref{beq4_5} переменных $\nu_i$ здесь не существенна, так как любая гладкая функция на сфере может быть продолжена и в некоторый шаровой слой.

Итак, симметрия потенциального силового поля означает, что потенциальная энергия \eqref{beq3_3} может быть записана в виде \eqref{beq4_7}:
\begin{equation}\label{beq4_8}
    \Pi = \Pi(\nu_1, \nu_2, \nu_3).
\end{equation}

Заметим, что метрика \eqref{beq3_1} по своему определению левоинвариантна и, в частности, сохраняется преобразованиями \eqref{beq4_2}. Поэтому, предполагая выполнение свойства \eqref{beq4_8}, попадаем в ситуацию, которой посвящен $\S$~\ref{ssec13}. В этом случае можно ставить вопрос о существовании у рассматриваемой системы первого интеграла вида \eqref{aeq3_6}.
По теореме \ref{the136}, для существования такого интеграла необходимо и достаточно, чтобы поле \eqref{beq4_3} и форма гироскопических сил $\varkappa$ системы \eqref{beq3_2} удовлетворяли условию \eqref{aeq3_18} с некоторой функцией $f$ на~$\mathcal{M}$. Согласно \eqref{aeq3_21} функция $f$ сохраняется группой $\Psi$, т.е. искать ее нужно в виде
\begin{equation}\label{beq4_9}
    F = F(\nu_1, \nu_2, \nu_3).
\end{equation}
Используя лемму \ref{lem221}, выражения \eqref{beq2_11}, \eqref{beq2_12} и свойства \eqref{beq4_4}, \eqref{beq4_9}, находим
\begin{equation}\label{beq4_10}
\begin{array}{c}
    i_v\varkappa = (\varkappa_3 \nu_2 - \varkappa_2 \nu_3)\Lambda_1 + (\varkappa_1 \nu_3 - \varkappa_3 \nu_1)\Lambda_2 + (\varkappa_2 \nu_1 - \varkappa_1 \nu_2)\Lambda_3, \\
    \ds -df = \left(\nu_2 \frac{\partial f}{\partial \nu_3} - \nu_3 \frac{\partial f}{\partial \nu_2} \right) \Lambda_1 + \left(\nu_3 \frac{\partial f}{\partial \nu_1} - \nu_1 \frac{\partial f}{\partial \nu_3} \right) \Lambda_2 + \left(\nu_1 \frac{\partial f}{\partial \nu_2} - \nu_2 \frac{\partial f}{\partial \nu_1} \right) \Lambda_3.
\end{array}
\end{equation}
Введем вектор-функции
\begin{equation}\label{beq4_11}
   \ds \mbs{\varkappa} = \left( \begin{array}{c} \varkappa_1\\ \varkappa_2 \\ \varkappa_3 \end{array} \right),\quad
   \grad f = \left( \begin{array}{c} {\partial f}/{\partial \nu_1} \\ {\partial f}/{\partial \nu_2} \\ {\partial f}/{\partial \nu_3}\end{array} \right).
\end{equation}
Подстановка \eqref{beq4_10}, \eqref{beq4_11} в \eqref{aeq3_18} приводит к тождеству
\begin{equation}\label{beq4_12}
    (\mbs{\varkappa} - \grad f) {\times} \mbs{\nu} = 0.
\end{equation}
Как отмечалось в $\S$~\ref{ssec13}, из замкнутости формы $\mbs{\varkappa}$ и условия \eqref{aeq3_18} следует ее $\Psi$-сохранность:
\begin{equation}\label{beq4_13}
    {\bf L}_v\mbs{\varkappa} = 0.
\end{equation}
С другой стороны, формы $\Lambda_i$ левоинвариантны и, в частности, ${\bf L}_v\Lambda_i = 0$. Поэтому в представлении \eqref{beq2_11}
\begin{equation}\label{beq4_14}
    {\bf L}_v\mbs{\varkappa} = (v\varkappa_1)\Lambda_3 \wedge \Lambda_2 + (v\varkappa_2)\Lambda_1 \wedge \Lambda_3 + (v\varkappa_3)\Lambda_2 \wedge \Lambda_1.
\end{equation}
Сравнивая \eqref{beq4_13}, \eqref{beq4_14}, заключаем, что $v\varkappa_i = 0$, т.е.
\begin{equation}\label{beq4_15}
    \varkappa_i = \varkappa_i(\nu_1, \nu_2, \nu_3), \qquad i = 1, 2, 3.
\end{equation}
Из формул \eqref{beq4_12}, \eqref{beq4_15} получаем следующее утверждение.

\begin{theorem}\label{the241}
Для существования интеграла площадей в системе \eqref{beq3_2} с потенциалом \eqref{beq4_8} необходимо и достаточно, чтобы вектор, составленный из коэффициентов формы гироскопических сил, мог быть записан в виде
\begin{equation}\label{beq4_16}
    \mbs{\varkappa} = F \mbs{\nu} + \grad f,
\end{equation}
где $F$ и $f$ -- некоторые функции, зависящие только от $\nu_1, \nu_2, \nu_3$. Если выполнено условие \eqref{beq4_16}, то интеграл площадей имеет следующее выражение:
\begin{equation}\label{beq4_17}
    G(Q, \mbs{\omega}) = A_1\omega_1 \nu_1 + A_2\omega_2 \nu_2 + A_3\omega_3 \nu_3 + f(\nu_1, \nu_2, \nu_3).
\end{equation}
\end{theorem}

Для доказательства достаточно использовать теорему \ref{the136} и равенства \eqref{aeq3_6}, \eqref{beq3_1}, \eqref{beq4_4}. Отметим, что в случае \eqref{beq4_16} равенство \eqref{beq3_4} выполняется автоматически.

Пусть выполнены условия существования интеграла площадей. Для полной интегрируемости рассматриваемой системы необходимо указать интеграл $K$ на $\tm$, не зависящий от $H$ и $G$ и находящийся в инволюции с ними. Согласно предложению \ref{pro138} этот интеграл должен иметь вид
\begin{equation}\label{beq4_18}
    K = K(\omega_1, \omega_2, \omega_3, \nu_1, \nu_2, \nu_3).
\end{equation}

Итак, при наличии интегралов \eqref{beq4_17}, \eqref{beq4_18} задача может быть сведена к квадратурам. Опыт показывает, однако, что в большинстве случаев эти квадратуры крайне сложны и не поддаются качественному исследованию. Значительная информация о движениях может быть получена на основе изучения одних лишь первых интегралов. Подобным вопросам посвящена следующая глава.

Укажем, как видоизменяются уравнения \eqref{beq3_15} при наличии группы симметрий. Поскольку в правые части \eqref{beq3_15} входят теперь лишь переменные $\nu_1, \nu_2, \nu_3$, то система замыкается уже первой группой уравнений \eqref{beq3_11}. Тем самым с учетом зависимости уравнений \eqref{beq3_11} порядок системы реально понижен на единицу. Фазовое пространство уравнений \eqref{beq3_15} и первой группы уравнений \eqref{beq3_11} соответствует многообразию \eqref{aeq5_25}. Более подробно оно будет рассмотрено в главе 3.

Известно, что с топологической точки зрения $\mathcal{M} = \bbR P^3$. В самом деле, сопоставим каждому вращению $Q$ его вектор конечного поворота \cite{bib31}. Получим гомеоморфизм многообразия $\mathcal{M}$ и трехмерного пространства, оснащенного бесконечно удаленными точками (повороты на $\pm \pi$), причем две бесконечно удаленные точки на одной прямой, проходящей через начало координат, необходимо отождествлять (повороты на $\pm \pi$ суть одно и то же). А это и есть $\bbR P^3$. Можно поступить иначе. Введем параметры Родрига-Гамильтона $(\lambda_0, \lambda_1, \lambda_2, \lambda_3) \in S^3$. Каждому вращению $Q$ отвечает два набора параметров -- диаметрально противоположные точки трехмерной сферы. Отождествляя их, вновь получим $S^3 / \mathbb{Z}_2 = \bbR P^3$.
Группа целочисленных когомологий $H^2(\bbR P^3)$ есть $\mathbb{Z}_2$ \cite[с. 115, 121]{bib55}. По теореме об универсальных коэффициентах \cite{bib12} находим $H^2(\mathcal{M}, \bbR) = 0$. В силу теоремы де Рама (см., например, \cite{bib53}) это означает, что всякая замкнутая 2-форма на $\mathcal{M}$ является точной. В частности, в нашей задаче форма гироскопических сил всегда точна, и согласно \eqref{aeq1_8}, \eqref{beq2_10} система представляется функцией Лагранжа
\begin{equation}\label{beq4_19}
    \ds L = \frac{1}{2} (A_1\omega_1^2 + A_2\omega_2^2 + A_3\omega_3^2) + \lambda_1 \omega_1 + \lambda_2 \omega_2 + \lambda_3 \omega_3 -\Pi,
\end{equation}
где $\lambda_i:\mathcal{M} \to \bbR$.

В общем случае выражения $\varkappa_i$ через $\lambda_j$, полученные из условия
\begin{equation}\label{beq4_20}
    \varkappa = d\lambda,
\end{equation}
весьма громоздки, и делать такую замену в \eqref{beq3_15} бессмысленно. Рассмотрим случай \eqref{beq4_15}. По замечанию \ref{rem133} в $\S$~\ref{ssec13} равенство \eqref{beq4_20} выполнено для некоторой $\Psi$-сохранной 1-формы $\lambda$, т.е. можно считать
\begin{equation}\label{beq4_21}
    \lambda_i = \lambda_i(\nu_1, \nu_2, \nu_3), \qquad (i = 1, 2, 3).
\end{equation}
Дифференцируя \eqref{beq2_10} с учетом \eqref{beq2_12}, из равенства \eqref{beq4_20} получаем
\begin{equation}\label{beq4_22}
\begin{array}{c}
    \ds \varkappa_1 = \lambda_1 + \left(\nu_2 \frac{\partial \lambda_2}{\partial \nu_1} -  \nu_1 \frac{\partial \lambda_2}{\partial \nu_2}\right) -
    \left(\nu_1 \frac{\partial \lambda_3}{\partial \nu_3} -  \nu_3 \frac{\partial \lambda_3}{\partial \nu_1}\right). \\
    (123)
\end{array}
\end{equation}
(символ (123) показывает, что равенство сохраняется при циклических перестановках индексов.)
Введем вектор-функцию
\begin{equation}\label{beq4_23}
    \lambda(\nu_1, \nu_2, \nu_3) = \left( \begin{array}{c} \lambda_1(\nu_1, \nu_2, \nu_3) \\ \lambda_2(\nu_1, \nu_2, \nu_3) \\ \lambda_3(\nu_1, \nu_2, \nu_3)\end{array} \right).
\end{equation}
Равенства \eqref{beq4_22} запишем в компактном виде
\begin{equation}\label{beq4_24}
    \mbs{\varkappa} = -\mbs{\nu} \ddiv \mbs{\lambda} + \grad (\mbs{\lambda} \cdot \mbs{\nu}).
\end{equation}
Сравнивая формулу \eqref{beq4_24} с \eqref{beq4_16}, находим вид функций $f$ и $F$:
\begin{equation}\label{beq4_25}
    f(\mbs{\nu}) = \mbs{\lambda}(\mbs{\nu}) \cdot \mbs{\nu}, \qquad F(\mbs{\nu}) = -\ddiv \mbs{\lambda}.
\end{equation}
Используя формулу \eqref{beq4_24}, уравнения \eqref{beq3_15} можно представить как систему в векторной форме:
\begin{equation}\label{beq4_26}
    {\bf A} \astup{\omega} + \mbs{\omega} {\times} [{\bf A}\mbs{\omega} + \mbs{\lambda}+(\Lambda - E \ddiv \mbs{\lambda})\mbs{\nu}] = \mbs{\nu} {\times} \grad \Pi.
\end{equation}
Здесь $\ds \Lambda^T = (\frac{\partial \lambda_i}{\partial \nu_j})$ -- матрица Якоби, $E$ -- единичная матрица.
Интеграл площадей \eqref{beq4_17} с учетом \eqref{beq4_25} запишем следующим образом:
\begin{equation}\label{beq4_27}
    G = (A_1 \omega_1 + \lambda_1)\nu_1 + (A_2 \omega_2 + \lambda_2)\nu_2 + (A_3 \omega_3 + \lambda_3)\nu_3.
\end{equation}
Вектор \eqref{beq4_23} играет, следовательно, роль гиростатического момента.

Систему \eqref{beq4_26} можно получить и классическим способом, записав уравнения Эйлера\,--\,Лагранжа (Пуанкаре\,--\,Больцмана\,--\,Гаммеля) \cite{bib31} с лагранжианом \eqref{beq4_19}, в котором $\lambda_i$ -- функции вида \eqref{beq4_21}. Дифференцирование по квазикоординате, соответствующей квазискорости $\omega_i$, -- это дифференцирование вдоль векторного поля \eqref{beq2_1}. Трехиндексные символы Больцмана (структурные константы алгебры Ли левоинвариантных векторных полей на $\mathcal{M}$ в базисе \eqref{beq2_1}) уже определены соотношениями \eqref{beq2_4}.

На практике гироскопические силы могут быть обусловлены, например, наличием роторов, оси которых закреплены в теле-носителе. При этом обычно полагается, что угловые скорости роторов либо являются заданными функциями времени \cite{bib72}, либо определяются действующими на систему силами \cite{bib39}. Ситуация \eqref{beq4_21} наблюдается, например, в обобщенных уравнениях типа Кирхгофа \cite{bib34,bib35}.

Форму \eqref{beq4_26} имеют уравнения движения наэлектризованного твердого тела в магнитном поле, напряженность которого постоянна и коллинеарна вектору $\mbs{\nu}$. При этом $\mbs{\lambda} = \Lambda \mbs{\nu}$, где матрица $\Lambda$ также постоянна и пропорциональна <<электрическому тензору инерции>>.

Отметим еще один хорошо известный, но важный частный случай -- задачу о движении гиростата в потенциальном силовом поле. Он характеризуется равенством
\begin{equation}\label{beq4_28}
    \mbs{\lambda} = \cons.
\end{equation}
При этом условии выделим из \eqref{beq3_11}, \eqref{beq4_26} замкнутую подсистему
\begin{equation}\label{beq4_29}
\begin{array}{c}
    A_1 \dot{\omega}_1 + (A_3 - A_2)\omega_2 \omega_3 + \omega_2 \lambda_3 - \omega_3 \lambda_2 = \nu_2 \Pi_{\nu_3} - \nu_3 \Pi_{\nu_2}, \\
    A_2 \dot{\omega}_2 + (A_1 - A_3)\omega_3 \omega_1 + \omega_3 \lambda_1 - \omega_1 \lambda_3 = \nu_3 \Pi_{\nu_1} - \nu_1 \Pi_{\nu_3}, \\
    A_3 \dot{\omega}_3 + (A_2 - A_1)\omega_1 \omega_2 + \omega_1 \lambda_2 - \omega_2 \lambda_1 = \nu_1 \Pi_{\nu_2} - \nu_2 \Pi_{\nu_1}, \\
    \dot{\nu}_1 = \nu_2 \omega_3 - \nu_3 \omega_2, \\
    \dot{\nu}_2 = \nu_3 \omega_1 - \nu_1 \omega_3, \\
    \dot{\nu}_3 = \nu_1 \omega_2 - \nu_2 \omega_1.
\end{array}
\end{equation}
Левые части уравнений Эйлера -- это, конечно, компоненты относительной производной вектора ${\bf M} = {\bf A} \mbs{\omega} + \mbs{\lambda}$, называемого вектором момента количества движения гиростата или, короче, кинетическим моментом. Первый интеграл \eqref{beq4_27} выражает постоянство проекции кинетического момента на вертикаль.

Отметим возможность дальнейшего понижения порядка в системе \eqref{beq3_11}, \eqref{beq4_26}. Согласно теореме \ref{the154} и замечанию в конце $\S$~\ref{ssec15} индуцированное векторное поле на фиксированном уровне $G = g$ интеграла \eqref{beq4_17} с помощью отображения $Tp/\Psi$, где $p$ -- отображение \eqref{beq4_6}, переводится в векторное поле, задающее динамику некоторой системы с гироскопическими силами на многообразии $TS^2$. Его решения -- это, очевидно, траектории <<в теле>> направляющего вектора вертикали, соответствующие движениям с заданной постоянной площадей. В структуре прямого произведения \eqref{beq1_15} имеем
$Tp(Q, \mbs{\omega}) = (p(Q),p(Q){\times}\mbs{\omega})$, поэтому $(Tp / \Psi)(\mbs{\nu}, \mbs{\omega}) = (\mbs{\nu},\mbs{\nu} {\times} \mbs{\omega})$. Воспользовавшись этим фактом и формулами \eqref{beq2_6}, \eqref{beq2_7}, \eqref{beq2_12}, \eqref{beq4_8}, \eqref{beq4_9}, \eqref{beq4_16}, найдем элементы приведенной системы на $S^2$:
\begin{equation*}
\begin{array}{c}
    \ds \widetilde{s}_\nu \left(\astup{\nu}^1, \astup{\nu}^2 \right) = \frac{A_1A_2A_3}{I(\mbs{\nu})} \left( \frac{\dot{\nu}_1^1 \dot{\nu}_1^2}{A_1} + \frac{\dot{\nu}_2^1 \dot{\nu}_2^2}{A_2} + \frac{\dot{\nu}_3^1 \dot{\nu}_3^2}{A_3}\right), \\
    \ds \Pi_g(\mbs{\nu}) = \Pi(\mbs{\nu}) + \frac{[g-f(\mbs{\nu})]^2}{2I(\mbs{\nu})}, \\
    \ds \varkappa_g = \frac{[g-f(\mbs{\nu})]\Delta(\mbs{\nu})}{I^2(\mbs{\nu})}(\nu_1d\nu_2 \wedge d\nu_3 + \nu_2d\nu_3 \wedge d\nu_1 + \nu_3d\nu_1 \wedge d\nu_2) - \\
    \ds - \frac{I(\mbs{\nu})F(\mbs{\nu}) + \mbs{A}\mbs{\nu} \cdot \grad f(\mbs{\nu})}{I^2(\mbs{\nu})}(A_1\nu_1d\nu_2 \wedge d\nu_3 + A_2\nu_2d\nu_3 \wedge d\nu_1 + A_3\nu_3d\nu_1 \wedge d\nu_2).
\end{array}
\end{equation*}
Здесь обозначено
\begin{equation*}
\begin{array}{c}
    I(\mbs{\nu}) = A_1 \nu_1^2 + A_2 \nu_2^2 + A_3 \nu_3^2, \\
    \Delta(\mbs{\nu}) = (A_2 + A_3 - A_1)A_1 \nu_1^2 + (A_3 + A_1 - A_2)A_2 \nu_2^2 + (A_1 + A_2 - A_3)A_3 \nu_3^2.
\end{array}
\end{equation*}
Соответствующие дифференциальные уравнения второго порядка на $S^2$ весьма громоздки и в дальнейшем нам не понадобятся. В частном случае, когда в исходной системе отсутствуют гироскопические силы, они приведены, например, в \cite{bib50}.

Отметим, что если исходная система $SO(3)$ натуральна (в нашем случае это означает $F(\mbs{\nu}) \equiv 0$ и  $f(\mbs{\nu}) \equiv 0$), то при $g \neq 0$ приведенная форма гироскопических сил с точностью до множителя $3g\Delta(\mbs{\nu}) / I^2(\mbs{\nu})$, нигде не обращающегося в нуль на $S^2$, совпадает с формой объема на $S^2$, индуцированной из $\bbRR$. В частности, эта форма неточна. Этот же факт вытекает из результатов, отмеченных в $\S$~\ref{ssec15}: группа $H^2(S^2, \mathbb{Z}) = \mathbb{Z}$; расслоение \eqref{beq4_6} дифференцируемо изоморфно нетривиальному расслоению единичных касательных векторов двумерной сферы. Таким образом, приведенная система в задаче о движении твердого тела в потенциальном поле при отличной от нуля постоянной площадей не допускает глобальной функции Лагранжа.

\clearpage

\section{Основные принципы\\  топологического и геометрического анализа\\  интегрируемых механических систем}\label{sec3}

Поведение большого класса объектов, изучаемых в механике, может быть с той или иной степенью адекватности описано системой обыкновенных дифференциальных уравнений первого порядка. Множество, на котором такая система задана, называют пространством состояний объекта или его фазовым пространством. Изучение различных свойств решений системы уравнений, соответствующих заданному объекту, получило название метода фазового пространства. В этой главе будут рассмотрены два аспекта метода -- исследование фазовой топологии системы и ее геометрический анализ.

\subsection{Фазовая топология динамической системы}\label{ssec31}

Динамическая система -- это гладкое векторное поле $X$ на конечномерном $C^\infty$-много\-обра\-зии~$\mathfrak{M}$. Многообразие $\mathfrak{M}$ есть по определению фазовое пространство системы $X$.

Фазовой топологией динамической системы $X$ называют кривую $x(t) \in \mathfrak{M}$, $t \in (a,b) \subset \bbR$, удовлетворяющую условию $T_t\,x(t, {\bf 1}) = X(x(t))$ (где ${\bf 1}$ -- единичное векторное поле на $\bbR$, направленное в сторону возрастания $t$), или, короче,
\begin{equation}\label{ceq1_1}
    \ds \frac{dx}{dt} = X(x).
\end{equation}

Конечная задача теории динамических систем состоит в том, чтобы дать описание всех фазовых траекторий изучаемой системы. Ясно, что в общем случае такое идеальное решение вряд ли возможно. Поэтому стремятся разбить множество траекторий на какие-либо как можно более узкие классы, указывая их специфические свойства. Естественное разбиение на подобные классы порождают первые интегралы динамической системы, т.е. функции, постоянные вдоль фазовых траекторий.

Пусть система \eqref{ceq1_1} обладает первыми интегралами
\begin{equation}\label{ceq1_2}
    K_1, \ldots, K_n: \mathfrak{M} \to \bbR \,.
\end{equation}
Тогда их совместный уровень
\begin{equation}\label{ceq1_3}
    J_k = \{x \in \mathfrak{M}: K_i(x) = k_i, i = 1, \ldots, n\}
\end{equation}
где $k = (k_1, \ldots, k_n)$ -- набор постоянных, является подмножеством фазового пространства, инвариантным относительно фазового потока системы \eqref{ceq1_1}. Это означает, что всякая траектория, пересекающая множество \eqref{ceq1_3}, целиком в нем и содержится. В частности, векторное поле~$X$ касается множества \eqref{ceq1_3}.

Множества \eqref{ceq1_3} называются интегральными многообразиями системы \eqref{ceq1_1}. Их классификация представляет собой значительный интерес при исследовании динамических систем. Во-первых, поскольку размерность этих многообразий меньше, чем размерность фазового пространства, можно ожидать, что изучение динамической системы на $J_k$ окажется проще, чем изучение исходной системы. Во-вторых, уже сама топологическая структура интегрального многообразия накладывает существенные ограничения на свойства соответствующих фазовых траекторий.

Под задачей изучения фазовой топологии системы \eqref{ceq1_1}, обладающей интегралами \eqref{ceq1_2}, будем понимать задачу задачу классификации топологического и дифференцируемого типа интегральных многообразий \eqref{ceq1_3} по всем наборам постоянных $(k_1, \ldots, k_n) \in \bbR^n$. Основную роль в такой классификации играет так называемое бифуркационное множество -- множество точек $k \in \bbR^n$, при переходе через которые меняется тип интегрального многообразия.

Пусть $f: A \to B$ -- дифференцируемое отображение многообразий. Для каждой точки $y \in B$ обозначим
\begin{equation}\label{ceq1_4}
    f_y = \{x \in A: f(x) = y\}.
\end{equation}
В терминологии расслоений $f_y$ -- слой над точкой $y$.

Отображение $f$ называется локально-тривиальным над точкой $y_0 \in B$, если существует такая окрестность $U$ точки $y_0$ в $B$, что $f^{-1}(U)$ диффеоморфно $f_{y_0} {\times} U$, и этот диффеоморфизм замыкает диаграмму ($p_2$ -- проекция на второй сомножитель):

$$
\begin{CD}
f^{-1}(U) @>{}>> f_{y_0} {\times} U\\
@VfVV @VVp_2V \\
U @>\id{U}>>U
\end{CD}
$$

Бифуркационным множеством отображения $f$ называется множество $\Sigma$ точек $y \in B$, над которыми $f$ не является локально-тривиальным \cite{bib41}.

Очевидно, если $y \in B \backslash \Sigma$, то отображение $f$ регулярно в любой точке $x \in f_y$. Отсюда следует, что множество $\hat{\Sigma}$ критических отображений $f$ содержится в $\Sigma$. Случаи, когда $\hat{\Sigma} \neq \Sigma$, являются, вообще говоря, патологическими и связаны в основном с некомпактностью слоев. В задачах, изучаемых в дальнейшем, такие ситуации не возникают, поэтому условимся сразу рассматривать лишь те отображения, для которых
\begin{equation}\label{ceq1_5}
    \Sigma = \hat{\Sigma}.
\end{equation}

Сделаем одно важное замечание. Пусть $B = B_1 {\times} B_2$. Отображение $f$ можно представить в виде $f = f_1 {\times} f_2$, где $f_i:A \to B_i$ $(i = 1, 2)$. Допустим, что $f_2$ всюду регулярно. Тогда для всех $y_2 \in B_2$ множество $f_2^{-1}(y_2)$ является гладким подмногообразием в $A$ \cite{bib46}. Обозначим через $\Sigma_{y_2}$ бифуркационное множество отображения
\begin{equation*}
    f[y_2] = \left. f_1 \right|_{f_2^{-1}(y_2)} : f_2^{-1}(y_2) \to B_1.
\end{equation*}

\begin{lemma}\label{lem311}
В этой ситуации бифуркационное множество отображения $f$ имеет вид
\begin{equation*}
    \Sigma = \bigcup \limits_{y_2 \in B_2} \Sigma_{y_2} {\times} \{y_2\}.
\end{equation*}
\end{lemma}

С учетом \eqref{ceq1_5} это утверждение легко вытекает из теоремы о неявной функции и в некоторых случаях значительно облегчает нахождение бифуркационных множеств.

Возвращаясь к системе \eqref{ceq1_1}, введем интегральное отображение
\begin{equation}\label{ceq1_6}
    J = K_1 {\times} K_2 {\times} \ldots K_n: \mathfrak{M} \to \bbR^n\,.
\end{equation}
Тем самым, обозначение \eqref{ceq1_3} согласовано с \eqref{ceq1_4}.

Назовем бифуркационным множеством системы \eqref{ceq1_1} с интегралами \eqref{ceq1_2} бифуркационное множество $\Sigma$ отображения \eqref{ceq1_6}.

Точки $\Sigma$ разбивают $\bbR^n$ на открытые связные области, в каждой из которых тип интегрального многообразия сохраняется. После того, как такое разбиение осуществлено, возникает задача установить тип в каждой из упомянутых областей. Дело осложняется тем, что встречающиеся на практике системы содержат числовые параметры, и естественно требовать описания фазовой топологии при всевозможных их значениях. С другой стороны, как мы сейчас увидим, наличие параметров может и упростить задачу.

Рассмотрим динамическую систему
\begin{equation}\label{ceq1_7}
    \ds \frac{dx}{dt} = X(x,c),
\end{equation}
зависящую от набора вещественных параметров $c = (c_1, c_2, \ldots, c_\ell)$ и обладающую интегралами
\begin{equation}\label{ceq1_8}
    K_i(x, c) = k_i, \qquad i = 1, \ldots, n.
\end{equation}

Естественно, что интегральные многообразия и бифуркационные множества зависят от параметров. Поэтому будем писать $J_k(c)$, $\Sigma(c)$. Возникает вопрос, можно ли, зная фазовую топологию при некотором фиксированном $c = c_*$, сказать что-либо о других значениях $c$. Чтобы ответить на него, поступим следующим образом. Дополним уравнения \eqref{ceq1_7} системой
\begin{equation}\label{ceq1_9}
    \ds \frac{dc_j}{dt} = 0, \qquad j = 1, \ldots, \ell,
\end{equation}
а соотношения \eqref{ceq1_8} -- соотношениями
\begin{equation}\label{ceq1_10}
    C_j(x,c) = c_j, \qquad j = 1, \ldots, \ell,
\end{equation}
где отображения $C_j: \mathfrak{M} {\times} \bbR^\ell \to \bbR$ на самом деле от $x$ не зависят, а вектору $c$ сопоставляют его компоненты.
Интегральное отображение расширенной системы \eqref{ceq1_7}, \eqref{ceq1_9} будет таково:
\begin{equation*}
    \widetilde{J} = J {\times} C_1 {\times} \ldots {\times} C_\ell: \mathfrak{M} {\times} \bbR^\ell \to \bbR^{n + \ell}.
\end{equation*}
Его бифуркационное множество обозначим через $\widetilde{\Sigma}$.

Совершенно ясно, что для $(k, c) \in \bbR^{n + \ell}$
\begin{equation}\label{ceq1_11}
    \widetilde{J}_{k,c} = J_k(c).
\end{equation}
Кроме того, отображение $C = C_1 {\times} \ldots {\times} C_\ell$ в силу \eqref{ceq1_10} тождественно по $c$ и, следовательно, всюду регулярно. По лемме \ref{lem311}
\begin{equation}\label{ceq1_12}
    \widetilde{\Sigma} = \bigcup \limits_{c \in \bbR^\ell} \Sigma(c) {\times} \{c\},
\end{equation}
так что дополнительных трудностей в построении $\widetilde{\Sigma}$ не возникает. В то же время тип интегральных многообразий сохраняется в связной подобласти $\bbR^{n + \ell} \backslash \widetilde{\Sigma}$.

\begin{defin}\label{def311}
Пусть $c, c_*$ -- два набора параметров динамической системы \eqref{ceq1_7}, а $U$ и $U_*$ -- связные области в $\bbR^n \backslash \Sigma(c)$ и $\bbR^n \backslash \Sigma(c_*)$ соответственно. Будем говорить, что области $U$ и $U_*$ эквивалентны, если для некоторой пары $k \in U$ и $k_* \in U_*$ найдется непрерывный путь в $\bbR^{n + \ell} \backslash \widetilde{\Sigma}$, соединяющий $(k, c)$ и $(k_*, c_*)$.
\end{defin}

Из определения следует, что эквивалентные области можно непрерывно деформировать одна в другую в $\bbR^{n+\ell}$, минуя $\widetilde{\Sigma}$.

\begin{theorem}\label{the312}
Если наборы констант интегралов \eqref{ceq1_8} $k$ и $k_*$ лежат в эквивалентных областях, то соответствующие интегральные многообразия диффеоморфны:
\begin{equation}\label{ceq1_13}
    J_k(c) \cong J_{k_*}(c_*).
\end{equation}
\end{theorem}

\begin{proof}
Согласно \eqref{ceq1_12} эквивалентные области принадлежат сечениям плоскостями $c = \cons$ одной и той же связной компоненты в $\bbR^{n+\ell} \backslash \widetilde{\Sigma}$. Поэтому в условиях теоремы $\widetilde{J}_{k,c} \cong \widetilde{J}_{k_*, c_*}$, так что \eqref{ceq1_13} вытекает из равенства \eqref{ceq1_11}.\end{proof}

Таким образом, чтобы установить тип $J_k(c)$ в некоторой области $\bbR^n \backslash \Sigma(c)$, достаточно проследить непрерывное изменение этой области до каких-либо значений параметров $c_*$, для которых интегральные многообразия легко устанавливаются. При этом, конечно, необходимо быть уверенным, что при значениях $c_*$ изучаемая область не вырождается.

Приведенные здесь элементарные соображения неявно использовались практически во всех исследованиях по фазовой топологии \cite{bib23,bib26,bib51,bib52}.

\subsection{Области возможности движения в механических системах}\label{ssec32}

Вопрос об исследовании типов интегральных многообразий $J_k$, определенных в предыдущем параграфе, во многих весьма важных случаях может быть решен с помощью классификации их проекций на такое многообразие, размерность которого, вообще говоря, не превосходит самих $J_k$.

В задачах механики подобное многообразие возникает естественным образом. Это -- пространство фазовых переменных, зависящих лишь от положения (конфигурации) реальной механической системы. Такие переменные назовем конфигурационными. Подчеркнем, что конфигурация реальной системы не обязана полностью определяться заданием этих переменных. Так, например, в динамике твердого тела в качестве конфигурационных переменных сохраняют обычно лишь компоненты в подвижных осях одного неизменного в пространстве вектора. Тем самым положение тела определено с точностью до поворотов вокруг соответствующей оси.
Исследование проекций фазовых траекторий на пространство конфигурационных переменных само по себе представляет интерес, поскольку эти проекции непосредственно связаны с движениями реальной системы.

Сформулируем теперь основные понятия в достаточно общем виде.

Пусть фазовое пространство $\mathfrak{M}$ динамической системы \eqref{ceq1_1} служит тотальным пространством локально-тривиального расслоения
\begin{equation}\label{ceq2_1}
    \pi:\mathfrak{M} \to \mathfrak{N}
\end{equation}
со слоем $F$, и пусть по-прежнему система \eqref{ceq1_1} обладает интегралами \eqref{ceq1_2}, т.е. определены интегральное многообразие \eqref{ceq1_6} и множества \eqref{ceq1_3}.
\begin{defin}\label{def321}
Областью возможности движения (ОВД) на многообразии $\mathfrak{N}$, соответствующей постоянным $k = (k_1, \ldots, k_n)$ первых интегралов, назовем образ интегрального многообразия \eqref{ceq1_3} при отображении~\eqref{ceq2_1}.
\end{defin}

Будем иметь дело со следующими основными ситуациями:

А) Натуральная механическая система с конфигурационным пространством $\mathcal{M}$ размерности ${\dim \mathcal{M} = m}$. Положим тогда
\begin{equation}\label{ceq2_2}
    \mathfrak{N} = \mathcal{M}, \qquad \mathfrak{M} = \tm, \qquad \pi = p_\mathcal{M}.
\end{equation}
Слой $F_x$ над точкой $z \in \mathfrak{N}$ есть $m$-мерное пространство $T_x\mathcal{M}$, так что $F = \bbR^m$.

 Б) Механическая система с симметрией. Предположим, что на конфигурационном пространстве $\mathcal{M}$ действует группа Ли $\Psi$ таким образом, что определено фактор-многообразие $\widetilde{\mathcal{M}} = \mathcal{M} / \Psi$ и соответствующее главное расслоение \cite{bib09} $p: \mathcal{M} \to \widetilde{\mathcal{M}}$. Пусть вдобавок преобразования, касательные к преобразованиям группы, сохраняют векторное поле на $\tm$, задающее динамику рассматриваемой системы\footnote[1]{Здесь мы не делаем никаких предположений относительно структуры системы. Весьма сложно может быть устроена и группа симметрий.}. Тогда это поле факторизуется по действию $\Psi$, и естественным образом определена динамическая система $X$ на $\tm / \Psi$. Это пространство представляет собой векторное расслоение над $\widetilde{\mathcal{M}}$ -- сумму Уитни касательного расслоения $T\widetilde{\mathcal{M}}$ и расслоения-произведения со слоем, изоморфным алгебре Ли группы $\Psi$. Положим
\begin{equation}\label{ceq2_3}
    \mathfrak{M} = \tm / \Psi, \qquad \mathfrak{N} = \widetilde{\mathcal{M}}, \qquad \pi = (p \circ p_\mathcal{M})/\Psi.
\end{equation}
Последняя запись означает, что каноническая проекция $\pi$ расслоения $\tm / \Psi$ над $\widetilde{\mathcal{M}}$ индуцирована $\Psi$-сохранным отображением $p \circ p_\mathcal{M}: \tm \to\widetilde{\mathcal{M}}$. Как и в ситуации (А), нетрудно убедиться, что слой $F$ изоморфен $\bbR^m$. В качестве $F_z = \pi^{-1}(z)$ может быть рассмотрено любое $T_x\mathcal{M}$, где $p(x) = z$ (факторизация $\tm / \Psi$ <<склеивает>> все касательные пространства в точках орбиты в одно, после чего сопоставляет векторам полученного пространства их <<точку приложения>> -- орбиту группы, т.е. элемент $\widetilde{\mathcal{M}}$). В силу этого все элементы $F$ называем ниже векторами.

Отметим, что определение фазового пространства согласно \eqref{ceq2_3} предполагает $\Psi$-сох\-ран\-ность первых интегралов исходной системы на $\tm$, что выполнено далеко не всегда. В случае однопараметрической группы $\Psi$, которой отвечает глобальный интеграл площадей $G$, это по предложению \ref{pro138} означает, что в рассмотрение можно включать лишь интегралы, находящиеся в инволюции с $G$.

По предположению каждая точка $\mathfrak{N}$ обладает окрестностью $U$, такой, что уравнения, задающие множество $J_k \bigcap \pi^{-1}(U)$, могут быть записаны в виде
\begin{equation}\label{ceq2_4}
    K_i(z,w) = k_i, \qquad i = 1, \ldots,n,
\end{equation}
где $w \in F$, а $z \in U$.

Обозначим через $U_k = \pi(J_k)$ ОВД, соответствующую набору $k$.
\begin{propos}\label{pro322}
Следующие утверждения равносильны:

а{\rm )} точка $z$ принадлежит $U_k${\rm ;}

б{\rm )} через точку $z$ проходит хотя бы одна проекция фазовой траектории с заданными постоянными $(k_1, \ldots, k_n)$ первых интегралов{\rm ;}

в{\rm )} множество решений системы \eqref{ceq2_4} относительно $w$ при заданном $z$ не пусто.
\end{propos}

Это непосредственно следует из определения \ref{def321}.

Вектор $w$, удовлетворяющий \eqref{ceq2_4} при заданном $z$, назовем допустимой скоростью в точке $z$. В случае \eqref{ceq2_2} он действительно является скоростью движения вдоль проекции фазовой траектории на конфигурационное пространство.

Ясно, что ОВД полностью определена распределением допустимых скоростей на $\mathfrak{N}$. Классификацию областей возможности движения с указанием структуры множества допустимых скоростей в каждой точке называют геометрическим анализом механической системы. Будем пользоваться этим термином, хотя он и не совсем удачен.

Пусть
\begin{equation}\label{ceq2_5}
    \pi_k: J_k \to \mathfrak{N}
\end{equation}
есть ограничение на $J_k$ отображения \eqref{ceq2_1}. Ясно, что множество допустимых скоростей в точке $z$ диффеоморфно $\pi_k^{-1}(z)$. Таким образом, геометрический анализ сводится к классификации прообразов точек $\mathfrak{N}$ при отображении \eqref{ceq2_5}, зависящем от параметров $k$. Следовательно, и здесь большую роль играет понятие бифуркации.

\begin{defin}\label{def323}
Назовем обобщенной границей ОВД $U_k = \pi(J_k)$ бифуркационное множество отображения \eqref{ceq2_5}. Обозначим это множество через $\hat{\partial}U_k$.
\end{defin}

Пусть $k \notin \Sigma$. Тогда $J_k$ -- гладкое многообразие, и $\hat{\partial}U_k$ есть образ множества критических точек гладкого отображения $\pi_k$. В теории особенностей гладких отображений есть термин <<каустика>>, однако он имеет и специфическое значение \cite{bib05}, поэтому в столь широком смысле применять его не будем.

Пусть $k \in \Sigma$. Тогда $J_k$ содержит в себе некоторое множество $\delta_k$ критических значений интегрального отображения $J$. В окрестности $\delta_k$ множество $J_k$ не является многообразием общего положения. Здесь либо падает его размерность, либо нарушается гладкость или локально-евклидова структура. Поэтому над точками $\pi_k(\delta_k)$ отображение \eqref{ceq2_5} не является локально-тривиальным и всегда $\pi_k(\delta_k) \subset \hat{\partial}U_k$. Остальная часть $\hat{\partial}U_k$ получается , как и в предыдущем случае, при отображении гладких участков $J_k$.

Геометрический анализ проводится, таким образом, в два этапа. На первом строится множество $\hat{\partial}U_k$. Если
\begin{equation}\label{ceq2_6}
    \dim J_k \gs \dim \mathfrak{N}
\end{equation}
(а только такой случай и имеет смысл рассматривать), то $\hat{\partial}U_k$ делит $\mathfrak{N}$ на открытые компоненты, в которых тип множества допустимых скоростей сохраняется. На втором этапе этот тип устанавливается в каждой из упомянутых компонент и в точках $\hat{\partial}U_k$.

Укажем путь построения обобщенной границы ОВД. Пусть $x \in J_k \subset \mathfrak{M}$. В локальном представлении $x = (z, w)$, где $ z \in \mathfrak{N}, w \in F$, касательное пространство к $\mathfrak{M}$ в точке $x$ изоморфно прямой сумме
\begin{equation}\label{ceq2_7}
    T_x\mathfrak{M} \cong T_z\mathfrak{N} \oplus T_w F,
\end{equation}
а отображение, касательное к \eqref{ceq2_1},
\begin{equation*}
    T_x\pi : T_z\mathfrak{N} \oplus T_wF \to T_z\mathfrak{N}
\end{equation*}
становится проекцией на первое слагаемое. Очевидно, $T_x \left( \left. \pi \right|_{J_k} \right) = \left. (T_x \pi) \right|_{T_xJ_k}$, и отображение \eqref{ceq2_5} регулярно в точке $x$ тогда и только тогда, когда в представлении \eqref{ceq2_7}
\begin{equation*}
    T_{(z,w)}J_k + T_wF = T_x\mathfrak{M}.
\end{equation*}
Иначе говоря, регулярность $\pi_k$ в точке $x \in J_k$ означает, что в этой точке подмногообразие $J_k$ трансверсально в $\mathfrak{M}$ слою $\pi^{-1} {\times} (\pi(x)) = F_{\pi(x)}.$

Отметим один факт из линейной алгебры. Пусть $V_1, V_2, V$ -- векторные пространства, $\ell_i : V \to V_i$ -- эпиморфизмы. Следующие условия равносильны:

а) ограничение $\ell_1$ на ядро $\ell_2$ есть эпиморфизм;

б) ограничение $\ell_2$ на ядро $\ell_1$ есть эпиморфизм.

Действительно, оба свойства имеют место в том и только в том случае, когда $\ker \ell_1 + \ker \ell_2 = V$.

\begin{lemma}\label{lem324}
Пусть $X, Y, V$ -- векторные пространства, $E: X {\times} Y \to V$ -- дифференцируемое отображение, $v \in V$ -- его регулярное значение, $p_v$ -- ограничение на $E^{-1}(v)$ естественной проекции $p: X {\times} Y \to X$. Точка $(x, y) \in E^{-1}(v)$ является критической для отображения $p_v$ тогда и только тогда, когда
\begin{equation}\label{ceq2_8}
    \rang \left[ \left. E'(x, y) \right|_Y \right] < \dim V.
\end{equation}
\end{lemma}

\begin{proof}
Касательное пространство в точке $(x, y)$ к $E^{-1}(v)$ есть ядро $E'(x, y)$, пространство $Y$ -- ядро проекции $p$, отображение $p'_v(x, y)$ -- ограничение $p$ на касательное пространство к $E^{-1}(v)$. Поэтому точка $(x, y)$ -- критическая для $p$, если
\begin{equation*}
    \rang \left[ \left. p \right|_{\ker E'(x, y)} \right] < \dim X.
\end{equation*}
В силу сказанного выше это неравенство равносильно \eqref{ceq2_8}.
\end{proof}

Возвращаясь к построению обобщенной границы, обозначим
\begin{equation*}
    \frac{\partial J}{\partial z}: T_z\mathfrak{N} \to T_k \bbR^n = \bbR^n, \quad \frac{\partial J}{\partial w}: T_wF \to T_k \bbR^n = \bbR^n
\end{equation*}
ограничения оператора $T_xJ$ на слагаемые \eqref{ceq2_7}.

\begin{theorem}\label{the325}
Образ точки $x \in J_k$ принадлежит обобщенной границе ОВД тогда и только тогда, когда в этой точке
\begin{equation}\label{ceq2_9}
    \rang \frac{\partial J}{\partial w} < n.
\end{equation}
\end{theorem}
\noindent Напомним, что $n$ -- размерность пространства, в которое $J$ отображает $\mathfrak{M}$.

\begin{proof}
Если $x$ -- регулярная точка интегрального отображения, то утверждение вытекает из леммы \ref{lem324}. Образ критической точки $J$, как отмечалось, всегда принадлежит обобщенной границе, но в такой точке $\rang T_xJ < n$, и тем более выполнено \eqref{ceq2_9}.
\end{proof}

\begin{remark}\label{rem321}
Пусть в окрестности точки $x$ многообразие $J_k$ задано явными уравнениями
\begin{equation*}
\begin{array}{ll}
    z_i = z_i(t_1, \ldots, t_\ell), & i = 1, \ldots, s, \\
    w_j = w_j(t_1, \ldots, t_\ell), & j = 1, \ldots, m,
\end{array}
\end{equation*}
где $z_1, \ldots, z_s$ -- локальные координаты на $\mathfrak{N}$, $w_1, \ldots, w_m$ -- локальные координаты на $F$, а точка $(t_1, \ldots, t_\ell)$ пробегает открытое множество в $\bbR^ \ell$. Согласно \eqref{ceq2_6} следует считать, что ${\ell \gs s}$. Непосредственно проверяется, что условие \eqref{ceq2_9} равносильно, как и следовало ожидать, неравенству
\begin{equation}\label{ceq2_10}
    \rang \frac{\partial (z_1, \ldots, z_s)}{\partial (t_1, \ldots, t_\ell)} < s.
\end{equation}
\end{remark}

В конкретных задачах условия \eqref{ceq2_9} или \eqref{ceq2_10} записываются некоторой системой уравнений (которую в первом случае необходимо дополнить системой \eqref{ceq2_4}), описывающей множество критических точек отображения \eqref{ceq2_5}. Образ этого множества есть обобщенная граница ОВД.

\subsection{Примеры перестроек областей возможности движения}\label{ssec33}

Рассмотрим причины изменения типа области возможности движения. Первой причиной является перестройка самих интегральных многообразий $J_k$ в тех случаях, когда точка $k$ в $\bbR^n$ пересекает бифуркационное множество $\Sigma$ отображения $J$. Изменения в ОВД определяются возникновением образа множества критических точек $J$. Вторая причина -- существование различных типов проекций диффеоморфных между собой интегральных многообразий -- более трудна для изучения, так как связана с глобальными свойствами множества особенностей гладкого отображения. Его исследованию в конкретных задачах посвящены следующие главы. Пока же разберем простые случаи, когда подобных трудностей не возникает.
\begin{example}\label{exa331}
(Система Лиувилля). Ограничимся случаем двух степеней свободы и рассмотрим систему с конфигурационным пространством $\bbR^2(x_1, x_2)$, определяемую функцией Лагранжа \cite{bib07}
\end{example}
\begin{equation*}
    \ds L=\frac{1}{2}[\lambda_1(x_1) + \lambda_2(x_2)][\alpha_1(x_1) \dot{x}_1^2 + \alpha_2(x_2) \dot{x}_2^2] - \frac{\varphi_1(x_1) + \varphi_2(x_2)}{\lambda_1(x_1) + \lambda_2(x_2)},
\end{equation*}
где $\lambda_i, \alpha_i, \varphi_i$ -- произвольные функции, подчиненные условиям\footnote[1]{В системах, конфигурационное пространство которых есть двумерное многообразие, условия \eqref{ceq3_1} могут нарушаться. При этом возникают так называемые естественные границы изменения переменных \cite{bib01}. Подобные случаи здесь для простоты исключены.}
\begin{equation}\label{ceq3_1}
    \lambda_1(x_1) + \lambda_2(x_2) \neq 0, \qquad \alpha_1(x_1) > 0, \qquad \alpha_2(x_2) > 0.
\end{equation}
Существует два независимых первых интеграла:
\begin{equation}\label{ceq3_2}
\begin{array}{c}
    \ds H = \frac{1}{2}(\lambda_1 + \lambda_2)(\alpha_1 \dot{x}_1^2 + \alpha_2 \dot{x}_2^2) + \frac{\varphi_1 + \varphi_2}{\lambda_1 + \lambda_2}, \\
    \ds K = \frac{1}{2}(\lambda_1 + \lambda_2)(\lambda_2 \alpha_1 \dot{x}_1^2 - \lambda_1 \alpha_2 \dot{x}_2^2) + \frac{\varphi_1 \lambda_2 - \varphi_2 \lambda_1}{\lambda_1 + \lambda_2}.
\end{array}
\end{equation}
Это соответствует случаю \eqref{ceq2_2}:
\begin{equation*}
\begin{array}{c}
    \mathcal{M} = \bbR^2(x_1, x_2), \qquad \tm = \bbR^4(x_1, x_2, \dot{x}_1, \dot{x}_2),\\[3mm]
    p_\mathcal{M}(x_1, x_2, \dot{x}_1, \dot{x}_2) = (x_1, x_2),\\[3mm]
    J = H {\times} K :\tm \to \bbR^2(h, k), \qquad F = p_\mathcal{M}^{-1}(x_1, x_2) = \bbR^2(\dot{x}_1, \dot{x}_2).
\end{array}
\end{equation*}
Для большинства $(h, k)$ $\dim J_{h,k} = 2 = \dim \mathcal{M}$, так что выполнено условие \eqref{ceq2_6}.

Найдем области возможности движения. Решая систему $H=h, K=k$ относительно $\dot{x}_i$, имеем
\begin{equation}\label{ceq3_3}
    \ds \frac{1}{2}[\lambda_1(x_1) + \lambda_2(x_2)]\alpha_i(x_i) \dot{x}_i^2 + \Lambda_i(x_i; h, k) = 0,
\end{equation}
где $\Lambda_i(x_i, h, k) = \varphi_i(x_i) - h \lambda_i(x_i) + (-1)^i k$, $i = 1, 2$.
Таким образом, ОВД определяется неравенствами
$$
\Lambda_1(x_1; h, k) \ls 0, \qquad \Lambda_2(x_2; h, k) \ls 0.
$$
Положим
\begin{equation}\label{ceq3_4}
    U_i = \{x_i:\Lambda_i(x_i; h, k)\ls 0\} \subset \bbR, \qquad i = 1, 2.
\end{equation}
Тогда
\begin{equation}\label{ceq3_5}
    U_{h,k} = U_1 {\times} U_2.
\end{equation}
Ее обобщенная граница согласно \eqref{ceq2_9} определяется условием
\begin{equation*}
    \ds \rang \frac{\partial (H, K)}{\partial (\dot{x}_1, \dot{x}_2)} < 2,
\end{equation*}
или в явном виде
\begin{equation*}
    [\lambda_1(x_1) + \lambda_2(x_2)]\alpha_1(x_1) \alpha_2(x_2) \det \left(
    \begin{array}{cc}
    \dot{x}_1 & \dot{x}_2 \\
    \lambda_2(x_2) \dot{x}_1 & -\lambda_1(x_1) \dot{x}_2
    \end{array}
    \right)=0.
\end{equation*}
Это условие выполняется в двух случаях: при $\dot{x}_1=0$ по формуле \eqref{ceq3_3} имеем $\Lambda_1(x_1; h, k) = 0$, при $\dot{x}_2 = 0$ -- соответственно $\Lambda_2(x_2; h, k) = 0$. Если $\Lambda_i$ как функции $x_i$ не имеют кратных корней, то обобщенная граница $\hat{\partial}U_{h,k}$ совпадает с топологической $\partial U_{h,k}$.

Из равенств \eqref{ceq3_4}, \eqref{ceq3_5} следует, что тип ОВД может изменяться лишь при переходе через такие значения $h, k$, когда хотя бы одно из $\Lambda_i$ имеет кратный корень. Непосредственно проверяется, что в этом и только этом случае градиенты функций \eqref{ceq3_2} линейно зависимы, т.е. $(h, k) \in \Sigma$, так что перестройки ОВД соответствуют точкам бифуркационного множества.

Если множества \eqref{ceq3_4} ограничены, то каждая связная компонента ОВД при $(h, k) \notin \Sigma$ есть прямоугольник, внутри которого четыре допустимые скорости, на сторонах -- две, а в вершинах -- одна (нулевая). Траектории в прямоугольнике подобны фигурам Лиссажу. Компонента интегрального многообразия -- двумерный тор с условно-периодическими движениями. Несложно разобрать и все возможные критические случаи, опираясь, например, на полный анализ одномерных систем вида \eqref{ceq3_3}, выполненный в \cite{bib01}. Один из вариантов (наиболее типичный) показан на рис.~\ref{fig_3_1}. Здесь $U_1 = [a_1, b_1]$, $U_2 = [a_2, b_2]$, а $\Lambda_1$ имеет кратный корень $x_1^* \in (a_1, b_1)$. Цифры на рис.~\ref{fig_3_1},\,{\it а} обозначают число допустимых скоростей, сплошные линии -- обобщенную границу (она шире топологической). На рис.~\ref{fig_3_1},\,{\it б} приведены проекции фазовых траекторий. Интегральное многообразие представляет собой произведение <<восьмерки>> на окружность (рис.~\ref{fig_3_2}). Его средняя линия есть $\alpha$- и $\omega$-предельный цикл для остальных фазовых траекторий.

\begin{figure}[ht]
\def\fs{0.4}
\begin{minipage}[h]{\fs\linewidth}
\center{\includegraphics[width=\linewidth]{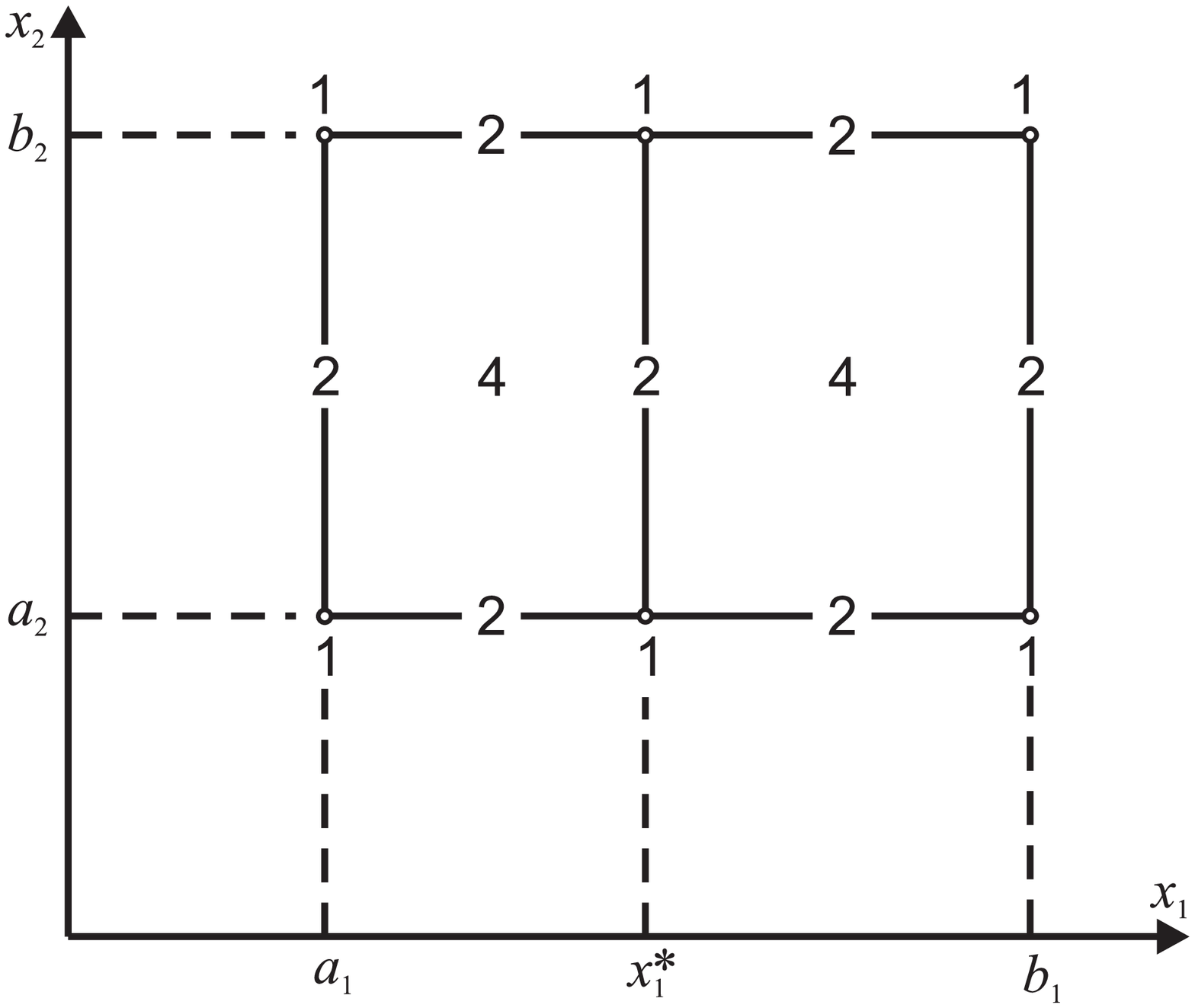} \\ а)}
\end{minipage}
\hfill
\begin{minipage}[h]{\fs\linewidth}
\center{\includegraphics[width=\linewidth]{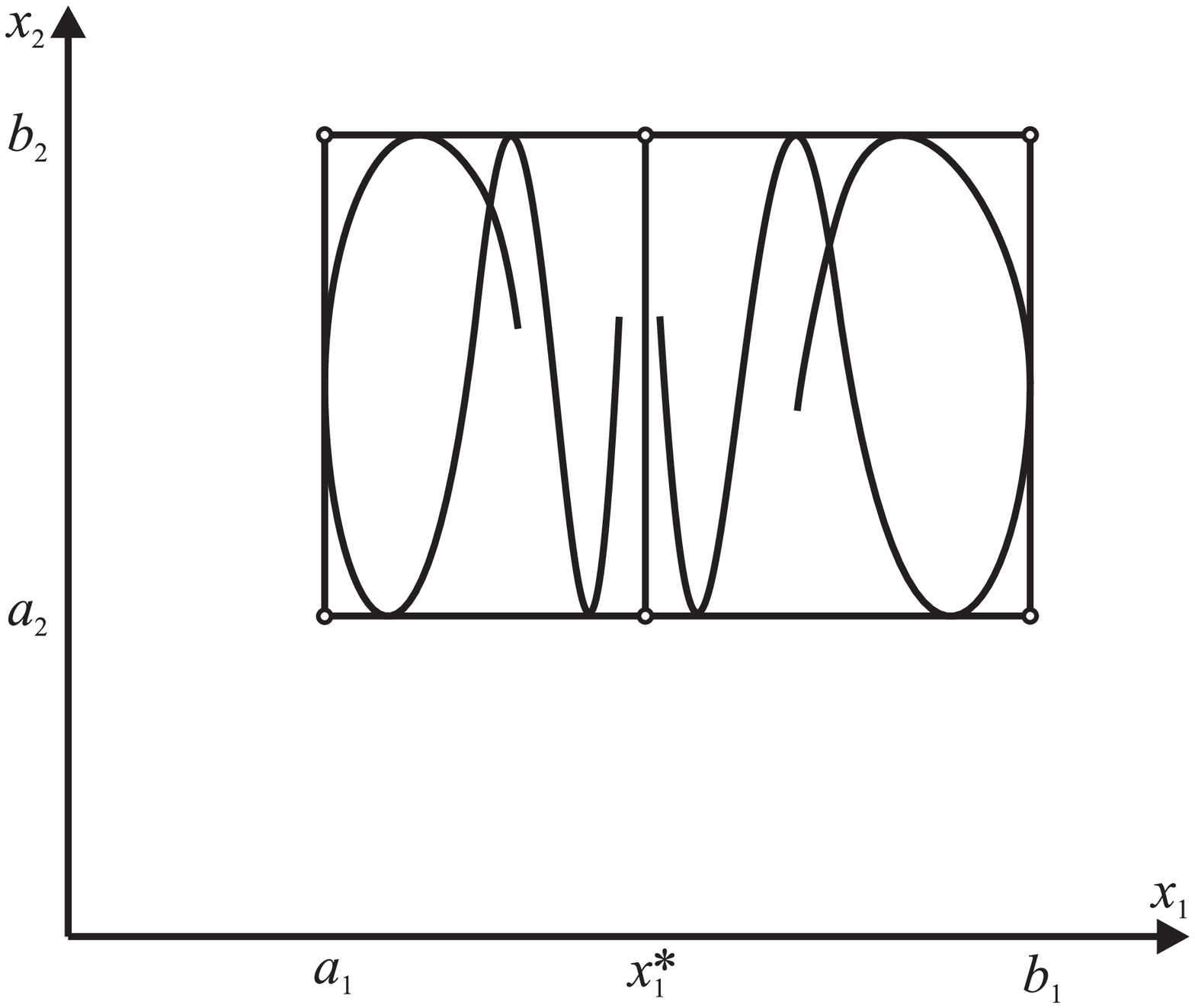} \\ б)}
\end{minipage}
\caption{}\label{fig_3_1}
\end{figure}

\begin{figure}[ht]
\def\fs{0.4}
\centering
\includegraphics[width=\fs\linewidth,keepaspectratio]{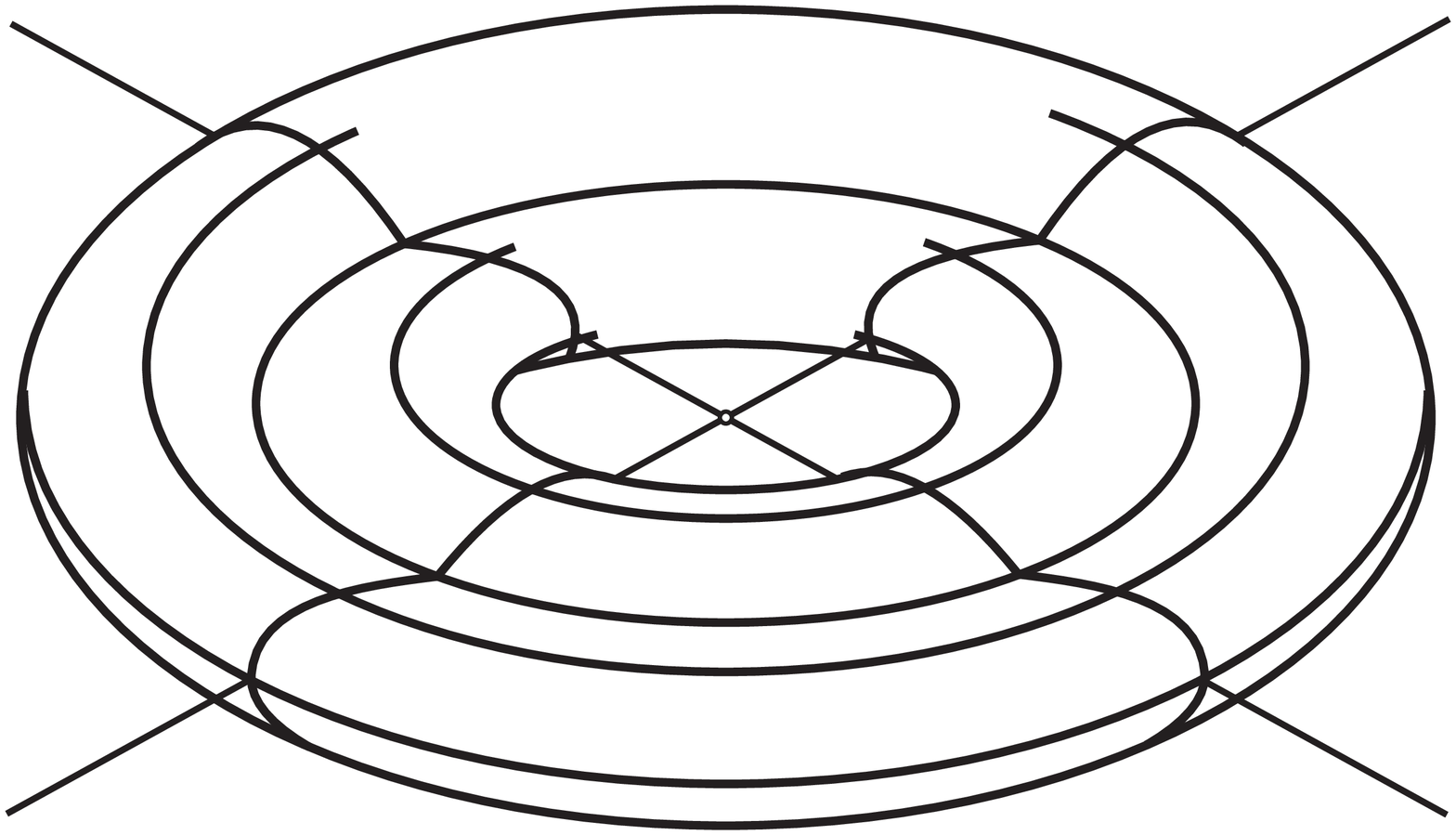}
\caption{}\label{fig_3_2}
\end{figure}

Интересные примеры систем Лиувилля на сфере дают геодезический поток на эллипсоиде \cite{bib22}, случай Эйлера\,--\,Пуансо и случай Клебша с нулевой постоянной площадей \cite{bib56,bib60}.

\begin{example}\label{exa332}
Рассмотрим механическую систему с гироскопическими силами, определенную четверкой \eqref{aeq2_1}, имеющую однопараметрическую группу симметрий \eqref{aeq3_1}, которая порождает глобальный интеграл площадей. Перенесем в эту ситуацию некоторые результаты С.\,Смейла \cite{bib41}, полученные для натуральных систем. Экономный вариант теории Смейла изложен в \cite{bib48,bib52}.
\end{example}

Данные условия соответствуют ситуации (Б) предыдущего параграфа, поэтому примем обозначения \eqref{ceq2_3}.

Как отмечалось, факторизация по $\Psi$ пространства $\tm$ устанавливает изоморфизм $\is_x : T_x\mathcal{M} \to F_z$ для всех $z \in \mathfrak{N}$ и $x \in p^{-1}(z)$, причем если $x_1 = \psi^\tau(x_2)$, то коммутативна диаграмма
$$
\begin{CD}
T_{x_1}\mathcal{M} @>T\psi^\tau>> T_{x_2}\mathcal{M}\\
@V\is_{x_1}VV @VV\is_{x_2}V \\
F_z @>\id{F_z}>>F_z
\end{CD}
$$
Поэтому соотношение
\begin{equation}\label{ceq3_6}
    \langle w_1, w_2 \rangle_z = \langle \is_x^{-1}(w_1), \is_x^{-1}(w_2)\rangle_x
\end{equation}
корректно определяет скалярное произведение на $F_z$, индуцированное $\Psi$-со\-хран\-ной римановой метрикой на $\mathcal{M}$.

Заметим, что порождающее векторное поле \eqref{aeq3_3} сохраняется группой $\Psi$:
\begin{equation*}
    \ds T\psi^t(v(x)) = T\psi^t \left. \frac{d}{d\tau} \right|_{\tau=0}\psi^\tau(x) = \left. \frac{d}{d\tau} \right|_{\tau=0}\psi^{t+\tau}(x) = v(\psi^\tau(x)),
\end{equation*}
и, следовательно, для любого $z \in \mathfrak{N}$ корректно определен вектор $v(z) \in F_z$, удовлетворяющий равенству
\begin{equation}\label{ceq3_7}
    \is_x(v(x)) = v(z).
\end{equation}
Обозначим через $V_z \subset F_z$ прямую\footnote[1]{Из наличия главного расслоения $p:\mathcal{M} \to \mathfrak{N}$ следует, что $v(x)$, а значит, и $v(z)$ нигде в нуль не обращаются.}, натянутую на $v(z)$, а через $N_z$ -- ортогональное дополнение к $V_z$ в структуре \eqref{ceq3_6}:
\begin{equation}\label{ceq3_8}
    F_z = N_z \oplus V_z.
\end{equation}

Из равенств \eqref{aeq3_9}, \eqref{aeq3_21} вытекает $\Psi$-сохранность функций $\Pi$ и $f$ на $\mathcal{M}$. Следовательно, они индуцируют гладкие функции на $\mathfrak{N}$, для которых оставим прежние обозначения. То же относится к функциям \eqref{aeq2_3} и \eqref{aeq3_6} (см. \eqref{aeq3_10}, \eqref{aeq3_20}), порождающим первые интегралы фактор-системы на $\mathfrak{M}$. Запишем их согласно \eqref{ceq3_6}, \eqref{ceq3_7} в виде
\begin{eqnarray}
& &    \ds H(z,w) = \frac{1}{2}\langle w, w \rangle_x + \Pi(z), \label{ceq3_9}\\
& &    G(z,w) = \langle v(z), w \rangle_z + f(z),\label{ceq3_10}
\end{eqnarray}
где $z \in \mathfrak{N}$, $w \in F_z$.

Интегральное отображение $J = H {\times} G : \mathfrak{M} \to \bbR^2$. Если $\dim(\mathcal{M}) = m$, то $\dim(\mathfrak{N}) = m - 1$, и размерность интегрального многообразия в общем случае $\dim J_{h,g} = \dim \mathfrak{M} - 2 = 2m - 3$.
Условие \eqref{ceq2_6} выполняется при $m \gs 2$. Это ограничение естественно, так как оно означает лишь, что $\dim \mathcal{M} > \dim \Psi$.

Фиксируя $z \in \mathfrak{N}$, $g \in \bbR$, обозначаем $G_g(z) = G^{-1}(g) \cap F_z$. Из записи \eqref{ceq3_10} следует, что $G_{f(z)}(z) = N_z$, а при произвольном $g$ множество $G_g(z)$ есть гиперплоскость в $F_z$, параллельная $N_z$. Поэтому существует единственный вектор $v_g(z)$, принадлежащий $G_g(z)$ и ортогональный $N_z$. Легко видеть, что
\begin{equation}\label{ceq3_11}
    v_g(z) = [g - f(z)]\|v(z)\|^{-2}v(z)
\end{equation}
и, в частности, соответствие $z \mapsto v_g(z)$ есть гладкое сечение расслоения $\pi: \mathfrak{M} \to \mathfrak{N}$ (конечно, вектор \eqref{ceq3_11} получен факторизацией из \eqref{aeq5_8}).

Эффективным потенциалом назовем функцию $\Pi_g$ на $\mathfrak{N}$, определенную равенством $\Pi_g = H \circ v_g$ (см. $\S$~\ref{ssec15}). Из \eqref{ceq3_9}, \eqref{ceq3_11} находим явное выражение
\begin{equation}\label{ceq3_12}
    \ds \Pi_g(z) = \Pi(z) + \frac{1}{2} [g - f(z)]\|v(z)\|^{-2}.
\end{equation}
Аналогично \cite{bib41} устанавливаем следующие факты.

\begin{propos}\label{pro333}
Эффективный потенциал обладает свойствами:

а{\rm )} $\Pi_g(z) = \min\limits_{w \in G_g(z)}H(z,w)${\rm ;}

б{\rm )} область возможности движения $U_{h,g} = \Pi_g^{-1}(-\infty, h]${\rm ;}

в{\rm )} множество критических точек функции $\left. H \right|_{G^{-1}(g)}$ совпадает с $v_g(\delta)$, где $\delta \subset \mathfrak{N}$ -- множество критических точек эффективного потенциала{\rm ;}

г) бифуркационное множество $\Sigma = \{(h, g): h \in \Pi_g(\delta)\}$.
\end{propos}

\begin{proof}
Свойство (а) следует из определений  и положительной определенности скалярного произведения на $F_z$.

Заметим, что $\max\limits_{w \in G_g(z)}H(z,w) = +\infty$, так что разрешимость уравнений $G = g, H = h$ относительно $w$ равносильна в силу (а) неравенству $\Pi_g(z) \ls h$, чем доказано (б).

Пусть $w \in G_g(z)$. Тогда, очевидно, $w = w^0 + v_g(z)$, где $w^0 \in N_z$. Из ортогональности слагаемых в \eqref{ceq3_8} и записи \eqref{ceq3_12} получаем
\begin{equation}\label{ceq3_13}
    \ds H(z, w) = \frac{1}{2} \|w^0\|^2 + \Pi_g(z).
\end{equation}

Обозначим через $\Gamma_z: F_z \to F_z^*$ оператор <<опускания индексов>> в структуре \eqref{ceq3_6}. Тогда
\begin{equation*}
    \ds \frac{\partial}{\partial w^0} \left( \left. H \right|_{G^{-1}(g)}\right) = \Gamma_z w^0, \qquad \frac{\partial}{\partial z} \left( \left. H \right|_{G^{-1}(g)}\right) = \frac{\partial \Pi_g}{\partial z} + \frac{1}{2} \frac{\partial \Gamma_z}{\partial z}w^0 \cdot w^0,
\end{equation*}
и равенство $d\left( \left. H \right|_{G^{-1}(g)}\right) = 0$ равносильно тому, что $w^0 = 0$, $z \in \delta$. Следовательно, $w = v_g(z) \in v_g(\delta)$ -- условие критичности функции $H_{G^{-1}(g)}$. Это соответствует утверждению (в).

Для доказательства (г) вспомним предположение \eqref{ceq1_5}. Утверждение следует теперь из регулярности функции $G$, леммы \ref{lem311} и свойства (в).
\end{proof}

Выясним, что представляет собой обобщенная граница ОВД. Для этого применим теорему~\ref{the325}. Так как $F_z$ здесь -- векторное пространство, естественно отождествить $T_w F_z$ с самим $F_z$. Условие \eqref{ceq2_9} означает тогда, что
\begin{equation*}
    \ds \frac{\partial H}{\partial w} (z, w) = \Gamma_z w  \quad \textrm{и} \quad \frac{\partial G}{\partial w} (z, w) = \Gamma_z v(z)
\end{equation*}
линейно зависимы как элементы $F_z^*$. Но $\Gamma_z$ -- изоморфизм. Следовательно, точка $(z, w) \in J_{h,g}$ проектируется в обобщенную границу ОВД тогда и только тогда, когда $w \in V_z$. С другой стороны, $G^{-1}(g) \cap V_z = \{v_g(z)\}$, и из определения эффективного потенциала находим
\begin{equation}\label{ceq3_14}
    \hat{\partial}U_{h,g} = \{z: \Pi_g(z) = h\}.
\end{equation}

Свойство (г) предложения \ref{pro333} можно сформулировать следующим образом: точка $(h, g) \in \Sigma$ тогда и только тогда, когда на обобщенной границе соответствующей ОВД есть критическая точка эффективного потенциала. Из свойства (б) и теоремы Морса \cite{bib33} вытекает теперь, что перестройки ОВД, как и в примере \ref{exa331}, происходят лишь при пересечении множества $\Sigma$. В точках $(h, g) \notin \Sigma$ обобщенная граница \eqref{ceq3_14} совпадает с топологической границей ОВД.

Топологический тип интегрального многообразия $J_{h,g}$ устанавливается точно так же, как и в случае натуральных систем \cite{bib41}. На фиксированном уровне $G = g$ воспользуемся представлением \eqref{ceq3_13} и запишем уравнение $J_{h,g}$ в виде
\begin{equation*}
    \|w^0\|^2 = 2[h - \Pi_g(z)].
\end{equation*}
Следовательно, $J_{h,g}$ есть расслоение $(m-2)$-мерных сфер над $U_{h,g}$, стянутых в точку над $\hat{\partial}U_{h,g}$.

Таким образом, при наличии лишь интегралов энергии и площадей гироскопические силы не вносят ничего принципиально нового в топологию изучаемой системы.

\subsection{Свойства интегрируемых задач динамики твердого тела}\label{ssec34}

Начиная с настоящего момента, ограничимся практически наиболее важным случаем, когда твердое тело помещено в поле осесимметричных потенциальных сил \eqref{beq4_8}, а гироскопические силы заданы 2-формой $\varkappa$ в виде \eqref{beq2_11}, удовлетворяющей условию \eqref{beq4_16}. По теореме \ref{the241} в дополнение к интегралу энергии существует интеграл площадей, соответствующий действию группы $\Psi = S^1$ с левыми сдвигами \eqref{beq4_2} на конфигурационном пространстве $\mathcal{M} = SO(3)$.

Из \eqref{aeq3_10} и предложения \ref{pro134} следует, что определяющее динамику гамильтоново векторное поле на $\tm$ сохраняется действиями группы. Следовательно, имеет место ситуация (Б) $\S$~\ref{ssec32}. Выясним, что представляют собой объекты \eqref{ceq3_2}.

Отображение факторизации $\mathcal{M}$ по действию $\Psi$ есть \eqref{beq4_6}. При этом, как указано, $\mathcal{M} / \Psi = S^2$ -- сфера Пуассона.

Воспользуемся изоморфизмом \eqref{beq1_15} $\tm = \mathcal{M} {\times} \bbRR$. Компоненты внутренней угловой скорости по лемме \ref{rem321} сохраняются левыми сдвигами, поэтому действие $\Psi$ на $\mathcal{M} {\times} \bbRR$ совпадает с действием $\Psi$ на $\mathcal{M}$ и тривиально на $\bbRR$. Таким образом, фазовое пространство профакторизованной системы будет
\begin{equation}\label{ceq4_1}
    \mathfrak{M} = \tm / \Psi = (\mathcal{M} / \Psi) {\times} \bbRR = S^2 {\times} \bbRR.
\end{equation}
Многообразие конфигурационных переменных
\begin{equation}\label{ceq4_2}
    \mathfrak{N} = S^2,
\end{equation}
а проекция $\pi = (p \circ p_\mathcal{M})/ \Psi$ определена естественным образом:
\begin{equation}\label{ceq4_3}
    \pi: \mathfrak{M} \to \mathfrak{N}, \qquad \pi({\boldsymbol \nu}, {\boldsymbol\omega}) = {\boldsymbol\nu}.
\end{equation}

Динамическую систему $X$ на $\mathfrak{M}$ определяют уравнения Эйлера\,--\,Пуассона (замкнутая система, отделяющаяся от \eqref{beq3_11}, \eqref{beq3_15} в предположениях \eqref{beq4_8}, \eqref{beq4_15}):
\begin{equation}\label{ceq4_4}
    {\bf A} \astup{\omega} + \mbs{\omega} {\times} ({\bf A}\mbs{\omega} + \mbs{\varkappa}) = \mbs{\nu} {\times} \grad \Pi,  \qquad \astup{\nu} = \mbs{\nu} {\times} \mbs{\omega}.
\end{equation}
Ее первые интегралы, индуцированные на $\mathfrak{M}$ функциями \eqref{beq3_12}, \eqref{beq4_17}, обозначим, как и в примере \ref{exa332}:
\begin{eqnarray}
& &    H(\mbs{\nu}, \mbs{\omega}) = \frac{1}{2}(A_1 \omega_1^2 + A_2 \omega_2^2 + A_3 \omega_3^2) + \Pi(\nu_1, \nu_2, \nu_3),\label{ceq4_5}\\
& &    G(\mbs{\nu}, \mbs{\omega}) = A_1 \omega_1 \nu_1 + A_2 \omega_2 \nu_2 + A_3 \omega_3 \nu_3 + f(\nu_1, \nu_2, \nu_3).\label{ceq4_6}
\end{eqnarray}

При условии отсутствия гироскопических сил бифуркации интегралов \eqref{ceq4_5}, \eqref{ceq4_6} и соответствующие интегральные многообразия подробно изучены как для классических решений \cite{bib23,bib67}, так и в общем случае для линейного поля сил \cite{bib51,bib52}. Обобщение имеющихся результатов на гироскопический случай -- важная задача, представляющая значительный самостоятельный интерес в механике, выходящая, однако, за рамки настоящей работы.

Далее будут изучены некоторые случаи наличия дополнительного интеграла вида \eqref{beq4_18}
\begin{equation}\label{ceq4_7}
    K = K(\mbs{\nu}, \mbs{\omega}).
\end{equation}
При этом исходная динамическая система на $\tm$ является вполне интегрируемой. Почти все ее непустые интегральные многообразия состоят из конечного числа связных компонент, диффеоморфных трехмерному тору
$$
{\bf T}^3 = \{(\varphi_1, \varphi_2, \varphi_3) \mod 2 \pi \},
$$
несущему условно-периодические движения \cite{bib05}
\begin{equation}\label{ceq4_8}
    \dot{\mbs{\varphi}} = \cons.
\end{equation}
В критических случаях интегральные многообразия будут, вообще говоря, сложными трехмерными поверхностями с самопересечениями и нарушениями гладкости, изучение которых затруднительно как с аналитической, так и с геометрической точки зрения. Заметим, что и в классическом подходе, и в современных исследованиях основную роль играет все же система \eqref{ceq4_4} с фазовым пространством \eqref{ceq4_1}. По теореме Сарда \cite{bib46} почти все ее непустые интегральные многообразия
\begin{equation}\label{ceq4_9}
    J_{h,k,g} = \{x \in \mathfrak{M}: H(x) = h, K(x) = k, G(x) = g\}
\end{equation}
($h, k, g$ -- константы) являются гладкими двумерными поверхностями. Поэтому исследование фазовой топологии системы \eqref{ceq4_4} -- гораздо более благодарная работа в смысле наглядности результатов.

Поставим задачу изучить фазовую топологию системы \eqref{ceq4_4} в некоторых весьма важных классических случаях наличия интеграла \eqref{ceq4_7}. Введем интегральное отображение
\begin{equation}\label{ceq4_10}
    J = H {\times} K {\times} G: \mathfrak{M} \to \bbRR
\end{equation}
и обозначим через $\Sigma \subset \bbRR$ его бифуркационное множество.

\begin{propos}\label{pro341}
Пусть на связной компоненте $E$ многообразия \eqref{ceq4_9} нет критических точек отображения \eqref{ceq4_10}. Тогда $E$ диффеоморфна двумерному тору, несущему условно-периодические движения.
\end{propos}

\begin{proof}
Всякое множество \eqref{ceq4_9} получено из интегрального многообразия системы на $\tm$ факторизацией по группе $\Psi = S^1$. В нашем случае отображение факторизации есть
\begin{equation*}
    p {\times} \iid: \tm = \mathcal{M} {\times} \bbRR \to S^2 {\times} \bbRR = \mathfrak{M},
\end{equation*}
где $p$ -- отображение \eqref{beq4_6}. В силу $\Psi$-сохранности интегрального отображения исходной системы множество $I = (p {\times} \iid )^{-1}(E)$ является связной компонентой интегрального многообразия, и на $I$ интегралы $H(Q, \omega), K(Q, \omega), G(Q, \omega)$ независимы. Тогда по сказанному ранее $I = {\bf T}^3$, и существует изоморфизм, переводящий фазовый поток на $I$ в поток \eqref{ceq4_8} Вследствие $\Psi$-сохранности фазового потока на $\tm$ индуцированное действие $\Psi$ на ${\bf T}^3$ сохраняет условно-периодический характер движения. Тогда $E = I/ \Psi = {\bf T}^3 / S^1 = {\bf T}^2$, и траектории на нем условно-периодические.
\end{proof}

Последнее утверждение не совсем строго. Оно будет доказано в более общей ситуации. Здесь же просто продемонстрирована геометрия сформулированного предложения.

Для точек $(h, k, g) \in \bbRR \subset \Sigma$ любая связная компонента $J_{h, k, g}$ удовлетворяет условиям предложения \ref{pro341}.
Таким образом, для полного исследования фазовой топологии необходимо:

$1^\circ$) построить бифуркационное множество $\Sigma$;

$2^\circ$) в каждой области множества $\bbRR \backslash \Sigma$ указать количество связных компонент $J_{h, k, g}$;

$3^\circ$) для точек $\Sigma$ указать критическую интегральную поверхность и устройство фазовых траекторий.

В реализации п.~$1^\circ$ помогает следующий прием. Зафиксируем постоянную $g \in \bbR$ и рассмотрим отображение
\begin{equation}\label{ceq4_11}
    J[g] = \left. H {\times} K \right|_{G^{-1}(g)}: G^{-1}(g) \to \bbR^2.
\end{equation}

\begin{lemma}\label{lem342}
Пусть $\Sigma_g \subset \bbR^2$ бифуркационное множество отображения \eqref{ceq4_11}. Тогда
\begin{equation*}
    \Sigma = \bigcup \limits_{g \in \bbR} \Sigma_g {\times} \{g\}.
\end{equation*}
\end{lemma}

Это является следствием леммы \ref{lem311} и регулярности функции \eqref{ceq4_6} на $\mathfrak{M}$. Теперь при выяснении типов $J_{h,k,g}$ (п. $2^\circ$) удобно считать $g$ параметром и применять теорему \ref{the312}, утверждение которой в данном случае очевидно.

Критические поверхности (п.~$3^\circ$) устанавливаются в результате геометрического анализа задачи. Области возможности движения здесь -- это образы многообразий \eqref{ceq4_9} при отображении \eqref{ceq4_3}. Образы фазовых траекторий на многообразии \eqref{ceq4_2} -- траектории <<в теле>> постоянного в пространстве направляющего вектора $\mbs{\nu}_0$ оси симметрии силового поля. Поэтому геометрический анализ проясняет движение тела с точностью до вращений вокруг вертикали. Отметим, что условие \eqref{ceq2_9}, определяющее обобщенную границу ОВД, принимает вид
\begin{equation}\label{ceq4_12}
    \ds \det \frac{\partial (H, K, G)}{\partial (\omega_1, \omega_2, \omega_3)} = 0.
\end{equation}

Рассмотрим достаточно распространенную ситуацию, когда от уравнений
\begin{equation}\label{ceq4_13}
    H(\mbs{\nu}, \mbs{\omega}) = h, \quad K(\mbs{\nu}, \mbs{\omega}) = k, \quad G(\mbs{\nu}, \mbs{\omega}) = g,
\end{equation}
определяющих интегральное многообразие в $S^2 {\times} \bbRR$, можно в некоторой области перейти к одному уравнению вида
\begin{equation}\label{ceq4_14}
    \Phi(x_1, x_2, z; c) = 0,
\end{equation}
в котором через $c$ обозначен набор констант $(h, k, g)$, переменные $(x_1, x_2)$ -- локальные координаты на сфере Пуассона, $z$ -- вспомогательная переменная, такая, что имеются зависимости $\omega_i = \omega_i(x_1, x_2, z; c)$, тождественно удовлетворяющие \eqref{ceq4_13}.

Будем считать, что интегральное многообразие гладкое ($c$ не принадлежит бифуркационному множеству). Тогда в рассматриваемой области $U$ пространства $\bbRR = \bbR^2(x_1, x_2) {\times} \bbR(z)$ уравнение \eqref{ceq4_14} определяет гладкую регулярную поверхность $P$, зависящую от $c$.
В таком случае отображению \eqref{ceq4_3} соответствует проекция $\bbRR \to \bbR^2(x_1, x_2)$. Обозначим ее также через $\pi$. В критических точках ограничения $\pi$ на $P$ касательная плоскость к $P$ параллельна оси~$z$:
\begin{equation}\label{ceq4_15}
    \Phi'_z(x_1, x_2, z; c) = 0
\end{equation}
(тривиальный частный случай леммы \ref{lem324}).

Сделаем предположение, выполненное в типичном случае \cite{bib14}: пусть решение системы \eqref{ceq4_14}, \eqref{ceq4_15} есть гладкая кривая $\gamma$ в $\bbRR$ без особых точек. Для этого нужно, чтобы в области $U$ ранг соответствующей матрицы Якоби равнялся двум.

Образ кривой $\gamma$ на плоскости $\bbR^2(x_1, x_2)$ представляет собой исследуемый участок обобщенной границы ОВД. Пусть на этом участке есть особая точка (точка возврата). Это отвечает особенности типа сборки на поверхности $P$ по отношению к проекции $\pi$. В точке сборки ненулевой касательный вектор к $\gamma$ должен быть параллелен оси $z$. Дифференцируя \eqref{ceq4_15} вдоль такого вектора, получаем
\begin{equation}\label{ceq4_16}
    \Phi''_{zz}(x_1, x_2, z; c) = 0.
\end{equation}
Отметим, что при этом в силу сделанного предположения
\begin{equation}\label{ceq4_17}
    \det \left( \begin{array}{cc}
    \Phi'_{x_1} & \Phi'_{x_2} \\
    \Phi''_{z x_1} & \Phi''_{z x_2}
    \end{array} \right) \neq 0.
\end{equation}

Качественные перестройки ОВД связаны с изменением структуры множества особых точек обобщенной границы. При фиксированном $c$ это множество находится во взаимно-одно\-знач\-ном соответствии с множеством $Q_c$ решений системы \eqref{ceq4_14} -- \eqref{ceq4_16}. Выясним, при каких $c$ происходят бифуркации $Q_c$. Объединение множеств $Q_c$ можно представить следующим образом:
\begin{equation*}
    \bigcup \limits_c Q_c = \widetilde{\Phi}^{-1}(0),
\end{equation*}
где $\widetilde{\Phi} = \Phi {\times} \Phi'_z {\times} \Phi''_{zz} : \bbRR(c) {\times} \bbRR(x_1, x_2, z) \to \bbRR$. Пусть $p_0$ -- ограничение естественной проекции
\begin{equation*}
    p:\bbRR(c) {\times} \bbRR(x_1, x_2, z) \to \bbRR(c)
\end{equation*}
на множество $\widetilde{\Phi}^{-1}(0)$. Тогда $Q_c=p_0^{-1}(c)$, и структура $Q_c$ меняется при переходе через критические значения отображения $p_0$. По лемме \ref{lem324} в критических точках $p_0$ выполнено условие
\begin{equation*}
    \ds \det \frac{\partial (\Phi, \Phi'_z, \Phi''_{zz})}{\partial(x_1, x_2,z)} = 0.
\end{equation*}
С учетом \eqref{ceq4_17} приходим к соотношению
\begin{equation}\label{ceq4_18}
    \Phi'''_{zzz}(x_1, x_2, z; c) = 0.
\end{equation}

Система \eqref{ceq4_14} -- \eqref{ceq4_16}, \eqref{ceq4_18} равносильна одному скалярному ограничению на параметры $h, k, g$. Соответствующая поверхность в $\bbRR(h, k, g)$ служит разделяющей для классификации областей возможности движения. Указанный метод ее построения будет использован в последующих главах.

В динамике твердого тела встречаются задачи, в которых функция \eqref{ceq4_7} является не общим, а частным интегралом уравнений \eqref{ceq4_4}, т.е. когда производная $\dot{K}$ в силу \eqref{ceq4_4} обращается в нуль лишь на фиксированном множестве \eqref{ceq4_9}. Верно ли в этом случае предложение~\ref{pro341}? Один из возможных ответов дает теорема В.В.\,Козлова \cite{bib27}, касающаяся движения тела в поле одних лишь потенциальных сил. Докажем несколько более общее утверждение, применимое и к гироскопическому случаю. Его формулировке предпошлем некоторые рассуждения.

Вновь зафиксируем постоянную площадей $g \in \bbR$ и рассмотрим гладкое четырехмерное подмногообразие $G_g = G^{-1}(g)$ в $\mathfrak{M}$. С точностью до диффеоморфизма (см. замечание в $\S$~\ref{ssec15}) оно является фактор-многообразием уровня момента исходной системы на $\tm$ по действию группы $\Psi$. Согласно \eqref{aeq5_16}, ограничение $X_g = \left. X \right|_{G_g}$ есть гамильтоново векторное поле с гамильтонианом $H_g = \left. H \right|_{G_g}$ в некоторой, специальным образом определенной, симплектической структуре $\sigma_g$ на $G_g$. Обозначим через $\{ , \}$ скобку Пуассона функций на $G_g$ в структуре $\sigma_g$. Для любой функции $K$ на $\mathfrak{M}$ имеем
\begin{equation}\label{ceq4_19}
    \left. \dot{K} \right|_{G_g} = \left. (XK) \right|_{G_g} = X_g \left( \left. K \right|_{G_g} \right) = X_g K_g = \{H_g, K_g\}.
\end{equation}

Будем считать, что все рассматриваемые объекты заданы в открытой окрестности многообразия $\mathfrak{M}$ в пространстве $\bbR^6(\mbs{\nu}, \mbs{\omega})$. Рассматривая $\Gamma(\mbs{\nu}, \mbs{\omega}) = |\mbs{\nu}|^2$ как функцию на $\bbR^6$, определим многообразие $G_g$ в $\bbR^6$ уравнениями
\begin{equation}\label{ceq4_20}
    G(\mbs{\nu}, \mbs{\omega}) = g, \qquad \Gamma(\mbs{\nu}, \mbs{\omega}) = 1.
\end{equation}
Легко видеть, что в точках множества \eqref{ceq4_20} (и даже в точках $\mathfrak{M}$) $G$ и $\Gamma$ независимы как функции на $\bbR^6$. Обозначая через $\ds \nabla = (\partial/\partial \mbs{\nu}, \partial/\partial \mbs{\omega})$ оператор градиента в $\bbR^6$, получаем, что линейно независимы векторы $\nabla G$ и $\nabla \Gamma$.

Касательное пространство к $G_g$ в точке $x = (\mbs{\nu}, \mbs{\omega})$ задается в $\bbR^6$ уравнениями
\begin{equation}\label{ceq4_21}
    dG(x) = 0, \qquad d\Gamma(x) = 0.
\end{equation}

Пусть $\Phi$ -- некоторая функция на $\bbR^6$, $\Phi_g = \left. \Phi \right|_{G_g}$. Тогда для точки $x \in G_g$ в силу \eqref{ceq4_21} имеем
\begin{equation}\label{ceq4_22}
    d\Phi_g(x) = \left. d\Phi(x) \right|_{T_xG_g} = \left. d\Phi(x) \right|_{\ker dG(x) \cap \ker d\Gamma (x)}.
\end{equation}

\begin{lemma}\label{lem343}
Для точки $x \in G_g$ следующие условия равносильны:

а{\rm )} $d\Phi(x) = 0${\rm ;}

б{\rm )} векторы $\nabla \Phi(x), \nabla G(x), \nabla \Gamma(x)$ компланарны.
\end{lemma}

Это очевидное следствие \eqref{ceq4_22} и леммы об аннуляторе.

\begin{theorem}[\cite{bib27}]\label{the344} Пусть $(N^{2n}, \sigma^2)$ -- симплектическое многообразие с заданными на нем гладкими функциями $\Phi_1,\ldots,\Phi_n$. Пусть $X$ -- гамильтоново векторное поле с гамильтонианом $\Phi_1$, $E$ -- связная компонента множества
$$
\{x \in N: {\Phi_1(x) = 0}, \ldots, {\Phi_n(x) = 0}\}.
$$
Предположим, что в точках $E$

{\rm 1)} формы $d\Phi_1, \ldots, d\Phi_n$ линейно независимы{\rm ;}

{\rm 2)} скобки Пуассона $\{ \Phi_i, \Phi_j\} = 0, \qquad i, j = 1, \ldots, n${\rm ;}

{\rm 3)} $d\{ \Phi_i, \Phi_j\} = 0, \qquad i, j = 1, \ldots, n$.

В этом случае: $E$ -- гладкое $n$-мерное многообразие, инвариантное относительно фазового потока поля $X$; решения динамической системы $X|_E$ находятся с помощью квадратур. Если $E$ компактно, то $E$ диффеоморфно \mbox{$n$-м}ерному тору. Фазовый поток поля $X$ определяет на $E$ условно-пе\-ри\-оди\-чес\-кое движение.
\end{theorem}

\begin{theorem}\label{the345}
Пусть $E$ --связная компонента множества
$$
\{x \in \bbR^6: H(x) = h, \, G(x) = g, \, K(x) = k, \, \Gamma(x) = 1 \},
$$
где $H, G$ -- интегралы энергии и площадей, $K$ -- некоторая функция переменных Эйлера\,--\,Пуассона. Предположим, что в точках множества $E$

{\rm 1)} функции $H, G, K, \Gamma$ независимы;

{\rm 2)} производная $\dot{K}$ в силу уравнений Эйлера\,--\,Пу\-ас\-сона обращается в нуль;

{\rm 3)} векторы $\nabla \dot K, \nabla G, \nabla \Gamma$ компланарны.

Тогда $E$ -- гладкое интегральное многообразие уравнений Эйлера\,--\,Пу\-ас\-сона, и существует диффеоморфизм $E$ на двумерный тор ${\bf T}^2$, переводящий фазовые траектории в условно-периодические движения.
\end{theorem}

\begin{proof}
В силу независимости $H, G, K, \Gamma$ множество $E$ -- гладкое двумерное многообразие. Поскольку $\dot{K} = 0$ на $E$, то траектория, пересекающая $E$, целиком в нем содержится. Таким образом, $E$ инвариантно относительно фазового потока уравнений Эйлера\,--\,Пуассона.
В частности, $E$ -- интегральное многообразие системы $X_g$ на $G_g$, определяемое как $\{x \in G_g: H_g(x) = h, K_g(x) = k\}$. Из предположения 2 и равенства \eqref{ceq4_19} следует, что $\{H_g, K_g\} = 0$ на множестве $E$, а из предположения 3 и леммы \ref{lem343} имеем $d\{H_g, K_g\} (x) = 0$ для всех $x \in E$. Доказательство завершается применением теоремы \ref{the344} к симплектическому многообразию $(G_g, \sigma_g)$ и функциям $H_g, K_g$.
\end{proof}

Подчеркнем, что данное доказательство оперирует лишь с глобальными объектами. Ограничиться какой бы то ни было системой локальных координат (например, углами Эйлера) нельзя -- соответствующие уравнения Гамильтона будут иметь неустранимые особенности. Это связано с тем, что, как показано в $\S$~\ref{ssec24}, симплектическая форма $\sigma_g$, вообще говоря, неточна. Следовательно, при отличной от нуля постоянной площадей не существует глобального симплектического диффеоморфизма $TS^2$ на кокасательное расслоение $T^*S^2$, снабженное канонической симплектической структурой.

Если $K$ не частный, а общий интеграл (т.е. $\dot{K} \equiv 0$ во всем фазовом пространстве), из теоремы \ref{the345} вытекает предложение \ref{pro341}, которое теперь, таким образом, строго доказано.

\clearpage

\section{Топологический анализ задачи\\  о движении гиростата по инерции}\label{sec4}

Вплоть до настоящего времени знаменитое решение Эйлера задачи о движении твердого тела по инерции служит объектом приложения многих математических теорий. Результаты таких приложений часто имеют весьма общий характер и позволяют выявить новые свойства широкого круга задач механики (см., например, \cite{bib05,bib26}).

Здесь приводится топологический анализ решения Жуковского и исследование областей возможности движения в случае Эйлера. Последнее раскрывает <<внутреннюю>> геометрию движения тела. Представим себе аппарат, движущийся по инерции с относительно небольшой скоростью. Какие траектории будут выписывать на его экране звезды? В этой ситуации можно считать, что вектор $\overrightarrow{OS}$, направленный из центра масс аппарата на звезду, неподвижен в пространстве. Свяжем с этим вектором систему уравнений Эйлера\,--\,Пуассона. Тогда изучение траекторий <<в теле>> вектора $\overrightarrow{OS}$ в терминологии главы \ref{sec3} является задачей геометрического анализа рассматриваемой системы. Ее решение, интересное и само по себе, может послужить основой для изучения близких случаев, вообще говоря, неинтегрируемых.

\subsection{Бифуркационное множество и интегральные многообразия}\label{ssec41}

Полагая в системе уравнений \eqref{beq4_29}
\begin{equation}\label{deq1_1}
    \Pi = 0,
\end{equation}
получим уравнения движения гиростата по инерции.

Введем вектор кинетического момента, имеющий в подвижном пространстве вид
\begin{equation}\label{deq1_2}
    {\bf M} = (A_1 \omega_1 + \lambda_1, A_2 \omega_2 + \lambda_2, A_3 \omega_3 + \lambda_3)^T.
\end{equation}
Уравнения Эйлера с учетом \eqref{beq1_22}, \eqref{beq1_23}, \eqref{deq1_1} дают $\dot{{\bf M}} = 0$, и, следовательно, вектор ${\bf M}^0 \in \bbRRR$, соответствующий вектору \eqref{deq1_2}, постоянен. В частности, в дополнение к интегралам энергии
\begin{equation}\label{deq1_3}
    \ds H = \frac{1}{2}(A_1 \omega_1^2 + A_2 \omega_2^2 + A_3 \omega_3^2)
\end{equation}
и площадей \eqref{beq4_27} существует интеграл
\begin{equation}\label{deq1_4}
    K = (A_1 \omega_1 + \lambda_1)^2 + (A_2 \omega_2 + \lambda_2)^2 + (A_3 \omega_3 + \lambda_3)^2.
\end{equation}

Уравнения, определяющие интегральное многообразие \eqref{ceq4_9} в пространстве $\mathfrak{M} = S^2 {\times} \bbRR$, запишем в векторной форме:
\begin{equation}\label{deq1_5}
    {\bf A}\mbs{\omega} \cdot \mbs{\omega} = 2h, \qquad ({\bf A}\mbs{\omega} + \mbs{\lambda}) \cdot ({\bf A}\mbs{\omega} + \mbs{\lambda}) = k,
\end{equation}
\begin{equation}\label{deq1_6}
    ({\bf A}\mbs{\omega} + \mbs{\lambda}) \cdot \mbs{\nu} = g.
\end{equation}

Для нахождения критических точек интегрального отображения \eqref{ceq4_10} составим функцию с неопределенными множителями $\ell_1 {\bf A}\mbs{\omega} \cdot \mbs{\omega} + \ell_2 {\bf M} \cdot {\bf M} + 2 \ell_3 {\bf M} \cdot \mbs{\nu} + \ell_4 \mbs{\nu} \cdot \mbs{\nu}$ и приравняем нулю ее градиент по переменным Эйлера\,--\,Пуассона. В результате получим
\begin{equation}\label{deq1_7}
    \ell_1 \mbs{\omega} + \ell_2 {\bf M} + \ell_3 \mbs{\nu} = 0, \qquad \ell_3 {\bf M} + \ell_4 \mbs{\nu} = 0.
\end{equation}
Ненулевой набор $(\ell_1, \ldots, \ell_4)$, удовлетворяющий \eqref{deq1_7}, существует, если выполнено одно из двух условий:
\begin{equation}\label{deq1_8}
    {\bf M} {\times} \mbs{\omega} = 0
\end{equation}
или
\begin{equation}\label{deq1_9}
    {\bf M} {\times} \mbs{\nu} = 0.
\end{equation}

Обозначим
\begin{equation}\label{deq1_10}
    a = 1/A_1, \qquad b = 1/A_2, \qquad c = 1/A_3.
\end{equation}
Условимся рассматривать несимметричное твердое тело и положим без ограничения общности
\begin{equation}\label{deq1_11}
    a > b > c.
\end{equation}
Пусть
\begin{equation}\label{deq1_12}
    \lambda_1 \lambda_2 \lambda_3 \neq 0.
\end{equation}
При условии \eqref{deq1_8} введем параметр $\sigma$, полагая $\mbs{\omega} = {\bf M} \sigma$. Тогда постоянные интегралов связаны соотношениями
\begin{equation}\label{deq1_13}
\begin{array}{c}
  \ds h = \frac{\sigma^2}{2}\left(\frac{a \lambda_1^2}{(a-\sigma)^2}+\frac{b \lambda_2^2}{(b-\sigma)^2}+\frac{c \lambda_3^2}{(c-\sigma)^2}\right) ,\\
  \ds k = \frac{a^2 \lambda_1^2}{(a-\sigma)^2}+\frac{b^2 \lambda_2^2}{(b-\sigma)^2}+\frac{c^2 \lambda_3^2}{(c-\sigma)^2},
\end{array}
\end{equation}
\begin{equation}\label{deq1_14}
    g^2 \ls k.
\end{equation}

Кривая \eqref{deq1_13} непрерывна при $\sigma = \infty$, имеет разрывы с параболической асимптотикой при $\sigma \in \{a, b, c\}$, имеет ровно две точки возврата $A_*, A^*$ при некоторых вполне определенных значениях $\sigma_* \in \{c, b\}, \: \sigma^* \in \{b, a\}$. Пусть $h = F_i(k) \: (i = 1, \ldots, 6)$ -- однозначные ветви этой кривой, пронумерованные снизу вверх при достаточно больших $k$. Тогда в случае \eqref{deq1_9} имеем
\begin{equation}\label{deq1_15}
    k = g^2, \qquad F_1(k) \ls h \ls F_6(k).
\end{equation}

Объединение множества \eqref{deq1_13}, \eqref{deq1_14} и множества \eqref{deq1_15} в $\bbRR(h,k,g)$ дает бифуркационное множество $\Sigma$. Его сечения $\Sigma_g$ плоскостями $g = \cons$ показаны на
рис.~\ref{fig_4_1} ({\it а}: $g = 0$, {\it б}: $g \neq 0$). Одинаковыми римскими цифрами обозначены области, эквивалентные в смысле определения \ref{def311}. Роль параметра играет постоянная площадей $g$ (см. лемму \ref{lem342}).

\begin{figure}[ht]
\centering
\def\fs{0.4}
\begin{minipage}[h]{\fs\linewidth}
\center{\includegraphics[width=1\linewidth]{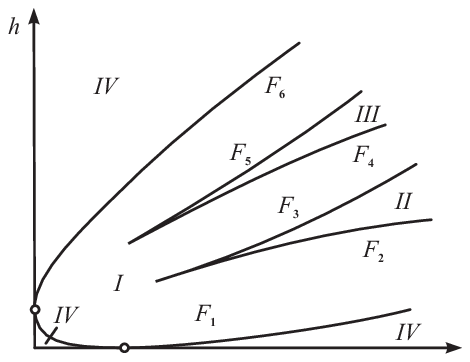} \\ а)}
\end{minipage}
\hspace{5mm}
\begin{minipage}[h]{\fs\linewidth}
\center{\includegraphics[width=1\linewidth]{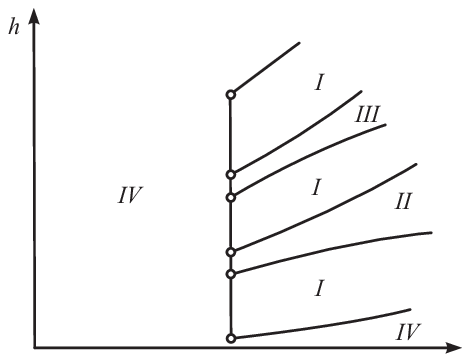} \\ б)}
\end{minipage}
\caption{}\label{fig_4_1}
\end{figure}

Переходя к анализу интегральных многообразий, примем некоторые обозначения, сохраняющие силу в дальнейшем.

\begin{soglash}\label{sog411}
Если $U$ -- гладкое связное многообразие, $n$ -- натуральное число, то $nU$ будет означать многообразие, состоящее из $n$ связных компонент, диффеоморфных $U$. Символами $\{ \cdot \}, \; S^1 \dot{\cup} S^1, \; S^1 \ddot{\cup} S^1$ обозначим соответственно точку, букет двух окружностей (<<восьмерку>>) и множество, состоящее из двух окружностей, пересечение которых -- две точки. Вообще, если $U, V, W$ -- многообразия, то $U{\stackrel{W}{\cup}}V$ означает склейку $U$ и $V$ по подмногообразию, диффеоморфному $W$.
\end{soglash}

Обратимся к системе \eqref{deq1_5}, \eqref{deq1_6}.

Уравнения \eqref{deq1_5} задают в $\bbRR$ линию пересечения эллипсоидов с соответственно параллельными главными осями. Обозначим ее через $P_{h,k}$. Очевидно, это есть не что иное, как <<полодия>> -- годограф $\mbs{\omega}$ в $\bbRR$. Введем множества
\begin{equation*}
    X_i = \{(k, h): h = F_i(k)\} \backslash \{A_*,A^*\} \; (i=1, \ldots 6).
\end{equation*}
Виды $P_{h,k}$ указаны в табл.~\ref{tab41}. Через $S_\Lambda^1$ обозначена топологическая окружность с угловой точкой.

\vspace{-2mm}

\begin{center}
\tabcolsep=1mm
\begin{tabular}{|l|c|c|c|c|c|c|c|}
  \multicolumn{8}{r}{\footnotesize Таблица \myt\label{tab41}} \\
  \hline
  \begin{tabular}{l} Случай \end{tabular} & \hbox to 8mm {\hfil 1 \hfil} & 2 & 3 & 4 & 5 & 6 & 7 \\
  \hline
\begin{tabular}{l}
  Область пара- \\
  метров $k,h$
\end{tabular}
  & $I$ & $II \cup III$ & $IV$ & $X_1 \cup X_6$ & $X_2 \cup X_5$ & $X_3 \cup X_4$ & $\{A_*, A^* \}$ \\
  \hline
  \begin{tabular}{l}\rule[-5pt]{0pt}{16pt}$P_{h,k}$\end{tabular} & $S^1$ & $2S^1$ & ${\varnothing}$ & $\{ \cdot \}$ & $\{ \cdot \} \cup S^1$ & $S^1 \dot{\cup}  S^1$& $S_\Lambda^1$ \\
  \hline
\end{tabular}
\end{center}

Зависимость $\mbs{\omega}(t)$ -- периодическая на окружностях в случаях 1, 2, 5; $\mbs{\omega(t)} = \cons$ для изолированных точек в случаях 4, 5, в центре <<восьмерки>> 6 и в угловой точке 7. Остальные траектории в случаях 6,7 -- двоякоасимптотические к неподвижной точке.

Пусть $\mbs{\omega} \in \bbRR$ фиксировано. Множество точек сферы Пуассона, удовлетворяющих \eqref{deq1_6} обозначим $S_g(\mbs{\omega})$. При этом $S_g(\mbs{\omega}) = {\varnothing}, \{ \cdot \}, S^1$, если $k < g^2, k = g^2, k > g^2$ соответственно. Интегральные многообразия находим по формуле $J_{h,k,g} = \bigcup \limits_{\mbs{\omega} \in P_{h,k}} S_g(\mbs{\omega})$. Таким образом,
\begin{equation}\label{deq1_16}
    J_{h,k,g} = \left\{ \begin{array}{lc}
                          P_{h,k} {\times} S^1, & k > g^2, \\
                          P_{h,k}, & k = g^2, \\
                          {\varnothing}, & k < g^2.
                        \end{array}
    \right.
\end{equation}
В частности, в области $I$ (рис.~\ref{fig_4_1}) $J_{h,k,g} = {\bf T}^2$, в областях $II, III$ имеем $J_{h,k,g} = 2{\bf T}^2$, и траектории на торах условно-периодические. Критические поверхности так же легко определить по формуле \eqref{deq1_16} и табл. \ref{tab41}.

Перейдем от \eqref{deq1_12} к другому крайнему случаю
\begin{equation}\label{deq1_17}
    \lambda_1 = \lambda_2 = \lambda_3 = 0
\end{equation}
(в промежуточных вариантах результаты устанавливаются соответствующими предельными переходами). Исключим из дальнейшего рассмотрения тривиальную возможность $h=k=g=0$, когда интегральное многообразие есть сфера Пуассона, заполненная неподвижными точками. В остальном непустые интегральные многообразия возможны лишь при
\begin{equation}\label{deq1_18}
    k>0.
\end{equation}
Кривая \eqref{deq1_13} вырождается в объединение лучей $2h = \alpha k$, $k \gs 0$, $\alpha \in \{a, b, c\}$. Удобно ввести новые параметры
\begin{equation}\label{deq1_19}
    \ds \varepsilon = \frac{g^2}{k}, \qquad \sigma = \frac{2h}{k}.
\end{equation}
Образ бифуркационного множества на плоскости $(\varepsilon, \sigma)$ состоит из точек
\begin{equation}\label{deq1_20}
    \varepsilon = 0, \qquad \sigma \in \{a, b, c\},
\end{equation}
отвечающих $\Sigma_0 \backslash \{0\}$, и отрезков
\begin{equation}\label{deq1_21}
    0 < \varepsilon \ls 1, \qquad \sigma \in \{a, b, c\},
\end{equation}
\begin{equation}\label{deq1_22}
    \varepsilon = 1, \qquad c \ls \sigma \ls a,
\end{equation}
в которые переходит множество $\Sigma_g$ при $g \neq 0$. Область $I$ исчезла. В областях $II - IV$ вид интегральных многообразий не изменился по теореме \ref{the312}. Непустые полодии имеют вид $2\{ \cdot \}, 2S^1, S^1 \ddot{\cup} S^1$ при $\sigma \in \{c, a\}, (c, b) \cup (b, a), \{b\}$ соответственно. Критические поверхности указаны в табл.~\ref{tab42}.

\vspace{-2mm}

\begin{center}
\begin{tabular}{|l|c|c|c|c|c|}
  \multicolumn{6}{r}{\footnotesize Таблица \myt\label{tab42}} \\
  \hline
   \begin{tabular}{l}Случай\end{tabular} & \multicolumn{2}{|c|}{$0 \ls \varepsilon < 1$} & \multicolumn{3}{|c|}{$\varepsilon = 1$} \\
  \cline{2-6}
   & 1 & 2& 3 & 4 & 5 \\
  \hline
 \begin{tabular}{l}
 Множество \\
  значений $\sigma$
\end{tabular}
& $\{a,c\}$ & $\{ b \}$ & $\{a,c\}$ & $(c,b) \cup (b,a)$ & $\{ b \}$ \\
  \hline
   \begin{tabular}{l}\rule[-5pt]{0pt}{16pt}$J_{h,k,g}$\end{tabular} & $2S^1$ & $(S^1 \ddot{\cup} S^1){\times} S^1$ & $2\{ \cdot \}$ & $2S^1$ & $S^1 \ddot{\cup} S^1$ \\
  \hline
\end{tabular}
\end{center}
Характер фазовых траекторий устанавливается без труда. В случаях 1, 4 это периодические решения, в случае 3 -- неподвижные точки (устойчивые равномерные вращения тела). В ситуации 5 имеем две гиперболические неподвижные точки (вращения вокруг средней оси инерции) и четыре отрезка сепаратрисы. В случае 2 проекция фазовой траектории на второй сомножитель периодична, а на первом возникает отображение последования с двумя гиперболическими точками.

\subsection{Вывод уравнений обобщенных границ}\label{ssec42}

Начиная с этого момента ограничимся классическим решением Эйлера \eqref{deq1_17}, в котором геометрический анализ проводится до конца аналитическим путем. Численное исследование общего случая выявило лишь один существенно новый вид области возможности движения. Остальные изменения типов ОВД связаны в основном с нарушением симметрии относительно главных сечений сферы Пуассона.

Напомним, что областью возможности движения мы называем образ интегрального многообразия $J_{h, k, g} \subset \mathfrak{M}$ на сфере Пуассона при отображении $\pi: \mathfrak{M} = S^2 {\times} \bbRR \to S^2$ проектирования на первый сомножитель.
Таким образом, точка $\mbs{\nu}$ сферы $S^2$ принадлежит ОВД тогда и только тогда, когда система \eqref{deq1_5}, \eqref{deq1_6} имеет хотя бы одно решение $\mbs{\omega} \in \bbRR$ (предложение \ref{pro322}). Упростим эту систему, <<параметризуя>> решение уравнения \eqref{deq1_6} при $\lambda = 0$.

Рассмотрим касательное расслоение $TS^2$ сферы Пуассона как множества точек $(\mbs{\nu}, \mbs{\xi})$ в $S^2 {\times} \bbRR$, заданное соотношением
\begin{equation}\label{deq2_1}
    \mbs{\nu} \cdot \mbs{\xi} =0.
\end{equation}
Отображение $TS^2$ в $\mathfrak{M}$, определенное на слоях формулой Г.В.\,Колосова \cite{bib29}
\begin{equation}\label{deq2_2}
    \ds \mbs{\omega} = \frac{1}{{\bf A} \mbs{\nu} \cdot \mbs{\nu}}(\mbs{\xi} {\times} {\bf A} \mbs{\nu} + g \mbs{\nu}),
\end{equation}
есть диффеоморфизм $TS^2$ на подмногообразие $G_g = G^{-1}(g) \subset \mathfrak{M}$ (взаимная однозначность на множествах \eqref{deq1_6} и \eqref{deq2_1} проверяется непосредственно, обратное отображение имеет вид $\mbs{\xi} = \mbs{\nu} {\times} \mbs{\omega}$ и поэтому дифференцируемо).

Условимся многообразия $G_g$ и $TS^2$ впредь не различать, считая \eqref{deq2_2} просто заменой переменных, содержащей параметр $g$. Интегральные многообразия $J_{h,k,g}$, следовательно, вложены в $TS^2$, проекция $\pi$ (точнее, ее ограничение на $G_g$) отождествляется с канонической проекцией $p_{S^2}:TS^2 \to S^2$.

Подстановка формулы \eqref{deq2_2} в \eqref{deq1_3}, \eqref{deq1_4} с учетом обозначения \eqref{deq1_10} приводит к выражению функций, индуцированных на многообразии $TS^2$ первыми интегралами $H$ и $K$:
\begin{equation}\label{deq2_3}
\begin{array}{c}
    \ds H_g(\mbs{\nu}, \mbs{\xi}) = \frac{a \xi_1^2 + b \xi_2^2 + c \xi_3^2 + abcg^2}{2 (bc \nu_1^2 + ca \nu_2^2 + ab \nu_3^2)},\\[3mm]
    \ds K_g(\mbs{\nu}, \mbs{\xi}) = \frac{1}{(bc \nu_1^2 + ca \nu_2^2 + ab \nu_3^2)} \left\{a^2 \xi_1^2 + b^2 \xi_2^2 + c^2 \xi_3^2 - \right. \\[3mm]
    - (a \nu_1 \xi_1 + b \nu_2 \xi_2 + c \nu_3 \xi_3)^2 + 2g \left[ a^2(b-c) \nu_2 \nu_3 \xi_1 + b^2(c-a) \nu_3 \nu_1 \xi_2 + \right. \\[3mm]
    \left. \left. c^2(a-b) \nu_1 \nu_2 \xi_3 \right] + g^2(b^2 c^2 \nu_1^2 + c^2 a^2 \nu_2^2 + a^2 b^2 \nu_3^2) \right\}.
\end{array}
\end{equation}
Система \eqref{deq1_5}, \eqref{deq1_6} равносильна двум уравнениям на $TS^2$:
\begin{equation}\label{deq2_4}
    H_g(\mbs{\nu}, \mbs{\xi}) = h, \qquad K_g(\mbs{\nu}, \mbs{\xi}) =k.
\end{equation}

Исключим избыточные уравнения. Для этого введем эллиптические локальные координаты $\lambda, \mu$, удовлетворяя \eqref{beq4_5}:
\begin{equation}\label{deq2_5}
    \ds \nu_1 = \sqrt{\frac{(a-\lambda)(a-\mu)}{(a-b)(a-c)}}, \quad \nu_2 = \sqrt{\frac{(\lambda-b)(b-\mu)}{(a-b)(b-c)}}, \quad \nu_3 = \sqrt{\frac{(\lambda-c)(\mu - c)}{(a-c)(b-c)}}.
\end{equation}
В зависимости от выбора знаков радикала, формулы \eqref{deq2_5} осуществляют гомеоморфизм прямоугольника
\begin{equation}\label{deq2_6}
    c \ls \mu \ls b \ls \lambda \ls a
\end{equation}
на соответствующий октант сферы. Здесь учтены неравенства \eqref{deq1_11}. На рис.~\ref{fig_4_2} показаны образы прямых $\lambda = \cons, \mu = \cons$. Множество \eqref{deq2_6} в этой главе называется основным прямоугольником.
Дифференцируя \eqref{deq2_5}, найдем выражение для компонент касательного вектора ${\mbs{\xi} \in T_\mbs{\nu}S^2}$:
\begin{equation}\label{deq2_7}
\begin{array}{c}
  \ds \xi_1 = -\frac{1}{2\sqrt{(a-b)(a-c)}} \frac{(a-\mu)\dot{\lambda}+(a-\lambda)\dot{\mu}}{\sqrt{(a-\lambda)(a-\mu)}}, \\
  \ds \xi_2 = -\frac{1}{2\sqrt{(a-b)(b-c)}} \frac{(b-\mu)\dot{\lambda}+(\lambda-b)\dot{\mu}}{\sqrt{(\lambda-b)(b-\mu)}}, \\
  \ds \xi_3 = -\frac{1}{2\sqrt{(a-c)(b-c)}} \frac{(\mu-c)\dot{\lambda}+(\lambda-c)\dot{\mu}}{\sqrt{(\lambda-c)(\mu-c)}}.
\end{array}
\end{equation}
Тем самым удовлетворено соотношение \eqref{deq2_1}.

\begin{figure}[ht]
\centering
\includegraphics[width=0.3\linewidth,keepaspectratio]{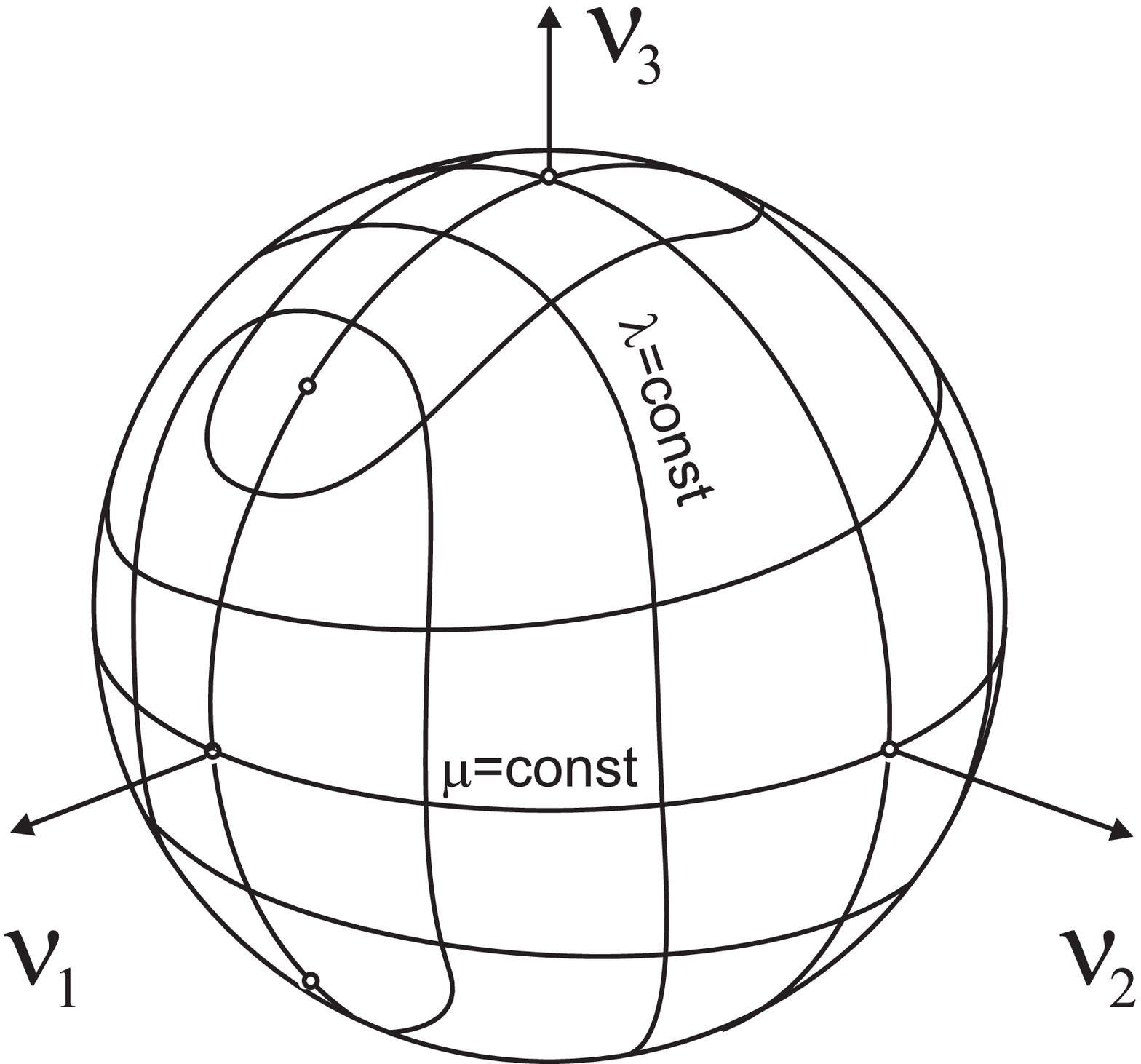}
\caption{}\label{fig_4_2}
\end{figure}

Из выражений \eqref{deq2_3} видно, что область возможности движений симметрична относительно координатных сечений сферы Пуассона (система \eqref{deq2_4} допускает одновременную замену знака у величин $\nu_\ell, \xi_m, \xi_n$, где все $\ell,m,n$ различны). Поэтому при исследовании обобщенных границ достаточно ограничиться первым октантом и считать все радикалы в \eqref{deq2_5}, \eqref{deq2_7} неотрицательными.

При подстановке \eqref{deq2_5}, \eqref{deq2_7} система \eqref{deq2_4} примет вид
\begin{equation}\label{deq2_8}
\begin{array}{c}
  \ds u^2+v^2 = \frac{\lambda \mu \sigma - abc \varepsilon}{\lambda - \mu}, \\
  \ds \frac{1}{\mu} \left[ u - \frac{\sqrt{-\varepsilon \lambda \varphi(\mu)}}{\lambda - \mu}\right]^2 + \frac{1}{\lambda} \left[v + \frac{\sqrt{\varepsilon \mu \varphi(\lambda)}}{\lambda - \mu}\right]^2 = \frac{\lambda \mu}{\lambda - \mu}(1-\varepsilon).
\end{array}
\end{equation}
Здесь учтены неравенство \eqref{deq1_18}, соотношения \eqref{deq1_19} и введены обозначения
\begin{equation}\label{deq2_9}
\begin{array}{c}
  \varphi(z) = (a-z)(b-z)(c-z), \quad
  \ds u = \frac{1}{2}\sqrt{\frac{\lambda}{k \varphi(\lambda)}} \dot{\lambda}, \quad v = \frac{1}{2}\sqrt{-\frac{\mu}{k \varphi(\mu)}} \dot{\mu}.
\end{array}
\end{equation}

Условимся областью возможности движения называть также ее образ в основном прямоугольнике \eqref{deq2_6}. Тогда можно сказать, что точка $(\lambda, \mu)$ принадлежит ОВД, если система \eqref{deq2_8} имеет (вещественное) решение относительно $u, v$. Условно вектор $(u, v)$, удовлетворяющий \eqref{deq2_9}, будем называть допустимой скоростью\footnote[1]{Настоящая допустимая скорость здесь -- вектор \eqref{deq2_2}, связанный с $(u, v)$ посредством \eqref{deq2_7}, \eqref{deq2_9}. Вектор \eqref{deq2_7} -- это скорость вдоль траектории $\mbs{\nu}(t)$, а $(\dot{\lambda}, \dot{\mu})$ -- обобщенная скорость, соответствующая координатам $(\lambda, \nu)$.} в точке $(\lambda, \mu)$.

Из системы \eqref{deq2_8} видно, что ОВД полностью определена двумя параметрами \eqref{deq1_19}, поэтому обозначим ее через $U_{\varepsilon, \sigma}$.

Запишем в новых переменных уравнения обобщенных границ ОВД. В результате подстановки \eqref{deq1_2} -- \eqref{deq1_4} условие \eqref{ceq4_12} принимает вид $(\mbs{\omega}{\times}{\bf M}) \cdot \mbs{\nu} = 0$ или в развернутой форме
\begin{equation*}
    (A_2 - A_3) \omega_2 \omega_3 \nu_1 + (A_3 - A_1)\omega_3 \omega_1 \nu_2 + (A_1 - A_2) \omega_1 \omega_2 \nu_3 = 0.
\end{equation*}
Вводя в это равенство выражения \eqref{deq2_2}, \eqref{deq2_5}, \eqref{deq2_7}, \eqref{deq2_9}, получаем
\begin{equation}\label{deq2_10}
    \ds \frac{(\lambda - \mu)\sqrt{\lambda \mu}}{abc} \det \left(
    \begin{array}{cc}
    u & v \\
    \ds \frac{1}{\mu} \left[ u - \frac{\sqrt{-\varepsilon \lambda \varphi(\mu)}}{\lambda - \mu}\right] & \ds \frac{1}{\lambda} \left[v + \frac{\sqrt{\varepsilon \mu \varphi(\lambda)}}{\lambda - \mu}\right]
    \end{array}
    \right) = 0.
\end{equation}
Геометрически это условие означает, что окружность и эллипс, определяемые в плоскости $(u, v)$ уравнениями \eqref{deq2_8}, имеют в некоторой точке общую касательную. Изучим некоторые особые случаи \eqref{deq2_10}.

Условие \eqref{deq2_10}, очевидно, выполнено, если
\begin{equation}\label{deq2_11}
    \lambda = \mu = b.
\end{equation}
При этом, однако, теряют смысл уравнения \eqref{deq2_8}, и исследование точек сферы
\begin{equation}\label{deq2_12}
    \ds \nu_1^2 = \frac{a-b}{a-c}, \qquad \nu_2 = 0, \qquad \nu_3^2 = \frac{b-c}{a-c},
\end{equation}
соответствующих \eqref{deq2_11}, необходимо проводить в переменных $\mbs{\nu}, \mbs{\xi}$. Непосредственно проверяется, что точки \eqref{deq2_12} лежат на обобщенной границе ОВД, если
\begin{equation}\label{deq2_13}
    \sigma = b + (a + c - 2b)\varepsilon - 2\sqrt{(a-b)(b-c)\varepsilon (1 - \varepsilon)}.
\end{equation}
Особым является лишь случай
\begin{equation}\label{deq2_14}
    \varepsilon = 0, \qquad \sigma = b,
\end{equation}
когда уравнения \eqref{deq2_4} определяют в касательной плоскости сферы в точках \eqref{deq2_12} две совпадающие окружности. При переходе через остальные значения $(\varepsilon, \sigma)$, принадлежащие кривой \eqref{deq2_13}, никакой перестройки обобщенной границы не происходит, так как точки \eqref{deq2_12} лежат на сфере <<внутри ребра>> (см. рис.~\ref{fig_4_2}).

Пусть теперь в определителе \eqref{deq2_10} вторая строка нулевая:
\begin{equation}\label{deq2_15}
    \ds u = \frac{\sqrt{-\varepsilon \lambda \varphi(\mu)}}{\lambda - \mu}, \qquad v = -\frac{\sqrt{\varepsilon \mu \varphi(\lambda)}}{\lambda - \mu}.
\end{equation}
Это произойдет, если в пространстве $\bbRR$ второй эллипсоид \eqref{deq1_5} и плоскость \eqref{deq1_6} имеют единственную общую точку. В этой точке выполнено условие \eqref{deq1_9}, и, следовательно, она является критической для интегрального отображения. Соответствующие ей критические значения в параметрах $\varepsilon, \sigma$ имеют вид \eqref{deq1_22}. В этом же убеждаемся непосредственно, подставляя \eqref{deq2_15} во второе уравнение \eqref{deq2_8}. Образ критического множества, соответствующего значению \eqref{deq1_6}, получим, подставив \eqref{deq2_15} в первое уравнение \eqref{deq2_8}:
\begin{equation}\label{deq2_16}
    \lambda + \mu = a + b + c -\sigma.
\end{equation}

Множество \eqref{deq2_16} назовем критической прямой (на сфере это будет, вообще говоря, пара замкнутых кривых).
В случае, если равенства \eqref{deq2_15} не выполняются, строки определителя \eqref{deq2_10} должны быть пропорциональны:
\begin{equation}\label{deq2_17}
    \ds u = \frac{\tau}{\mu} \left[ u - \frac{\sqrt{-\varepsilon \lambda \varphi(\mu)}}{\lambda - \mu}\right], \qquad v= \frac{\tau}{\lambda} \left[v + \frac{\sqrt{\varepsilon \mu \varphi(\lambda)}}{\lambda - \mu}\right].
\end{equation}
Отсюда следует, что в точках $(\lambda, \mu)$ обобщенной границы допустимая скорость, реализующая условие \eqref{deq2_10}, такова:
\begin{equation}\label{deq2_18}
    \ds u = \frac{\tau \sqrt{-\varepsilon \lambda \varphi(\mu)}}{(\lambda - \mu)(\tau - \mu)}, \qquad v = \frac{\tau \sqrt{\varepsilon \mu \varphi(\lambda)}}{(\lambda - \mu)(\lambda - \tau)}.
\end{equation}
Очевидно, имеется особенность, если $\lambda = \tau$ или $\mu = \tau$. Тогда по формуле \eqref{deq2_17} $\varphi(\tau) = 0$, так что, согласно \eqref{deq2_9} $\tau \in \{a, b, c\}$.

Пусть $\theta$ -- перестановка множества из двух элементов, тождественная при ${\lambda = \tau}$ и являющаяся транспозицией при $\mu = \tau$.
Обозначим
\begin{equation}\label{deq2_19}
    \{z, \overline{z} \} = \theta \{\lambda, \mu \}, \quad \{\psi, \overline{\psi} \} = \theta \{\varphi, -\varphi \}, \quad \{w, \overline{w} \} = \theta \{u, v \},
\end{equation}
так что в любом случае $z = \tau \in \{a, b, c\}$. Пусть, кроме того,
\begin{equation}\label{deq2_20}
    \alpha, \beta, \gamma
\end{equation}
это величины $a, b, c$, расставленные в таком порядке, что $\tau = \alpha$. Используя формулы \eqref{deq2_8}, \eqref{deq2_17}, находим
\begin{equation}\label{deq2_21}
    \ds z = \alpha, \qquad \overline{z} = \alpha + \varepsilon\frac{(\beta - \alpha)(\gamma - \alpha)}{\sigma - \alpha},
\end{equation}
\begin{equation}\label{deq2_22}
    \ds w=\frac{\alpha\sqrt{\varepsilon \alpha \overline{\psi}(\overline{z})}}{(\alpha - \overline{z})^2}, \qquad  \overline{w}^2=\frac{\alpha^2\overline{z}\Delta_\alpha}{\varepsilon |\alpha-\overline{z}|(\beta - \alpha)(\gamma - \alpha)}.
\end{equation}
Здесь
\begin{equation}\label{deq2_23}
    \Delta_\alpha = (\sigma - \alpha)^2 - \varepsilon(\beta+\gamma-2\alpha)(\sigma-\alpha)+\varepsilon(\beta-\alpha)(\gamma - \alpha).
\end{equation}

Поскольку должно быть $\lambda \gs \mu$, знак перестановки $\theta$ определяется по правилу
\begin{equation}\label{deq2_24}
    \ds \sgn \theta = \sgn \frac{(\beta + \alpha)(\gamma - \alpha)}{\alpha - \sigma}.
\end{equation}
При этом в силу \eqref{deq2_22} точка \eqref{deq2_21} принадлежит обобщенной границе ОВД лишь в случае
\begin{equation}\label{deq2_25}
    (\beta - \alpha)(\gamma - \alpha)\Delta_\alpha \gs 0.
\end{equation}

К исследованию поведения обобщенной границы в окрестности точек вида \eqref{deq2_21} вернемся позже.

Далее обозначим
\begin{equation}\label{deq2_26}
    x = (\lambda - \tau)(\tau - \mu), \qquad y = (\lambda - \tau) - (\tau - \mu).
\end{equation}

Предполагая $x \neq 0$, подставляем выражения \eqref{deq2_18} в \eqref{deq2_8}. После упрощений получаем
\begin{equation}\label{deq2_27}
    \frac{\varepsilon \varphi(\tau)}{x} = \tau - \sigma, \qquad \frac{\varepsilon [\varphi'(\tau)x - \varphi(\tau) y]}{x^2} = 1.
\end{equation}

Допустим, что
\begin{equation}\label{deq2_28}
    \tau = \sigma.
\end{equation}
Тогда $\varphi(\sigma) = 0$, и, следовательно, $(\varepsilon, \sigma)$ -- одно из критических значений \eqref{deq1_20}, \eqref{deq1_21}. Второе уравнение \eqref{deq2_27} принимает вид
\begin{equation}\label{deq2_29}
    (\lambda - \sigma)(\sigma - \nu) = \varepsilon \varphi'(\sigma).
\end{equation}
Это образ критических точек интегрального отображения, в которых выполнено условие \eqref{deq1_8}. Множество \eqref{deq2_29} назовем критической гиперболой. При $\varepsilon = 0$ критическая гипербола распадается на пару координатных прямых.

Пусть
\begin{equation}\label{deq2_30}
    \tau \neq \sigma.
\end{equation}
Из соотношений \eqref{deq2_27} находим
\begin{equation}\label{deq2_31}
    \ds x(\tau) = \varepsilon \frac{\varphi(\tau)}{\tau - \sigma}, \qquad y(\tau) = \varepsilon \frac{(\tau - \sigma)\varphi'(\tau) - \varphi(\tau)}{(\tau - \sigma)^2}.
\end{equation}
Интересно отметить тождество
\begin{equation}\label{deq2_32}
    x'(\tau) = y(\tau),
\end{equation}
из которого следует, что полная производная $x$ по $\tau$ совпадает с производной по явно входящему $\tau$ в формуле \eqref{deq2_26}.

Обращая замену \eqref{deq2_26}, получаем параметрические уравнения той части обобщенной границы, прообраз которой не содержит критических точек интегрального отображения:
\begin{equation}\label{deq2_33}
\begin{array}{c}
  \ds \lambda = \frac{1}{2}\left[2\tau + y(\tau) + \sqrt{y^2(\tau)+4x(\tau)} \right],\quad
  \ds \mu = \frac{1}{2}\left[2\tau + y(\tau) - \sqrt{y^2(\tau)+4x(\tau)} \right].
\end{array}
\end{equation}
Заметим, что, предположив здесь
\begin{equation}\label{deq2_34}
    x(\tau)=0,
\end{equation}
приходим к равенствам \eqref{deq2_21}, подчиненным правилу \eqref{deq2_24}. Поэтому уравнения \eqref{deq2_31}, \eqref{deq2_33} остаются справедливыми и в ситуации \eqref{deq2_34}. Следовательно, они могут быть использованы для более детального изучения этого случая.

\subsection{Особые точки обобщенных границ и разделяющие кривые}\label{ssec43}

Две области возможности движения будем считать эквивалентными, если существует диффеоморфизм сферы в себя, переводящий одну область на другую с сохранением структуры множества допустимых скоростей в каждой точке. Интуитивно ясно, что такой диффеоморфизм можно <<достроить>> до диффеоморфизма интегральных многообразий, сохраняющего критическое множество проекции $\pi$. Очевидно, эквивалентные ОВД имеют диффеоморфные обобщенные границы. Исследовав перестройки обобщенных границ, получим тем самым и классификацию областей возможности движения.

Множество в пространстве параметров $(\varepsilon, \sigma)$, при переходе через которое меняется дифференцируемый тип обобщенной границы, назовем разделяющим. По общим соображениям разделяющее множество должно состоять из некоторого числа кривых, которые также будем называть разделяющими. В их число, конечно, входят кривые, образующие бифуркационное множество интегрального отображения (в нашем случае это точки \eqref{deq1_20} и отрезки \eqref{deq1_21}, \eqref{deq1_22}), так как на них меняется тип самих интегральных многообразий. Соответствующие перестройки ОВД связаны с появлением в составе обобщенных границ критической прямой \eqref{deq2_16} или критической гиперболы \eqref{deq2_29}.

Разделяющие кривые второго типа отвечают различным проекциям на сферу диффеоморфных между собой интегральных многообразий. Найдем их, изучив эволюцию особых точек кривой \eqref{deq2_33}.

Учитывая тождество \eqref{deq2_32}, продифференцируем уравнение \eqref{deq2_33} по $\tau$:
\begin{eqnarray}
&  \ds \lambda' = \frac{(\lambda - \tau)(y' + 2)}{\sqrt{y^2+4x}}, \qquad \mu' = \frac{(\tau - \mu)(y' + 2)}{\sqrt{y^2+4x}},\label{deq3_1}\\
& \begin{array}{l}
  \lambda'' = \left\{2(y'+2)^2 x + \left[ (\lambda - \tau)y'' - (y' + 2)\right] (y^2 + 4x \right\} (y^2 + 4x)^{-3/2}, \\
  \mu'' = \left\{-2(y'+2)^2 x + \left[ (\tau - \mu)y'' + (y' + 2)\right] (y^2 + 4x \right\} (y^2 + 4x)^{-3/2}.
\end{array}\label{deq3_2}
\end{eqnarray}
Согласно \eqref{deq3_1} кривая \eqref{deq2_33} имеет экстремумы при условии \eqref{deq2_34}, что выполняется лишь в точках вида \eqref{deq2_21}. Таким образом, в экстремальных точках кривая \eqref{deq2_33} касается прямых, содержащих стороны основного прямоугольника. Характер экстремума определяется значениями \eqref{deq3_2}:
\begin{equation}\label{deq3_3}
\begin{array}{lll}
  \lambda = \tau \Rightarrow \lambda' = 0, & \mu' = y' + 2, & \ds \lambda'' = -\frac{y' + 2}{\lambda - \mu}, \\
  \mu = \tau \Rightarrow \mu' = 0, & \lambda' = y' + 2, & \ds \mu'' = \frac{y' + 2}{\lambda - \mu}.
\end{array}
\end{equation}
Ключевую роль играет знак выражения
\begin{equation}\label{deq3_4}
\begin{array}{ll}
  \ds y'(\tau) + 2 & \ds = \frac{\varepsilon}{(\tau - \sigma)^3} [2(\tau - \sigma)^3 + \varepsilon \varphi''(\tau)(\tau - \sigma)^2 -\\
  &- 2\varepsilon \varphi'(\tau)(\tau - \sigma) + 2 \varepsilon \varphi(\tau) ] = \\
  &\ds = \frac{2\varepsilon}{(\tau - \sigma)^3}\left[(1-\varepsilon)(\tau - \sigma)^3 + \varepsilon \varphi(\sigma) \right].
\end{array}
\end{equation}
Напомним обозначения \eqref{deq2_20}, \eqref{deq2_23}. Полагая в \eqref{deq3_4} $\tau = \alpha$, находим
\begin{equation}\label{deq3_5}
    \ds y'(\alpha) + 2 = \frac{2 \varepsilon}{(\sigma - \alpha)^2}\Delta_\alpha.
\end{equation}

Из формул \eqref{deq3_3}, \eqref{deq3_5} получаем информацию об экстремальных точках кривой \eqref{deq2_33}, отраженную в табл.~\ref{tab43}.

\vspace{-2mm}

\begin{center}
\tabcolsep=1mm
\begin{tabular}{|c|c|c|c|c|c|c|c|c|}
\multicolumn{1}{c}{} & \multicolumn{1}{c}{\rule{12mm}{0pt}} & \multicolumn{1}{c}{\rule{12mm}{0pt}} & \multicolumn{1}{c}{\rule{12mm}{0pt}} & \multicolumn{1}{c}{\rule{12mm}{0pt}} & \multicolumn{1}{c}{\rule{12mm}{0pt}} & \multicolumn{1}{c}{\rule{12mm}{0pt}} & \multicolumn{1}{c}{\rule{12mm}{0pt}} & \multicolumn{1}{c}{\rule{12mm}{0pt}}\\
  \multicolumn{9}{r}{\footnotesize Таблица \myt\label{tab43}} \\
  \hline
  {Случай} & {1} & {2} & {3} & {4} & {5} & {6} & {7} & {8} \\
  \hline
  $\tau$ & \multicolumn{2}{c|}{$a$} & \multicolumn{2}{c|}{$b$} & \multicolumn{2}{c|}{$c$} & \multicolumn{2}{c|}{$d$} \\
  \hline
  $\sigma \in$ & \multicolumn{2}{c|}{$[c, a)$} & \multicolumn{2}{c|}{$[c, b)$} & \multicolumn{2}{c|}{$(b, a]$} & \multicolumn{2}{c|}{$(c, a]$} \\
  \hline
  $x'$ & \multicolumn{2}{c|}{$\mbs{-}$} & \multicolumn{2}{c|}{$\mbs{+}$} & \multicolumn{2}{c|}{$\mbs{-}$} & \multicolumn{2}{c|}{$\mbs{+}$} \\
  \hline
  $\lambda$ & \multicolumn{2}{c|}{$a$} & \multicolumn{2}{c|}
  {\footnotesize $\ds b + \varepsilon \frac{(a-b)(b-c)}{b-\sigma}$}
   & \multicolumn{2}{c|}{$b$} & \multicolumn{2}{c|}
   {\footnotesize \rule[-12pt]{0pt}{11mm}  $\ds c + \varepsilon \frac{(a-b)(b-c)}{\sigma-c}$} \\
  \hline
  $\mu$ & \multicolumn{2}{c|}
  {\footnotesize\rule[-12pt]{0pt}{11mm}$\ds a - \varepsilon \frac{(a-b)(a-c)}{a-\sigma}$}
  & \multicolumn{2}{c|}{$b$} & \multicolumn{2}{c|}{\footnotesize$\ds b - \varepsilon \frac{(a-b)(b-c)}{\sigma-b}$} & \multicolumn{2}{c|}{$c$} \\
  \hline
  $\Delta_\tau$ & $\mbs{+}$& $\mbs{-}$ & $\mbs{+}$ & $\mbs{-}$ & $\mbs{+}$ & $\mbs{-}$ & $\mbs{+}$ & $\mbs{-}$ \\
  \hline
  $\lambda'$ & {0} & {0} & $\mbs{+}$ & $\mbs{-}$ & {0} & {0} & $\mbs{+}$ & $\mbs{-}$ \\
  \hline
  $\mu'$ & $\mbs{+}$ & $\mbs{-}$ & {0} & {0} & $\mbs{+}$ & $\mbs{-}$ & {0} & {0} \\
  \hline
  $\lambda''$ & $\mbs{-}$ & $\mbs{+}$ & {$\mbs{\pm}$} & {$\mbs{\pm}$} & $\mbs{-}$ & $\mbs{+}$ & {$\mbs{\pm}$} & {$\mbs{\pm}$} \\
  \hline
  $\mu'$ & {$\mbs{\pm}$} & {$\mbs{\pm}$} & $\mbs{+}$ & $\mbs{-}$ & {$\mbs{\pm}$} & {$\mbs{\pm}$} & $\mbs{+}$ & $\mbs{-}$ \\
  \hline
\end{tabular}
\end{center}

В случаях 2, 3, 5, 8 пересечение кривой \eqref{deq2_33} в окрестности экстремума с основным прямоугольником может быть лишь изолированной точкой. Обращаясь к условию \eqref{deq2_25}, видим, что такая точка -- постороннее решение и поэтому не включается в обобщенную границу ОВД.

Анализируя возможности 1, 4, 6, 7 и учитывая неравенства \eqref{deq2_6}, приходим к следующему результату: кривая \eqref{deq2_33} имеет точку касания со стороной $\lambda = a$ прямоугольника \eqref{deq2_6} -- в области, ограниченной отрезками\footnote[1]{Напомним, что в целом область изменения параметров $\varepsilon, \sigma$ есть множество $[0, 1] {\times} [c, a]$.}
\begin{eqnarray}
& &    \sigma = a - \varepsilon(a - c), \label{deq3_6}\\
& &    \sigma = a - \varepsilon(a - b),\label{deq3_7}
\end{eqnarray}
и кривой (рис.~\ref{fig_4_3},\,{\it а});
\begin{equation}\label{deq3_8}
    \Delta_a = 0;
\end{equation}
точку касания со стороной $\mu = b$ -- в области, заключенной между отрезком
\begin{equation}\label{deq3_9}
    \sigma = b - \varepsilon(b - c),
\end{equation}
и частью кривой
\begin{equation}\label{deq3_10}
    \Delta_b = 0,
\end{equation}
находящейся в пределах $c \ls \sigma \ls b$ (рис.~\ref{fig_4_3},\,{\it б});
точку касания со стороной $\lambda = b$ -- в области, лежащей между отрезком
\begin{equation}\label{deq3_11}
    \sigma = b + \varepsilon(a - b),
\end{equation}
и частью кривой \eqref{deq3_10} в пределах $b \ls \sigma \ls a$ (рис.~\ref{fig_4_3},\,{\it в});
точку касания со стороной $\mu = c$ -- в области, ограниченной отрезками
\begin{eqnarray}
& &     \sigma = c + \varepsilon (b - c), \label{deq3_12}\\
& &    \sigma = c + \varepsilon (a - c),\label{deq3_13}
\end{eqnarray}
и кривой (рис.~\ref{fig_4_3},\,{\it г})
\begin{equation}\label{deq3_14}
    \Delta_c = 0.
\end{equation}

\begin{figure}[!ht]
\def\fs{0.3}
\centering
\begin{tabular}{c}
\begin{minipage}[h]{\fs\linewidth}
\center{\includegraphics[width=1\linewidth]{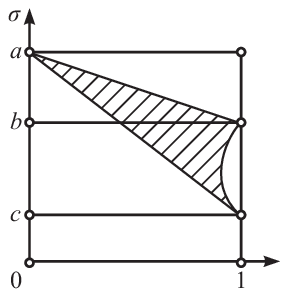} \\ а)}
\end{minipage}
\hfill
\begin{minipage}[h]{\fs\linewidth}
\center{\includegraphics[width=1\linewidth]{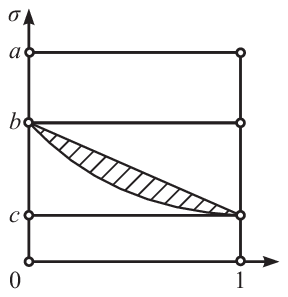} \\ б)}
\end{minipage}
   \\
\begin{minipage}[h]{\fs\linewidth}
\center{\includegraphics[width=1\linewidth]{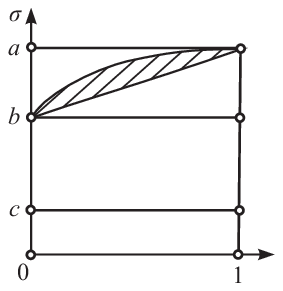} \\ в)}
\end{minipage}
\hfill
\begin{minipage}[h]{\fs\linewidth}
\center{\includegraphics[width=1\linewidth]{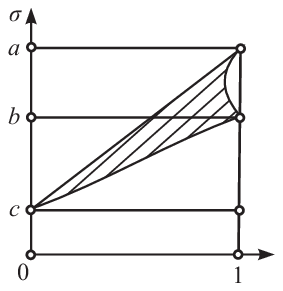} \\ г)}
\end{minipage}
  \\
\begin{minipage}[h]{\fs\linewidth}
\center{\includegraphics[width=1\linewidth]{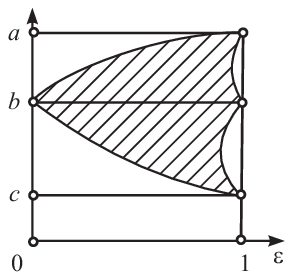} \\ д)}
\end{minipage}
  \\
\end{tabular}
\caption{}\label{fig_4_3}
\end{figure}

Обратимся к уравнению
\begin{equation}\label{deq3_15}
    y'+2=0.
\end{equation}
Из представления \eqref{deq3_4} следует, что вне бифуркационного множества ($\varepsilon \neq 1$, $\varphi(\sigma) \neq 0$) оно имеет единственный простой корень $\tau_0 \neq \sigma$. Согласно \eqref{deq3_1} кривая \eqref{deq2_33} имеет при $\tau = \tau_0$ точку возврата. Выясним, когда она принадлежит прямоугольнику \eqref{deq2_6}.

Разделяющий случай -- точка возврата находится на границе прямоугольника. Нетрудно убедиться, что для кривой \eqref{deq2_33} в обозначениях \eqref{deq2_19} справедливо тождество
\begin{equation}\label{deq3_16}
    (\tau - \sigma)[y'(\tau) + 2] x^2(\tau) = \varepsilon^2[(\tau - z)^3 \overline{\psi}(\overline{z}) - (\overline{z} - \tau)^2 \psi(z)].
\end{equation}
Поэтому, если предположить, что в точке возврата $z(\tau_0) = \alpha \in \{a, b, c\}$, то $\psi(z) = 0$ и, следовательно, $z(\tau_0) = \tau_0$. Таким образом, точка возврата совпадает с экстремальной точкой \eqref{deq2_21}, а левая часть уравнения \eqref{deq3_15} принимает значение \eqref{deq3_5}. Итак, разделяющее множество в данном случае -- это уже отмеченные кривые \eqref{deq3_8}, \eqref{deq3_10}, \eqref{deq3_14}. Точка возврата находится внутри основного прямоугольника, если $(\varepsilon, \sigma)$ принадлежит области, заштрихованной на рис.~\ref{fig_4_3},\,{\it д}.

Равенства \eqref{deq3_1} показывают, что других особых точек, кроме \eqref{deq3_3}, \eqref{deq3_15}, кривая \eqref{deq2_33} внутри прямоугольника не имеет. Следовательно, разделяющее множество есть объединение бифуркационного множества $\Sigma$, кривых $\Delta_\alpha = 0$, где  $\alpha = \{a, b, c\}$, и отрезков \eqref{deq3_6}, \eqref{deq3_7}, \eqref{deq3_9}, \eqref{deq3_11}-\eqref{deq3_13}.

Предположим, что кривая \eqref{deq2_33} построена. Совместно с критическими множествами \eqref{deq2_16} и \eqref{deq2_29} (если таковые определены) она разбивает сферу Пуассона на открытые связные подобласти. Для полного описания ОВД необходимо установить правила, по которым определяется число допустимых скоростей в каждой такой подобласти.
С этой целью изучим решения уравнений \eqref{deq2_8} на обобщенной границе $\varepsilon \neq 0$.

Отметим сразу, что на критической прямой \eqref{deq2_16}, существующей при $\varepsilon = 1$, допустимая скорость одна -- вектор \eqref{deq2_15}.
На остальных участках обобщенной границы окружность и эллипс \eqref{deq2_8} имеют точку касания \eqref{deq2_18}. При $\lambda = \tau$ или $\mu = \tau$, раскрывая неопределенность в правой части соответствующего равенства \eqref{deq2_18} с помощью формул \eqref{deq3_3}, \eqref{deq3_5}, приходим к выражениям \eqref{deq2_22} в обозначениях \eqref{deq2_19}. Таким образом, и в экстремальном случае $x(\tau) = 0$ формулы \eqref{deq2_18} сохраняют естественный смысл.

Если $(u, v)$ -- вектор \eqref{deq2_18}, то остальные решения \eqref{deq2_8} запишем так:
\begin{equation}\label{deq3_17}
\begin{array}{l}
  \ds u_{1,2} = \frac{[\tau(\lambda + \mu) - 2 \lambda \mu]u \mp 2v\sqrt{\lambda \mu(\lambda - \tau)(\tau - \mu)}}{\tau(\lambda - \mu)},\\
  \ds v_{1,2} = \frac{\mp 2u\sqrt{\lambda \mu(\lambda - \tau)(\tau - \mu)} - [\tau(\lambda + \mu) - 2 \lambda \mu]v}{\tau(\lambda - \mu)}.
\end{array}
\end{equation}
Они вещественны лишь при условии $x(\tau) \gs 0$. Если $x(\tau) = 0$, то $(u_1, v_1) = (u_2, v_2)$. В точке возврата \eqref{deq3_15} воспользуемся тождеством \eqref{deq3_16}. Получим, что один из векторов \eqref{deq3_17} совпадает с $(u, v)$. Следовательно, в точках экстремума и возврата -- две допустимые скорости.
Если одновременно $x(\tau) = 0$ и $y'(\tau) + 2 = 0$ (точка возврата попала на границу основного прямоугольника), то уравнения  \eqref{deq2_8} имеют одно решение \eqref{deq2_28} кратности четыре.

Анализируя выражения \eqref{deq3_17} при условии \eqref{deq2_28}, учтем, что $\varphi'(a) < 0$, $\varphi'(b) < 0$, $\varphi'(c) < 0$. Поэтому на критической гиперболе \eqref{deq2_29} имеется одна допустимая скорость, если $\sigma = c$ или $\sigma = a$, и три, если $\sigma = b$. Исключение составляет лишь случай $\sigma = b$, $\varepsilon = 1$, в котором единственная общая точка прямоугольника \eqref{deq2_6} и гиперболы -- вершина $(\lambda, \mu) = (a, c)$, и в ней единственной допустимой скоростью является нулевая.

Таким образом, число допустимых скоростей на обобщенной границе установлено во всех ситуациях.

Участок кривой \eqref{deq2_33}, не содержащий особых точек, по общей классификации Уитни отображений двумерных многообразий \cite{bib14} может быть лишь образом <<складки>>.\footnote[1]{Точка возврата обобщенной границы является образом <<сборки>>, а точка экстремума на сфере дает трансверсальное пересечение образов <<складок>>.} Поэтому при переходе через такой участок число допустимых скоростей (прообразов при проектировании интегрального многообразия на сферу) меняется на два. Конечно, этот факт нетрудно установить и аналитически, исследуя \eqref{deq2_8} в окрестности неособой точки \eqref{deq2_33}.

Заметим теперь, что в вершинах основного прямоугольника, соответствующих пересечениям сферы Пуассона с координатными осями, число допустимых скоростей устанавливается явно ($\varphi(\lambda) = \varphi(\mu) = 0$, и уравнения \eqref{deq2_8} разрешимы) и полностью определяется расположением $(\varepsilon, \sigma)$ относительно отрезков \eqref{deq3_6}, \eqref{deq3_7}, \eqref{deq3_9}, \eqref{deq3_11}-\eqref{deq3_13}. А поскольку это число нигде, кроме точек \eqref{deq2_12} в случае \eqref{deq2_14}, не превосходит четыре, получаем очевидный алгоритм его вычисления в любой точке сферы.

\subsection{Классификация областей возможности движения}\label{ssec44}

Рассмотрим случай, когда постоянная площадей равна нулю. На плоскости $\varepsilon, \sigma$ -- это отрезок $\{\varepsilon = 0, c \ls \sigma \ls a\}$. При этом уравнения \eqref{deq2_8} разрешимы в явном виде:
\begin{equation*}
    \ds u^2 = \frac{\lambda \mu ^ 2}{(\lambda - \mu)^2}(\lambda - \sigma), \qquad v^2 = \frac{\lambda^2 \mu}{(\lambda - \mu)^2}(\sigma - \mu).
\end{equation*}
Отсюда следует, что области возможности движения $U_{0, \sigma}$ есть пересечение основного многоугольника\eqref{deq2_6} с множеством $\{(\lambda, \mu): \lambda \gs \sigma \gs \mu \}$.
Обобщенная граница состоит из отрезков прямых $\lambda = \sigma$ или $\mu = \sigma$. Внутри ОВД существует четыре допустимых скорости, на обобщенной границе -- две. Исключение составляет отмеченный выше случай \eqref{deq2_14}, когда в вершине $\lambda = \mu = b$ имеется целая окружность допустимых скоростей. Множества $U_{0, \sigma}$ таковы: кольцо со средней линией $\nu_3 = 0$ при $c < \sigma < b$ (при $\sigma = 0$ оно вырождается в окружность $\nu_3 = 0$); кольцо со средней линией $\nu_1 = 0$ при $b < \sigma < a$ (при $\sigma = a$ происходит вырождение в окружность $\nu_1 = 0$). В этих случаях обобщенная граница совпадает с топологической. Множество $U_{0, b}$ -- вся сфера, обобщенная граница --  $\nu_2 = 0$. Типичные траектории вектора $\nu$ в областях $U_{0, \sigma}$ топологически устроены так же, как геодезические на трехосном эллипсоиде \cite{bib22}.
Перестройки ОВД происходят лишь в точках \eqref{deq1_17} бифуркационного множества. Это, естественно, поскольку при $g = 0$ система, описывающая движение $\mbs{\nu}$ принадлежит типу Лиувилля (см. пример \ref{exa331}).

Связь с интегральными многообразиями весьма проста: в случаях $\sigma = a$ и $\sigma = b$ обе окружности проектируются в соответствующее координатное сечение $S^2$; при $c < \sigma < b$ и $b < \sigma < a$ оба интегральных тора отображаются на одну и ту же полосу с различными направлениями движения по ней. Интересен случай $\sigma = b$: имеются две периодические траектории -- движения по сечению $\nu_2 = 0$. Они служат $\alpha$- и $\omega$-предельными циклами для всех остальных траекторий, каждая из которых через равные промежутки времени пересекает пару диаметрально противоположных точек \eqref{deq2_12}. Пучки траекторий, содержащих фиксированную пару, образуют в фазовом пространстве двумерные торы. Эти два тора пересекаются по двум предельным циклам, что и дает множество $(S^1 \ddot{\cup} S^1) {\times} S^1$.

Переходя к случаю $\varepsilon \neq 0$, отметим некоторые свойства кривой \eqref{deq2_33}. Воспользовавшись выражениями \eqref{deq2_31}, запишем агрегаты, входящие в правые части \eqref{deq2_33}, в виде функций от $\tau, \sigma$:
\begin{equation}\label{deq4_1}
  \ds 2\tau + y(\tau) = 2\sigma + \frac{1}{(\tau - \sigma)^2} \bigl[ 2(1-\varepsilon) (\tau - \sigma)^3 + \frac{\varepsilon}{2} \varphi''(\sigma)(\tau - \sigma)^2 - \varepsilon \varphi(\sigma)\bigr],
\end{equation}
\begin{equation}\label{deq4_2}
\begin{array}{l}
  \ds y^2(\tau) + 4x(\tau) = \frac{\varepsilon}{(\tau - \sigma)^4}\bigl\{-4(1-\varepsilon)(\tau - \sigma)^6 + 2(1-\varepsilon) \varphi''(\sigma)(\tau - \sigma)^5 + \bigr. \\[3mm]
  \qquad +\ds \bigl[4 \varphi'(\sigma) + \frac{\varepsilon}{4}\varphi''^2(\sigma) \bigr](\tau - \sigma)^4+ 4(1 + \varepsilon) \varphi(\sigma)(\tau - \sigma)^3 -\\[3mm]
  \qquad \bigl. -\varepsilon \varphi(\sigma) \varphi''(\sigma)(\tau - \sigma)^2 + \varepsilon \varphi^2(\sigma) \bigr\}.
\end{array}
\end{equation}

Пусть $T_{\varepsilon,\sigma} \subset \bbR$ -- множество значений $\tau$, при которых величина \eqref{deq4_2} неотрицательна.

При $\varepsilon < 1$ имеем $\lim \limits_{\tau \to \infty}(y^2 + 4x) = -\infty$. Поэтому промежуток $T_{\varepsilon, \sigma}$ ограничен. Его верхней и нижней границей являются наибольший $\tau^*$ и наименьший $\tau_*$ корни уравнения $y^2 + 4x = 0$. Из \eqref{deq2_33} следует, что $\lambda(\tau^*) = \mu(\tau^*), \;  \lambda(\tau_*) = \mu(\tau_*)$. Таким образом, кривая обрывается на биссектрисе $\lambda = \mu$.

В случае $\varepsilon = 1$
\begin{eqnarray}
& & \begin{array}{c}
  \ds \lim \limits_{\tau \to \infty}(y^2 + 4x) = 4\varphi'(\sigma) + \frac{1}{4} \varphi''^2(\sigma) = a^2 + b^2 + c^2 -\\
  -2(bc + ca + ab) + 2(a+b+c)\sigma - 3\sigma ^2,
\end{array}\label{deq4_3}\\
& &  \ds \lim \limits_{\tau \to \infty}(2\tau + y) = 2\sigma + \frac{1}{2}\varphi''(\sigma) = a + b + c - \sigma.\label{deq4_4}
\end{eqnarray}
Поскольку значение \eqref{deq4_3} положительно на всем отрезке $\sigma \in [c, a]$, существует
\begin{equation}\label{deq4_5}
    \lim \limits_{\tau \to \infty}(\lambda(\tau), \mu(\tau)) = (\lambda_\infty, \mu_\infty).
\end{equation}
Кроме того, из \eqref{deq4_4}
$$
\lambda_\infty + \mu_\infty = a + b + c - \sigma,
$$
а из \eqref{deq3_1}
$$
\ds \lim \limits_{\tau \to \infty} \frac{d\lambda}{d\tau} = \lim \limits_{\tau \to \infty} \frac{\lambda - \tau}{\tau - \mu} = -1,
$$
так что в точке \eqref{deq4_5} кривая \eqref{deq2_33} касается критической прямой \eqref{deq2_16}.

Изучим значения $\tau$, близкие к $\sigma$. Если $\varphi(\sigma) \neq 0$, то найдется проколотая окрестность $\sigma$, лежащая в $T_{\varepsilon, \sigma}$. При $\tau \to \sigma$, используя выражения \eqref{deq4_1}, \eqref{deq4_2}, получаем
\begin{equation*}
    \begin{array}{c}
       \ds \lambda = \sigma - \frac{\varepsilon}{4} \Bigl[ \frac{2 \varphi(\sigma)}{(\tau - \sigma)^2} - \varphi''(\sigma) \Bigr][1 - \sgn \varphi(\sigma)]+o(1)], \\
       \ds \lambda = \sigma - \frac{\varepsilon}{4} \Bigl[ \frac{2 \varphi(\sigma)}{(\tau - \sigma)^2} - \varphi''(\sigma) \Bigr][1 + \sgn \varphi(\sigma)]+o(1)].
     \end{array}
\end{equation*}
В результате имеем
\begin{equation}\label{deq4_6}
\begin{array}{l}
  \varphi(\sigma) > 0 \Rightarrow \lim \limits_{\tau \to \infty}(\lambda(\tau), \mu(\tau)) = (\sigma, -\infty), \\
  \varphi(\sigma) < 0 \Rightarrow \lim \limits_{\tau \to \infty}(\lambda(\tau), \mu(\tau)) = (+\infty, \sigma).
\end{array}
\end{equation}

Допустим, что $\varphi(\sigma) = 0$, т.е. $\sigma \in \{a, b, c\}$. Пусть $\{\alpha, \beta, \gamma\}$ -- перестановка чисел $a, b, c$, такая, что
\begin{equation}\label{deq4_7}
    \sigma = \alpha, \beta > \gamma.
\end{equation}
Выражение \eqref{deq4_2} принимает вид
\begin{equation}\label{deq4_8}
\begin{array}{c}
  y^2(\tau) + 4x(\tau) =\varepsilon\{-4(1-\varepsilon)(\tau-\alpha)^2+4(1-\varepsilon)(\beta+\gamma-2\alpha)(\tau-\alpha)+\\
  +\varepsilon(\beta+\gamma-2\alpha)^2-4(\beta-\alpha)(\gamma-\alpha)\}.
\end{array}
\end{equation}
Это выражение является квадратным трехчленом с отрицательным старшим коэффициентом. При $\tau = \beta, \gamma$ оно равно $\varepsilon(\beta-\gamma)^2$. Таким образом, отрезок $[\gamma, \beta]$ содержится в $\mathop{\rm int}\nolimits T_{\varepsilon, \alpha}$. В частности, существуют экстремальные точки, соответствующие $\beta, \gamma$.
Значение $\tau = \alpha$ лежит в $T_{\varepsilon, \alpha}$ только в том случае, если
\begin{equation}\label{deq4_9}
    \ds \varepsilon \gs \frac{4(\beta-\alpha)(\gamma-\alpha)}{\varepsilon(\beta+\gamma-2\alpha)^2}.
\end{equation}
Однако при этом точки экстремума нет, так как нарушено условие \eqref{deq2_30}. Переходя к пределу в \eqref{deq2_31}, находим
\begin{equation}\label{deq4_10}
\begin{array}{c}
  [\lambda(\alpha) - \alpha][\alpha - \mu(\alpha)] = \varepsilon \varphi'(\alpha), \\
  \lambda(\alpha) + \mu(\alpha) - 2\alpha = \varepsilon \varphi''(\alpha).
\end{array}
\end{equation}
Равенства \eqref{deq4_10} означают, что при $\tau = \alpha$ кривая \eqref{deq2_33} касается критической гиперболы \eqref{deq2_29}.

Приступим к окончательному построению обобщенных границ и классификации ОВД.

Рассмотрим значения $\varepsilon = 1$, $\sigma \in \{a, b, c\}$. В обозначениях \eqref{deq4_7} из \eqref{deq4_1}, \eqref{deq4_8} получаем
\begin{equation*}
    2\tau + y(\tau) = \beta + \gamma, \qquad y^2(\tau) + 4x(\tau) = (\beta - \gamma)^2,
\end{equation*}
так что кривая \eqref{deq2_33} вырождается в точку
\begin{equation}\label{deq4_11}
    (\lambda(\tau), \mu(\tau)) \equiv (\beta, \gamma).
\end{equation}
Обобщенная граница состоит из критической прямой $\lambda + \mu = \beta + \gamma$, критической гиперболы $(\lambda - \alpha)(\alpha - \mu) = (\beta - \alpha)(\alpha - \gamma)$ и <<точки>> \eqref{deq4_11}. В связи с тем, что вне обобщенной границы допустимых скоростей нет, она полностью исчерпывает ОВД.
На рис.~\ref{fig_4_4},\,{\it а-в} представлены области возможности движения при $\varepsilon  = 1$ в случаях: а) $\sigma = c$; б) $\sigma = b$; в) $\sigma = c$.

\begin{figure}[!ht]
\def\fs{0.4}
\centering
\begin{tabular}{c}
\begin{minipage}[h]{\fs\linewidth}
\center{\includegraphics[width=1\linewidth]{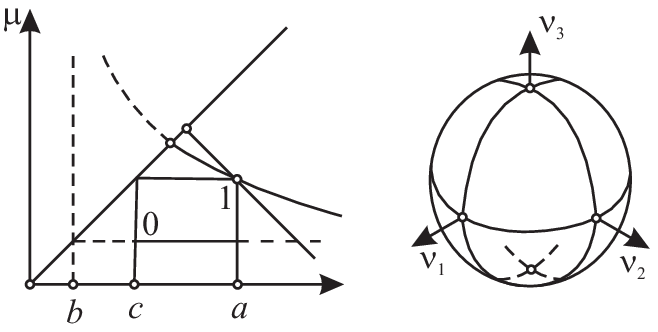} \\а)}
\end{minipage}
\hfill
\begin{minipage}[h]{\fs\linewidth}
\center{\includegraphics[width=1\linewidth]{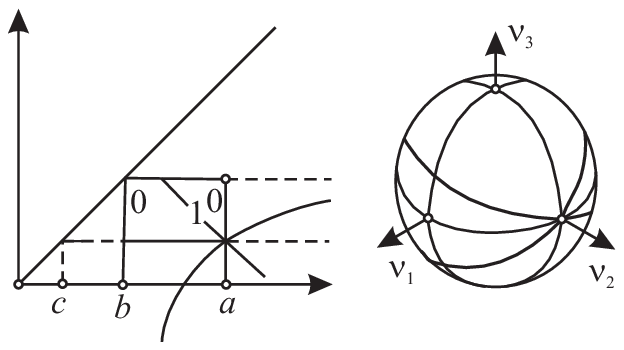} \\б)}
\end{minipage}
\\
\begin{minipage}[h]{\fs\linewidth}
\center{\includegraphics[width=1\linewidth]{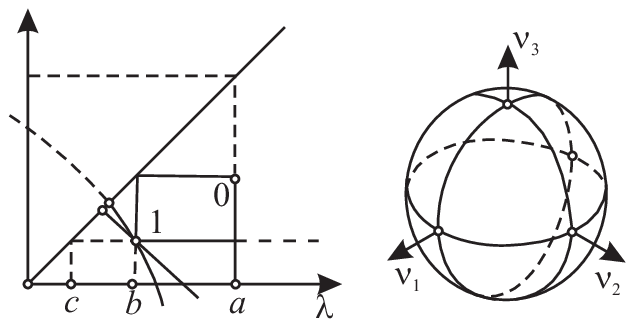} \\в)}
\end{minipage}
\end{tabular}
\caption{}\label{fig_4_4}
\end{figure}

Пусть $\varepsilon  = 1$, $c < \sigma < b$. При этом критическая прямая пересекает стороны $\lambda = a$ и $\mu = b$ основного прямоугольника. Критическая гипербола исчезла. <<Точка>> \eqref{deq4_11} стала настоящей кривой. Опишем ее свойства, вытекающие из \eqref{tab43} и результатов данного параграфа. Область определения $T_{1, \sigma} = \bbR \backslash \{\sigma\}$. Кривая начинается при $\tau = -\infty$ в точке \eqref{deq4_5} на критической прямой, при $\tau = c$ касается прямой $\mu = c$ сверху, уходит в бесконечность с асимптотой $\mu = \sigma$, когда $\tau \to \sigma$ (см. \eqref{deq4_6}). При $\tau = b$ происходит касание с прямой $\mu = b$ (снизу), а при $\tau = a$ -- с прямой $\lambda = a$ (справа). Последняя экстремальная точка принадлежит прямоугольнику, но, как отмечалось, является посторонним решением (касание внешнее). В результате получаем картину, изображенную на рис.~\ref{fig_4_5},\,{\it а}.

\begin{figure}[ht]
\def\fs{0.4}
\centering
\begin{tabular}{c}
\begin{minipage}[h]{\fs\linewidth}
\center{\includegraphics[width=1\linewidth]{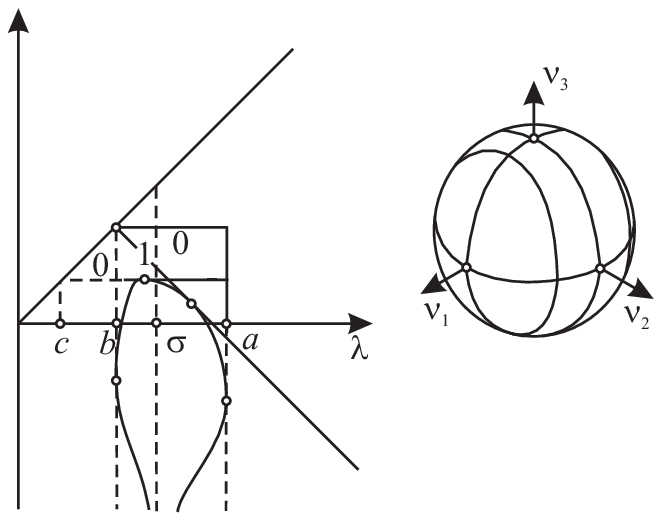} \\а)}
\end{minipage}
\\
\begin{minipage}[h]{\fs\linewidth}
\center{\includegraphics[width=1\linewidth]{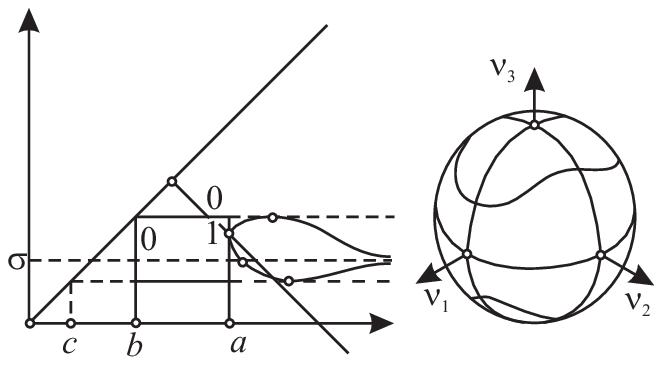} \\б)}
\end{minipage}
\end{tabular}
\caption{}\label{fig_4_5}
\end{figure}

Случай $\varepsilon = 1$, $b < \sigma < a$ разбирается аналогично. Результат приведен на рис.~\ref{fig_4_5},\,{\it б}.
Область возможности движения вновь совпадает со своей обобщенной границей. Вспоминая структуру интегральных многообразий (табл. \ref{tab42}), видим, что они проектируются на сферу гомеоморфно.

Отметим непрерывную зависимость кривых от параметра $\sigma$. Когда он стремится к значениям $a, b, c$, геометрический образ кривой \eqref{deq2_33} <<складывается вдвое>>, приближаясь к части критической гиперболы.

Пусть $\varepsilon<1$, а $\sigma=\alpha \in \{a, b,c\}$. В этом случае критическая гипербола имеет вид \eqref{deq2_29} и пересекает внутренность основного прямоугольника. Кривая \eqref{deq2_33} заменяет собой отрезок исчезнувшей критической прямой и при $\tau = \alpha$ касается гиперболы в силу \eqref{deq4_10}. Область определения $T_{\varepsilon, \alpha} = [\tau_*, \tau^*]$.

Поясним случай $\alpha = c$. Если $\varepsilon$ близко к единице, то неравенство \eqref{deq4_9} строгое и $\tau_* < c$. Кривая \eqref{deq2_33} начинается на прямой $\lambda = \mu$ ($\tau = \tau_*$), касается критической гиперболы $(\tau = c)$, прямых $\mu = b$ ($\tau = b$), $\lambda = a$ ($\tau = a$) и заканчивается на прямой $\lambda = \mu$ ($\tau = \tau^*$). Эта ситуация показана на рис.~\ref{fig_4_6},\,{\it а}. При обращении \eqref{deq4_9} в равенство произойдет слияние $\tau_* = c$, точка касания с критической гиперболой попадет на прямую $\lambda = \mu$ и при дальнейшем уменьшении $\varepsilon$ исчезнет. Случай $\alpha = a$ аналогичен (рис.~\ref{fig_4_6},\,{\it б}).

\begin{figure}[ht]
\def\fs{0.4}
\centering
\begin{tabular}{c}
\begin{minipage}[h]{\fs\linewidth}
\center{\includegraphics[width=1\linewidth]{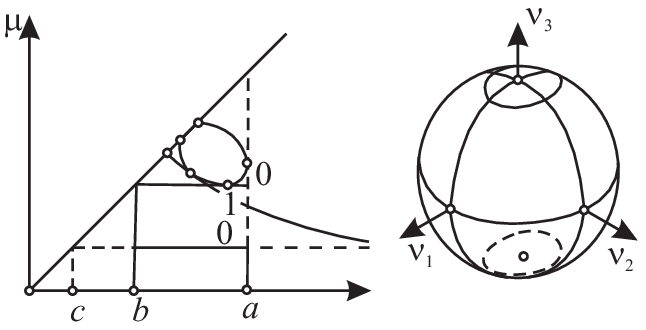} \\а)}
\end{minipage}
\hfill
\begin{minipage}[h]{\fs\linewidth}
\center{\includegraphics[width=1\linewidth]{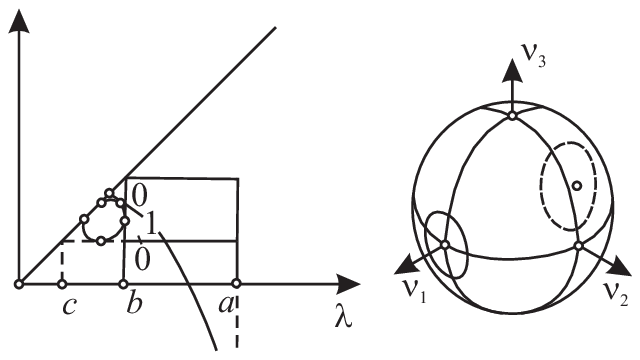} \\б)}
\end{minipage}
\end{tabular}
\caption{}\label{fig_4_6}
\end{figure}

Положим $\alpha = b$. Неравенство \eqref{deq4_9} выполнено для всех допустимых $\varepsilon$ и $\tau_* < c$, $a< \tau^*$. Вид ОВД зависит от точки $(\varepsilon, b)$ относительно разделяющих отрезков \eqref{deq3_6}, \eqref{deq3_13}. На рис.~\ref{fig_4_7},\,{\it а-в} приведены области возможности движения при $\sigma=b$ в предположении $a + c > 2b$ для случаев: а)~$\ds \frac{a-b}{a-c} < \varepsilon < 1$; б)~$\ds \frac{b-c}{a-c} < \varepsilon< \frac{a-b}{a-c}$; в)~$\ds 0 < \varepsilon < \frac{b-c}{a-c}$. Распределение допустимых скоростей определяется по правилам предыдущего параграфа с учетом того, что в вершине $(a, c)$ четыре скорости, а на отрезках от $\tau_*$ и $\tau^*$ до ближайших экстремальных точек ${x(\tau) < 0}$. Связь с интегральным многообразием $(S^1 \ddot{\cup} S^1) {\times} S^1$ легче всего увидеть на рис.~\ref{fig_4_7},\,{\it а}. Критическая гипербола дает на сфере две периодические траектории, а две полосы, охватывающие сферу суть образы торов, пересекающихся по этим циклам. Поскольку всякая фазовая  траектория стремится к периодическим, ясно, как устроен ее образ на сфере.

\begin{figure}[ht]
\def\fs{0.4}
\centering
\begin{tabular}{c}
\begin{minipage}[h]{\fs\linewidth}
\center{\includegraphics[width=1\linewidth]{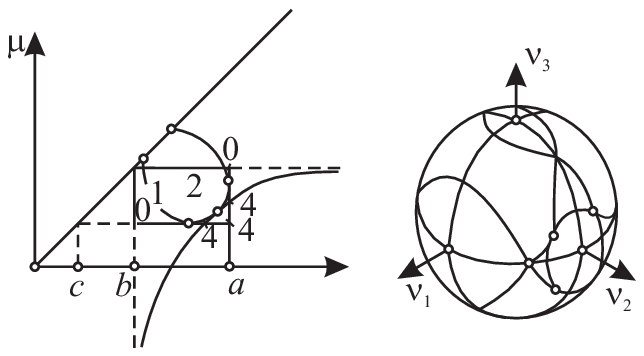} \\а)}
\end{minipage}
\hfill
\begin{minipage}[h]{\fs\linewidth}
\center{\includegraphics[width=1\linewidth]{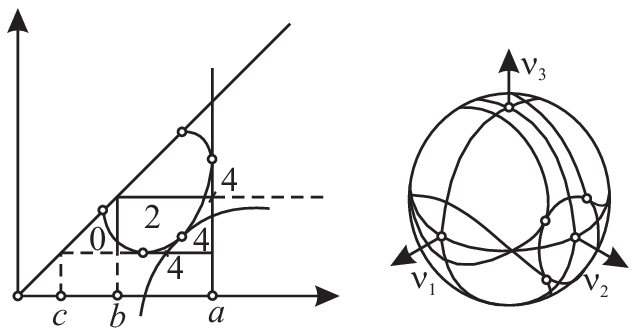} \\б)}
\end{minipage}
\\
\begin{minipage}[h]{\fs\linewidth}
\center{\includegraphics[width=1\linewidth]{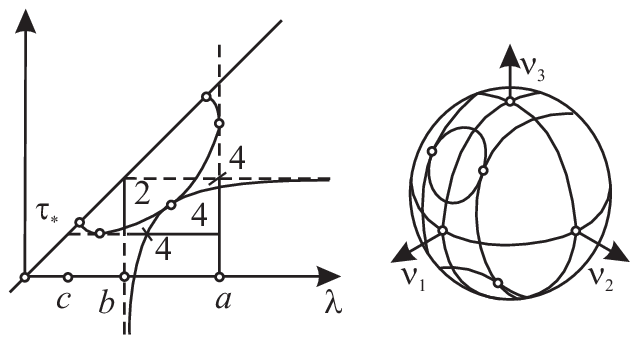} \\в)}
\end{minipage}
\end{tabular}
\caption{}\label{fig_4_7}
\end{figure}

Рассмотрим некритические случаи. На рис.~\ref{fig_4_8} обозначены открытые области, на которые множество $[0, 1] {\times} [c, a] \subset \bbR^2(\varepsilon, \sigma)$ разбивается разделяющими кривыми \eqref{deq3_6} -- \eqref{deq3_14}, если $a + c > 2b$. При  $a + c = 2b$ исчезает область~10. Если  $a + c < 2b$, то область, аналогичная области~10, возникает ниже прямой $\sigma = b$. Выясним, что происходит при отходе от бифуркационного множества в области 1, 5, 14, 18.

\begin{figure}[ht]
\centering
\center{\includegraphics[width=0.35\linewidth]{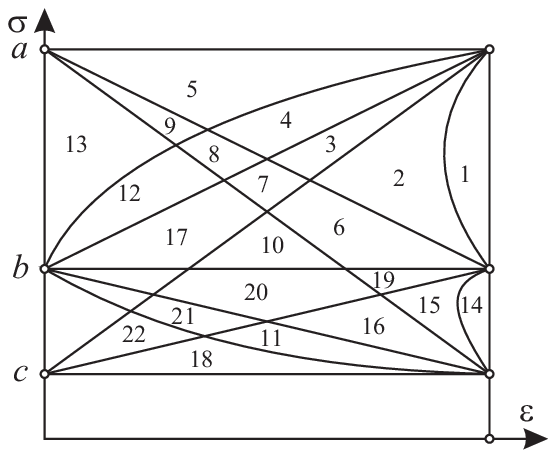}}
\caption{}\label{fig_4_8}
\end{figure}

Исследуем переход из точки $\varepsilon = 1$, $b < \sigma < a$ в область 1. При этом исчезает критическая прямая. Область изменения $\tau$ в \eqref{deq2_33} становится ограниченной: $T_{\varepsilon, \sigma} =[\tau_*, \tau^*] \backslash \{\sigma\}$. Если $\varepsilon$ достаточно близко к единице, то $[c, a] \subset [\tau_*, \tau^*]$. При некотором $\tau = \tau_0$ существует точка возврата, причем так как $\varphi(\sigma) > 0$, то $\tau_0 < \sigma$ (см. \eqref{deq3_4}). Поскольку точка возврата лежит здесь вне основного прямоугольника, то на самом деле $\tau_0 < c$. При $\tau \to \sigma$ имеем первый случай \eqref{deq4_6}. Характер экстремальных точек такой же, как и на рис.~\ref{fig_4_5},\,{\it б} (точка $(\varepsilon, \sigma)$ разделяющих отрезков не пересекает). Кривая \eqref{deq2_33} принимает вид, представленный на рис.~\ref{fig_4_9},\,{\it а}. Образно говоря, произошло раздвоение ограниченного участка критической прямой.

\begin{figure}[ht]
\centering
\begin{tabular}{cc}
\includegraphics[width=0.35\linewidth]{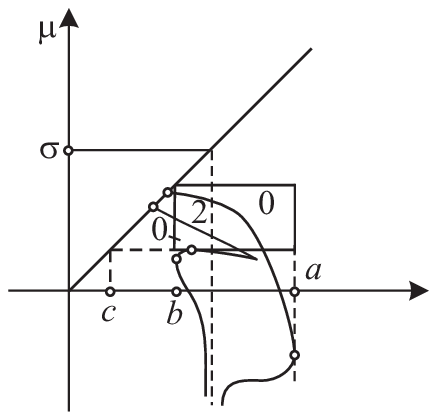} & \includegraphics[width=0.35\linewidth]{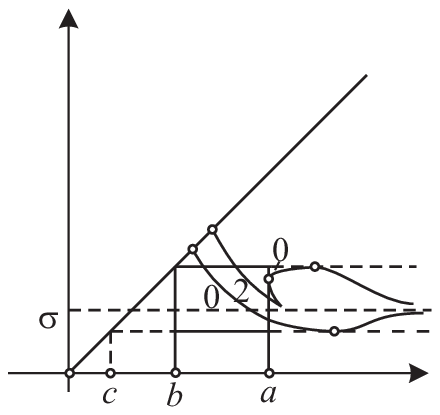}
\\
а) & б) \\
\includegraphics[width=0.35\linewidth]{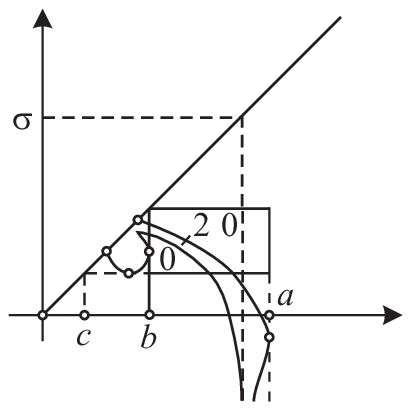} &
\includegraphics[width=0.35\linewidth]{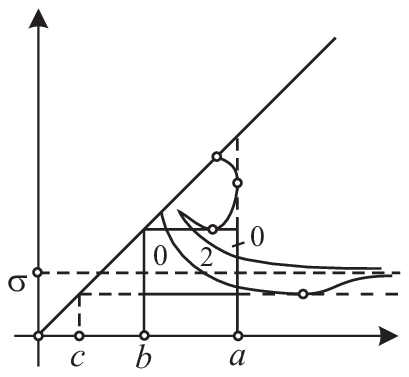}
\\
в) & г) \\
\end{tabular}
\caption{}\label{fig_4_9}
\end{figure}

Аналогичное явление происходит при переходе из точки $\varepsilon = 1$, $c < \sigma < b$ в область 14. Кривые рис.~\ref{fig_4_5},\,{\it а} трансформируются в кривую, показанную на рис.~\ref{fig_4_9},\,{\it б}.

Рассмотрим переход из точки $0 < \varepsilon <1$, $\sigma = a$ в область 5. Исчезает критическая гипербола. Точка касания кривой \eqref{deq2_33} и гиперболы при $\tau = a$ превращается в точку касания с прямой $\lambda = a$, возникают вертикальная асимптота при $\tau \to \sigma$ и точка возврата на участке $\tau < \sigma$. Последняя лежит вне прямоугольника \eqref{deq2_6}. Общий вид кривой \eqref{deq2_33} приведен на рис.~\ref{fig_4_9},\,{\it б}. Он получается из кривых на рис.~\ref{fig_4_6},\,{\it б} раздвоением неограниченного участка гиперболы.

Переход из точки $0 < \varepsilon <1$, $\sigma = c$ в область 18 сопровождается раздвоением неограниченного участка критической гиперболы на рис.~\ref{fig_4_6},\,{\it а}. Результат показан на рис.~\ref{fig_4_9},\,{\it г}.

\begin{figure}[ht]
\def\fs{0.15}
\def\mps{1}
\centering
\begin{minipage}[h]{\fs\linewidth}
\center{\includegraphics[width=\mps\linewidth]{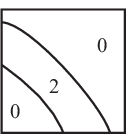} \\ а)}
\end{minipage}
\hfill
\begin{minipage}[h]{\fs\linewidth}
\center{\includegraphics[width=\mps\linewidth]{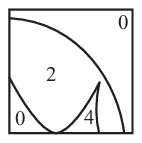} \\ б)}
\end{minipage}
\begin{minipage}[h]{\fs\linewidth}
\center{\includegraphics[width=\mps\linewidth]{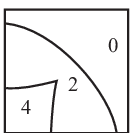} \\ в)}
\end{minipage}
\hfill
\begin{minipage}[h]{\fs\linewidth}
\center{\includegraphics[width=\mps\linewidth]{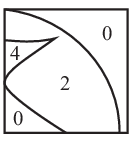} \\ г)}
\end{minipage}
\hfill
\begin{minipage}[h]{\fs\linewidth}
\center{\includegraphics[width=\mps\linewidth]{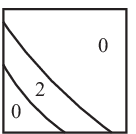} \\ д)}
\end{minipage}
\\
\begin{minipage}[h]{\fs\linewidth}
\center{\includegraphics[width=\mps\linewidth]{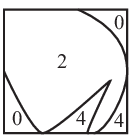} \\ е)}
\end{minipage}
\hfill
\begin{minipage}[h]{\fs\linewidth}
\center{\includegraphics[width=\mps\linewidth]{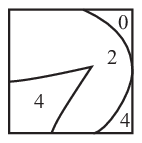} \\ ж)}
\end{minipage}
\begin{minipage}[h]{\fs\linewidth}
\center{\includegraphics[width=\mps\linewidth]{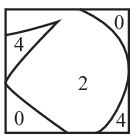} \\ з)}
\end{minipage}
\hfill
\begin{minipage}[h]{\fs\linewidth}
\center{\includegraphics[width=\mps\linewidth]{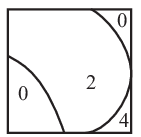} \\ и)}
\end{minipage}
\hfill
\begin{minipage}[h]{\fs\linewidth}
\center{\includegraphics[width=\mps\linewidth]{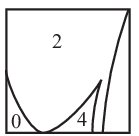} \\ к)}
\end{minipage}
\\
\begin{minipage}[h]{\fs\linewidth}
\center{\includegraphics[width=\mps\linewidth]{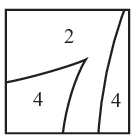} \\ л)}
\end{minipage}
\hfill
\begin{minipage}[h]{\fs\linewidth}
\center{\includegraphics[width=\mps\linewidth]{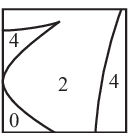} \\ м)}
\end{minipage}
\begin{minipage}[h]{\fs\linewidth}
\center{\includegraphics[width=\mps\linewidth]{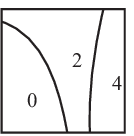} \\ н)}
\end{minipage}
\hfill
\begin{minipage}[h]{\fs\linewidth}
\center{\includegraphics[width=\mps\linewidth]{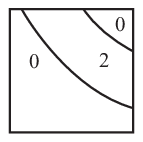} \\ о)}
\end{minipage}
\hfill
\begin{minipage}[h]{\fs\linewidth}
\center{\includegraphics[width=\mps\linewidth]{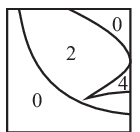} \\ п)}
\end{minipage}
\\
\begin{minipage}[h]{\fs\linewidth}
\center{\includegraphics[width=\mps\linewidth]{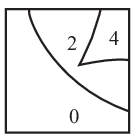} \\ р)}
\end{minipage}
\hfill
\begin{minipage}[h]{\fs\linewidth}
\center{\includegraphics[width=\mps\linewidth]{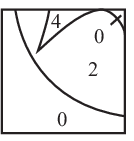} \\ с)}
\end{minipage}
\begin{minipage}[h]{\fs\linewidth}
\center{\includegraphics[width=\mps\linewidth]{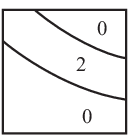} \\ т)}
\end{minipage}
\hfill
\begin{minipage}[h]{\fs\linewidth}
\center{\includegraphics[width=\mps\linewidth]{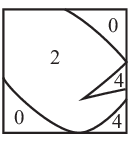} \\ у)}
\end{minipage}
\hfill
\begin{minipage}[h]{\fs\linewidth}
\center{\includegraphics[width=\mps\linewidth]{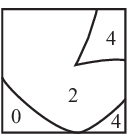} \\ ф)}
\end{minipage}
\\
\begin{minipage}[h]{\fs\linewidth}
\center{\includegraphics[width=\mps\linewidth]{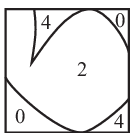} \\ х)}
\end{minipage}
\hfill
\begin{minipage}[h]{\fs\linewidth}
\center{\includegraphics[width=\mps\linewidth]{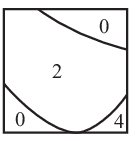} \\ ц)}
\end{minipage}
\begin{minipage}[h]{\fs\linewidth}
\center{\includegraphics[width=\mps\linewidth]{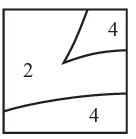} \\ ч)}
\end{minipage}
\hfill
\begin{minipage}[h]{\fs\linewidth}
\center{\includegraphics[width=\mps\linewidth]{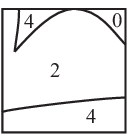} \\ ш)}
\end{minipage}
\hfill
\begin{minipage}[h]{\fs\linewidth}
\center{\includegraphics[width=\mps\linewidth]{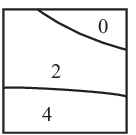} \\ щ)}
\end{minipage}
\caption{}\label{fig_4_10}
\end{figure}

Дальнейшие перестройки областей возможности движения полностью определяются по смыслу разделяющих кривых. Окончательный итог для областей 1 -- 25 указан в основном прямоугольнике на рисунке~\ref{fig_4_10},\,{\it а-щ} соответственно. При переносе ОВД на сферу следует учесть особенность равенств \eqref{deq2_7} на главных сечениях (кривая, трансверсальная стороне прямоугольника, будет ортогональна главному сечению сферы, а кривая, имеющая простое касание со стороной прямоугольника, образует с главным сечением острый угол). Назовем компонентой ОВД ее часть, отвечающую одной связной компоненте интегрального многообразия -- двумерному тору. Для большей наглядности возникающие здесь компоненты ОВД указаны отдельно на рис.~\ref{fig_4_11},\,{\it а-б}. Они устойчивы не только по отношению к параметрам $(\varepsilon, \sigma)$, но и в смысле устойчивости гладкого отображения. На разделяющих кривых реализуются промежуточные, неустойчивые виды ОВД. Во всех этих областях траектории представляют собой образ условно-периодической обмотки тора.

\begin{figure}[ht]
\def\fs{0.12}
\centering
\begin{tabular}{ccc}
\includegraphics[width=\fs\linewidth]{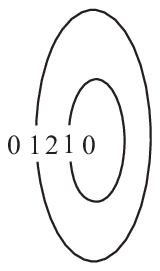}
&
\includegraphics[width=\fs\linewidth]{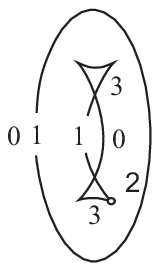}
&
\includegraphics[width=\fs\linewidth]{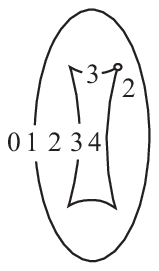}
\\
а) & б) & в)
\end{tabular}
\caption{}\label{fig_4_11}
\end{figure}

Области возможности движения на сфере, за исключением случая $\sigma = b$, всегда состоят из двух центрально-симметричных друг другу компонент одного из указанных видов. При этом в областях 8 -- 13, 19 -- 25 эти компоненты частично перекрываются, что объясняет наличие четырех допустимых скоростей и вне участков, заключенных внутри треугольных <<хвостиков>>. Случаи 1 -- 25 отличаются типом компонент, ориентацией их относительно сечений сферы координатными плоскостями и расположением участков суперпозиций компонент. Отметим, что такие участки всегда содержат пересечения сферы с одной из координатных осей.

\clearpage

\subsection{Комментарий к главе 4}\label{ssec45}

Обобщение случая Эйлера, открытое Н.Е.\,Жуковским \cite{bib20}, сравнительно мало изучено. В.Н.\,Ру\-ба\-нов\-ский полностью исследовал равномерные вращения свободного гиростата \cite{bib40}, их бифуркации и устойчивость. В работе \cite{bib72} предложен новый способ способ сведения задач к квадратурам, послуживший основой для классификации <<полодий>> -- подвижных годографов вектора угловой скорости тела-носителя.

Фазовая топология случая Эйлера\,--\,Пуансо изучалась в работах \cite{bib56,bib68} с несколько иных позиций, чем в настоящей главе. А именно, в шестимерном фазовом пространстве $SO(3) {\times} \bbR^3$ фиксировались постоянная интеграла энергии и три проекции вектора кинетического момента на неподвижные оси.

В результате получены интегральные поверхности следующих типов: $2S^1$, $(S^1 \ddot{\cup} S^1) {\times} S^1$, $2{\bf T}^2$. При нашем подходе классификация аналогична, если равна нулю постоянная площадей. Это совпадение не случайно: поверхность кинетического момента в $SO(3) {\times} \bbRR$ и поверхность $\{K = k, G = 0\}$ в $\mathfrak{M} = S^2 {\times} \bbRR$ диффеоморфны между собой. По природе своей, однако, эти многообразия совершенно различны: в первом случае имеем левоинвариантное сечение касательного расслоения группы $SO(3)$, во втором же -- расслоение касательных векторов к двумерной сфере.

Укажем тесно примыкающий к $\S$ \ref{ssec44} цикл работ \cite{bib60,bib36,bib37}, посвященный геометрическому и топологическому анализу интегрируемой задачи Клебша, интерпретируемой как задача о движении твердого тела вокруг центра масс под действием линейного силового поля (в качестве потенциальной функции берутся квадратичные члены разложения ньютоновского потенциала). Методика исследования вплоть до замены переменных, совпадает с используемой в настоящей работе. Набор областей возможности движения в задаче Клебша, естественно, богаче. Однако, новых типов бифуркаций интегральных многообразий, кроме некоторых особых случаев, не возникает.

В заключение отметим, что по теореме А.Н.\,Колмогорова \cite{bib04} при воздействии на тело малых потенциальных и гироскопических сил большинство траекторий в фазовом пространстве уравнений Эйлера\,--\,Пуассона будут условно-периодическим образом заполнять двумерные торы, близкие к интегральным случаям Эйлера. Поскольку в свою очередь большинство последних отображается на сферу устойчиво относительно малых деформаций, то в задачах, близких к изученной, значительная часть траекторий в теле неподвижных в пространстве векторов будет заполнять области, эквивалентные изображенным на рис.~\ref{fig_4_11}. При желании этому замечанию можно придать достаточную строгость.

\clearpage

\section{Фазовая топология\\  решения Чаплыгина -- Сретенского}\label{sec5}

\subsection{Равномерные вращения}\label{ssec51}

Предположим, что гиростат, движение которого описывается системой \eqref{beq4_29}, подчинен следующим условиям: моменты инерции удовлетворяют соотношениям $A_1 = A_2 = 4A_3$, центр масс находится в экваториальной плоскости эллипсоида инерции, гиростатический момент направлен по оси динамической симметрии $(\lambda_1 = \lambda_2 = 0)$. Поместим гиростат в поле силы тяжести. Выберем подвижные оси так, чтобы центр масс лежал на первой из них. Полагая вектор $\mbs{\nu}_0$ направленным вертикально вниз, запишем потенциальную энергию системы в виде $\Pi(\mbs{\nu}) = -\gamma \nu_1$, где $\gamma > 0$ есть произведение веса гиростата на расстояние от центра масс до неподвижной точки.

Назначим $\sqrt{\gamma / A_3}$ единицей измерения угловой скорости. Обратная величина $\sqrt{A_3 / \gamma}$ становится тем самым единицей измерения времени. Полагая $\lambda = \lambda_3 / \sqrt{\gamma A_3}$, перепишем уравнения \eqref{beq4_29} в безразмерных переменных:
\begin{eqnarray}
&    4 \dot{\omega}_1 = 3 \omega_2 \omega_3 - \lambda \omega_2, \quad 4\dot{\omega}_2 = -3 \omega_3 \omega_1 + \lambda \omega_1 - \omega_3, \quad \dot{\omega_3} = \nu_2,\label{eeq1_1}\\
&    \dot{\nu_1} = \nu_2 \omega_3 - \nu_3 \omega_2, \quad  \dot{\nu_2} = \nu_3 \omega_1 - \nu_1 \omega_3, \quad  \dot{\nu_3} = \nu_1 \omega_2 - \nu_2 \omega_1.\label{eeq1_2}
\end{eqnarray}
Интеграл энергии \eqref{beq3_12} с точностью до постоянного множителя $\gamma$, имеет вид
\begin{equation}\label{eeq1_3}
    \ds H = 2(\omega_1^2 + \omega_2^2) + \frac{1}{2}\omega_3^2 - \nu_1.
\end{equation}
Будем рассматривать лишь движения с нулевой постоянной интеграла площадей \eqref{beq4_27}
\begin{equation}\label{eeq1_4}
    4(\omega_1 \nu_1 + \omega_2 \nu_2) + (\omega_3 + \lambda) \nu_3 = 0.
\end{equation}
На множестве \eqref{eeq1_4} система \eqref{eeq1_1}, \eqref{eeq1_2} обладает дополнительным, в терминологии С.А.\,Чап\-лы\-ги\-на -- частным, интегралом
\begin{equation}\label{eeq1_5}
    K = 2(\omega_3 - \lambda)(\omega_1^2 + \omega_2^2) + 2 \omega_1 \nu_3.
\end{equation}
Он впервые отмечен в работе Л.Н.\,Сретенского \cite{bib45}.

Изучим неподвижные точки уравнений Эйлера\,--\,Пуассона, отвечающие равномерным вращениям тела-носителя (в частности, положениям равновесия).

Пусть $(\mbs{\nu}, \mbs{\omega})$ -- стационарное решение. Тогда
\begin{eqnarray}
& &    (3 \omega_3 - \lambda)\omega_2 = 0,\label{eeq1_6}\\
& &    (3 \omega_3 - \lambda)\omega_1 + \nu_3 = 0,\label{eeq1_7}\\
& &    \nu_2 = 0,\label{eeq1_8}\\
& &    \nu_3 \omega_2 - \nu_1 \omega_2 = 0,\label{eeq1_9}\\
& &    \nu_3 \omega_1 - \nu_1 \omega_3 = 0.\label{eeq1_10}
\end{eqnarray}
При всех значениях $\lambda$ эти условия выполнены в точках
\begin{equation}\label{eeq1_11}
    \omega_1 = \omega_2 = \omega_3 = 0, \quad \nu_1 = \pm 1, \quad \nu_2 = \nu_3 = 0.
\end{equation}
Им соответствуют постоянные интегралов \eqref{eeq1_3}, \eqref{eeq1_5}
\begin{equation}\label{eeq1_12}
    h = \mp 1, \qquad k = 0.
\end{equation}

Допуская, что $3 \omega_3 - \lambda = 0$, по формулам \eqref{eeq1_7} -- \eqref{eeq1_10}  находим $\nu_2 = \nu_3 = 0$, $\omega_2 = \omega_3 = 0$. Используя \eqref{eeq1_4}, приходим к \eqref{eeq1_11}. Пусть $3 \omega_3 - \lambda \neq 0$. Из равенств \eqref{eeq1_6}, \eqref{eeq1_7} следует
\begin{eqnarray}
&    \omega_2 = 0,\label{eeq1_13}\\
&    \ds \omega_1 = -\frac{\nu_3}{3 \omega_3 - \lambda}.\label{eeq1_14}
\end{eqnarray}
Случай $\nu_3 = 0$ вновь приводит к равенствам \eqref{eeq1_11}. Полагая $\nu_3 \neq 0$, из \eqref{eeq1_4}, \eqref{eeq1_3}, \eqref{eeq1_14} получим
\begin{equation}\label{eeq1_15}
    \ds \nu_1 = \frac{1}{4}(3 \omega_3 - \lambda)(\omega_3 + \lambda).
\end{equation}
Теперь из \eqref{eeq1_10} вытекает
\begin{equation}\label{eeq1_16}
    \ds \pm \nu_3 = \frac{1}{2}\sqrt{-\omega_3(\omega_3 + \lambda)(3 \omega_3 - \lambda)^2}.
\end{equation}
Условие единичности вектора $\mbs{\nu}$ с учетом \eqref{eeq1_8} определяет величину $\omega_3$:
\begin{equation}\label{eeq1_17}
    (\omega_3 + \lambda)(3 \omega_3 - \lambda)^3 + 16 = 0.
\end{equation}
При этом согласно \eqref{eeq1_16} должно быть
\begin{equation}\label{eeq1_18}
    \omega_3(\omega_3 + \lambda) \ls 0.
\end{equation}
Легко видеть, что при $\lambda^2 < 4/3$ уравнение \eqref{eeq1_17} корней не имеет. В случае ${\lambda^2 = 4/3}$ корень один: $\omega_3 = -2\lambda / 3$. Если же
\begin{equation}\label{eeq1_19}
    4/3 < \lambda^2 < 4,
\end{equation}
то оба корня \eqref{eeq1_17} удовлетворяют строгому неравенству \eqref{eeq1_18}. При $\lambda^2 = 4$ один из корней нулевой, и соответствующая стационарная точка совпадает с точкой \eqref{eeq1_11}, взятой с нижним знаком. Наконец, когда $\lambda^2 > 4$, условие \eqref{eeq1_18} выполнено лишь для одного корня уравнения \eqref{eeq1_17}.
Значения постоянных интегралов $H$ и $K$ в стационарной точке \eqref{eeq1_8}, \eqref{eeq1_13} -- \eqref{eeq1_16} таковы:
\begin{equation}\label{eeq1_20}
    \ds h = \frac{1}{4}(\lambda^2 - 4\omega_3 \lambda - 3 \omega_3^2), \qquad k = \omega_3^2(\omega_3 + \lambda).
\end{equation}

Очевидно, решения \eqref{eeq1_11} -- это положения равновесия тела (центр масс лежит на вертикали, проходящей через точку закрепления), а \eqref{eeq1_8}, \eqref{eeq1_13} -- \eqref{eeq1_16} суть равномерное вращение, при котором гиростатический момент лежит в вертикальной плоскости, содержащей центр масс и неподвижную точку.

\subsection{Бифуркационное множество\\  и его связь с разделением переменных}\label{ssec52}

Интегральное отображение определим согласно \eqref{ceq4_11}, где функции $H$ и $K$ имеют вид \eqref{eeq1_3}, \eqref{eeq1_5}, а постоянная $g = 0$.

Исследуем отдельно на критичность точки, в которых $\omega_1^2 + \omega_2^2 = 0$. Полагая вдобавок $\nu_3 = 0$, видим, что обращается в нуль градиент по переменным Эйлера\,--\,Пуассона функции $K$. Таким образом, значения
\begin{equation}\label{eeq2_1}
    k = 0, \quad h = -1
\end{equation}
всегда принадлежат бифуркационному множеству $\Sigma_0$. Луч $\eqref{eeq2_1}$ содержит и точки \eqref{eeq1_12}, в которых зависимы интегралы энергии и площадей.

Для нахождения остальных критических точек введем новые переменные
\begin{equation}\label{eeq2_2}
    x, y, \alpha, \beta,
\end{equation}
полагая
\begin{equation}\label{eeq2_3}
\begin{array}{c}
    \ds \omega_1 = \frac{1}{2}\sqrt{-xy} \sin \frac{\alpha + \beta}{2}, \quad \omega_2 = \frac{1}{2}\sqrt{-xy} \cos \frac{\alpha + \beta}{2}, \quad \omega_3 = x + y - \lambda, \\
    \ds \nu_1 = \frac{x \cos \alpha - y \cos \beta}{x - y}, \quad \nu_2 = \frac{-x \sin \alpha + y \sin \beta}{x - y}, \\
    \ds \nu_3 = \frac{2\sqrt{-xy}}{x - y} \sin \frac{\alpha - \beta}{2}.
\end{array}
\end{equation}
Тем самым удовлетворены соотношения \eqref{beq4_5} и \eqref{eeq1_4}. Областью изменения новых переменных назначим
\begin{equation}\label{eeq2_4}
    \ds x \gs 0, \quad y \ls 0, \quad 0 \ls \frac{\alpha - \beta}{2} < 2\pi, \quad 0 \ls \frac{\alpha + \beta}{2} < 2\pi.
\end{equation}
Замена \eqref{eeq2_3} теряет смысл, если
\begin{equation}\label{eeq2_5}
    x = y = 0.
\end{equation}
В соответствующих точках поверхности $G^{-1}(0) \subset \mathfrak{M}$ имеем
\begin{equation}\label{eeq2_6}
    \omega_1 = \omega_2 = 0, \quad \omega_3 = -\lambda.
\end{equation}
Образ этого множества при интегральном отображении представляет собой отрезок
\begin{equation}\label{eeq2_7}
    \ds k = 0, \quad  -1 + \frac{\lambda^2}{2} \ls h \ls 1 + \frac{\lambda^2}{2},
\end{equation}
принадлежащий уже отмеченному лучу \eqref{eeq2_1}, так что равенствами \eqref{eeq2_3} можно пользоваться без дальнейших оговорок.
Функции \eqref{eeq1_2}, \eqref{eeq1_5} примут вид
\begin{equation}\label{eeq2_8}
\begin{array}{c}
  \ds H = \frac{1}{2}(x^2 + xy + y^2) - \frac{x \cos \alpha - y \cos \beta}{x - y} - \lambda (x + y) + \frac{1}{2} \lambda^2, \\
  \ds K = \frac{1}{2}xy(x + y) + \frac{x \cos \alpha - y \cos \beta}{x - y} - \lambda xy.
\end{array}
\end{equation}

Отметим тождества, вытекающие из приведенных выражений:
\begin{equation}\label{eeq2_9}
\begin{array}{l}
    \ds x \cos \alpha = \frac{1}{2}x(x - \lambda)^2 - xH - K,\quad
    \ds y \cos \beta = \frac{1}{2}y(y - \lambda)^2 - yH - K.
\end{array}
\end{equation}

Частные производные функции \eqref{eeq2_8} связаны соотношениями
$$
(K_x, K_\alpha) = -y(H_x, H_\alpha), \qquad (K_y, K_\beta) = -x(H_y, H_\beta).
$$
Отсюда следует, что всякая точка, критическая для $H$, является критической для $K$, поэтому для исследования зависимости $H$ и $K$ достаточно ввести неопределенный множитель $\sigma$, потребовав $\grad_4K = \sigma \grad_4H$, или в развернутом виде
\begin{equation}\label{eeq2_10}
\begin{array}{c}
    (y - \sigma)H_x = 0, \qquad (y - \sigma)H_\alpha = 0, \\
    (x - \sigma)H_y = 0, \qquad (x - \sigma)H_\beta = 0.
\end{array}
\end{equation}

Пусть сначала $x = \sigma$. Система \eqref{eeq2_10} с учетом \eqref{eeq2_8} дает
\begin{equation}\label{eeq2_11}
\begin{array}{c}
    \cos \alpha = \pm 1, \\
    \ds x^2 - \frac{1}{2}xy - \frac{1}{2}y^2 + \frac{y}{x - y}(\cos \alpha - \cos \beta) - \lambda (x - y) = 0.
\end{array}
\end{equation}
Исключая с помощью \eqref{eeq2_11} $y$ и $\beta$ из выражений \eqref{eeq2_8}, находим соответствующие критические значения интегрального отображения:
\begin{equation}\label{eeq2_12}
\begin{array}{c}
  \ds h = \frac{3}{2} \sigma^2 - 2 \lambda \sigma + \frac{1}{2}\lambda^2 \mp 1, \qquad
  k = -\sigma^3 + \lambda \sigma^2.
\end{array}
\end{equation}
При этом, согласно \eqref{eeq2_4}, следует считать $\sigma \gs 0$.

Полагая $y = \sigma$, в силу очевидной симметрии имеем
\begin{equation}\label{eeq2_13}
\begin{array}{c}
    \cos \beta = \pm 1, \\
    \ds y^2 - \frac{1}{2}xy - \frac{1}{2}x^2 + \frac{x}{x - y}(\cos \alpha - \cos \beta) + \lambda (x - y) = 0,
\end{array}
\end{equation}
что снова приводит к \eqref{eeq2_12}, но уже при $\sigma \ls 0$.

В соответствии с возможным выбором знака удобно разбить множество \eqref{eeq2_12} на две бифуркационные кривые, обозначая второй параметр через $\sigma_*$:
\begin{equation}\label{eeq2_14}
    \ds h = \frac{3}{2}\sigma^2 - 2\lambda \sigma + \frac{1}{2}\lambda^2 - 1, \quad k = -\sigma^3 + \lambda \sigma^2,
\end{equation}
\begin{equation}\label{eeq2_15}
    \ds h = \frac{3}{2}\sigma_*^2 - 2\lambda \sigma_* + \frac{1}{2}\lambda^2 + 1, \quad k = -\sigma_*^3 + \lambda \sigma_*^2.
\end{equation}
(кривые $\Bif$ и $\Bif_*$ соответственно). Они, очевидно, конгруэнтны, поэтому достаточно отметить свойства лишь первой из них.

Кривая $\Bif$ выпукла для $\sigma < 2\lambda / 3$, вогнута для $\sigma > 2\lambda / 3$, имеет точку возврата $(h, k) = (-1 - \lambda^2 / 6, 4\lambda^3 / 27)$ при $\sigma = 2\lambda / 3$, точку минимума $(h, k) = (-1 + \lambda^2 / 2, 0)$ при $\sigma = 0$ и проходит через точку $(-1, 0)$ при $\sigma = \lambda$.

Пусть теперь в \eqref{eeq2_10} $x \neq \sigma, y \neq \sigma$. Тогда одновременно выполняются равенства \eqref{eeq2_11}, \eqref{eeq2_13}, и, следовательно, пара $(h, k)$ удовлетворяет соотношениям \eqref{eeq2_12}, если в последних заменить $\sigma$ на $x$, а затем на $y$. Так как $x \neq y$ и кривые $\Bif$ и $\Bif_*$ не имеют самопересечений, то при этом знаки в \eqref{eeq2_12} необходимо брать различными. Таким образом, $(h, k)$ -- общая точка бифуркационных кривых. Выясним, когда это возможно. Условие пересечения \eqref{eeq2_14}, \eqref{eeq2_15} запишем в виде
\begin{equation}\label{eeq2_16}
    \ds \sigma_*^2 - \sigma^2 = \frac{4}{3}[\lambda(\sigma_* - \sigma) - 1], \quad \sigma_*^3 - \sigma^3 = \lambda(\sigma_*^2 - \sigma^2).
\end{equation}
Полагая $\tau = \sigma_* + \sigma - \lambda$, находим из первого уравнения \eqref{eeq2_16} $\sigma_* - \sigma = 4 / (\lambda - 3 \tau)$, так что
\begin{equation*}
    \ds \sigma = \frac{1}{2}\Bigl(\tau + \lambda + \frac{4}{3\tau - \lambda} \Bigr), \quad \sigma_* = \frac{1}{2}\Bigl(\tau + \lambda - \frac{4}{3\tau - \lambda} \Bigr).
\end{equation*}
Найденные значения внесем в \eqref{eeq2_14}, \eqref{eeq2_15} и второе уравнение \eqref{eeq2_16}. Получим
\begin{equation*}
    \begin{array}{c}
       \ds h = \frac{1}{4}(\lambda^2 - 4\lambda \tau - 3\tau^2), \quad k = \tau^2(\tau + \lambda), \\
       (\tau + \lambda)(3\tau - \lambda)^3 + 16 = 0.
    \end{array}
\end{equation*}
Сравнивая эти соотношения с \eqref{eeq1_17}, \eqref{eeq1_20}, заключаем, что кривые $\Bif$ и $\Bif_*$ пересекаются между собой при условии $\lambda^2 \gs 4/3$ по точкам, соответствующим равномерным вращениям твердого тела.

\begin{figure}[ht]
\centering
\includegraphics[width=0.8\linewidth]{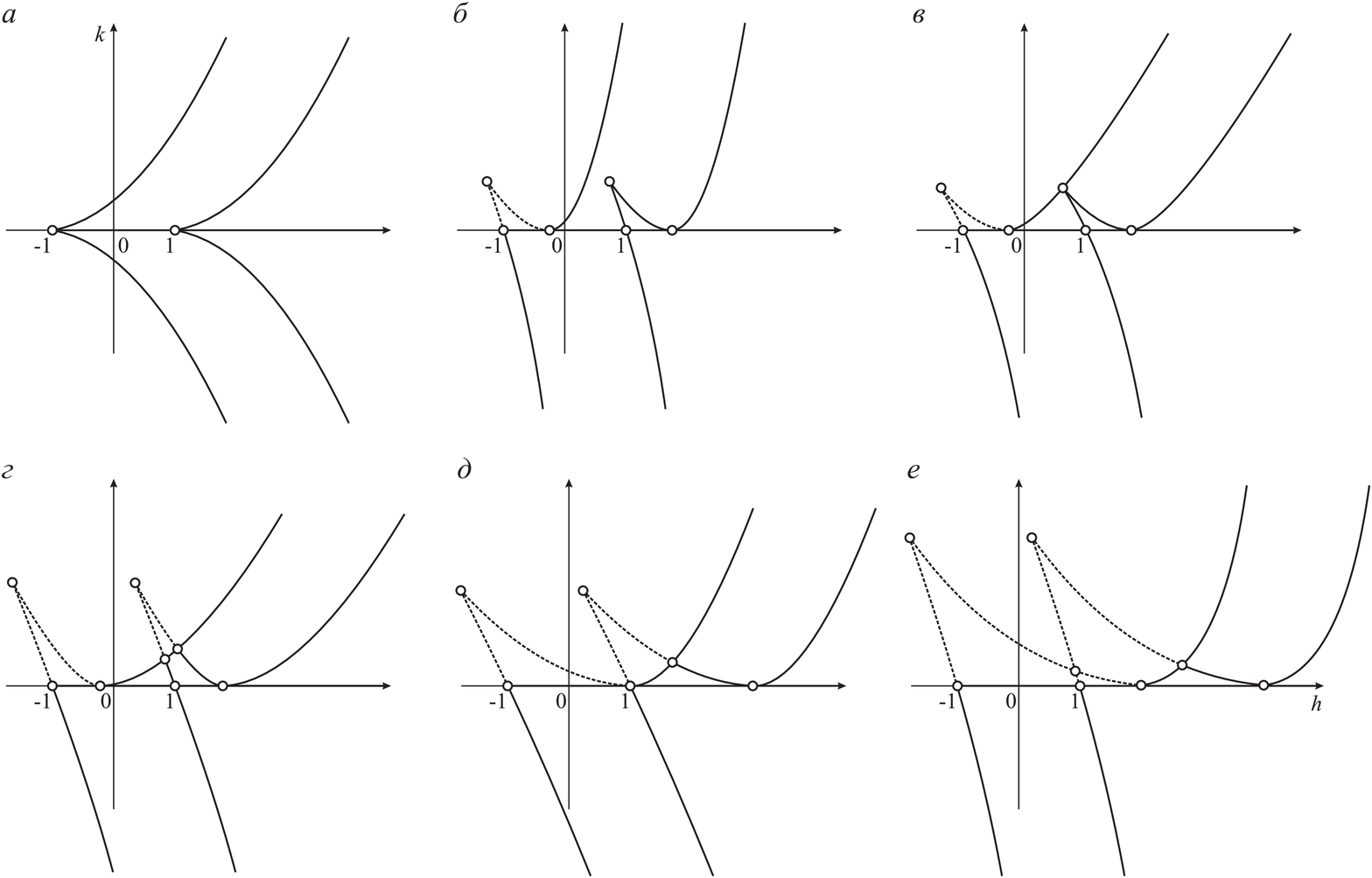}
\caption{}\label{fig_5_1}
\end{figure}

Виды бифуркационного множества $\Sigma_0(\lambda)$ показаны на рис.~\ref{fig_5_1},\,{\it а-e} для значений: а)~${\lambda = 0}$; б)~$0 < \lambda < 2/\sqrt{3}$; в)~$\lambda = 2/\sqrt{3}$; г)~$2/\sqrt{3} < \lambda < 2$; д)~$\lambda = 2$; е)~$2 < \lambda$. При замене $\lambda$ на $-\lambda$ происходит зеркальное отражение относительно оси $k = 0$. Случай $\lambda = 0$ отвечает классическому решению Горячева\,--\,Чап\-лы\-ги\-на. Подчеркнем, что $h$ и $k$ здесь безразмерные. В реальных величинах кривые $\Bif$ и $\Bif_*$ отстоят друг от друга на $2\gamma$. На участках, обозначенных пунктиром, как будет показано ниже, движений не существует.

Выполним подстановку \eqref{eeq2_3} в уравнения Эйлера\,--\,Пуассона. Система \eqref{eeq1_1} дает
\begin{equation}\label{eeq2_17}
\begin{array}{rcl}
  \ds y \frac{dx}{dt} + x \frac{dy}{dt} &=& \ds -\frac{xy}{x-y}(\sin \alpha - \sin \beta) \\
  \ds \frac{dx}{dt} + \frac{dy}{dt} &=& \ds -\frac{1}{x-y}(x\sin \alpha - y\sin \beta),
\end{array}
\end{equation}
а из первых двух уравнений \eqref{eeq1_2} имеем
\begin{equation}\label{eeq2_18}
\begin{array}{l}
  \ds x\frac{d\alpha}{dt} \sin \alpha - y \frac{d\beta}{dt} \sin \beta = \left(x + \frac{1}{2}y - \lambda \right) x \sin \alpha - \\
\qquad  \ds - \left( \frac{1}{2}x + y - \lambda \right) y \sin \beta + \\
\qquad  \ds + \frac{xy}{(x - y)^2}(\sin \alpha + \sin \beta)(\cos \alpha - \cos \beta), \\
  \ds x\frac{d\alpha}{dt} \cos \alpha - y \frac{d\beta}{dt} \cos \beta = \left(x + \frac{1}{2}y - \lambda \right) x \cos \alpha -\\
\qquad  \ds - \left( \frac{1}{2}x + y - \lambda \right) y \cos \beta + \\
\qquad  \ds + \frac{xy}{(x - y)^2}(\cos \alpha + \cos \beta)(\cos \alpha - \cos \beta).
\end{array}
\end{equation}
Решая уравнения \eqref{eeq2_17}, \eqref{eeq2_18} относительно производных, находим
\begin{equation}\label{eeq2_19}
\begin{array}{l}
  \ds \frac{dx}{dt} = -\frac{x}{x -y}\sin \alpha, \quad \frac{dy}{dt} = -\frac{y}{x -y}\sin \beta, \\
  \ds \frac{d\alpha}{dt} = x + \frac{1}{2}y - \lambda + \frac{y}{(x - y)^2}(\cos \alpha - \cos \beta), \\
  \ds \frac{d\beta}{dt} = \frac{1}{2}x + y - \lambda - \frac{x}{(x - y)^2}(\cos \alpha - \cos \beta).
\end{array}
\end{equation}
Эти уравнения образуют каноническую систему с гамильтонианом $H$, причем $x, y$ играют роль импульсов. Аналогичные \eqref{eeq2_2} канонические переменные в случае Горячева\,--\,Чаплыгина введены В.В.\,Козловым \cite{bib26}.

Зафиксируем постоянные $h$ и $k$ и исключим переменные $\alpha, \beta$ из \eqref{eeq2_18} с помощью тождеств \eqref{eeq2_9}. После преобразований получим
\begin{equation}\label{eeq2_20}
    \ds 2(x - y)\frac{dx}{dt} = -\sqrt{W(x)}, \quad 2(x - y)\frac{dy}{dt} = -\sqrt{W(y)}
\end{equation}
где
\begin{equation}\label{eeq2_21}
    W(z) = -Z(z) Z_*(z);
\end{equation}
\begin{equation}\label{eeq2_22}
    Z(z) = z(z - \lambda)^2 - 2(h + 1)z - 2k, \quad Z_*(z) = z(z - \lambda)^2 - 2(h - 1)z - 2k.
\end{equation}
Дискриминанты многочленов $Z(z)$ и $Z_*(z)$ таковы:
\begin{equation*}
\begin{array}{c}
  \ds \Delta = 4 \Bigl\{ 27 \Bigl[ k + \frac{2}{3}(h + 1) \lambda - \frac{\lambda^3}{27} \Bigr]^2 - 8 \Bigl[ h + 1 + \frac{\lambda^2}{6} \Bigr]^3 \Bigr\}, \\
  \ds \Delta_* = 4 \Bigl\{27 \Bigl[ k + \frac{2}{3}(h - 1) \lambda - \frac{\lambda^3}{27} \Bigr]^2 - 8 \Bigl[ h - 1 + \frac{\lambda^2}{6} \Bigr]^3 \Bigr\}.
\end{array}
\end{equation*}
Легко убедиться, что уравнение $\Delta = 0$ определяет на плоскости $(h, k)$ кривую \eqref{eeq2_14}. При этом параметр $\sigma$ играет роль кратного корня $Z$. То же самое относится к уравнению $\Delta_* = 0$ и кривой \eqref{eeq2_15}. Кроме того, при $k = 0$ (т.е. на луче \eqref{eeq2_1}) многочлены $Z(z)$ и $Z_*(z)$ имеют общий корень $z = 0$. Таким образом, бифуркационное множество интегрального отображения представляет собой некоторую часть кривой кратных корней многочлена \eqref{eeq2_21}, определяющего правые части уравнений \eqref{eeq2_20}. Это свойство вообще характерно для задач с разделяющимися переменными \cite{bib01} и может быть легко доказано в случае, если разделяющее преобразование не содержит постоянных интегрирования.

\subsection{Исследование основного многочлена}\label{ssec53}

В данной задаче естественно изучить проекции $V_{h,k}$ интегрального многообразия на плоскость переменных $(x, y)$. После этого, установив свойства отображения проектирования, будет нетрудно выяснить структуру самого $J_{h,k}$. Согласно формулам \eqref{eeq2_4}, \eqref{eeq2_20} область $V_{h,k}$ задается неравенствами
\begin{equation}\label{eeq3_1}
    x \gs 0, \quad y \ls 0, \quad W(x) \gs 0, \quad W(y) \gs 0.
\end{equation}

Очевидно, множество \eqref{eeq3_1} полностью ограничено числом, взаимным расположением и знаками корней многочлена $W(z)$. Получению этой вспомогательной информации и посвящен настоящий параграф.

Отметим сразу же, что согласно определениям \eqref{eeq2_21}, \eqref{eeq2_22}, $\deg W = 6$ и старший коэффициент отрицателен, так что область $V_{h,k}$ всегда ограничена. Кроме того, если все корни $W(z)$ одного знака, то $V_{h,k} = \varnothing$.

На рис.~\ref{fig_5_2} показаны области\footnote[1]{В силу отмеченной ранее симметрии при замене $(\lambda, k) \to (-\lambda, -k)$ все дальнейшие результаты излагаются для случая $\lambda \gs 0$.} с различным расположением корней многочлена $Z(z)$. Римские цифры соответствуют открытым связным множествам на плоскости $(h, k)$, в которых отсутствуют кратные корни, арабские -- участкам кривой $\Bif$ и прямой $k = 0$, узловые точки $N_1, N_2, N_3$ выделены особо.

\begin{figure}[ht]
\center{\includegraphics[width=0.3\linewidth]{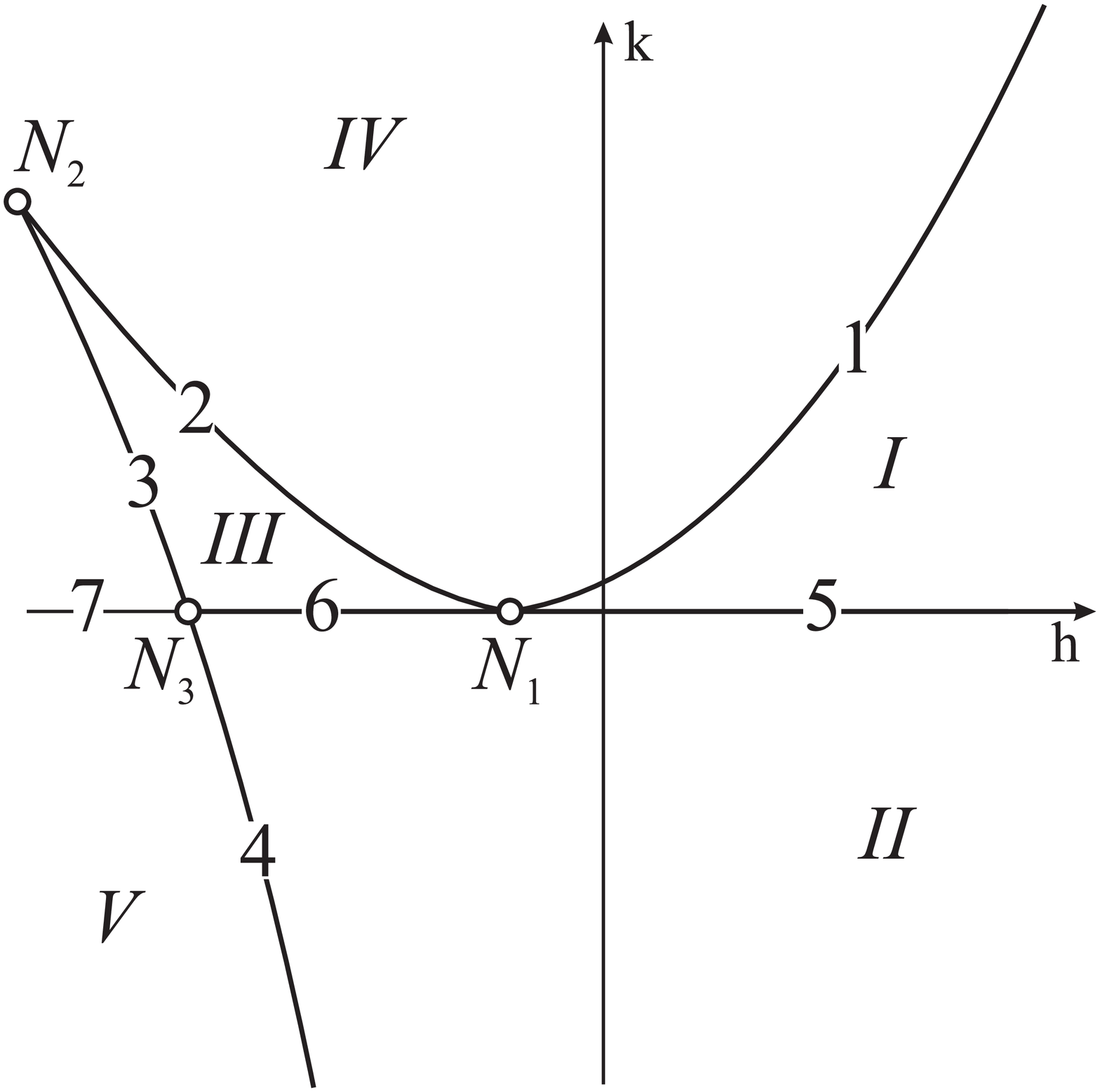}}
\caption{}\label{fig_5_2}
\end{figure}

Условимся кратный корень $Z(z)$ обозначать через $\zeta_0$. Если простой корень один, то обозначим его $\zeta$. Если же корней три, то обозначим их $\zeta_1, \zeta_2, \zeta_3$. Полная информация о многочлене $Z(z)$ приведена в табл. \ref{tab51}.

\begin{center}
\tabcolsep=1mm
\begin{tabular}{|c|c|c|c|c|c|}
  \multicolumn{6}{r}{\footnotesize Таблица \myt\label{tab51}} \\
  \hline
  Область & Корни $z$ & Область & Корни $z$ & Область & Корни $z$ \\
  \hline
  $I$ & $\zeta_1 < \zeta_2 < 0 < \zeta_3 $ & 1 & $\zeta_0 < 0 < \zeta$ & 6 & $0 = \zeta_1 < \zeta_2 <  \zeta_3 $ \\
  \hline
  $\ts{II}$ & $\zeta_1 < 0 < \zeta_2 < \zeta_3 $ & 2 & $0 < \zeta_0 < \zeta$ & 7 & $0 = \zeta$ \\
  \hline
  $\ts{III}$ & $0 < \zeta_1 < \zeta_2 <  \zeta_3 $ & 3 & $0 < \zeta < \zeta_0$ & $N_1$ & $0 = \zeta_0 < \zeta$ \\
  \hline
  $\ts{IV}$ & $0 < \zeta$ & 4 & $ \zeta < 0 < \zeta_0$ & $N_2$ & $0 < \zeta_0 = \zeta$ \\
  \hline
  $V$ & $\zeta < 0$ & 5 & $\zeta_1 < \zeta_2 = 0 <  \zeta_3 $ & $N_3$ & $0 = \zeta < \zeta_0$ \\
  \hline
\end{tabular}
\end{center}

Подобные обозначения для многочлена $Z_*(z)$ отметим индексом <<$*$>>. Результат совершенно аналогичен приведенному в табл. \ref{tab51}. Взаимное расположение корней $Z(z)$ и $Z_*(z)$ легко устанавливается с учетом равенств
\begin{equation}\label{eeq3_2}
    Z(0) = Z_*(0) = -2k, \quad Z_*(z) - Z(z) = 4z.
\end{equation}
Области на плоскости $(h, k)$ с различным типом множеств\eqref{eeq3_1} получают двойную нумерацию $\ts{I} - \ts{IV}_*, 1 - 2_*, 5 - 5_*, N_1 - 7_*, 5 - N_3^*$ и т.п. В наиболее полном наборе возможные варианты, соответствующие \eqref{eeq1_19}, приведены на рис.~\ref{fig_5_3}.

\begin{figure}[ht]
\center{\includegraphics[width=0.7\linewidth]{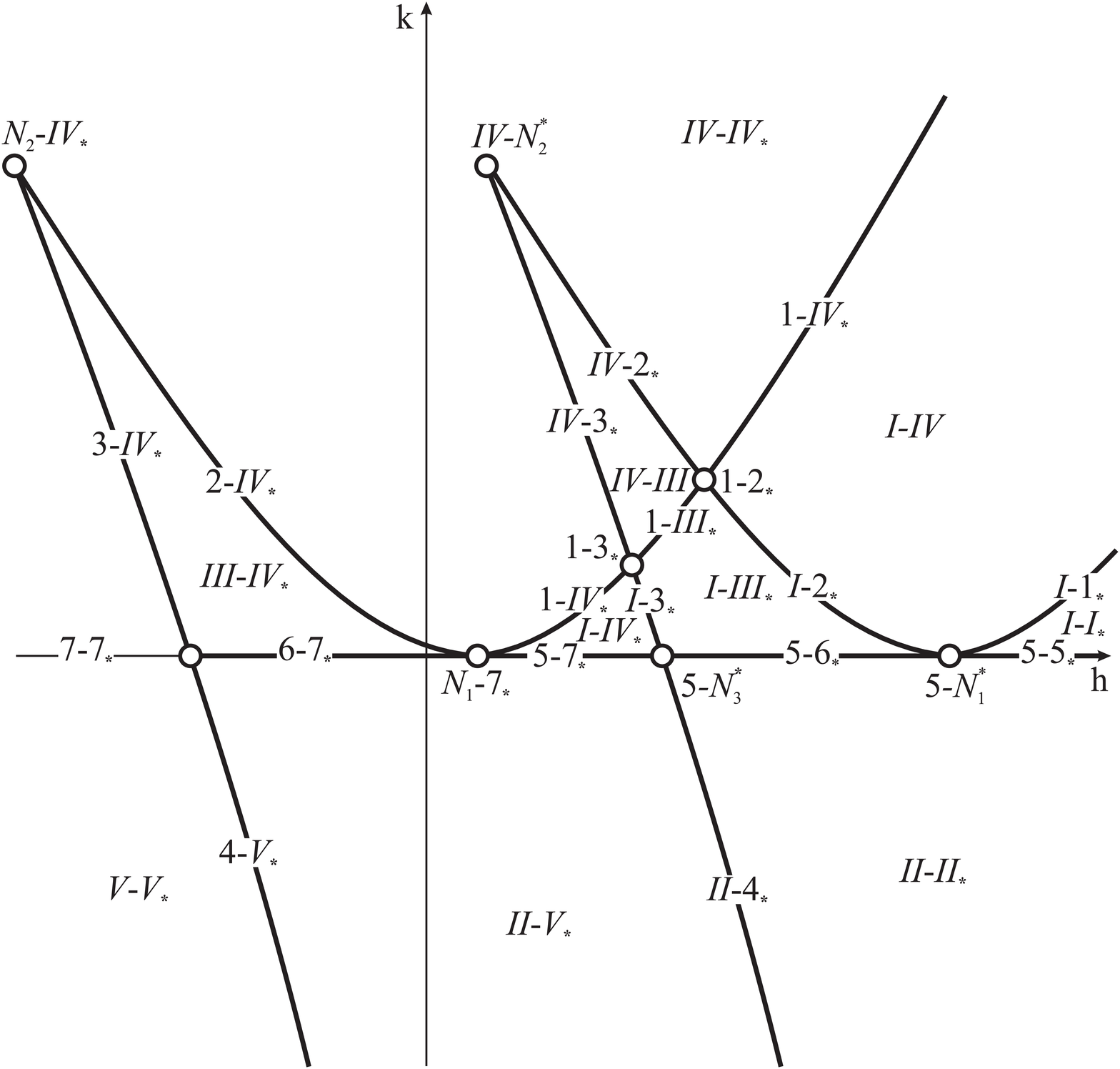}}
\caption{}\label{fig_5_3}
\end{figure}

Из \eqref{eeq3_2} следует, что, если у многочлена $Z(z)$ все корни одного знака, тот же знак имеют и все корни $Z_*(z)$, а, следовательно, и все корни $W(z)$. По сделанному выше замечанию при этом $V_{h, k} = \varnothing$. Таким образом, для всех значений $(h, k)$, шифр которых содержит $\ts{IV}, \ts{V}, 2, 3, N_2$, движений не существует (на рис.~\ref{fig_5_3} это точки, лежащие слева от всех утолщенных линий). В остальных случаях многочлен $W(z)$ имеет корни, указанные в табл.~\ref{tab52}.

Прежде чем перейти к анализу интегральных многообразий, условимся о некоторых обозначениях.

Будем говорить, что функция $z(t)$, определенная на $\bbR$, принадлежит классам $\Con(a)$, $\Osc[a, b]$, $\Asp[a, b)$ (или $\Asp(b, a]$), если она, соответственно, постоянна и равна $a$; совершает периодические колебания на отрезке $[a, b]$; принимает значение $a$ при единственном $t = t_0$, монотонна на $(-\infty, t_0), (t_0, +\infty)$ и стремится к $b$ при $t \to \pm \infty$.


\begin{center}
\tabcolsep=1mm
\renewcommand{\arraystretch}{1.2}
\small
\begin{tabular}{|c|c|c|c|}
  \multicolumn{4}{r}{\footnotesize Таблица \myt\label{tab52}} \\
  \hline
Область & Корни $W$ &  Область & Корни $W$ \\
\hline
  $I - I_*$ & $\zeta_1 < \zeta_1^* < \zeta_2^* < \zeta_2 < 0 < \zeta_3^* < \zeta_3$ & $I - 1_*$ & $\zeta_1 < \zeta_0^* < \zeta_2 < 0 < \zeta^* < \zeta_3$ \\
  \hline
  $I - \ts{IV}_*$ & $\zeta_1 < \zeta_2 < 0 < \zeta^* < \zeta_3$ & $4 - V_*$ & $\zeta < \zeta^* < 0 < \zeta_0$ \\
  \hline
  $I - \ts{III}_*$ & $\zeta_1 < \zeta_2 < 0 < \zeta_1^* < \zeta_2^* < \zeta_3^* < \zeta_3$ & $II - 4_*$ & $\zeta_1 < \zeta^* < 0 < \zeta_2 < \zeta_0^* < \zeta_3$ \\
  \hline
  $\ts{II} - V_*$ & $\zeta_1 < \zeta^* < 0 < \zeta_2 < \zeta_3$ & $N_3 - 7_*$ & $0 = \zeta^* = \zeta < \zeta_0$ \\
  \hline
  $\ts{II} - \ts{II}_*$ & $\zeta_1 < \zeta_1^* < 0 < \zeta_2 < \zeta_2^* < \zeta_3^* < \zeta_3$ & $6 - 7_*$ & $0 = \zeta^* = \zeta_1 < \zeta_2 < \zeta_3$ \\
  \hline
  $1 - \ts{IV}_*$ & $\zeta_0 < 0 < \zeta^* < \zeta$ & $N_1 - 7^*$ & $0 = \zeta^* = \zeta_0 < \zeta$ \\
  \hline
  $1 - 2_*$ & $\zeta_0 < 0 < \zeta_0^* < \zeta^* < \zeta$ & $5 - 7_*$ & $\zeta_1 < 0 = \zeta^* = \zeta_2 < \zeta_3$ \\
  \hline
  $1 - \ts{III}_*$ & $\zeta_0 < 0 < \zeta_1^* < \zeta_2^* < \zeta_3^* < \zeta$ & $5 - N_3^*$ & $\zeta_1 < 0 = \zeta^* = \zeta_2 < \zeta_0^* < \zeta_3$ \\
  \hline
  $1 - 3_*$ & $\zeta_0 < 0 < \zeta^* < \zeta_0^* < \zeta$ & $5 - 6_*$ & $\zeta_1 < 0 = \zeta_1^* = \zeta_2 < \zeta_2^* < \zeta_3^* < \zeta_3$ \\
  \hline
  $I - 2_*$ & $\zeta_1 < \zeta_2 < 0 < \zeta_0^* < \zeta^* < \zeta_3$ & $5 - N_1^*$ & $\zeta_1 < 0 = \zeta_0^* = \zeta_2 < \zeta^* < \zeta_3$ \\
  \hline
  $I - 3_*$ & $\zeta_1 < \zeta_2 < 0 < \zeta^* < \zeta_0^* < \zeta_3$ & $5 - 5_*$ & $\zeta_1 < \zeta_1^* < 0 = \zeta_2 = \zeta_2^* < \zeta_3^* < \zeta_3$\\
  \hline
\end{tabular}
\end{center}

Рассмотрим одну предельно простую ситуацию. Пусть задано уравнение
\begin{equation}\label{eeq3_3}
    \ds \left(\frac{dz}{dt}\right)^2 = P(z),
\end{equation}
где $P(z)$ -- многочлен, отрицательный на бесконечности. Покажем, как связаны его решения $z(t)$, множества $\{z:P(z) \gs 0\}$ и кривые, определяемые \eqref{eeq3_3} на плоскости $(z, dz/dt)$. Пусть $E$ -- связная компонента множества $\{z:P(z) \gs 0\}$, $\Gamma$ -- соответствующая связная компонента кривой \eqref{eeq3_3} на фазовой плоскости. Имеем следующие элементарные возможности.

1. $E = \{a \}$. Тогда $a$ -- корень многочлена $P$ четной кратности, $z(t) \in \Con(a)$, $\Gamma = \{ (a, 0) \}$.

2. $E = [a, b]$, где $a, b$ соседние простые корни $P$. Тогда $z(t) \in \Osc[a, b]$, кривая $\Gamma$ диффеоморфна окружности $S^1$.

3. $E = [a, b]$, где $a$ простой корень $P$, $b$ -- соседний с $a$ корень нечетной кратности больше единицы. Тогда либо $z(t) \in \Con(b)$, либо $z(t) \in \Asp[a, b)$. Кривая $\Gamma$ -- петля (топологическая окружность с угловой точкой). В этом случае пишем $\Gamma = S_\wedge^1$.

4. $E = [a, b]$, где $a$ и $b$ -- простые корни $P$, между которыми лежит корень $c$ четной кратности. Тогда $z(t) \in \Con(b) \cup \Asp[a, c) \cup \Asp(c, b]$ и $\Gamma = S^1 \dot{\cup} S^1$ (восьмерка).

Очевидно, $\Gamma = \hat{S}^0(E)$ -- расслоение нуль-мерных сфер, стянутых в точку над корнями $P(z)$. В системах с интегралом энергии эта конструкция хорошо известна \cite{bib41}.

Отрезок, ограниченный простыми корнями $P(z)$ и не содержащий других (случай 2), обозначим через $I$. В случае 3 пишем $E = I_\wedge$, а в случае 4 $E = I \dot{\cup} I$. Таким образом,
\begin{equation*}
\begin{array}{c}
  \hat{S}^0(\{ \cdot \}) = \{ \cdot \}, \qquad \hat{S}^0(I) = S^1, \qquad \hat{S}^0(I_\wedge) = S_\wedge^1, \qquad  \hat{S}^0(I \dot{\cup} I) = S^1 \dot{\cup} S^1.
\end{array}
\end{equation*}

\subsection{Интегральные многообразия и фазовые траектории}\label{ssec54}

Поскольку постоянная площадей все время равна нулю, опустим последний индекс в обозначении множества \eqref{ceq4_9}. Таким образом,
\begin{equation}\label{eeq4_1}
    J_{h, k} = \{(\mbs{\nu}, \mbs{\omega}) \in \mathfrak{M}: H(\mbs{\nu}, \mbs{\omega}) = h, \; K(\mbs{\nu}, \mbs{\omega}) = k, \; G(\mbs{\nu}, \mbs{\omega}) = 0 \}.
\end{equation}

Нетрудно убедиться, что рассматриваемый случай интегрируемости удовлетворяет условиям теоремы \ref{the345}. Поэтому, если на связной компоненте $J_{h, k}$ нет точек зависимости функций $H, K, G$, то она представляет собой двумерный тор с условно-периодическими траекториями. Остается выяснить число компонент при $(h, k) \notin \Sigma_0$ и установить, как устроены критические интегральные поверхности.

Обратимся к условиям \eqref{eeq2_20}. Введем <<приведенное время>> $\tau$, положив на фиксированной траектории
\begin{equation}\label{eeq4_2}
    \ds \frac{d}{d\tau} = 2[x(t) - y(t)]\frac{d}{dt}.
\end{equation}

Такая замена неправомочна лишь в случае \eqref{eeq2_5}. Соответствующая точка \eqref{eeq2_6} ни при каком $\mbs{\nu}$ не является стационарным решением уравнений Эйлера\,--\,Пуассона. Следовательно, тождество $x \equiv y \equiv 0$ невозможно.

Таким образом, траектория $(x(t), y(t))$ может лишь пересекать точку $(0, 0)$ при изолированных значениях $t$. Эти значения разбивают ось времени на конечное или бесконечное число интервалов, в которых замена \eqref{eeq4_2} возможна. Каждый интервал соответствует значениям $\tau \in (-\infty, +\infty)$, причем $\int \limits_{-\infty}^{+\infty}[x(\tau) - y(\tau)]d\tau$ сходится. Как отмечалось, эта ситуация возможна лишь на множестве \eqref{eeq2_7}. К ее детальному исследованию обратимся далее.

Вернемся к общему случаю. Из второго соотношения \eqref{eeq3_2} следует, что в области \eqref{eeq3_1}
\begin{equation*}
  Z(x) \ls 0 \ls Z_*(x), \qquad z_*(y) \ls 0 \ls Z(y).
\end{equation*}
Положим
\begin{equation}\label{eeq4_3}
  X = \sqrt{-Z(x)}, \quad X_* = \sqrt{Z_*(x)}, \quad Y = \sqrt{Z(y)}, \quad Y_* = \sqrt{-Z_*(y)}.
\end{equation}
Обозначая штрихом дифференцирование по $\tau$, запишем систему \eqref{eeq2_20} в виде
\begin{equation}\label{eeq4_4}
  x' = -XX_*, \quad y' = -YY_*.
\end{equation}
В пространстве $\bbR^4 = \bbR^2(x, x') {\times} \bbR^2(y, y')$ эти уравнения задают поверхность $\Gamma$, естественным образом представимую как произведение $\Gamma = \Gamma_x {\times} \Gamma_y$, где $\Gamma_x \subset \bbR^2(x, x')$, $\Gamma_y \subset \bbR^2(y, y')$. Эта поверхность есть образ интегрального многообразия \eqref{eeq4_1} при отображении
\begin{equation}\label{eeq4_5}
  D: (\mbs{\nu}, \mbs{\omega}) \mapsto (x, y, x', y'),
\end{equation}
определенном формулами \eqref{eeq2_3}, \eqref{eeq2_19}. Всюду на множестве $G^{-1}(0)$, кроме точек \eqref{eeq2_5}, отображение \eqref{eeq4_5} -- локальный диффеоморфизм. В частности,
\begin{equation}\label{eeq4_6}
  \left. D \right|_{J_{h,k}}: J_{h,k} \to \Gamma
\end{equation}
является накрытием. Выясним его кратность.

Из соотношений \eqref{eeq2_9}, \eqref{eeq2_22}, \eqref{eeq4_3} находим
\begin{equation}\label{eeq4_7}
  \ds 2\sqrt{x} \sin \frac{\alpha}{2} = X, \quad 2\sqrt{x} \cos \frac{\alpha}{2} = X_*, \quad 2\sqrt{-y} \sin \frac{\beta}{2} = Y, \quad 2\sqrt{-y} \cos \frac{\beta}{2} = Y_*.
\end{equation}
Подстановка этих значений в формулу \eqref{eeq2_3} приводит к выражениям
\begin{equation}\label{eeq4_8}
\begin{array}{c}
  \ds \omega_1 = \frac{1}{8}(XY_* + X_*Y), \quad \omega_2 = \frac{1}{8}(X_*Y_* - XY), \\
  \omega_3 = x + y - \lambda, \\
  \ds \nu_1 = 1 - \frac{X^2 + Y^2}{2(x-y)} =  \frac{X_*^2 + Y_*^2}{2(x-y)} - 1,\\
  \ds \nu_2 = -\frac{1}{2(x-y)}(XX_* + YY_*), \quad \nu_3 = \frac{1}{2(x-y)}(XY_* - X_*Y).
\end{array}
\end{equation}

Пусть фиксирован набор
\begin{equation}\label{eeq4_9}
  (x, y, x', y'),
\end{equation}
удовлетворяющий условиям \eqref{eeq3_1}, \eqref{eeq4_4}. Точка \eqref{eeq4_9} не изменится, если произвести одновременную замену знака у радикалов $X, X_*$ или $Y, Y_*$. Однако при такой однократной замене меняют знак величины $\omega_1, \omega_2, \nu_3$ в \eqref{eeq4_8}. Таким образом, отображение \eqref{eeq4_6} двулистно. При этом возможны две ситуации: связная компонента $\Gamma$ накрывается либо двумя компонентами $J_{h,k}$, либо одной двукратно.

Рассмотрим некритический случай $(h, k) \notin \Sigma_0$. Пусть $E = E_x {\times} E_y$ -- компонента связности поверхности \eqref{eeq4_4}. В обозначениях предыдущего параграфа $E_x = \hat{S}^0(I_x)$, $E_y = \hat{S}^0(I_y)$, где $I_x$, $I_y$ -- отрезки, концами которых являются соседние (соответственно положительные и отрицательные) простые корни многочлена $W(z)$.

\begin{propos}\label{pro541}
Если хотя бы одна из переменных $x, y$ принимает значения между корнями $Z(z)$ и $Z_*(z)$, то $D^{-1}(E)$ связно. В противном случае $D^{-1}(E)$ состоит из двух компонент.
\end{propos}
\begin{proof}
Пусть, например, функция $x(\tau)$, удовлетворяющая \eqref{eeq4_4}, колеблется между корнями $Z(z)$ и $Z_*(z)$. Зафиксируем $y$ и дадим переменной $x$ пробежать свой период. Оба радикала $X, X_*$ изменят знак. Следовательно, существует непрерывный переход вдоль $J_{h, k}$ от точки $(\omega_1, \omega_2, \omega_3, \nu_1, \nu_2, \nu_3)$ к точке $(-\omega_1, -\omega_2, \omega_3, \nu_1, \nu_2, -\nu_3)$.

Если же при изменении $x$ меняет знак одна из величин $X, X_*$, а при изменении $y$ -- одна из $Y, Y_*$, то никаким непрерывным путем в области $I_x {\times} I_y$ нельзя изменить знак $\omega_1, \omega_2, \nu_3$, что и требовалось доказать.
\end{proof}

\begin{remark}\label{rem541}
Аналогичное правило имеет место и в случае, когда одна из переменных постоянна ($I_x$ либо $I_y$ вырождается в точку). При этом $E = S^1$, а $D^{-1}(E) = S^1$, если оставшаяся переменная колеблется между корнями $Z$ и $Z_*$ и $D^{-1}(E) = 2S^1$ в противном случае.
\end{remark}

Обращаясь к табл. \ref{tab52}, видим, что удвоение прообраза происходит при $k \neq 0$ в области $\ts{I} - \ts{III}_*$ и согласно замечанию \ref{rem541} в критических случаях $1 - 2_*$, $1 - \ts{III}_*$, $\ts{I} - 2_*$. Кроме того, нечто подобное должно быть в точках $1 - 3_*$, $\ts{I} - 3_*$. На рис.~\ref{fig_5_4} показана трансформация множества $\Gamma_x$ при переходах $(\ts{I} - \ts{III}_*) \to (\ts{I} - 3_*) \to (\ts{I} - \ts{IV}_*)$ и $(1 - \ts{III}_*) \to (1 - 3_*) \to (1 - \ts{IV}_*)$. В точке бифуркации естественно обозначить $\Gamma_x = S^1 \dot{\cup} S^1 \dot{\cup} S^1$. В процессе обоих переходов $\Gamma_y$ топологически не меняется: $\Gamma_y = S^1$ при первом и $\Gamma_y = \{ \cdot \}$ при втором. Таким образом, при первом переходе бифуркация $J_{h, k}$ имеет вид
\begin{equation*}
  3{\bf T}^2 \to S^1 {\times} (S^1 \dot{\cup} S^1 \dot{\cup} S^1) \to {\bf T}^2.
\end{equation*}

\begin{figure}[ht]
\center{\includegraphics[width=0.8\linewidth]{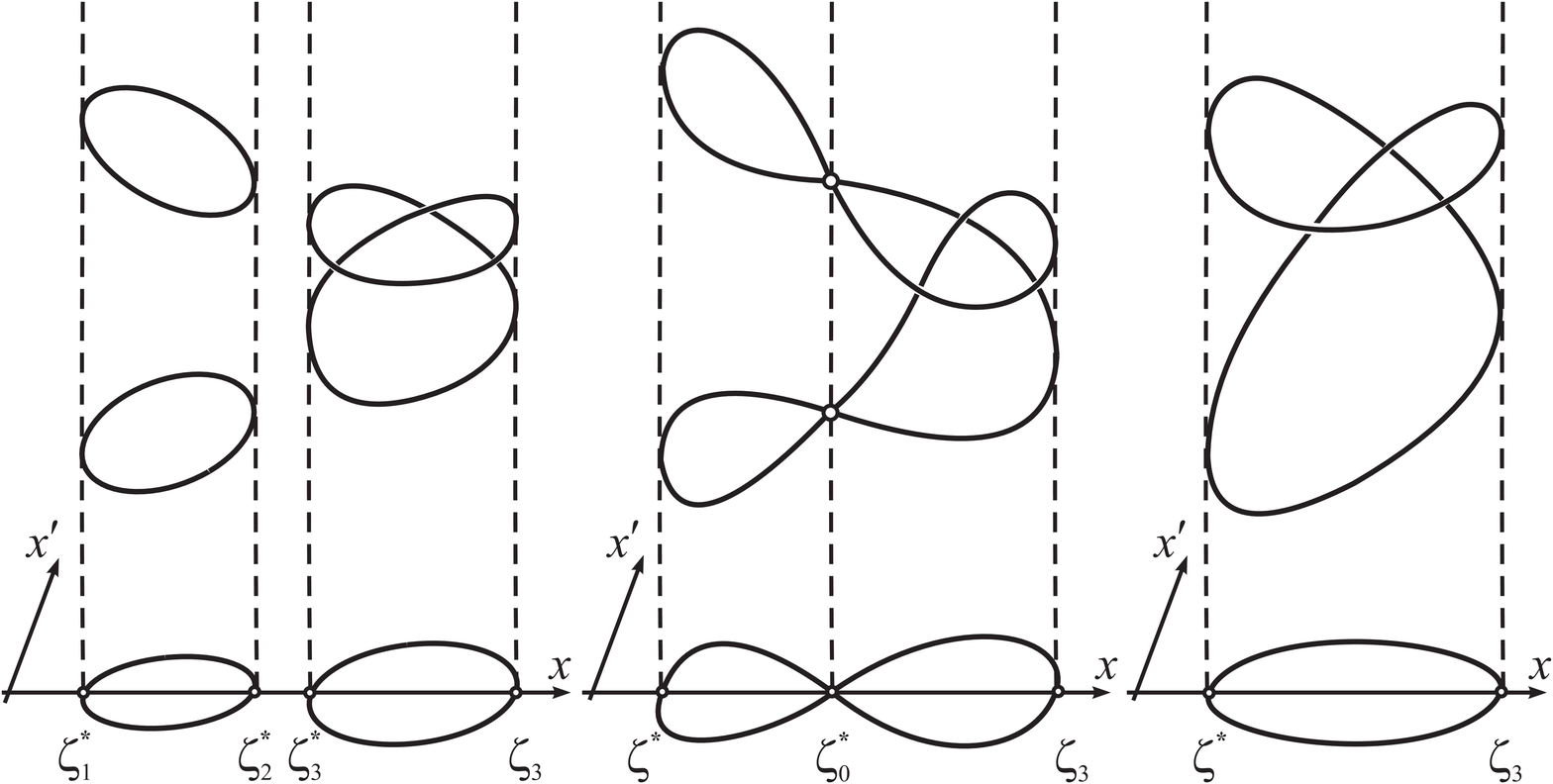}}
\caption{}\label{fig_5_4}
\end{figure}

Сводка результатов для $k \neq 0$ в наиболее полном варианте \eqref{eeq1_19} дана в табл. \ref{tab53}. Для остальных значений $\lambda$ интегральные многообразия определяются по теореме \ref{the312}. Критические поверхности $J_{h, k}$ также непрерывно зависят от $\lambda$ до тех пор, пока сохраняется расположение корней многочлена $W(z)$.

В табл. \ref{tab53} добавлена точка $1 - N_2$, существующая, когда
\begin{equation}\label{eeq4_10}
  \lambda = 2\sqrt{3},
\end{equation}
и не имеющая аналога при условии \eqref{eeq1_19}. Через $S^1_{\wedge \wedge}$ обозначена топологическая окружность с двумя угловыми точками. Рассмотрим соответствующие движения аналитически. Здесь
\begin{equation}\label{eeq4_11}
  \ds h = \frac{7}{9}, \quad k = \frac{32}{81 \sqrt{3}}.
\end{equation}
При условиях \eqref{eeq4_10}, \eqref{eeq4_11} из второго уравнения  \eqref{eeq4_4} получаем
\begin{equation}\label{eeq4_12}
  \ds y(t) \equiv -\frac{2}{3 \sqrt{3}},
\end{equation}
а первое с учетом \eqref{eeq4_2} принимает вид
\begin{equation*}
  \ds \frac{dx}{dt} = -\frac{1}{2}\left( x - \frac{4}{3 \sqrt{3}} \right)\sqrt{\left(\frac{16}{3 \sqrt{3}} - x \right)\left( x - \frac{4}{3 \sqrt{3}} \right)}.
\end{equation*}
Отсюда либо
\begin{equation}\label{eeq4_13}
  \ds x(t) \equiv -\frac{4}{3 \sqrt{3}},
\end{equation}
либо
\begin{equation}\label{eeq4_14}
  \ds x(t) = \frac{4}{3 \sqrt{3}} \frac{t^2+12}{t^2+3}.
\end{equation}

Подстановка \eqref{eeq4_12}, \eqref{eeq4_13} в \eqref{eeq4_8} приводит к стационарным решениям
\begin{equation}\label{eeq4_15}
  \ds (\mbs{\nu}, \mbs{\omega}) = \left(-\frac{1}{3}, 0, \pm\frac{2\sqrt{2}}{3}, \pm\frac{1}{3}\sqrt\frac{2}{3}, 0, -\frac{4}{3\sqrt{3}}\right).
\end{equation}
Это, конечно, равномерные вращения \eqref{eeq1_8}, \eqref{eeq1_13} - \eqref{eeq1_16}.

В случае \eqref{eeq4_12}, \eqref{eeq4_14} получаем следующую пару решений:
\begin{equation}\label{eeq4_16}
\begin{array}{lll}
  \ds \omega_1 = \pm\frac{1}{3}\sqrt{\frac{2}{3}} \frac{t(t^2+9)}{(t^2+3)^{3/2}}, & \ds \omega_2 = \frac{2\sqrt{2}}{(t^2+3)^{3/2}}, & \ds \omega_3 = \frac{4}{3\sqrt{3}} \frac{6 - t^2}{t^2 + 3}, \\
  \ds \nu_1 = \frac{27 - 18t^2 - t^4}{3(t^2 +2)^2},  & \ds \nu_2 = \mp \frac{8 \sqrt{3}t}{(t^2+3)^2},  & \ds \nu_3 = \pm \frac{2\sqrt{2}t}{3(t^2+3)^{1/2}}.
\end{array}
\end{equation}
Каждое из них имеет точки \eqref{eeq4_15} своими пределами при $t \to \pm \infty$.
Решения \eqref{eeq4_16} интересны тем, что переменные Эйлера\,--\,Пуассона являются элементарными алгебраическими функциями времени.

\begin{center}
\tabcolsep=0.8mm
\begin{tabular}{|c|c|c|c|c|c|}
  \multicolumn{6}{r}{\footnotesize Таблица \myt\label{tab53}} \\
  \hline
   & $x(t)$ & $y(t)$ & $\Gamma_x$ & $\Gamma_y$ & $J_{h, k}$ \\
  \hline
  $I - I_*$ &
  $\Osc[\zeta_3^*, \zeta_3]$ &
  \begin{tabular}{c}$\Osc[\zeta_1, \zeta_1^*]$, \\$\Osc[\zeta_2^*, \zeta_2]$\end{tabular}
  & $S^1$ & $2S^1$ & $2{\bf T}^2$ \\
  \hline
  $I - \ts{IV}_*$ & $\Osc[\zeta^*, \zeta_3]$ & $\Osc[\zeta_1, \zeta_2]$ & $S^1$ & $S^1$ & ${\bf T}^2$ \\
  \hline
  $I - \ts{III}_*$  &
  \begin{tabular}{c}$\Osc[\zeta_1^*, \zeta_2^*]$, \\$\Osc[\zeta_3^*, \zeta_3]$ \end{tabular}
  & $\Osc[\zeta_1, \zeta_2]$ & $2S^1$ & $S^1$ & $3{\bf T}^2$ \\
  \hline
  $\ts{II} - V_*$ & $\Osc[\zeta_2, \zeta_3]$ & $\Osc[\zeta_1, \zeta^*]$ & $S^1$ & $S^1$ & ${\bf T}^2$ \\
  \hline
  $\ts{II} - \ts{II}_*$  &
  \begin{tabular}{c}$\Osc[\zeta_2, \zeta_2^*]$, \\$\Osc[\zeta_3^*, \zeta_3]$ \end{tabular}
  & $\Osc[\zeta_1, \zeta_1^*]$
  & $2S^1$ & $S^1$ & $2{\bf T}^2$ \\
  \hline
  $1 - \ts{IV}_*$ & $\Osc[\zeta^*, \zeta]$ & $\Con(\zeta_0)$ & $S^1$ & $\{\cdot\}$ & $S^1$ \\
  \hline
  $1 - 2_*$ &
  \begin{tabular}{c}$\Con(\zeta_0)$,\\ $\Osc[\zeta^*, \zeta]$\end{tabular} & $\Con(\zeta_0)$ & $\{\cdot\} \cup S^1$ & $\{\cdot\}$ & $2\{\cdot\} \cup S^1$ \\
  \hline
  $1 - \ts{III}_*$  &
  \begin{tabular}{c}$\Osc[\zeta_1^*, \zeta_2^*]$,\\ $\Osc[\zeta_3^*, \zeta]$ \end{tabular}
  & $\Con(\zeta_0)$ & $2S^1$ & $\{\cdot\}$ & $3S^1$ \\
  \hline
  $1 - 3_*$  & \begin{tabular}{c}$\Con(\zeta_0)$,\\$ \Asp[\zeta^*, \zeta_0^*)$,\\ $\Asp(\zeta_0^*, \zeta]$ \end{tabular}& $\Con(\zeta_0)$ & $S^1 \dot{\cup} S^1$ & $\{\cdot\}$ & $S^1 \dot{\cup} S^1 \dot{\cup} S^1$  \\
  \hline
  $I - 2_*$  &
  \begin{tabular}{c}$\Con(\zeta_0)$, \\$\Osc[\zeta^*, \zeta_3]$ \end{tabular}
  & $\Osc[\zeta_1, \zeta_2]$ & $\{ \cdot \} \cup S^1 $ & $S^1$ & $2S^1 \cup {\bf T}^2$  \\
  \hline
  $I - 3_*$ & \begin{tabular}{c}$\Asp[\zeta^*, \zeta_0)$, \\$\Con(\zeta_0^*)$, \\ $\Asp(\zeta_0^*, \zeta_3]$ \end{tabular}& $\Osc[\zeta_1, \zeta_2]$ & $S^1 \dot{\cup} S^1$ & $S^1$ & \begin{tabular}{c}$(S^1 \dot{\cup} S^1 \dot{\cup} S^1) {\times}$ \\ ${\times} S^1$ \end{tabular}\\
  \hline
  $I - 1_*$ & $\Osc[\zeta^*, \zeta_3]$ & \begin{tabular}{c}$\Asp[\zeta_1, \zeta_0^*)$,\\ $\Con(\zeta_0^*)$,\\ $\Asp(\zeta_0^*, \zeta_2]$ \end{tabular}& $S^1$ & $S^1 \dot{\cup} S^1$ & $S^1 {\times} (S^1 \dot{\cup} S^1)$ \\
  \hline
  $4 - V_*$ & $\Osc[\zeta^*, \zeta_3]$ & $\Osc[\zeta, \zeta^*]$ & $\{ \cdot \}$ & $S^1$ & $S^1$ \\
  \hline
  $\ts{II} - 4_*$ & \begin{tabular}{c}$\Asp[\zeta_2, \zeta^*)$, \\ $\Con(\zeta_0^*)$, \\ $\Asp(\zeta_0^*, \zeta_3]$ \end{tabular}& $\Osc[\zeta_1, \zeta^*]$ & $S^1 \dot{\cup} S^1$ & $S^1$ & $(S^1 \dot{\cup} S^1) {\times} S^1$ \\
  \hline
  $1 - N_2$ &
  \begin{tabular}{c}$\Con(\zeta_0^*)$, \\$\Asp(\zeta_0^*, \zeta]$ \end{tabular}
  & $\Con(\zeta_0)$ & $S^1_\wedge$ & $\{ \cdot \}$ & $S^1_{\wedge \wedge}$ \\
  \hline
\end{tabular}
\end{center}

Обратимся к исследованию точек луча \eqref{eeq2_1}.

Пусть $h = 1$. Тогда из интеграла энергии следует, что $\mbs{\omega} \equiv 0$, $\mbs{\nu} \equiv (1, 0, 0)$. Это -- одно из стацонарных решений \eqref{eeq1_11}. Оно соответствует устойчивому равновесию тела, когда центр масс находится в нижнем положении и $J_{-1, 0} = \{ \cdot \}$.

Рассмотрим участок $6 - 7_*$ (рис.~\ref{fig_5_3}):
$$
h \in \left(-1,  \min\left(\frac{\lambda^2}{2}-1, 1 \right) \right).
$$
Обозначив $\sqrt{2(h+1)} = a$, $\sqrt{2(h-1)} = b$, получим
$$
W(z) = z^2[a^2-(z-\lambda)^2][(z-\lambda)^2+b^2],
$$
так что из \eqref{eeq2_20} $y(t) \equiv 0$, $x(t) \in \Osc[\lambda - a, \lambda + a]$. Поверхность $J_{h, 0} = S^1$.

Точка $N_1 - 7_*$ существует при условии $\lambda^2 < 4$. В этом случае $a = \lambda$ и
$$W(z) = z^3(2\lambda - z)[(z - \lambda)^2 + (4 - \lambda^2)].$$
По-прежнему $y(t) \equiv 0$, а первое уравнение \eqref{eeq2_20} примет вид
\begin{equation*}
  \ds \frac{dx}{dt} = -\frac{1}{2}\sqrt{(2\lambda - x)x}\sqrt{(z - \lambda)^2 + (4 - \lambda^2)},
\end{equation*}
поэтому $x(t) \in \Osc[0, 2\lambda]$. Вновь $J_{h, 0} = S^1$.

Участок $5 - 7_*$ характеризуется неравенством $\lambda^2 - 2 < 2h < 2$ (в частности, $\lambda^2 < 4$). При этом $a > \lambda$, и корни многочленов $Z,Z_*$ таковы:
\begin{equation*}
  \zeta_1 = \lambda - a < 0 = \zeta_2 = \zeta^* < \zeta_3 = a + \lambda.
\end{equation*}
Кривую
\begin{equation}\label{eeq4_17}
  \xi(\tau) = (x(\tau), y(\tau)), \quad -\infty < \tau < +\infty,
\end{equation}
удовлетворяющую \eqref{eeq4_4} и не вырождающуюся в точку $(0, 0)$, назовем псевдотраекторией. Очевидно,
\begin{equation}\label{eeq4_18}
  \lim \limits_{\tau \to \infty}\xi(\tau) = (0, 0).
\end{equation}
Так как в реальном времени $t$ точка $(0, 0)$ не является неподвижной, траектория $x(t), y(t)$ представляет собой последовательность псевдотраекторий
\begin{equation}\label{eeq4_19}
  \ldots, \xi_{-n} \ldots, \xi_{-1}, \xi_0, \xi_1, \ldots, \xi_n, \ldots.
\end{equation}

\begin{remark}\label{rem542}
Псевдотраектория $\xi_{i+1}(\tau)$ является продолжением псевдотраектории $\xi_i(\tau)$ тогда и только тогда, когда касательные в <<конце>> $\xi_i$ (т.е. предел касательной к $\xi_i$ при $\tau \to +\infty$) и в <<начале>> $\xi_{i+1}$ (в аналогичном смысле) симметричны друг другу относительно прямой $x = y$. Для доказательства достаточно продолжить отображение \eqref{eeq2_3} в область $x \ls 0$, $y \gs 0$ и убедиться, что точкам $(x, y, \alpha, \beta)$ и $(y, x, \beta, \alpha)$ соответствуют одни и те же значения переменных Эйлера\,--\,Пуассона.
\end{remark}

Система \eqref{eeq4_4} имеет два особых решения:
\begin{equation}\label{eeq4_20}
  \{x \equiv 0, y(\tau) \in \Asp[\zeta_1, 0)\}, \quad \{x(\tau) \in \Asp(0, \zeta_3], y \equiv 0 \}.
\end{equation}
В силу сказанного, в реальном времени они замыкаются в периодическую траекторию.
Остальные решения \eqref{eeq4_4} представлены на рис.~\ref{fig_5_5}. Одна из кривых (попадающая в угол $(\zeta_3, \zeta_1)$) геометрически проходится дважды. Выясним возможность образования замкнутых циклов, отличных от \eqref{eeq4_20}.

\begin{figure}[ht]
\center{\includegraphics[width=0.5\linewidth]{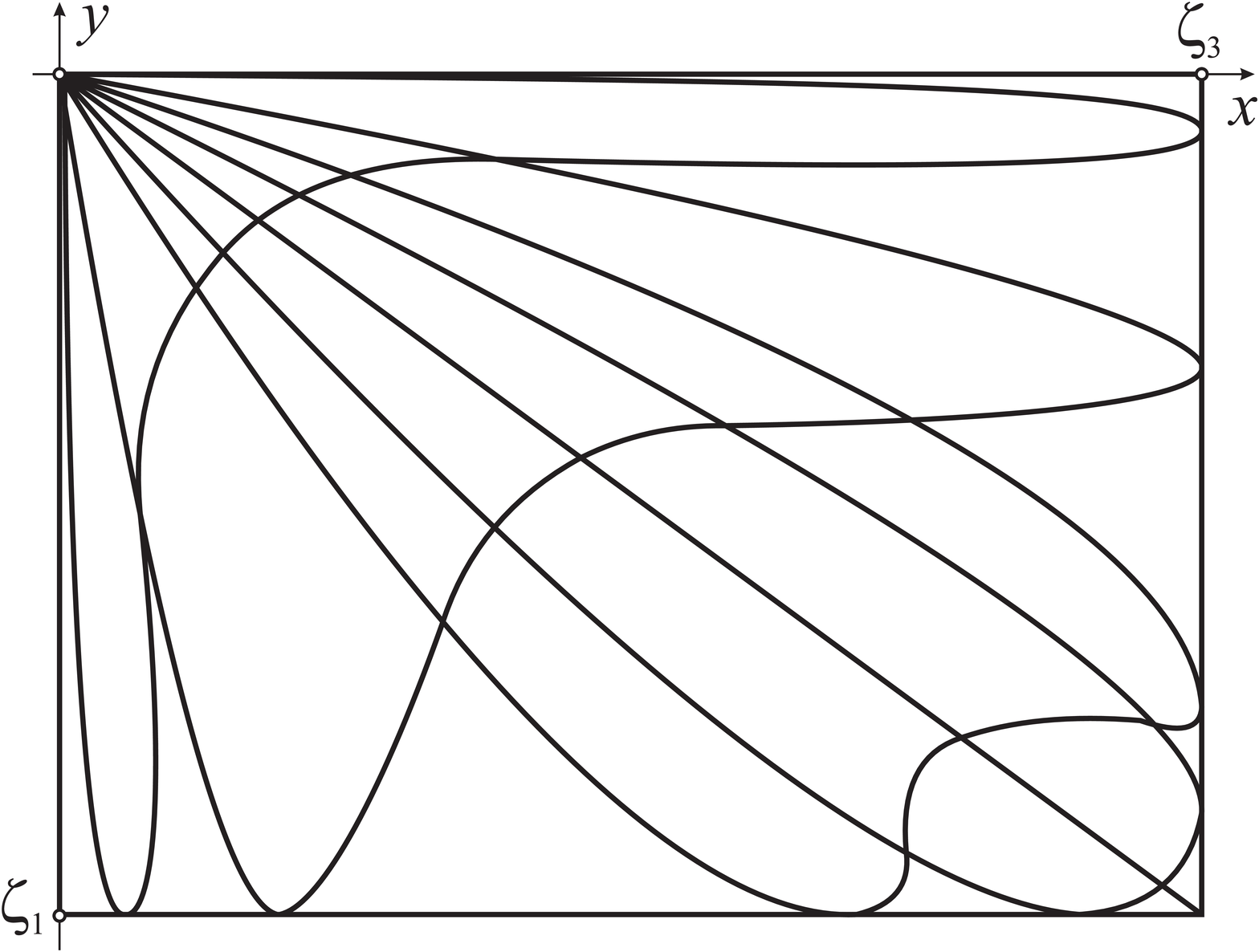}}
\caption{}\label{fig_5_5}
\end{figure}

Каждой псевдотраектории \eqref{eeq4_17} сопоставим два числа:
\begin{equation}\label{eeq4_21}
  \ds \varphi_+(\xi) = \lim \limits_{\tau \to +\infty} \frac{y(\tau)}{x(\tau)}, \quad \varphi_-(\xi) = \lim \limits_{\tau \to -\infty} \frac{y(\tau)}{x(\tau)}.
\end{equation}

\begin{lemma}\label{lem542}
Пределы \eqref{eeq4_21} существуют и конечны. Их произведение
\begin{equation}\label{eeq4_22}
  \theta = \varphi_+(\xi) \varphi_-(\xi)
\end{equation}
не зависит от псевдотраектории.
\end{lemma}
\begin{proof}
Обозначим $f(z) = \sqrt{(z - \lambda)^2 + (4 - \lambda^2)} > 0$ и зададим на множестве
$$
[\zeta_1, 0) {\times} [\zeta_1, 0) \cup (0, \zeta_3] {\times} (0, \zeta_3]
$$
функцию двух переменных
\begin{equation}\label{eeq4_23}
  \ds C(u, v) = \int \limits_u^v \frac{dz}{zf(z)\sqrt{(\zeta_3 - z)(z - \zeta_1)}},
\end{equation}
считая радикал в знаменателе положительным.

Рассмотрим кривую \eqref{eeq4_17}. В силу \eqref{eeq4_18} найдется $\tau_0$, такое, что
\begin{equation*}
  \begin{array}{c}
     x_0 = x(\tau_0) < \zeta_3, \quad x'(\tau_0) > 0,\\
     y_0 = y(\tau_0) > \zeta_1, \quad y'(\tau_0) < 0.
   \end{array}
\end{equation*}
Положим
\begin{equation}\label{eeq4_24}
  \tau_1 = \tau_0 + C(x_0, \zeta_3), \quad \tau_2 = \tau_0 + C(\zeta_1, y_0).
\end{equation}
В данном случае $W(z) = z^2f^2(z)(\zeta_3 - z)(z - \zeta_1)$, поэтому из системы \eqref{eeq4_4} находим
\begin{equation}\label{eeq4_25}
  C(x(\tau), \zeta_3) = |\tau - \tau_1|, \quad C(y(\tau), \zeta_1) = |\tau - \tau_2|.
\end{equation}
Подынтегральную функцию \eqref{eeq4_23} запишем в виде $\ds \frac{s}{z} + \Psi(z)$, где
$$
s = \frac{1}{2\sqrt{-\zeta_1\zeta_3}},
$$
а $\Psi(z)$ непрерывна при $z = 0$. Тогда
\begin{equation}\label{eeq4_26}
  \ds C(u,v) = s \ln\frac{v}{u} + \int \limits_u^v \Psi(z)dz.
\end{equation}
Из выражений \eqref{eeq4_25}, \eqref{eeq4_26} при $\tau > \max(\tau_1, \tau_2)$ имеем
\begin{equation*}
  \begin{array}{c}
    \ds \tau - \tau_1 = s \ln \frac{\zeta_3}{x(\tau)} + \int \limits_{x(\tau)}^{\zeta_3} \Psi(z)dz, \qquad
    \ds \tau - \tau_2 = s \ln \frac{\zeta_1}{y(\tau)} + \int \limits_{y(\tau)}^{\zeta_1} \Psi(z)dz.
  \end{array}
\end{equation*}
Учитывая \eqref{eeq4_24}, из предыдущих соотношений исключим $\tau$:
\begin{equation*}
\begin{array}{l}
    \ds s \ln \left| \frac{y(\tau)}{x(\tau)} \right| = 2s \ln \left| \frac{\zeta_1}{\zeta_3} \right| - s \ln \left| \frac{y_0}{x_0} \right| - \int \limits_{x_0}^{\zeta_3} \Psi(z)dz - \int \limits_{\zeta_1}^{y_0} \Psi(z)dz -\\
   \qquad \ds - \int \limits_{x(\tau)}^{\zeta_3} \Psi(z)dz - \int \limits_{\zeta_1}^{y(\tau)} \Psi(z)dz.
\end{array}
\end{equation*}
Используя \eqref{eeq4_18}, находим
\begin{equation}\label{eeq4_27}
\begin{array}{l}
    \ds s \ln \lim \limits_{\tau \to +\infty} \left| \frac{y(\tau)}{x(\tau)} \right| = 2s \ln \left| \frac{\zeta_1}{\zeta_3} \right| - s \ln \left| \frac{y_0}{x_0} \right| - \int \limits_{x_0}^{\zeta_3} \Psi(z)dz - \\
  \qquad \ds  - \int \limits_{\zeta_1}^{y_0} \Psi(z)dz - \int \limits_{\zeta_1}^{\zeta_3} \Psi(z)dz.
\end{array}
\end{equation}
Точно так же при $\tau < \min(\tau_1, \tau_2)$
\begin{equation*}
  \ds s \ln \left| \frac{y(\tau)}{x(\tau)} \right| = s \ln \left| \frac{y_0}{x_0} \right| + \int \limits_{x_0}^{\zeta_3} \Psi(z)dz + \int \limits_{\zeta_1}^{y_0} \Psi(z)dz - \int \limits_{x(\tau)}^{\zeta_3} \Psi(z)dz - \int \limits_{\zeta_1}^{y(\tau)} \Psi(z)dz,
\end{equation*}
откуда
\begin{equation}\label{eeq4_28}
  \ds s \ln \lim \limits_{\tau \to -\infty} \left| \frac{y(\tau)}{x(\tau)} \right| =  s \ln \left| \frac{y_0}{x_0} \right| + \int \limits_{x_0}^{\zeta_3} \Psi(z)dz + \int \limits_{\zeta_1}^{y_0} \Psi(z)dz - \int \limits_{\zeta_1}^{\zeta_3} \Psi(z)dz.
\end{equation}

Равенства \eqref{eeq4_27}, \eqref{eeq4_28} доказывают существование пределов \eqref{eeq4_21}. Складывая \eqref{eeq4_27}, \eqref{eeq4_28}, находим
\begin{equation*}
  \ds s\ln[\varphi_+ (\xi)\varphi_- (\xi)]=2\left[ s\ln\left|\frac{\zeta_1}{\zeta_3} \right| - \int \limits_{\zeta_1}^{\zeta_3}\Psi(z)dz\right],
\end{equation*}
поэтому величина
\begin{equation}\label{eeq4_29}
  \ds \theta = \frac{\zeta_1^2}{\zeta_3^2}\exp \left[-\frac{2}{s} \int \limits_{\zeta_1}^{\zeta_3}\Psi(z)dz\right]
\end{equation}
не зависит от выбора псевдотраектории, что и требовалось.
\end{proof}

Запишем величину $\theta$ в явном виде
\begin{equation}\label{eeq4_30}
\begin{array}{c}
  \ds \theta = \left(\frac{a-\lambda}{a+\lambda}\right)^2\exp\left\{-2\int \limits_{-a}^a \frac{(\zeta - \lambda)(\zeta^2-a^2+4)}{\sqrt{(\zeta^2+4-\lambda^2)(a^2-\zeta^2)}} {\times} \right.\\
  {\times} \left. \left[2\sqrt{a^2-\lambda^2} + \sqrt{(\zeta^2+4-\lambda^2)(a^2-\zeta^2)}\right]^{-1}d\zeta \right\}.
\end{array}
\end{equation}

Рассмотрим теперь траекторию \eqref{eeq4_19}. По замечанию \ref{rem542} для всех $n \in \mathbb{Z}$ выполняется условие
\begin{equation}\label{eeq4_31}
  \varphi_-(\xi_{n+1}) = 1/\varphi_+(\xi_n).
\end{equation}
Траектория замкнута тогда и только тогда, когда для некоторого $m$ справедливо $\varphi_-(\xi_0) = 1/\varphi_+(\xi_m)$, что в силу \eqref{eeq4_22}, \eqref{eeq4_31} дает $\theta^{m+1} = 1$.
Итак, имеем две возможности. Либо
\begin{equation}\label{eeq4_32}
  \theta = 1,
\end{equation}
и тогда все траектории замкнуты (в реальном времени замыкается уже каждая псевдотраектория), либо
\begin{equation}\label{eeq4_33}
  \theta \neq 1.
\end{equation}
В случае \eqref{eeq4_33} единственная замкнутая траектория \eqref{eeq4_20} является предельным циклом для всех остальных при $t \to \pm \infty$.

Равенство \eqref{eeq4_32}, очевидно, имеет место при $\lambda = 0$. Это объясняется тем, что на участке $k = 0$ в классическом решении существует дополнительный интеграл Горячева \cite{bib15}, расслаивающий поверхность $J_{h, 0}$ на периодические траектории.
Заметим, что
\begin{equation}\label{eeq4_34}
  \ds \left. \frac{\partial \theta}{\partial \lambda} \right|_{\lambda = 0} = \frac{2}{a}\bigl[\Phi(a) - 2\bigr],
\end{equation}
где
\begin{equation*}
  \ds \Phi(a) = \int \limits_{-\pi/2}^{\pi/2} \frac{(4-a^2 \cos^2\varphi)d\varphi}{\sqrt{4 + a^2\sin^2\varphi}\,\bigl[2+\sqrt{4 + a^2\sin^2\varphi}\cos\varphi\bigr]}.
\end{equation*}
Нетрудно показать, что $\Phi(a) < 2$, если $0 < a^2 < 4$. Тогда из \eqref{eeq4_34} следует, что при малых $\lambda$ выполнено \eqref{eeq4_33}\footnote[1]{Более того, численный анализ интеграла \eqref{eeq4_30} дает $\theta < 1$ на всем участке $5 - 7_*$.}. Таким образом, аналога интеграла Горячева при $\lambda \neq 0$ не существует.

Выбирая на множестве $\{(x, y): x \in [0, \zeta_3], \; y \in [\zeta_1, 0]\}$ непрерывные векторные поля из допустимых скоростей, определяемых системой \eqref{eeq2_20}, получаем части интегральной поверхности, показанные на рис.~\ref{fig_5_6}. Они складываются в прямоугольник, при необходимом отождествлении сторон которого возникает двумерный тор с нестягиваемым путем $a_+b_+b_-a_- = \delta$, отвечающим периодическому решению \eqref{eeq4_20}. Формулы \eqref{eeq4_8} показывают, что точки, в которых $x = 0$ или $y = 0$ при $k = 0$, имеют один прообраз при отображении \eqref{eeq4_6}. Следовательно, $D^{-1}(\delta) = S^1$. Поскольку у остальных точек два прообраза, то $J_{h, 0}$ есть расслоение над $S^1$ со слоем $S^1 \dot{\cup} S^1$.

\begin{figure}[ht]
\center{\includegraphics[width=0.5\linewidth]{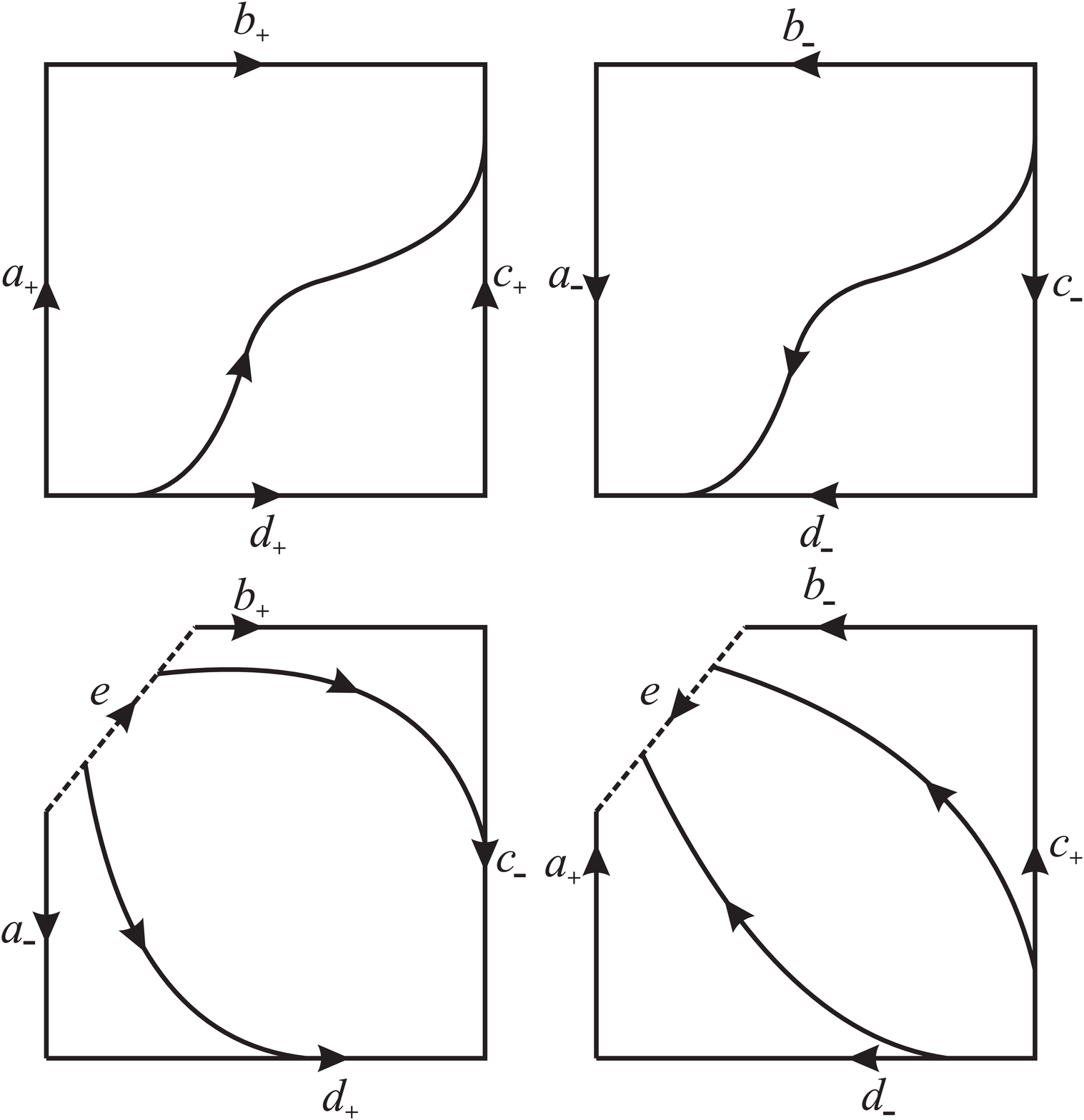}}
\caption{}\label{fig_5_6}
\end{figure}

Проведем явное доказательство того, что два листа асимптотической поверхности $J_{h, 0}$ пересекаются по предельному циклу трансверсально.
Обозначим, как и раньше
\begin{equation}\label{eeq4_35}
  a = \sqrt{2(h+1)}.
\end{equation}
На участке $5 -7_*$ (рис.~\ref{fig_5_3}) имеем $\lambda < a< 2$. Полагая в \eqref{eeq1_1} - \eqref{eeq1_3}
\begin{equation}\label{eeq4_36}
  \omega_1 = \omega_2 = \nu_3 =0,
\end{equation}
находим выражение остальных переменных через функции Якоби с модулем $a/2$:
\begin{equation}\label{eeq4_37}
  \nu_1 = 1-(a^2/2)\sn^2\varphi, \quad \nu_2 = -2a\sn\varphi\dn\varphi, \quad \omega_3 = a\cn\varphi, \quad \varphi = 2(t - t_0).
\end{equation}
Решение \eqref{eeq4_36}, \eqref{eeq4_37} соответствует \eqref{eeq4_20}. К этой траектории проведем через точку $\varphi = 0$ ортогональную гиперплоскость $\nu_2 = 0$. Ее пересечение с $J_{h, 0}$ определяется уравнениями
\begin{equation}\label{eeq4_38}
  \begin{array}{c}
    \nu_1^2 + \nu_3^2 = 1, \quad 4\omega_1 \nu_1 + (\omega_3 + \lambda)\nu_3 = 0, \\
    (\omega_3 - \lambda)(\omega_1^2 + \omega_2^2) + \omega_1\nu_3 = 0, \quad 4(\omega_1^2 + \omega_2^2) + \omega_3^2 - 2\nu_1 = 2h.
  \end{array}
\end{equation}
Выбирая в качестве малого параметра $\varepsilon$ в окрестности предельного цикла величину  $\nu_3$, получим параметрическое задание сечения поверхности $J_{h, 0}$:
\begin{equation*}
  \begin{array}{c}
    \ds \nu_1 = 1 - \frac{\varepsilon^2}{2}, \quad \nu_2 = 0, \quad \nu_3 = \varepsilon, \\
    \ds \omega_1 = -\frac{\varepsilon}{4}(a + \lambda), \quad \omega_2 = \pm \frac{\varepsilon}{4}\sqrt{\frac{a + \lambda}{a - \lambda}(4 + \lambda^2 - a^2)}, \quad \omega_3 = a - \frac{\varepsilon^2}{a - \lambda}.
  \end{array}
\end{equation*}
Это, очевидно, <<крест>>, что и требовалось доказать.

Отметим, что уравнения \eqref{eeq4_36} - \eqref{eeq4_38} могут послужить основой для изучения расщепления сепаратрис \cite{bib26} в случаях, близких к решению Сретенского.

\begin{figure}[ht]
\centering
\center{\includegraphics[width=0.5\linewidth]{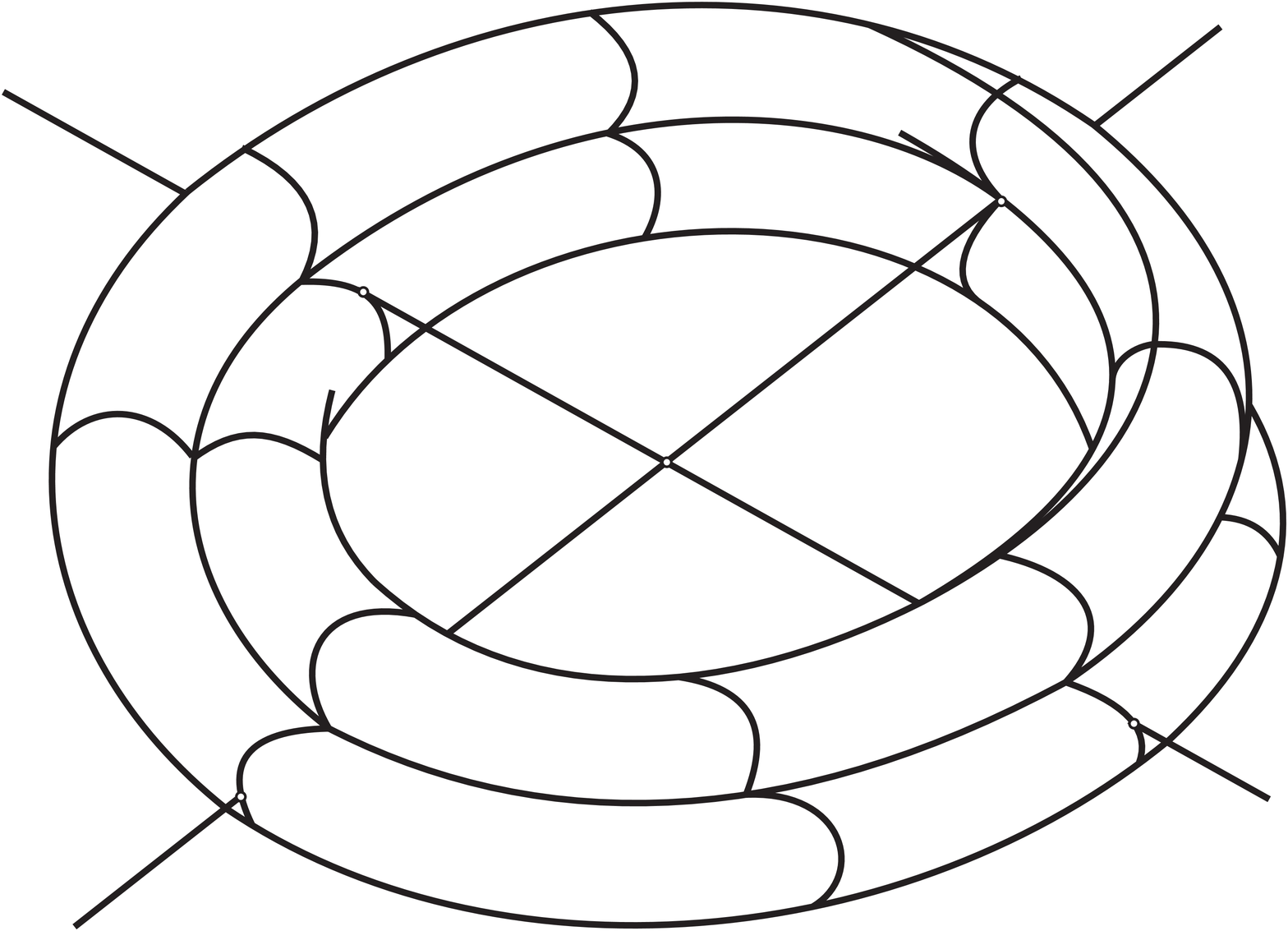}}
\caption{}\label{fig_5_7}
\end{figure}

Установим окончательно глобальный вид поверхности $J_{h, 0}$ в случае $5 - 7_*$. Ориентируемые расслоения над $S^1$ со слоем $S^1 \dot{\cup} S^1$ могут быть двух типов: уже встречавшееся прямое произведение $S^1 {\times} (S^1 \dot{\cup} S^1)$ (см., например, рис.~\ref{fig_3_2}) и нетривиальное расслоение, в котором при обходе окружности <<восьмерка>> отображается на себя центрально-симметрично. Соответствующая поверхность показана на рис.~\ref{fig_5_7}. Обозначим ее через $S^1 * (S^1 \dot{\cup} S^1)$.

При переходе $(I - \ts{IV}_*) \to (5 - 7_*) \to (\ts{II} \to \ts{IV}_*)$ (рис.~\ref{fig_5_3}) происходит перестройка ${\bf T}^2 \to J_{h, 0} \to {\bf T}^2$ (см. табл. \ref{tab53}). Ясно, что это возможно лишь при условии $J_{h, 0} = S^1 * (S^1 \dot{\cup} S^1)$. Соответствующую трансформацию двумерного тора можно по аналогии с одномерным случаем назвать <<бифуркацией удвоения длины>>.

Рассмотрим теперь точку $h = 1$, $k = 0$. Она имеет шифр: $5 - N_3^*$, если $0 < \lambda < 2$; $N_1 - N_3^*$, если $\lambda = 2$; $6 - N_3^*$, если $\lambda > 2$. В двух последних случаях корни $W(z)$ неотрицательны, и потому $y(t) \equiv 0$. Для $x(t)$ имеем уравнение
\begin{equation*}
  \ds \frac{dx}{dt} = \frac{1}{2}(x - \lambda)\sqrt{[(\lambda + 2) - x][x - (\lambda - 2)]}.
\end{equation*}
В силу \eqref{eeq2_3} получаем
\begin{equation}\label{eeq4_39}
  \ds \omega_1 = \omega_2 \equiv 0, \quad \frac{d\omega_3}{dt} = \frac{1}{2} \omega_3 \sqrt{(2 - \omega_3)(\omega_3 + 2)}.
\end{equation}
Эти зависимости описывают движения физического маятника, полная энергия которого равна потенциальной энергии в верхнем (неустойчивом) положении равновесия. Интегральное многообразие, таким образом, есть $S^1 \dot{\cup} S^1$.

Пусть $0 < \lambda < 2$. В точке $5 - N_3^*$ система \eqref{eeq4_4} имеет особые решения:
\begin{gather}
  x \equiv \lambda, \quad  y \equiv 0,\label{eeq4_40}\\
  x \equiv 0, \quad y(\tau) \in \Asp[\lambda - 2, 0),\label{eeq4_41}\\
  \ds x = \frac{\lambda}{2}\left[1 \pm \frac{e^u-1}{e^u+1} \right], \quad y = 0,\label{eeq4_42}\\
  x(\tau) \in \Asp (\lambda, \lambda + 2], \quad y \equiv 0,\label{eeq4_43}\\
  x \equiv \lambda, \quad y(\tau) \in \Asp[\lambda - 2, 0),\label{eeq4_44}
\end{gather}
где $u = u(\tau)$ -- монотонная функция, удовлетворяющая уравнению
\begin{equation*}
  \ds \frac{du}{d\tau} = \frac{\sqrt{4(1+e^u)^2 - \lambda^2}}{1 + e^u}.
\end{equation*}

Решения \eqref{eeq4_40} -- \eqref{eeq4_43} соответствуют движениям физического маятника \eqref{eeq4_39}. При этом псевдотраектории \eqref{eeq4_41}, \eqref{eeq4_42}, смыкаясь в точке $x = y = 0$, образуют траекторию $\omega_3(t) \in \Asp[-2, 0)$. В целом решения системы \eqref{eeq4_4} приведены на рис.~\ref{fig_5_8}. Там же указана схема их прохождения в реальном времени $t$. Таким образом, все траектории в рассматриваемом случае гомоклинические, т.е. своей $\alpha-$ и $\omega-$предельной точкой они имеют неустойчивое положение равновесия тела.

\begin{figure}[ht]
\center{\includegraphics[width=0.5\linewidth]{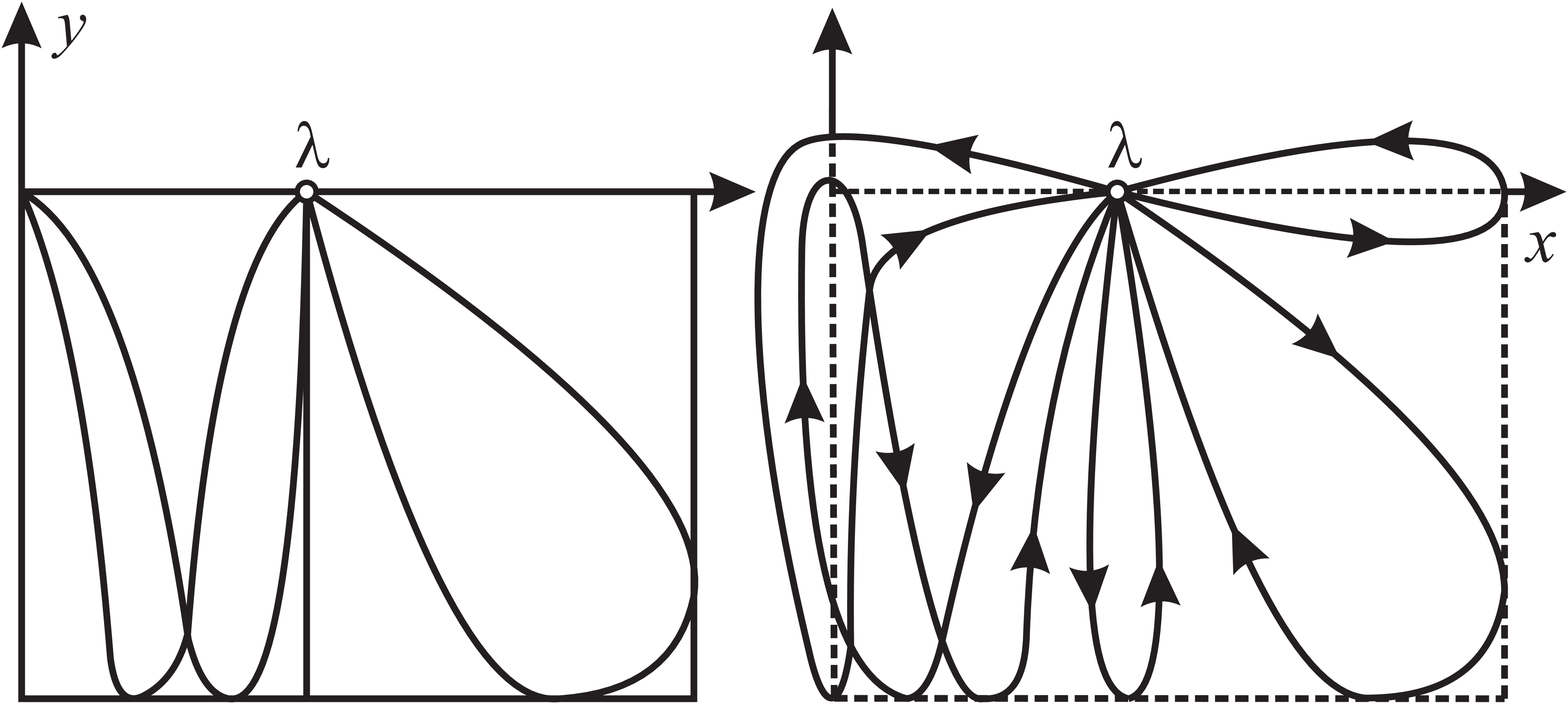}}
\caption{}\label{fig_5_8}
\end{figure}

\begin{figure}[ht]
\center{\includegraphics[width=0.5\linewidth]{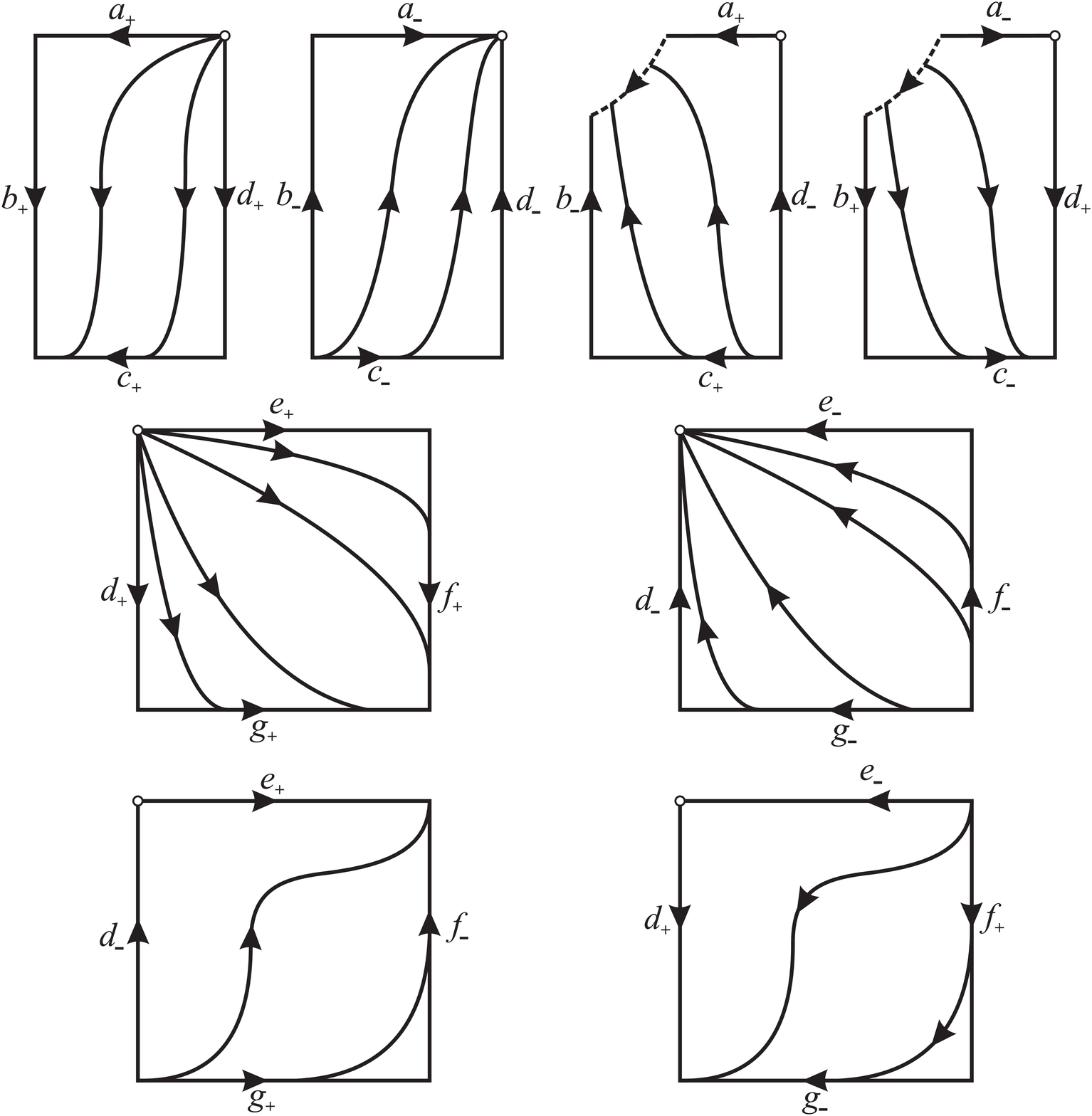}}
\caption{}\label{fig_5_9}
\end{figure}

Процесс построения интегральной поверхности показан на рис.~\ref{fig_5_9}. Прямоугольники в верхней части рисунка соответствуют траекториям, на которых $x \ls \lambda$, в нижней части -- $x \gs \lambda$ (отрезки $c_\pm, f_\pm, g_\pm$ являются не решениями, а лишь линиями склейки). Производя должные отождествления и учитывая удвоение всех точек, кроме \eqref{eeq4_40} -- \eqref{eeq4_43}, получаем в случае $x \ls \lambda$ описанную выше поверхность $S^1 * (S^1 \dot{\cup} S^1)$, а в случае $x \gs \lambda$ произведение $S^1 {\times} (S^1 \dot{\cup} S^1)$. Обе поверхности целиком <<сотканы>> из петель-траекторий, двоякоасимптотических к верхнему неустойчивому положению равновесия тела-носителя. Они образуют интегральное многообразие, пересекаясь по <<восьмерке>> $d_+d_- \dot{\cup} d_+d_-$. В принятых обозначениях (соглашение \ref{sog411}) имеем
\begin{equation*}
  J_{1,0} = (S^1 {\times} (S^1 \dot{\cup} S^1)) \bigcup \limits^{S^1 \dot{\cup} S^1} (S^1 * (S^1 \dot{\cup} S^1)).
\end{equation*}
Достаточно ясно, как эта поверхность получается из $J_{h, 0}$ для случая $5 - 7_*$ (рис.~\ref{fig_5_3}). Возьмем на $S^1 * (S^1 \dot{\cup} S^1)$ две гомотопные между собой <<восьмерки>> и отождествим их по отображению, индуцированному гомотопией. Результат будет гомеоморфен $J_{1, 0}$. Конечно, в трехмерном пространстве реализовать $J_{1, 0}$ нельзя. Однако легко построить наглядную модель, заменив <<восьмерку>> отрезком: отождествление между собой двух секущих отрезков листа Мёбиуса равносильно склейке секущего отрезка листа Мёбиуса шириной~$l$ с образующей цилиндра высотой~$l$.

Обратимся к участку $5 - 6_*$. Имеются две возможности (см. табл. \ref{tab52}):
\begin{equation}\label{eeq4_45}
  x \in [0, \zeta_2^*], \quad y \in [\zeta_1, 0]
\end{equation}
или
\begin{equation}\label{eeq4_46}
  x \in [\zeta_3^*, \zeta_3], \quad y \in [\zeta_1, 0].
\end{equation}

Вариант \eqref{eeq4_45} аналогичен участку $5 -7_*$. Одна периодическая траектория склеивается из псевдотраекторий $x \equiv 0$, $y(\tau) \in \Asp[\zeta_1, 0)$ и $x(\tau) \in \Asp(0, \zeta_2^*]$, $y\equiv 0$ (вращение маятника с большой энергией), остальные стремятся к этому движению при $t \to \infty$. В случае \eqref{eeq4_46} замена \eqref{eeq4_2} всюду невырождена, и движения таковы: одно периодическое $x(t) \in \Osc[\zeta_3^*, \zeta_3]$, ${y \equiv 0}$, остальные -- двоякоасимптотические $x(t) \in \Osc[\zeta_3^*, \zeta_3]$,
$y  \in \Asp[\zeta_1, 0)$.

Первому варианту в пространстве $(\mbs{\nu}, \mbs{\omega})$ отвечает косое произведение <<восьмерки>> на окружность, второму -- $S^1 {\times} (S^1 \dot{\cup} S^1)$. Несвязное объединение этих поверхностей есть $J_{h, 0}$. Непрерывность трансформация интегральной поверхности при переходе $(5 - 7_*) \to (5 - N_1^*) \to (5 - 6^*)$ совершенно наглядна.

В точке $5 - N_1^*$ $(k = 0, h = 1 + \lambda^2 / 2)$ ноль становится трехкратным корнем $W(z)$. Аналогично случаю $N_1 - 7_*$ имеем решение
\begin{equation}\label{eeq4_47}
  x \equiv 0, \quad y(t) \in \Osc[\zeta_1, 0].
\end{equation}
Это решение изолировано. Остальные $x(t) \in \Osc[\zeta_3^*, \zeta_3]$, ${y \equiv 0}$ и $x(t) \in \Osc[\zeta_3^*, \zeta_3]$, $y  \in \Asp[\zeta_1, 0)$, как и в предыдущем случае, дают поверхность $S^1 {\times} (S^1 \dot{\cup} S^1)$. Многообразие $J_{h, 0}$ -- несвязное объединение данной поверхности и окружности, отвечающей \eqref{eeq4_47}.

Наконец, в точках с шифром $5 - 5_*$ $(k = 0, h > 1 + \lambda^2 / 2)$ система \eqref{eeq2_20} имеет решения
\begin{equation}\label{eeq4_48}
  x \equiv 0, \quad y(t) \in \Osc[\zeta_1, \zeta_1^*]; \qquad x(t) \in \Osc[\zeta_3^*, \zeta_3], \quad y \equiv 0
\end{equation}
(вращения маятника) и
\begin{equation}\label{eeq4_49}
  x(t) \in \Osc[\zeta_3^*, \zeta_3], \quad y(t) \in \Osc[\zeta_1, \zeta_1^*].
\end{equation}
Интегральное многообразие есть несвязное объединение тора с условно-пе\-ри\-оди\-чес\-кими траекториями \eqref{eeq4_49} (по предложению \ref{pro541} раздвоения нет) и двух окружностей, соответствующих периодическим решениям \eqref{eeq4_48}:
$$
J_{h, 0} = 2S^1 \cup {\bf T}^2.
$$
Отметим, что ввиду явной несимметрии многочлена $W(z)$ число вращения на ${\bf T}^2$, равное отношению периодов функций \eqref{eeq4_49}, зависит от $\lambda$. Это вновь доказывает несуществование при $\lambda \neq 0$ интеграла, аналогичного интегралу Горячева.

\subsection{Геометрический анализ случая Горячева -- Чаплыгина}\label{ssec55}
Проекции интегральных многообразий на сферу Пуассона изучим в классическом случае
\begin{equation}\label{eeq5_1}
  \lambda = 0.
\end{equation}
Известно \cite{bib26}, что при этом для почти всех начальных движений средняя прецессия тела равна нулю. Поэтому знание того, как с течением времени меняется положение вертикали в подвижном триэдре, в значительной мере проясняет и характер движения самого твердого тела.

Уравнения, определяющие поверхность \eqref{eeq4_1} в $\mathfrak{M}$, при условии \eqref{eeq5_1} запишем в виде
\begin{eqnarray}
& &  4(\omega_1 \nu_1 + \omega_2 \nu_2) + \omega_3 \nu_3 = 0,\label{eeq5_2}\\
& &  4(\omega_1^2 + \omega_2^2) + \omega_3^2 -2 \nu_1 = 2h,\label{eeq5_3}\\
& &  4\omega_3(\omega_1^2 + \omega_2^2) + 2\omega_1 \nu_3 = 2k.\label{eeq5_4}
\end{eqnarray}

Область возможности движения, т.е. образ $J_{h, k}$ при отображении \eqref{ceq4_3}, обозначим $U_{h, k}$. В обобщенную границу ОВД отображаются точки $J_{h, k}$, удовлетворяющие условию \eqref{ceq4_12}, которое в данном случае имеет вид
\begin{equation}\label{eeq5_5}
  2[2(\omega_1^2 + \omega_2^2) -\omega_3^2](\nu_1 \omega_2 - \nu_2\omega_1) + (\nu_2 \omega_3 - \nu_3\omega_2)\nu_3 = 0.
\end{equation}
Запишем его в переменных $x, y$. Для этого выполним подстановку выражений \eqref{eeq4_8}, \eqref{eeq4_3} с учетом \eqref{eeq5_1}. В результате получим
\begin{equation}\label{eeq5_6}
  \begin{array}{l}
\ds    \left[2\left(x + \frac{1}{2}y\right)\left(\frac{1}{2}x + y\right)(x - y)^2 - \right.\\[3mm]
\ds\qquad     \left. - \left(x + \frac{1}{2}y\right)Z(x) -\left(\frac{1}{2}x + y\right)Z(y)\right]X_*Y_* -\\[3mm]
\ds    -\left[2\left(x + \frac{1}{2}y\right)\left(\frac{1}{2}x + y\right)(x - y)^2 - \right.\\[3mm]
\ds\qquad  -\left. \left(x + \frac{1}{2}y\right)Z_*(x) -\left(\frac{1}{2}x + y\right)Z_*(y)\right]XY = 0.
  \end{array}
\end{equation}
Для построения обобщенной границы необходимо теперь отыскать множество решений уравнения \eqref{eeq5_6}, удовлетворяющих \eqref{eeq3_1}, и применить к этому множеству отображение $(x, y) \to \mbs{\nu}$, определяемое равенствами \eqref{eeq4_8}.

Избавляясь от иррациональностей в \eqref{eeq5_6}, найдем
\begin{equation}\label{eeq5_7}
  \begin{array}{l}
    3x^5y^5 +8hx^4y^4 + 4[h^2 - 1 - 3k(x+y) + h(x+y)^2]x^3y^3 - \\[3mm]
    \qquad - 6k[k + 4h(x+y) - 2(x+y)^3]x^2y^2- [4k^2h + 8k(h^2 - 1)(x + y) -\\[3mm]
    \qquad - 9k^2(x + y)^2 - 8kh(x + y)^3 + 4(h^2 - 1)(x + y)^4]xy +\\[3mm]
   \qquad +2k[3k^2 + 8kh(x + y) + 4(h^2 - 1)(x + y)^2 - \\[3mm]
   \qquad -6k(x + y)^3 - 4h(x + y)^4](x + y) = 0.
  \end{array}
\end{equation}
Это уравнение удобно для численного анализа.

Обобщенная граница ОВД делит сферу Пуассона на области, внутри которых неизменно количество допустимых скоростей, т.е. число решений относительно $\omega_1, \omega_2, \omega_3$ системы \eqref{eeq5_2}~-- \eqref{eeq5_4}. Укажем способ определения этого числа.
Очевидно, каждая из упомянутых областей, будучи открытым множеством, имеет точки, в которых
\begin{equation}\label{eeq5_8}
  \nu_2\nu_3 \neq 0.
\end{equation}
Из уравнения \eqref{eeq5_3} находим
\begin{equation}\label{eeq5_9}
  4(\omega_1^2 + \omega_2^2) = 2(h + \nu_1) - \omega_3^2,
\end{equation}
после чего из \eqref{eeq5_2}, \eqref{eeq5_4}
\begin{equation}\label{eeq5_10}
  \begin{array}{c}
    \ds \omega_1 = \frac{1}{4\nu_3}[\omega_3^2 - 2(h + \nu_1)\omega_3 + 2k], \\
    \ds \omega_2 = -\frac{1}{4\nu_2\nu_3}\left\{\nu_1\omega_3^3 - [2(h + \nu_1)\nu_1-\nu_3^2]\omega_3 + 2k\nu_1\right\}.
  \end{array}
\end{equation}
Подстановка значений \eqref{eeq5_10} в \eqref{eeq5_9} приводит к уравнению, определяющему $\omega_3$:
\begin{equation}\label{eeq5_11}
\begin{array}{l}
  (\nu_1^2+\nu_2^2)\omega_3^6 + 2[\nu_1\nu_3^3 - 2(\nu_1^2+\nu_2^2)(h + \nu_1)]\omega_3^4 + 4k(\nu_1^2+\nu_2^2)\omega_3^3+ \\[3mm] \qquad
  +[\nu_3^4 + 4\nu_2^2\nu_3^2 - 4\nu_1\nu_3^2(h + \nu_1) + 4(\nu_1^2+\nu_2^2)(h + \nu_1)^2]\omega_3^2 + \\[3mm] \qquad
  +4k[\nu_1\nu_3^2 - 2(\nu_1^2+\nu_2^2)(h + \nu_1)]\omega_3 + \\[3mm] \qquad
  +4[k^2(\nu_1^2  + \nu_2^2) - 2\nu_2^2\nu_3^2(h + \nu_1)]=0.
\end{array}
\end{equation}

Обозначим через $P(\omega_3)$ выражение, стоящее в левой части \eqref{eeq5_11}, рассматриваемое как многочлен от $\omega_3$. Поскольку $\omega_1$, $\omega_2$ при фиксированном $\mbs{\nu}$ из области \eqref{eeq5_8} однозначно выражены через $\omega_3$ равенством \eqref{eeq5_10}, то число допустимых скоростей в точке $\mbs{\nu}$ равно числу различных действительных корней многочлена $P(\omega_3)$. В точках обобщенной границы $P(\omega_3)$ имеет кратный корень. В частности, уравнение обобщенной границы непосредственно в координатах $\nu_1, \nu_2, \nu_3$ можно получить, приравняв нулю результат многочлена $P(\omega_3)$ и его производной.

Как и ранее (см. $\S$ \ref{ssec43}), назовем разделяющим множество точек $(h, k) \in \mathbb{R}^2$, при переходе через которое меняется вид ОВД. Разделяющее множество состоит из некоторых кривых, которые также назовем разделяющими.

Заметим, что уравнения \eqref{beq4_5}, \eqref{eeq5_2} -- \eqref{eeq5_4} допускают одновременную замену знаков у $\omega_1, \omega_2, \nu_3$ и у $\omega_2, \nu_2$. В силу этого ОВД симметрична относительна сечений сферы Пуассона плоскостями
\begin{equation}\label{eeq5_12}
  \nu_2 = 0
\end{equation}
и
\begin{equation}\label{eeq5_13}
  \nu_3 = 0
\end{equation}

Отображение $(\nu_3, \omega_3) \mapsto (-\nu_3, -\omega_3)$ переводит точку $(h, k)$ в точку $(-h, -k)$. С другой стороны, в силу отмеченной симметрии, ОВД переводится сама в себя. Это означает, что ${U_{h, k}} = {U_{-h, -k}}$ и достаточно ограничиться значениями
\begin{equation}\label{eeq5_14}
  k > 0
\end{equation}
(в случае Горячева $k = 0$ все траектории замкнуты, и геометрический анализ ничего нового уже не дает).

Бифуркационное множество, отвечающее рассматриваемому случаю, указано на рис.~\ref{fig_5_1},\,{\it а}. В полуплоскости \eqref{eeq5_14} ОВД не пуста лишь в области между кривыми $\Bif$ и $\Bif_*$ и в области, лежащей правее $\Bif_*$. Для краткости присвоим этим областям шифры $I$ (старый шифр $I-\ts{IV}_*$) и $\ts{II}$ (старый шифр $I-I_*$).

Выясним, какие перестройки обобщенной границы ОВД происходят на сечениях \eqref{eeq5_12} и \eqref{eeq5_13}.

В случае \eqref{eeq5_13} из системы \eqref{beq4_5}, \eqref{eeq5_2} -- \eqref{eeq5_5} находим
\begin{equation}\label{eeq5_15}
  \ds \nu_1 = \frac{3}{2}k^{2/3} - h, \quad \nu_2 = \pm \left[1 - \left(\frac{3}{2}k^{2/3} - h \right)^2\right]^{1/2}.
\end{equation}
Эти значения вещественны лишь в области $I$. Таким образом, здесь в качестве разделяющих кривых выступают бифуркационные: в области $I$ обобщенная граница имеет две симметричные относительно плоскости \eqref{eeq5_12} точки пересечения с плоскостью \eqref{eeq5_13}, в области $\ts{II}$ пересечений нет.

Полезно отметить следующее. Точка $(1, 0, 0)$ всегда принадлежит ОВД (если, конечно, последняя не пуста). Точка $(-1, 0, 0)$ принадлежит внутренности ОВД для значений $(h, k)$ из области $\ts{II}$ и лежит на обобщенной границе, если $(h, k) \in \Bif_*$. Следовательно, для почти всех начальных данных траектория центра масс проходит сколь угодно близко к его нижнему положению. Для верхнего положения центра масс это свойство сохраняется лишь при условии $h \gs \frac{3}{2}k^{2/3} + 1$.

Полюсы $(0, \pm1, 0)$ сферы Пуассона принадлежат ОВД, если $h \gs \frac{3}{2}k^{2/3}$. Лишь при этом условии, следовательно, вторая главная ось инерции может занимать вертикальное положение. В этом смысле кривая
\begin{equation}\label{eeq5_16}
  h = \frac{3}{2}k^{2/3},
\end{equation}
конгруэнтная бифуркационным кривым, является разделяющей, хотя, как будет видно из дальнейшего, изменения типа ОВД при ее пересечении не происходит.

Пусть выполнено \eqref{eeq5_12}. Обозначая
\begin{equation*}
  \begin{array}{l}
    Q(\omega_3) = 8(h + \nu_1)\nu_1^2 - (1+3\nu_1^2)\omega_3^2 ,\\[3mm]
    R(\omega_3) = \nu_1\omega_3^3 + (1 - 2h\nu_1 - 3\nu_1^2)\omega_3 + 27 \nu_1,
  \end{array}
\end{equation*}
перепишем систему \eqref{eeq5_2} -- \eqref{eeq5_4} в виде
\begin{equation}\label{eeq5_17}
  4\nu_3\omega_1 = \omega_3^2 - 2(h + \nu_1)\omega_3 + 2k, \quad 16\nu_1^2\omega_2^2 = Q(\omega_3), \quad R(\omega_3) = 0.
\end{equation}
Отсюда следует, что точка $(\nu_1, 0, \nu_3)$ принадлежит ОВД тогда и только тогда, когда существует вещественный корень $R(\omega_3)$, на котором $Q(\omega_3) \gs 0$. Уравнение \eqref{eeq5_5} в подстановке \eqref{eeq5_17} становится таким: $Q(\omega_3)R'(\omega_3) = 0$. Следовательно, в точках обобщенной границы имеет решение относительно $\omega_3$ хотя бы одна из систем $R(\omega_3) = 0$, $R'(\omega_3) = 0$ или $R(\omega_3) = 0$, $Q(\omega_3) = 0$. Условия совместности этих систем имеют соответственно вид
\begin{eqnarray}
& \nu_1(2h - 3k^{2/3}) = 1 - 3\nu_1^2,\label{eeq5_18}\\
& (1-\nu_2^2)^2[8\nu_1^2 h^3 - 8\nu_1 h^2 + 2(1 + \nu_1^2)(1 - 3\nu_1^2)h + \nonumber \\
& + 2\nu_1(1 + \nu_1^2)^2] - (1+3\nu_1^2)^3k^2 = 0.\label{eeq5_19}
\end{eqnarray}
Разделяющий случай получим, когда уравнения \eqref{eeq5_18} и \eqref{eeq5_19} выполнены при одном и том же значении $\nu_1$, т.е. точка, в которой $R(\omega_3)$ имеет кратный корень, лежит на границе области совместности системы \eqref{eeq5_17}. Выразим $h$ из условия \eqref{eeq5_18} и подставим в \eqref{eeq5_19}:
\begin{equation}\label{eeq5_20}
\begin{array}{l}
  (1 - 9\nu_1^2)^2k^2 - 9\nu_1(1 - \nu_1^2)^2(1 - 9\nu_1^2)k^{4/3} +  \\[3mm]
\qquad   +  24\nu_1(1 - \nu_1^2)^2(1 - 3\nu_1^2)k^{2/3} - 16\nu_1^3(1 - \nu_1^2)^3 = 0.
\end{array}
\end{equation}
Таким образом, получено кубическое уравнение относительно $k^{2/3}$. Его корни легко найти, если заметить, что один из них кратный. Определив из \eqref{eeq5_20} $k = k(\nu_1)$, зависимость $h$ от $\nu_1$ установим из \eqref{eeq5_18}. В результате запишем параметрические уравнения следующих разделяющих кривых:
\begin{eqnarray}
& & \left\{  \begin{array}{l}
     \ds h = \frac{1 + 15\nu_1^4}{2\nu_1(1 - 9\nu_1^2)} \\
     \ds k = \left[\frac{4\nu_1(1 - \nu_1^2)}{1 - 9\nu_1^2} \right]^{3/2}
  \end{array} \right., \quad    \nu_1 \in [-1, -1/3) \cup (0, 1/3); \label{eeq5_21}\\
& & \left\{  \begin{array}{l}
    \ds h = \frac{1 - 3\nu_1^4}{2\nu_1} \\
    \ds k = \left[\nu_1(1 - \nu_1^2)\right]^{3/2}
  \end{array} \right., \quad     \nu_1 \in (0, 1].\label{eeq5_22}
\end{eqnarray}

Отметим здесь также кривую
\begin{equation}\label{eeq5_23}
  2h = k^2,
\end{equation}
аналогичную \eqref{eeq5_16}: если $2h \gs k^2$, то полюсы $(0, 0\pm1)$ сферы Пуассона принадлежат ОВД и ось динамической симметрии твердого тела при почти всех начальных данных проходит сколь угодно близко к вертикали.

Изучим теперь возможные перестройки, происходящие в области \eqref{eeq5_8}. Они появляются в типичном случае, когда на линии складки возникают (попарно) точки сборки, образуя в проекции на сферу криволинейный треугольник <<хвоста>>. Как было показано в предположении \eqref{eeq5_8}, система \eqref{eeq5_2} -- \eqref{eeq5_4} сводится к одному уравнению \eqref{eeq5_11}. Обобщенная граница ОВД находится из условий
\begin{equation}\label{eeq5_24}
  P(\omega_3) = 0, \quad P'(\omega_3) = 0.
\end{equation}
Таким образом, реализуется ситуация, подробно изученная в $\S$ \ref{ssec34}: особые точки (точки возврата) обобщенной границы удовлетворяют уравнению
\begin{equation}\label{eeq5_25}
  P''(\omega_3) = 0,
\end{equation}
а количество особых точек в области \eqref{eeq5_8} может измениться лишь при переходе через такие значения $h, k$, при которых в дополнение к \eqref{eeq5_24}, \eqref{eeq5_25} справедливо равенство
\begin{equation}\label{eeq5_26}
  P'''(\omega_3) = 0.
\end{equation}
Система \eqref{eeq5_24} -- \eqref{eeq5_26} приводит либо к случаю Горячева $k = 0$, который условились не рассматривать, либо, с учетом единичности вектора $\mbs{\nu}$, к параметрическим уравнениям искомой разделяющей кривой $(-1 < \nu_3 < 1)$:
\begin{equation}\label{eeq5_27}
\left\{  \begin{array}{l}
    \ds h = -\frac{1}{2}(11 - 24\nu_3^2 + 12\nu_3^4)\left[\frac{4 - 3\nu_3^2}{(1 - \nu_3^2)(121 - 228\nu_3^2 + 108\nu_3^4)} \right]^{1/2}, \\
    \ds k = 8\nu_3^3\left[\frac{(1 - \nu_3^2)(4 - 3\nu_3^2)}{121 - 228\nu_3^2 + 108\nu_3^4}\right].
  \end{array}\right.
\end{equation}

На рис.~\ref{fig_5_10} нанесены разделяющие кривые \eqref{eeq5_21}, \eqref{eeq5_22}, \eqref{eeq5_27}. Совместно с бифуркационным множеством они разбивают область значений $(h, k)$ с непустыми интегральными многообразиями (точки, лежащие правее кривой $\Bif$\,) на 10 подобластей (\textit{а} - \textit{е}, \textit{а}$'$ - \textit{г}$'$). Для каждой из этих подобластей с помощью ЭВМ на основе уравнений \eqref{eeq5_6}, \eqref{eeq5_7}, \eqref{eeq5_11} выполнено построение обобщенных границ и определено число допустимых скоростей в связных компонентах их дополнения. Соответствующие ОВД показаны на рис.~\ref{fig_5_11}. Здесь требуются некоторые пояснения.

\begin{figure}[ht]
\center{\includegraphics[width=0.5\linewidth]{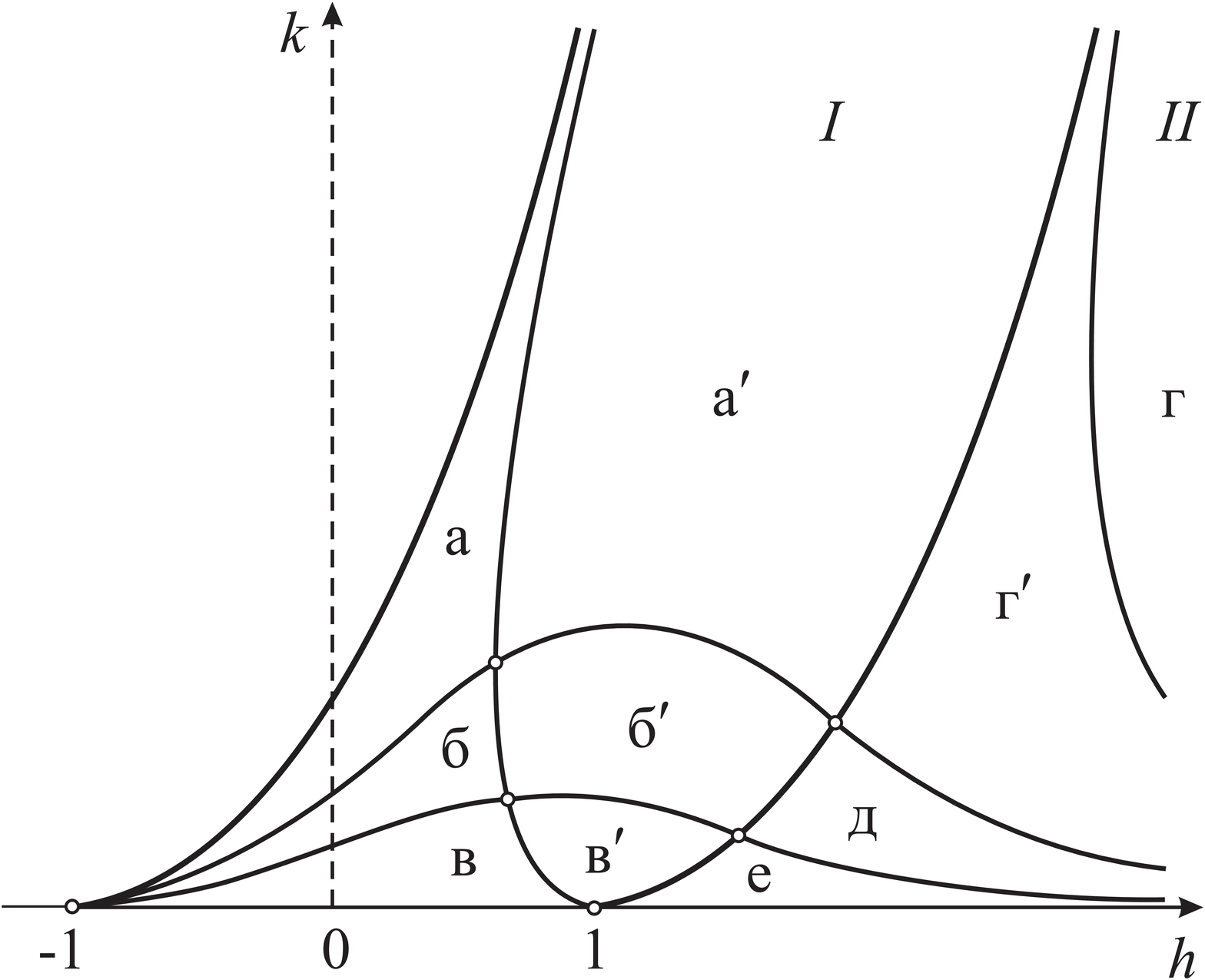}}
\caption{}\label{fig_5_10}
\end{figure}

На рис.~\ref{fig_5_11},\,\textit{а}--\textit{в} области возможности движения изображены в стереографической проекции из полюса $(-1, 0, 0)$. Сплошная линия соответствует подобластям \textit{а}\,--\,\textit{в} (рис.~\ref{fig_5_10}). Пунктиром показаны изменения, происходящие при переходе в подобласти \textit{а}$'$\,--\,\textit{в}$'$. Цифрами обозначено число допустимых скоростей~-- прообразов точки сферы при отображении на нее интегрального тора. При возникновении пунктирных <<треугольников>> число допустимых скоростей внутри них становится равно четырем.

\begin{figure}[ht]
\center{\includegraphics[width=0.6\linewidth]{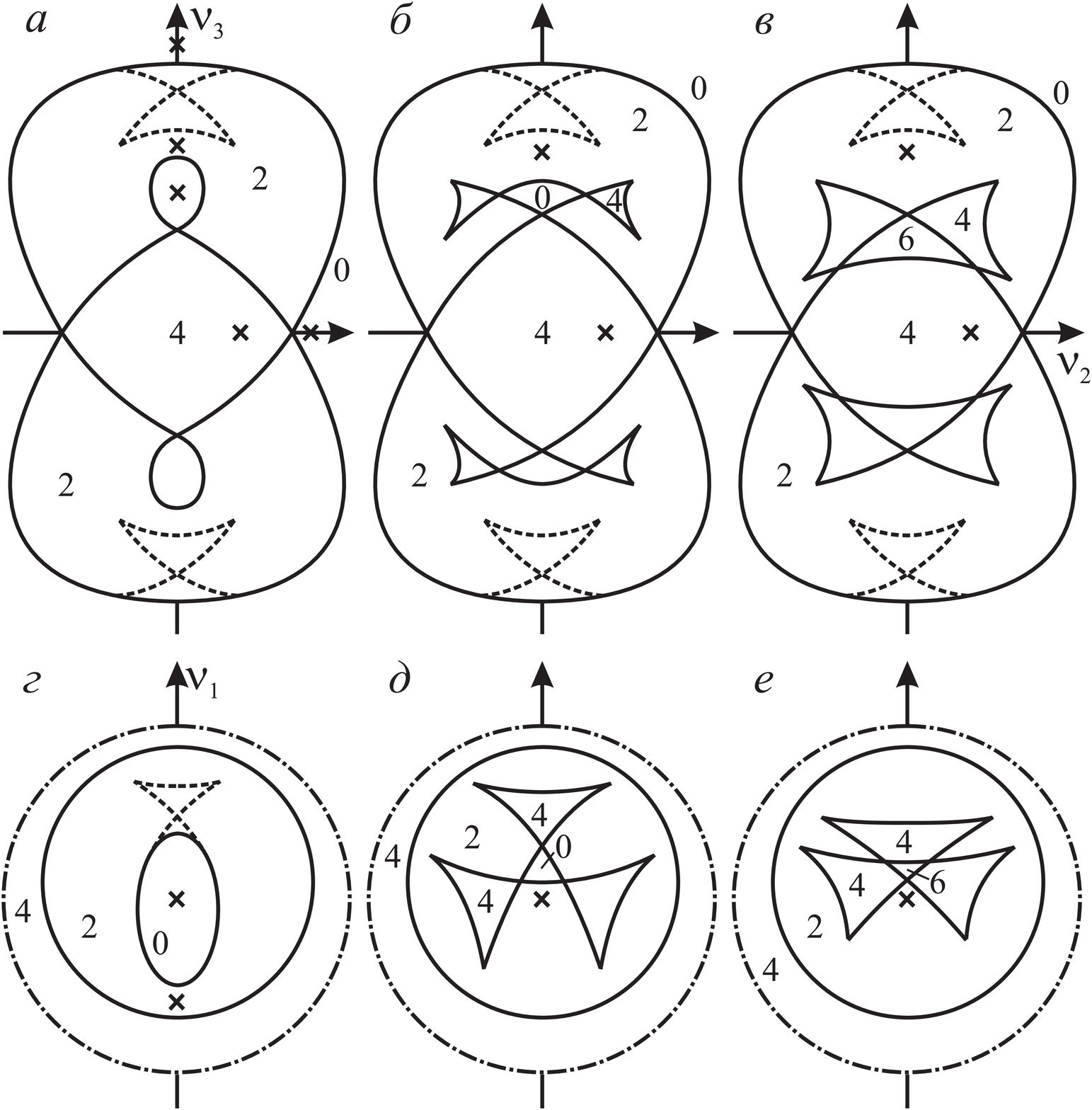}}
\caption{}\label{fig_5_11}
\end{figure}

Возможны различные положения точек $(0, 1, 0)$ и $(0, 0, 1)$ относительно ОВД. На рис.~\ref{fig_5_11} они отмечены крестиком. Какая именно возможность реализуется, нетрудно установить с учетом результатов, изложенных при выводе уравнений кривых \eqref{eeq5_16}, \eqref{eeq5_23}.

По доказанному в $\S$~\ref{ssec54} области $\ts{II}$ отвечает интегральное многообразие, состоящее из двух торов. Проекции каждого из них на сферу симметричны относительно плоскости $\nu_3 = 0$. На рис.~\ref{fig_5_11},\,\textit{г}--\textit{е} приведена картина, возникающая в полусферах $\nu_3 > 0$ и $\nu_3 < 0$ (штрихпунктирная окружность есть сечение $\nu_3 = 0$). Пунктиром показаны изменения, происходящие в подобласти \textit{г}$'$. Крестиком отмечены возможные положения точек $(0, 0, \pm1)$, полусферы для удобства изображения несколько деформированы. Цифры, как и ранее, означают число допустимых скоростей. Узкая полоска около сечения $\nu_3 = 0$ есть образ одного из торов (заметим, что при $k \to 0$ она стягивается к окружности $\nu_3 = 0$, предельные движения -- маятниковые). Эта полоса целиком лежит в образе другого тора, поэтому допустимых скоростей в ней четыре.

Поскольку фазовые траектории условно-периодические, становится полностью ясен характер движений <<в теле>> направляющего вектора вертикали.

Нетривиальные формы ОВД доказывают следующий глубокий факт: разделение переменных в случае Горячева\,--\,Чап\-лы\-ги\-на не может быть произведено точечным каноническим преобразованием -- в противном случае связная компонента области возможности движения была бы диффеоморфна прямоугольнику или кольцу, как в примере \ref{exa331}.

\subsection{Комментарий к главе 5}\label{ssec56}

Решение, анонсированное Л.Н.\,Сретенским в работе \cite{bib45} и более подробно изложенное в \cite{bib44}, представляет собой обобщение классического случая интегрируемости Горячева\,--\,Чап\-лы\-ги\-на \cite{bib15,bib64} на задачу о движении гиростата.

Анализом случая Горячева\,--\,Чап\-лы\-ги\-на занимались многие авторы. Быстрые вращения изучены методом малого параметра в работах Л.Н.\,Сре\-тен\-ско\-го~\cite{bib43} и Ю.А.\,Ар\-хан\-гель\-ско\-го~\cite{bib06}. А.И.\,Дошкевич~\cite{bib19} исследовал зависимость от времени переменных, в которых задача сводится к квадратурам. В работах \cite{bib15} -- \cite{bib17} дана геометрическая интерпретация отмеченного С.А.\,Чаплыгиным класса периодических движений тела. В наших обозначениях -- это движения, соответствующие точкам кривой $\Bif$ при отсутствии гиростатического момента. Аналогичную интерпретацию получили собственно критические периодические движения, отвечающие точкам кривой $\Bif_*$ \cite{bib12}. Движения асимптотического характера, существующие при тех же ограничениях на константы интегралов, не рассматривались.

В.В.\,Козловым \cite{bib26} проанализированы общие свойства волчка Горячева--Чап\-лы\-ги\-на в предположении независимости первых интегралов. Установлены теоремы о движениях оси динамической симметрии и линии узлов, характерные для интегрируемых задач динамики твердого тела.

Несмотря на то, что по числу свободных параметров решение Л.Н. Сретенского стоит в одном ряду с решениями Н.Е. Жуковского и С.В. Ковалевской, качественным исследованиям оно не подвергалось. Не изучались также и бифуркации первых интегралов в случае Горячева\,--\,Чап\-лы\-ги\-на. Затрагивая вопрос о количестве связных компонент интегрального многообразия в некритических случаях, В.В.\,Козлов \cite{bib26} фактически доказал следующее: если постоянные первых интегралов $h$ и $k$ таковы, что при обращении в нуль параметра Пуанкаре (т.е. величины $\gamma$, равной произведению веса тела на расстояние от центра масс до неподвижной точки) соответствующее многообразие $J_{h, k}$ остается непустым, то для достаточно малых $\gamma \neq 0$ оно состоит из двух торов. Это утверждение справедливо, конечно, и в случае Сретенского, но поскольку расстояние по оси $h$ между кривыми $\Bif$ и $\Bif_*$ в исходных <<размерных>> величинах равно $2\gamma$, то его условиям удовлетворяют лишь точки областей $I - I_*$, $\ts{II} - \ts{II}_*$ (рис.~\ref{fig_5_3}). Остальные области, в частности, область $I$ в случае Горячева\,--\,Чап\-лы\-ги\-на (рис.~\ref{fig_5_10}), при $\gamma = 0$ исчезает.

В данной главе выполнено исследование поведения вспомогательных переменных, непосредственно связанных с компонентами угловой скорости тела-носителя в сопутствующих осях, дано полное описание дифференцируемого типа интегральных многообразий и критических интегральных поверхностей, построены области возможности движения в случае Горячева\,--\,Чап\-лы\-ги\-на. В результате выявлены закономерности, специфичные для решения Чаплыгина\,--\,Сре\-тен\-ского.

\clearpage

\section{Фазовая топология решения Ковалевской}\label{sec6}

\subsection{Классы Аппельрота и критические значения\\  интегрального отображения}\label{ssec61}

Решение Ковалевской характеризуется следующими условиями: $A_1 = A_2 = 2A_3$, центр масс лежит в экваториальной плоскости эллипсоида инерции, гиростатический момент отсутствует $(\mbs{\lambda} = 0)$, потенциальное поле есть поле силы тяжести. Выберем подвижные оси так, чтобы центр масс лежал на первой из них. Полагая вектор $\mbs{\nu}_0$ направленным вертикально вниз, запишем потенциальную энергию $\Pi = -\gamma \nu_1$, где $\gamma \neq 0$ есть параметр Пуанкаре. Назначим $\sqrt{\gamma/A_3}$ единицей угловой скорости, а $\sqrt{A_3/\gamma}$ -- единицей времени. Систему \eqref{beq4_29} представим в безразмерных величинах:
\begin{equation}\label{feq1_1}
  \begin{array}{c}
    2\dot{\omega}_1 = \omega_2\omega_3, \quad 2\dot{\omega}_2 = -\omega_3\omega_1 - \nu_3, \quad \dot{\omega}_3 = \nu_2,\\[2mm]
    \dot{\nu}_1 = \omega_3\nu_2 - \omega_2\nu_3, \quad \dot{\nu}_2 = \omega_1\nu_3 - \omega_3\nu_1, \quad \dot{\nu}_3 = \omega_2\nu_1 - \omega_1\nu_2.
  \end{array}
\end{equation}
Ее первые интегралы таковы:
\begin{eqnarray}
& &  2(\omega_1^2 + \omega_2^2) + \omega_3^2 - 2\nu_1 = 2h,\label{feq1_2}\\
& &  2(\omega_1\nu_1 + \omega_2\nu_2) + \omega_3\nu_3 = 2g,\label{feq1_3}\\
& &  \nu_1^2 + \nu_2^2 + \nu_3^2 = 1,\label{feq1_4}\\
& &  (\omega_1^2 - \omega_2^2 + \nu_1)^2 + (2\omega_1\omega_2 + \nu_2)^2 = k.\label{feq1_5}
\end{eqnarray}

Введем переменные Ковалевской $s_1, s_2$:
\begin{equation}\label{feq1_6}
  s_{1,2} = h + \frac{R(x_1, x_2) \mp \sqrt{R(x_1)R(x_2)}}{(x_1 - x_2)^2}.
\end{equation}
Здесь
\begin{eqnarray}
& &  x_1 = \omega_1 + i\omega_2, \quad x_2 = \omega_1 - i\omega_2,\label{feq1_7}\\
& &  R(x_1, x_2) = -x_1^2x_2^2 + 2hx_1x_2 + 2g(x_1 + x_2) +1 - k,\label{feq1_8}\\
& &  R(x) = -x^4 + 2hx^2 + 4gx + 1 -k.\label{feq1_9}
\end{eqnarray}
Зависимость переменных \eqref{feq1_6} от времени определяется уравнениями
$$(s_1 - s_2)^2\dot{s}_1^2 = -2\Phi(s_1), \qquad (s_1 - s_2)^2\dot{s}_2^2 = -2\Phi(s_2),
$$
в которых
\begin{eqnarray}
& &   \Phi(s) = (s - h + \sqrt{k})(s - h - \sqrt{k})\varphi(s),\label{feq1_10}\\
& &  \varphi(s) = s^3 - 2hs^2+(h^2 + 1 - k)s - 2g^2,\label{feq1_11}
\end{eqnarray}
причем \eqref{feq1_11} есть резольвента Эйлера многочлена \eqref{feq1_9}.

Движения, соответствующие таким постоянным первых интегралов, при которых многочлен \eqref{feq1_10} имеет кратный корень, согласно классификации, данной Г.Г.\,Ап\-пель\-ротом \cite{bib03}, называют простейшими. Нетрудно видеть, что поверхность кратных корней в $\bbRR(h,k,g)$ состоит из плоскости
\begin{equation}\label{feq1_12}
  k = 0
\end{equation}
(1-й класс простейших движений по Аппельроту, или класс Делоне), поверхности
\begin{equation}\label{feq1_13}
  k = (h-2g^2)^2
\end{equation}
(2-й и 3-й классы простейших движений) и поверхности кратных корней резольвенты \eqref{feq1_11}, имеющей уравнение
\begin{equation}\label{feq1_14}
  (1-k)(h^2+1-k)^2 - 2\bigl[9h(1-k)+h^3\bigr]g^2 + 27g^4 = 0
\end{equation}
(4-й класс простейших движений). Поверхность \eqref{feq1_14} можно представить и в параметрической форме:
\begin{equation}\label{feq1_15}
  \ds h = s + \frac{g^2}{s^2}, \quad k = 1 - \frac{2g^2}{s} + \frac{g^4}{s^4}
\end{equation}
или
\begin{equation}\label{feq1_16}
  \ds h = x^2 - \frac{g}{x}, \quad k = 1 + 2gx + x^4.
\end{equation}
В записи \eqref{feq1_15} параметр $s$ играет роль кратного корня $\varphi(s)$, а в записи \eqref{feq1_16} $x$ есть кратный корень исходного многочлена \eqref{feq1_9}.

Изучим особенности системы \eqref{feq1_2} -- \eqref{feq1_5}. Критические точки интегрального отображения \eqref{ceq4_10} в данном случае удобно находить из условия
\begin{equation}\label{feq1_17}
  \rang \left(
  \begin{array}{cccccc}
    2\omega_1 & 2\omega_2 & -1 & 0 & \omega_3 & 0 \\
    2\nu_1 & 2\nu_2 & 2\omega_1 & 2\omega_2 & \nu_3 & \omega_3 \\
    0 & 0 & \nu_1 & \nu_2 & 0 & \nu_3 \\
    2(\omega_1 \eta_1 + \omega_2 \eta_2) & 2(\omega_1 \eta_2 - \omega_2 \eta_1) & \eta_1 & \eta_2 & 0 & 0
  \end{array}
  \right) <4.
\end{equation}
Здесь обозначено для краткости
\begin{equation}\label{feq1_18}
  \eta_1 = \omega_1^2 - \omega_2^2 + \nu_1, \quad \eta_2 = 2\omega_1 \omega_2 + \nu_2.
\end{equation}

Рассмотрим некоторые случаи вырождения переменных \eqref{feq1_6}. Пусть в соотношениях \eqref{feq1_2} -- \eqref{feq1_5}
\begin{equation}\label{feq1_19}
  \omega_2 = 0.
\end{equation}
Внесем значения $h, g, k$ из \eqref{feq1_2}, \eqref{feq1_3}, \eqref{feq1_5} в многочлен \eqref{feq1_9} и его производную, полагая попутно $x = \omega_1$. В результате преобразований получим
\begin{equation}\label{feq1_20}
  R(\omega_1) = (\omega_1 \omega_3 + \nu_3)^2, \quad R'(\omega_1) = 2\omega_3(\omega_1 \omega_3 + \nu_3).
\end{equation}
Заменим в матрице \eqref{feq1_17} последнюю строку на комбинацию строк с коэффициентами $\omega_1^2$, $\omega_1$, $1$, $-1$ соответственно. После элементарных преобразований, сводящихся к сокращению ненулевых множителей и перестановке строк либо столбцов, придем к условию
\begin{equation}\label{feq1_21}
  \rang \left(
  \begin{array}{cccccc}
    \omega_1 \omega_3 + \nu_3 & 0 & 0 & 0 & 0 & \omega_1(\omega_1 \omega_3 + \nu_3) \\
    0 & -1 & 0 & 0 & \omega_1 & \omega_3 \\
    \omega_3 & 2\omega_1 & \nu_2 & 0 & \nu_1 & \nu_3 \\
    \nu_3 & \nu_1 & 0 & \nu_2 & 0 & 0
  \end{array}
  \right) <4.
\end{equation}
Равенство нулю определителя, составленного из первых столбцов, дает
$$
\nu_2(\omega_1 \omega_3 + \nu_3) = 0.
$$
Если $\omega_1 \omega_3 + \nu_3 = 0$, то в силу \eqref{feq1_20} многочлен \eqref{feq1_9} имеет кратный корень, и выполнено \eqref{feq1_16} с $x = \omega_1$. Пусть
\begin{equation}\label{feq1_22}
  \nu_2 = 0, \quad \omega_1 \omega_3 + \nu_3 \neq 0.
\end{equation}
Условие \eqref{feq1_21} примет вид $(\omega_1^2 + \nu_1)(2\omega_1\nu_3 - \omega_3\nu_1) = 0$. Случай $\omega_1^2 + \nu_1 = 0$ вместе с \eqref{feq1_19}, \eqref{feq1_22} приводит к равенству \eqref{feq1_12}. Если же
\begin{equation}\label{feq1_23}
  2\omega_1\nu_3 - \omega_3\nu_1 = 0,
\end{equation}
то, используя формулы \eqref{feq1_4}, \eqref{feq1_19}, \eqref{feq1_22}, находим
$$
h - 2g^2 = -(\omega_1^2 + \nu_1), \qquad k = (\omega_1^2 + \nu_1)^2,
$$
откуда вытекает \eqref{feq1_13}. Итак, все критические значения, достигаемые в случае \eqref{feq1_19}, отвечают классам Аппельрота.

Исследуем еще одну возможность:
\begin{equation}\label{feq1_24}
  \omega_3 = 0, \quad \nu_3 = 0, \quad \omega_2 \neq 0.
\end{equation}
В матрице \eqref{feq1_17} последние два столбца нулевые. Приравнивая нулю оставшийся определитель четвертого порядка, получаем
$$
2(\omega_1^2  - \omega_2^2 + \nu_1)\nu_1\nu_2 - (2\omega_1\omega_2 + \nu_2)(\nu_1^2 - \nu_2^2) = 0,
$$
что позволяет ввести неопределенный множитель $\varkappa$:
\begin{equation}\label{feq1_25}
  \omega_1^2  - \omega_2^2 + \nu_1 = \varkappa(\nu_1^2 - \nu_2^2), \quad 2\omega_1\omega_2 + \nu_2 = 2\varkappa\nu_1\nu_2.
\end{equation}
Подстановка этих выражений в \eqref{feq1_5} с учетом \eqref{feq1_4} приводит к равенству
\begin{equation}\label{feq1_26}
  \varkappa = \pm \sqrt{k}.
\end{equation}
Из \eqref{feq1_4} и \eqref{feq1_25} находим
\begin{equation}\label{feq1_27}
  (\omega_1^2 + \omega_2^2)^2 = (\varkappa \nu_1 - 1)^2 + \varkappa^2\nu_2^2.
\end{equation}
Систему \eqref{feq1_25} запишем в виде
\begin{equation*}
  \begin{array}{c}
    (\varkappa \nu_1 - 1)(\omega_1 \nu_1 + \omega_2 \nu_2) + \varkappa \nu_2(\omega_2 \nu_1 - \omega_1 \nu_2) = \omega_1(\omega_1^2 + \omega_2^2),\\
    \varkappa \nu_2(\omega_1 \nu_1 + \omega_2 \nu_2) - (\varkappa \nu_1 - 1)(\omega_2 \nu_1 - \omega_1 \nu_2) = \omega_2(\omega_1^2 + \omega_2^2),
  \end{array}
\end{equation*}
откуда, учитывая \eqref{feq1_27}, имеем
\begin{equation*}
  \ds \omega_1 \nu_1 + \omega_2 \nu_2 = \frac{\omega_1}{\varkappa - \omega_1^2 - \omega_2^2}, \quad \omega_2 \nu_1 - \omega_1 \nu_2 = \frac{\omega_2}{\varkappa + \omega_1^2 + \omega_2^2}.
\end{equation*}
Следовательно,
\begin{equation*}
  \ds \nu_1 = \frac{\omega_1^2 - \omega_2^2 + \varkappa}{\varkappa^2 - (\omega_1^2 + \omega_2^2)^2}, \quad \nu_2 = \frac{2\omega_1 \omega_2}{\varkappa^2 - (\omega_1^2 + \omega_2^2)^2}.
\end{equation*}
Введя найденные выражения вместе с \eqref{feq1_24} в соотношения \eqref{feq1_2} -- \eqref{feq1_4}, получим
\begin{equation}\label{feq1_28}
  \begin{array}{c}
    \ds h = \omega_1^2 + \omega_2^2 - \frac{\omega_1^2 - \omega_2^2 + \varkappa}{\varkappa^2 - (\omega_1^2 + \omega_2^2)^2}, \quad k = \frac{\omega_1}{\varkappa^2 - \omega_1^2 - \omega_2^2}, \\[3mm]
    (\omega_1^2 + \omega_2^2)^4 - (1 + 2\varkappa^2)(\omega_1^2 + \omega_2^2)^2 -2 \varkappa (\omega_1^2 + \omega_2^2)^2 - \varkappa^2(1 - \varkappa^2) = 0,
  \end{array}
\end{equation}
откуда $h - 2g^2 + \varkappa = 0$. Последнее в силу \eqref{feq1_26} приводит к соотношению \eqref{feq1_13}. Таким образом, и при условии \eqref{feq1_24} критические значения отвечают классам Аппельрота.

Обратимся к общему случаю и покажем, что в предположении
\begin{equation}\label{feq1_29}
  \omega_3^2 + \nu_3^2 \neq 0, \quad \omega_2 \neq 0
\end{equation}
новых критических значений система \eqref{feq1_2} -- \eqref{feq1_5} не имеет.

Вначале отметим вытекающее из \eqref{feq1_2} -- \eqref{feq1_5}, \eqref{feq1_7}, \eqref{feq1_8} тождество
$$
R(x_1, x_2) = (\omega_1\omega_3 + \nu_3)^2 + \omega_2^2 \omega_3^2,
$$
так что при условии \eqref{feq1_29}
\begin{equation}\label{feq1_30}
  R(x_1, x_2) \neq 0, \quad x_1 \neq x_2.
\end{equation}
Введем в дополнение к \eqref{feq1_7} переменные $\xi_1 = \eta_1 + i\eta_2$, $\xi_2 = \eta_1 - i\eta_2$. Согласно \eqref{feq1_18} они связаны невырожденным преобразованием с $\nu_1, \nu_2$. Исключая из \eqref{feq1_2} -- \eqref{feq1_6} величины $\omega_3$, $\nu_3$, приходим к двум соотношениям \cite{bib24}
\begin{equation}\label{feq1_31}
  \begin{array}{c}
    \xi_1 \xi_2 = k, \\
    R(x_2)\xi_1 + R(x_1)\xi_2 + R_1(x_1, x_2) + (x_1 - x_2)^2k = 0,
  \end{array}
\end{equation}
где
\begin{equation}\label{feq1_32}
  R_1(x_1, x_2) = -2hx_1^2x_2^2 - 4g(x_1 + x_2)x_1x_2 - (1-k)(x_1 + x_2)^2 + 2(1-k)h - 4g^2.
\end{equation}
Первого неравенства \eqref{feq1_29} достаточно для того, чтобы условие \eqref{feq1_17} было равносильно системе, определяющей зависимость уравнений \eqref{feq1_31}:
\begin{eqnarray}
&  R(x_1)\xi_2 = R(x_2)\xi_1,\label{feq1_33}\\
&  \begin{array}{c}
    \ds R'(x_1)\xi_2 + \frac{\partial R_1(x_1, x_2)}{\partial x_1} + 2(x_1 - x_2)k = 0, \\
    \ds R'(x_2)\xi_1 + \frac{\partial R_1(x_1, x_2)}{\partial x_2} - 2(x_1 - x_2)k = 0.
  \end{array}\label{feq1_34}
\end{eqnarray}
Следуя \cite{bib24}, выразим $\xi_1, \xi_2$ из соотношений \eqref{feq1_31}:
\begin{equation}\label{feq1_35}
  \begin{array}{c}
    2R(x_2)\xi_1 = -[R_1(x_1, x_2) + (x_1 - x_2)^2 k] + W(x_1, x_2), \\
    2R(x_1)\xi_2 = -[R_1(x_1, x_2) + (x_1 - x_2)^2 k] - W(x_1, x_2),
  \end{array}
\end{equation}
где $W(x_1, x_2) = \bigl\{\left[R_1(x_1, x_2) - (x_1 - x_2)^2 k \right]^2 - 4kR^2(x_1, x_2) \bigr\}^{1/2}$.
Таким образом, условие \eqref{feq1_33} приводит к равенству $W(x_1, x_2) = 0$ или, принимая обозначение \eqref{feq1_26}, к соотношению
\begin{equation}\label{feq1_36}
  R_1(x_1, x_2) - (x_1 - x_2)^2\varkappa^2 = 2\varkappa R(x_1, x_2).
\end{equation}
Кроме того, подставляя $\xi_1, \xi_2$ из \eqref{feq1_35} в \eqref{feq1_34}, имеем
\begin{equation*}
  \begin{array}{c}
     \ds R(x_1) \bigl[\frac{1}{2} \frac{\partial R_1(x_1, x_2)}{\partial x_1} + (x_1 - x_2)^2 \varkappa^2 \bigr] - \frac{1}{4} R'(x_1)\left[R_1(x_1, x_2) - (x_1 - x_2)^2 \varkappa^2 \right] = 0,\\
     \ds R(x_2) \bigl[\frac{1}{2} \frac{\partial R_1(x_1, x_2)}{\partial x_2} + (x_1 - x_2)^2 \varkappa^2 \bigr] - \frac{1}{4} R'(x_2)\left[R_1(x_1, x_2) - (x_1 - x_2)^2 \varkappa^2 \right] = 0.
   \end{array}
\end{equation*}
Преобразуем эту систему с помощью матрицы
$$
\left(\begin{array}{cc}
1 & 1 \\
x_2 & -x_1
\end{array}\right).
$$
После сокращения на $x_1 - x_2 \neq 0$ получим
\begin{equation}
\begin{array}{l}
     2(1 - \varkappa^2)^2 + 2(1 - \varkappa^2)[h^2 - 1 + x_1^2x_2^2 + 3g(x_1 + x_2)] +  \\[3mm]
     \quad + 2h^2 x_1^2x_2^2 + h[-(x_1 + x_2)^2 + 4g(x_1 + x_2)x_1x_2 - 4g^2] + \\[3mm]
     \quad + (x_1^2 + x_2^2)x_1x_2 + 2g(x_1^2x_2^2 - 2)(x_1 + x_2)+\\[3mm]
     \quad+ 4g^2(x_1^2 + 3x_1x_2 + x_2^2) = 0,
   \end{array}\label{feq1_37}
\end{equation}
\begin{equation}
[R_1(x_1, x_2) - (x_1 - x_2)^2 \varkappa^2]g + (x_1 + x_2 + 2gx_1x_2)R(x_1,x_2) = 0.\label{feq1_38}
\end{equation}
Из \eqref{feq1_36}, \eqref{feq1_38} получим
$$
[x_1 + x_2 + 2g(x_1x_2 + \varkappa)]R(x_1, x_2) = 0,
$$
так что по предположению \eqref{feq1_30}
\begin{equation}\label{feq1_39}
  x_1 + x_2 = -2g(x_1x_2 + \varkappa).
\end{equation}
С помощью этого выражения исключим $x_1+x_2$ из \eqref{feq1_37}, \eqref{feq1_38} и в результате преобразований запишем
\begin{equation*}
  \begin{array}{l}
     (\varkappa - h + 2g^2)[4g^2(x_1 x_2 + \varkappa)^2 - (h + \varkappa)(x_1^2x_2^2 + 1 - \varkappa^2)] = 0, \\[3mm]
     (\varkappa - h + 2g^2)[(x_1x_2+\varkappa)^2 - 1] = 0.
   \end{array}
\end{equation*}
Отсюда либо $\varkappa - h + 2g^2 = 0$, что приводит к \eqref{feq1_13}, либо ${x_1x_2 = \pm 1 - \varkappa}$, ${x_1 + x_2 = \mp g}$, ${2g^2 = (h + \varkappa)(1 \mp \varkappa)}$. В последнем случае непосредственная проверка дает $R(x_1, x_2) = 0$, а это противоречит \eqref{feq1_30}. Тем самым доказано, что бифуркационное множество $\Sigma$ представляет собой часть поверхности кратных корней многочлена \eqref{feq1_10}, отвечающую действительным решениям системы \eqref{feq1_2} -- \eqref{feq1_5}.

В заключение отметим, что множество критических точек состоит из тех траекторий уравнений \eqref{feq1_1}, которые в терминологии Аппельрота соответствуют особо замечательным движениям, т.е. таким, в которых одна из величин \eqref{feq1_6} остается постоянной и равной кратному корню многочлена \eqref{feq1_10}.

\subsection{Бифуркационное множество и интегральные многообразия}\label{ssec62}

Вид уравнений, определяющих листы бифуркационной поверхности, показывает, что удобно рассмотреть сечения $\Sigma_g$ множества $\Sigma$ плоскостями ${g = \cons}$. Множества $\Sigma_g$ имеют и самостоятельное значение, как показано в $\S$~\ref{ssec34}. Поскольку $\Sigma_g = \Sigma_{-g}$, ограничимся случаем $g \gs 0$. Считаем, что $\Sigma_g\subset \bbR^2(k, h)$.

Уравнение \eqref{feq1_13} определяет на плоскости $\bbR^2(k, h)$ параболу с вершиной $(0, 2g^2)$. Здесь комментариев не требуется.

Обратимся к кривой \eqref{feq1_14}. При $g = 0$ она распадается на прямую $k = 1$ и параболу $k = 1 + h^2$. Если $g > 0$, воспользуемся записью \eqref{feq1_16}. Исследуемая кривая имеет точку возврата $x = -(g/2)^{1/3}$, вертикальную асимптоту $k = 1$ ($x \to \pm0$, $h \to \mp \infty$). При $x \to \pm \infty$ обе координаты $k, h$ стремятся к $+\infty$ так, что точка $(k(x), h(x))$ асимптотически приближается к соответствующей кривой $k = h^2 \pm 4g\sqrt{h} + 1$.

В общем случае кривая \eqref{feq1_16} имеет с параболой \eqref{feq1_13} точки пересечения
\begin{eqnarray}
& &  (k, h) = ((g^2 + 1)^2, g^2 - 1), \quad x = g,\label{feq2_1}\\
& &  (k, h) = ((g^2 - 1)^2, g^2 + 1), \quad x = -g\label{feq2_2}
\end{eqnarray}
и точку касания
\begin{equation}\label{feq2_3}
  \ds (k, h) = \left(\frac{1}{16g^4}, \frac{1}{4g^2} + 2g^2 \right), \quad x = -\frac{1}{2g}.
\end{equation}
Исключение составляет значение $g^2 = 1/2$, когда точки \eqref{feq2_3}, \eqref{feq2_2} совпадают между собой и с точкой возврата. При $g^2 = 4/(3\sqrt{3})$ точка возврата попадает на ось $k = 0$ (являющуюся, очевидно, частью границы области существования движений). Еще одна особенность возникает при $g = 1$: точка \eqref{feq2_2} переходит с одной ветви параболы \eqref{feq1_13} на другую.

В промежутке $4/(3\sqrt{3}) < g^2 < 1$ кривая \eqref{feq1_16} имеет две общие точки с прямой ${k = 0}$ при ${h > 2g^2}$. Их можно записать параметрически:
\begin{equation}\label{feq2_4}
  \ds h = \frac{1 + 3s^2}{2s}, \quad g^2 = \frac{(1+s^2)^2}{4s}, \quad 0 < s < 1.
\end{equation}
В случае $g^2 = 1$ одна из точек \eqref{feq2_4}, отвечающая большему значению $s$, равному в этот момент $1$, сливается с точкой \eqref{feq2_2} и при дальнейшем увеличении $g^2$ переходит на часть прямой $k = 0$ с $h < 2g^2$, где действительные движения отсутствуют.

Отметим, что кривая \eqref{feq2_4} и образ на плоскости $\{(h, g)\}$ множества точек \eqref{feq2_1}, \eqref{feq2_2}, состоящий из двух парабол $h = g^2 \pm 1$, образуют бифуркационное множество интегралов энергии и площадей, изучавшиеся в \cite{bib67}. А.\,Якоб ошибочно включил сюда и отрезки ${g^2 = 4/(3\sqrt{3})}$, ${\sqrt{3} < h < (9 + 4\sqrt{3})/9}$, не приводя никаких доводов в пользу этого включения. Нетрудно проверить непосредственно, что эти отрезки не являются критическими значениями соответствующего эффективного потенциала \cite{bib51}.

Проследим трансформацию кривой \eqref{feq1_16} при $g \to 0$. Ветвь $-\infty < x < 0$ переходит в луч $\{k = 1, h \gs 0\}$ и верхнюю часть параболы $k = 1 + h^2$ ($h \gs 0$). Точка возврата, следовательно, попадает в точку $(1, 0)$. Ветвь $0 < x < +\infty$ <<склеивается>> с лучом $\{k = 1, h \ls 0\}$ и той же половиной параболы $k = 1 + h^2$. Точка касания кривой \eqref{feq1_16} и параболы \eqref{feq1_13} при $g \to 0$ уходит в бесконечность.

\begin{figure}[h]
\centering
\includegraphics[width=0.9\linewidth]{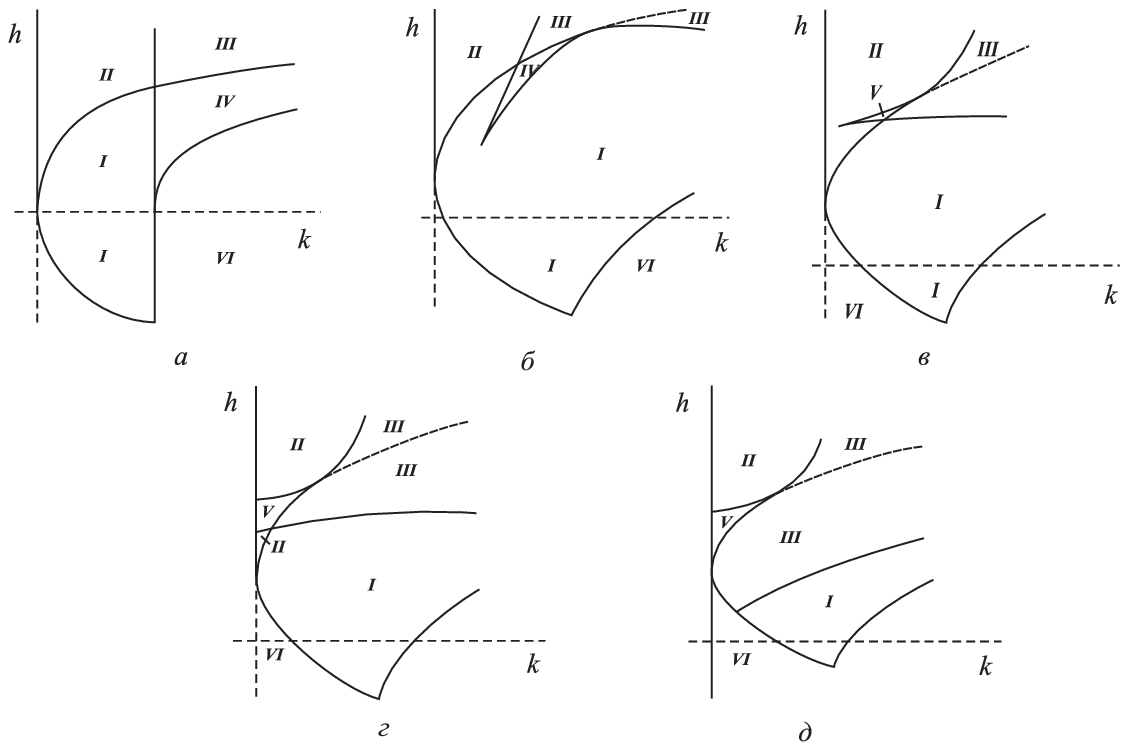}
\caption{}\label{fig_6_1}
\end{figure}

Качественно различные виды множеств $\Sigma_g$ приведены на рис.~\ref{fig_6_1},\,{\it а-д} для следующих значений постоянной площадей $g$, выбранной в качестве параметра: а)~$g = 0$; б)~$0 < g < 1/2$; в)~$1/2 < g^2 < 4/(3\sqrt{3})$; г)~$4/(3\sqrt{3}) < g^2 < 1$; д)~$g^2 > 1$.

Как уже отмечалось, в силу компактности интегральных многообразий их вид может изменяться лишь при переходе точки $(h, k, g)$ через критическое значение, т.е. через множество \eqref{feq1_12} -- \eqref{feq1_14}. Следовательно, $J_{h,k,g} = \varnothing$ в областях, объединенных под номером $\ts{VI}$ (поскольку они содержат точки с $k<0$ или $h < -1$).

Покажем, что в $\Sigma_g$ не включается часть верхней ветви
\begin{equation}\label{feq2_5}
  h = 2g^2+\sqrt{k}
\end{equation}
параболы \eqref{feq1_13}, лежащая правее точки \eqref{feq2_3} касания с кривой \eqref{feq1_16}. Для этого достаточно доказать, что при условии \eqref{feq2_5} действительным критическим точкам отвечают значения
\begin{equation}\label{feq2_6}
  \ds \sqrt{k} \ls \frac{1}{4g^2}
\end{equation}
(этот факт соответствует утверждению \cite{bib03,bib21} об отсутствии в 3-м классе при $\sqrt{k} > 1/(4g^2)$ особо замечательных движений).

Значения \eqref{feq2_5} достигаются, во-первых, если выполнены условия \eqref{feq1_19}, \eqref{feq1_22}, \eqref{feq1_23}
\begin{equation}\label{feq2_7}
  2\omega_1\nu_3 - \omega_3\nu_1 = 0, \quad \omega_2 = 0, \quad \nu_2 = 0,
\end{equation}
во-вторых, при условиях \eqref{feq1_28} и $\varkappa = -\sqrt{k}$, и в-третьих, в случае \eqref{feq1_39} с $\varkappa = -\sqrt{k}$. На основании двух последних возможностей приходим к уравнению
\begin{equation}\label{feq2_8}
  g(\omega_1^2 + \omega_2^2) + \omega_1 + g\sqrt{k} = 0.
\end{equation}
Оно имеет действительные решения относительно $\omega_1, \omega_2$ тогда и только тогда, когда имеет действительные корни трехчлен $g\omega_1^2  + \omega_1 + g\sqrt{k}$, что и приводит к неравенству \eqref{feq2_6}. Обратимся к случаю \eqref{feq2_7}. Из \eqref{feq1_3} и первого уравнения \eqref{feq2_7} находим
\begin{equation*}
  \ds \nu_1 = \frac{4g\omega_1}{4\omega_1^2 + \omega_3^2}, \quad \nu_2 = \frac{2g\omega_3}{4\omega_1^2 + \omega_3^2}.
\end{equation*}
Подставляя эти выражения в \eqref{feq1_4}, получаем $4\omega_1^2 + \omega_3^2 = 4g^2$, после чего соотношение \eqref{feq1_5} принимает вид \eqref{feq2_8}, т.е. и здесь необходимо выполнение неравенства \eqref{feq2_6}.

\begin{figure}[ht!]
\centering
\includegraphics[width=0.5\linewidth]{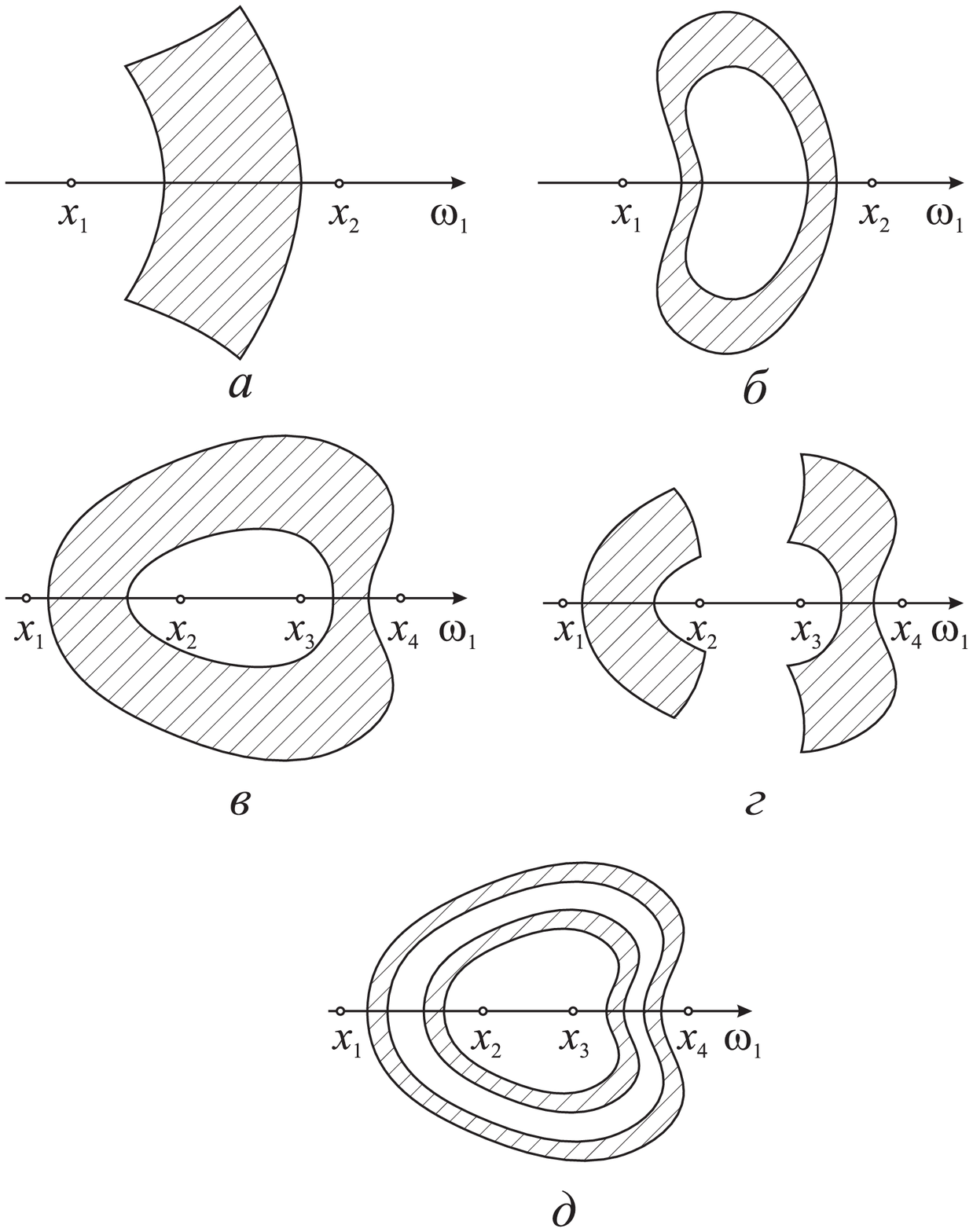}
\caption{}\label{fig_6_2}
\end{figure}

На рис.~\ref{fig_6_1},\,\textit{а}--\textit{д} одинаковыми цифрами обозначены области, объединение которых в пространстве $\bbI$ дает связную компоненту $\bbRR\backslash \Sigma$. Таким образом, существует пять компонент, в которых интегральные многообразия непусты. Чтобы установить число торов, входящих в $J_{h, k, g}$, рассмотрим проекцию последнего на плоскость $\omega_1\omega_2$ (что в \cite{bib03} названо областью действительных движений). Нетрудно установить связь между подобластями $I-V$ в $\bbRR\backslash \Sigma$ и случаями, рассмотренными в \cite{bib03}. В результате для подобластей $I-V$ получаем проекции, указанные соответственно на рис. \ref{fig_6_2},\,\textit{а}--\textit{д} (ось $\omega_2$ вертикальна, $x_i$ -- вещественные корни многочлена $R(x)$).

Установим число прообразов внутренних точек заштрихованных множеств. Поскольку все они имеют выход на ось $\omega_1$, достаточно сделать это при $\omega_2 = 0$. По формулам \eqref{feq1_2}, \eqref{feq1_3} имеем
\begin{equation}\label{feq2_9}
  \ds \nu_1 = \omega_1^2 + \frac{1}{2}\omega_3^2 - h, \quad \omega_3\nu_3 = -\omega_1\omega_3^2 + 2(g + h\omega_1 - \omega_1^3).
\end{equation}
В результате подстановки значений \eqref{feq2_9} в \eqref{feq1_4}, \eqref{feq1_5} с учетом обозначений \eqref{feq1_8}, \eqref{feq1_32} получаем
\begin{equation}\label{feq2_10}
  4R(\omega_1)\omega_3^2 = [R'(\omega_1)]^2, \quad 4R^2(\omega_1)\nu_2^2 = 4kR^2(\omega_1) - R_1^2(\omega_1, \omega_1).
\end{equation}
Пересечение области действительных движений с осью $\omega_1$ состоит из отрезков в множестве $\{\omega_1: R(\omega_1) > 0\}$, концы которого являются нулями функции
$$\ds F(\omega_1) = k - \frac{R_1^2(\omega_1, \omega_1)}{4R^2(\omega_1)}.
$$
Пусть $\omega_1$ -- внутренняя точка такого отрезка, т.е. $F(\omega_1) > 0$. Если при этом $R'(\omega_1) \neq 0$, то система \eqref{feq2_10} имеет четыре решения относительно $\omega_3, \nu_2$, а $\nu_1, \nu_3$ однозначно определяются из \eqref{feq2_9}. Если же $R'(\omega_1) = 0$, то непосредственно из \eqref{feq1_2} -- \eqref{feq1_5} находим $\omega_3 = 0$, $\nu_1 = g/\omega_1$, $\nu_2^2 = F(\omega_1)$, $\nu_3^2 = R(\omega_1)$ и решений снова четыре.

Итак, в каждую внутреннюю точку области действительных движений проектируется четыре точки интегрального многообразия. Следовательно, в область, диффеоморфную кольцу, отображается два тора, а в область, диффеоморфную прямоугольнику, -- один. Окончательно, в компоненте $I$ множества $\bbRR\backslash \Sigma$ интегральное многообразие состоит из одного двумерного тора, в компонентах $\ts{II}-\ts{IV}$ из двух, а в компоненте $V$ из четырех двумерных торов.

\subsection{Области возможности движения и\\   критические интегральные поверхности}\label{ssec63}
Обращаясь к вопросам геометрического анализа, ограничимся случаем
\begin{equation}\label{feq3_1}
  g=0.
\end{equation}
При этом средняя прецессия тела вокруг вертикали равна нулю \cite{bib26}. Поэтому, зная характер траектории вектора $\mbs{\nu}$ на сфере Пуассона, получаем представление о  движении тела с точностью до ограниченных колебаний около вертикальной оси. Кроме того, в предположении \eqref{feq3_1} выкладки значительно упрощаются, и выделяется геометрический аспект предлагаемого метода. Общий случай может послужить предметом отдельной фундаментальной работы (достаточно сослаться на монографию \cite{bib21}, посвященную в основном изучению одних лишь особо замечательных движений).

При условии \eqref{feq3_1} перепишем уравнения интегрального многообразия $J_{h, k}$ (третий, нулевой, индекс опущен) в пространстве $S^2 \times \bbRR$ следующим образом:
\begin{eqnarray}
& &  2(\omega_1^2 + \omega_2^2) + \omega_3^2 = 2(h + \nu_1),\label{feq3_2}\\
& &  2(\omega_1\nu_1 + \omega_2\nu_2) + \omega_3\nu_3 = 0,\label{feq3_3}\\
& &  (\omega_1^2 + \omega_2^2)^2 + 2[\nu_1(\omega_1^2 - \omega_2^2) + 2\nu_2\omega_1\omega_2] + \nu_1^2 + \nu_2^2 = k.\label{feq3_4}
\end{eqnarray}

Проекцию $J_{h, k}$ на сферу Пуассона обозначим, как обычно, через $U_{h, k}$. Ее обобщенная граница определяется уравнением \eqref{ceq4_12}:
\begin{equation}\label{feq3_5}
  \omega_3[(\omega_1^2 + \omega_2^2)(\nu_1\omega_2 - \nu_2\omega_3) - (\nu_1^2 + \nu_2^2)\omega_2] - \nu_3[\nu_2(\omega_1^2 - \omega_2^2) - 2\nu_1\omega_1\omega_2]=0.
\end{equation}

Выясним, когда область возможности движения содержит в себе пересечения сферы Пуассона с главными осями инерции тела (т.е. когда главные оси могут занимать вертикальное положение либо проходить сколь угодно близко от него).

Рассмотрим точку
\begin{equation}\label{feq3_6}
  \mbs{\nu} = (1, 0, 0).
\end{equation}
Из уравнения \eqref{feq3_3} следует, что $\omega_1 = 0$, а из \eqref{feq3_2}, \eqref{feq3_4} получим ${\omega_2^2 = 1 \pm \sqrt{k}}$, ${\omega_3^2 = 2(h \mp \sqrt{k})}$.
Обращаясь к рис.~\ref{fig_6_1},\,{\it а}, видим, что центр масс тела может занимать нижнее положение, если $(h, k)$ принадлежит области $\ts{II}$ (при этом в точке \eqref{feq3_6} имеется 8 допустимых скоростей) или областям $I, \ts{III}$ (4 допустимых скорости).

Аналогично для точки
\begin{equation}\label{feq3_7}
  \mbs{\nu} = (-1, 0, 0)
\end{equation}
имеем $\omega_1 = 0$, $\omega_2^2 = -1 \pm \sqrt{k}$, $\omega_3^2 = 2(h \mp \sqrt{k})$.
Поэтому центр масс тела может занимать верхнее положение (а для почти всех начальных условий проходить сколь угодно близко от него), если $(h, k)$ принадлежит области $\ts{III}$ (при этом в точке \eqref{feq3_7} существует 4 допустимых скорости).

При
\begin{equation}\label{feq3_8}
  \mbs{\nu} = (0, \pm 1, 0)
\end{equation}
найдем $\omega_1^2 = \sqrt{k-1}$, $\omega_2 = 0$, $\omega_3^2 = 2(h - \sqrt{k - 1})$, поэтому точки \eqref{feq3_8} принадлежат ОВД, если $(h, k)$ лежит в областях $\ts{III} - \ts{IV}$ (4 допустимых скорости).

Полагая
\begin{equation}\label{feq3_9}
  \mbs{\nu} = (0, 0, \pm 1)
\end{equation}
из уравнений \eqref{feq3_2} -- \eqref{feq3_4} получим $\omega_3 = 0$, $\omega_1^2 + \omega_2^2 = h$, $(\omega_1^2 + \omega_2^2)^2 = k$, так что полюсы \eqref{feq3_9} принадлежат ОВД лишь в критическом случае
\begin{equation}\label{feq3_10}
  k = h^2, \quad h \gs 0.
\end{equation}
При этом, если $h > 0$, то в точках \eqref{feq3_9} имеется целая окружность допустимых скоростей. Для всех значений $(h, k)$, кроме \eqref{feq3_10}, ось динамической симметрии в процессе движения тела отделена от вертикали, поскольку ОВД компактна (см. также \cite{bib03}).

Заметим теперь, что, как и в случае Горячева\,--\,Чап\-лы\-ги\-на, область возможности движения симметрична относительно сечений сферы
\begin{equation}\label{feq3_11}
  \nu_2 = 0
\end{equation}
и
\begin{equation}\label{feq3_12}
  \nu_3 = 0.
\end{equation}
Поэтому изучим перестройки обобщенной границы, возникающие на данных сечениях.

Пусть выполнено \eqref{feq3_11} и $\nu_1 \ne 0$ (иначе имеем случай \eqref{feq3_9}). Из системы \eqref{feq3_2} -- \eqref{feq3_4} получаем
\begin{equation}\label{feq3_13}
  \ds \omega_1 = -\frac{\omega_3\nu_3}{2\nu_1}, \quad \omega_2^2 = \frac{Q(\omega_3)}{4\nu_1^2}, \quad X(\omega_3) = 0,
\end{equation}
где
\begin{equation*}
  \begin{array}{c}
     Q(\omega_3) = 4(h + \nu_1)\nu_1^2 - (1+ \nu_1^2)\omega_3^2, \\
     X(\omega_3) = \nu_1 \omega_3^4 - 4[\nu_1(h + \nu_1) - 1]\omega_3^2 - 4\nu_1(h^2-k).
   \end{array}
\end{equation*}
Преобразуем уравнение \eqref{feq3_5}: $Q(\omega_3)X'(\omega_3) = 0$. Условия совместности систем $X(\omega_3) = 0$, $Q(\omega_3) = 0$ и $X(\omega_3) = 0$, $X'(\omega_3) = 0$ имеют соответственно вид
\begin{equation}\label{feq3_14}
  (1 - \nu_1^2)^2h^2 + 4\nu_1(1-\nu_1^2)h + 4\nu_1^2 - (1 + \nu_1^2)^2k = 0,
\end{equation}
и
\begin{equation}\label{feq3_15}
  2\nu_1(1 - \nu_1^2)h - (1-\nu_1^2)^2-\nu_1^2k = 0.
\end{equation}
Пересечение обобщенной границы ОВД с окружностью \eqref{feq3_11} состоит, таким образом, из точек, удовлетворяющих условиям \eqref{feq3_14} или \eqref{feq3_15}. Разделяющие случаи возможны, если одно из уравнений \eqref{feq3_14}, \eqref{feq3_15} имеет кратный корень или если эти уравнения имеют общий корень. В первом случае приходим к бифуркационным кривым $k = h^2$, $k = 1 + h^2$. Во втором -- находим параметрические уравнения разделяющей кривой
\begin{equation*}
  \ds h = \frac{1 + \nu_1^4}{\nu_1(1 - \nu_1^2)}, \quad k = \left(\nu_1 + \frac{1}{\nu_1}\right)^2, \quad \nu_1 \in (0, 1),
\end{equation*}
или в явном виде
\begin{equation}\label{feq3_16}
  \ds h = \frac{k - 2}{\sqrt{k-4}}, \quad k > 4.
\end{equation}

Обратимся к случаю \eqref{feq3_12}. Из \eqref{feq3_2}, \eqref{feq3_3}
\begin{equation}\label{feq3_17}
  \omega_1 = \pm\sqrt{\Omega}\,\nu_2, \quad \omega_2 = \mp\sqrt{\Omega}\,\nu_1, \quad \omega_3^2 = 2R(\Omega),
\end{equation}
где
\begin{equation}\label{feq3_18}
  \Omega = \omega_1^2 + \omega_2^2, \quad R(\Omega) = h + \nu_1 - \Omega.
\end{equation}
Подстановка \eqref{feq3_17}, \eqref{feq3_18} в \eqref{feq3_4}, \eqref{feq3_5} приводит к системе
\begin{equation*}
  Z(\Omega) = 0, \quad \Omega R(\Omega) Z'(\Omega) = 0,
\end{equation*}
в которой $Z(\Omega) = \Omega^2 - 2\nu_1 \Omega + 1 - k$. Здесь нетрудно разобрать все возможные ситуации. Кроме бифуркационных кривых разделяющим оказывается отрезок
\begin{equation}\label{feq3_19}
  \{h = 0, \quad 0 < k < 1\}.
\end{equation}

В табл. \ref{tab61} показано изменение числа допустимых скоростей на окружности \eqref{feq3_12} в зависимости от положения точки $(h, k)$ и значения переменной $\nu_1$. Через $\hat{\partial}U_{h,k}$ обозначена обобщенная граница ОВД, для особых значений переменной $\nu_1$ введены обозначения $\sigma_* = \sqrt{1 - k}$, $\sigma^* = \sqrt{h^2 + 1 -k}$.

\begin{center}
\begin{tabular}{|c|c|c|c|c|}
  \multicolumn{5}{r}{\footnotesize Таблица \myt\label{tab61}} \\
  \hline
  Параметры & $\hat{\partial}U_{h,k}$ &
  {\renewcommand{\arraystretch}{0.7}
  \begin{tabular}{c} \strut Число\\ допустимых\\ скоростей \end{tabular}}
   & $\mathop{\rm int}\nolimits U_{h,k}$ &
   {\renewcommand{\arraystretch}{0.7}
   \begin{tabular}{c} \strut Число\\ допустимых\\ скоростей \end{tabular}}  \\
  \hline
  \begin{tabular}{c} $-\sqrt{k} < h < 0$\\ $0 < k < 1$ \end{tabular} & $\sigma^*$ & 2 & $(\sigma^*, 1]$ & 4 \\
  \hline
  $0 < h < \sqrt{k}$ & $\sigma_*$ & 2 & $(\sigma_*,\sigma^*) $ & 8 \\
  $0 < k < 1$ & $\sigma^*$ & 6 & $(\sigma_*,1]$ & 4 \\
  \hline
  \begin{tabular}{c} $h > \sqrt{k}$\\$ 0 < k < 1$ \end{tabular} & $\sigma_*$ & 4 & $(\sigma_*, 1]$ & 8 \\
  \hline
  \begin{tabular}{c} $\sqrt{k-1} < h <\sqrt{k}$\\ $k>1$ \end{tabular} & \begin{tabular}{c} $-\sigma^*$ \\ $\sigma^*$ \end{tabular} & \begin{tabular}{c} 2 \\ 2 \end{tabular} & $(-\sigma^*, \sigma^*)$ & 4 \\
  \hline
  $h > \sqrt{k}, \quad k > 1$ & $\varnothing$ & $-$ & $[-1, 1]$& 4 \\
  \hline
\end{tabular}
\end{center}

Изучим теперь точки сферы, в которых
\begin{equation}\label{feq3_20}
  \nu_2 \nu_3 \ne 0.
\end{equation}
Введем следующие обозначения:
\begin{equation}\label{feq3_21}
  \Omega = \omega_1^2 + \omega_2^2, \quad x = \omega_1 \nu_1 + \omega_2 \nu_2, \quad y = \omega_2 \nu_1 - \omega_1 \nu_2, \quad z = \nu_1^2 + \nu_2^2.
\end{equation}
Исключая $\omega_3$ из \eqref{feq3_2}, \eqref{feq3_3}, находим
\begin{equation}\label{feq3_22}
  2x^2 = (z - 1)[\Omega - (h + \nu_1)].
\end{equation}
Используя связь между величинами \eqref{feq3_21} и преобразуя уравнение \eqref{feq3_4}, имеем
\begin{equation}\label{feq3_23}
  \begin{array}{c}
    x^2 + y^2 = z\Omega, \\
    z\Omega^2 + 2[\nu_1(x^2 - y^2) - 2\nu_2xy]+z(z - k) = 0.
  \end{array}
\end{equation}
Приравняем нулю результант \eqref{feq3_23} по $y$:
\begin{equation*}
  16x^4 + 8[\nu_1\Omega^2 - 2z\Omega + \nu_1(z - k)]x^2 + z[\Omega^2 - 2\nu_1\Omega + (z - k)]^2 = 0.
\end{equation*}
Подставляя сюда выражение \eqref{feq3_22}, получаем одно уравнение, заменяющее систему \eqref{feq3_2} -- \eqref{feq3_4}:
\begin{equation}\label{feq3_24}
  a_4\Omega^4 - 4a_3\Omega^3 + 2a_2\Omega^2 - 4a_1 \Omega +a_0 = 0,
\end{equation}
в котором
\begin{equation*}
  \begin{array}{l}
    a_4 = z, \qquad a_3 = \nu_1, \\
    a_2 = 2\nu_1^2 - 2(z - 1)h\nu_1 - (z^2 + kz - 2), \\
    a_1 = -(z + k - 2)\nu_1 - 2(z - 1)h, \\
    a_0 = 4(k - 1)(z - 1)\nu_1^2 + 4(z + k - 2)(z - 1)h\nu_1 + 4(z - 1)^2 h^2 + z(z - k)^2.
  \end{array}
\end{equation*}
Преобразуем уравнение \eqref{feq3_5}. Исключив $\omega_3$ с помощью \eqref{feq3_3}, в обозначениях \eqref{feq3_21} запишем
\begin{equation*}
  2(z\Omega - \nu_1)xy = \nu_2[(1+z)x^2 - (1 - z)y^2].
\end{equation*}
Вводя в полученное соотношение $x^2, xy, y^2$, найденные из \eqref{feq3_22}, \eqref{feq3_23}, придем к уравнению
\begin{equation}\label{feq3_25}
  a_4\Omega^3 - 3a_3\Omega^2 + a_2\Omega - a_1 = 0.
\end{equation}
Таким образом, обобщенная граница ОВД отвечает, как и следовало ожидать, кратным корням уравнения \eqref{feq3_24}. Корни кратности три
\begin{equation}\label{feq3_26}
  3a_4\Omega^2 - 6a_3\Omega + a_2 = 0
\end{equation}
соответствуют особым точкам обобщенной границы, а общее число особых точек может измениться лишь в случае наличия четырехкратного корня
\begin{equation}\label{feq3_27}
  a_4\Omega - a_3 = 0.
\end{equation}
Если отбросить уже изученный случай $\nu_2 = 0$, система \eqref{feq3_24} -- \eqref{feq3_27} дает уравнение новой (и последней) разделяющей кривой
\begin{equation}\label{feq3_28}
  \ds h^2 = k - \frac{1}{k}, \quad k \gs 1.
\end{equation}


\begin{figure}[ht]
\center{\includegraphics[width=0.9\linewidth]{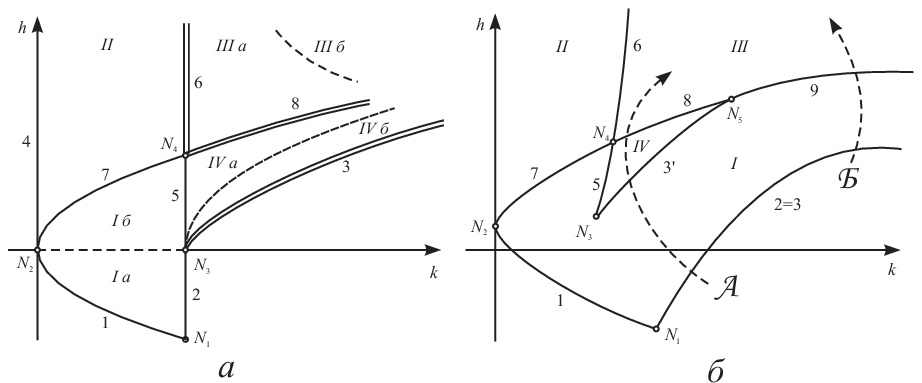}}
\caption{}\label{fig_6_3}
\end{figure}

На рис.~\ref{fig_6_3},\,{\it а} нанесены разделяющие множества \eqref{feq3_16}, \eqref{feq3_19}, \eqref{feq3_28}, арабскими цифрами обозначены качественно различные участки бифуркационных кривых (смысл двойных линий поясним далее) и выделены узловые точки $N_1 - N_4$. Имеется семь открытых областей с устойчивыми относительно возмущений $h$ и $k$ при фиксированном $g = 0$ видами ОВД. Для построения соответствующих рисунков можно использовать численные методы. С этой целью преобразуем уравнение обобщенной границы.

Комбинируя \eqref{feq3_24}, \eqref{feq3_25}, приходим к системе
\begin{equation}\label{feq3_29}
  U(z, \Omega)\nu_1 - V(z, \Omega) = 0, \quad V(z, \Omega)\nu_1 - W(z, \Omega) = 0,
\end{equation}
в которой
\begin{equation*}
  \begin{array}{l}
    U = 4\left\{\Omega^4 + [(3k - 2)z - 2k +1]\Omega^2 + 2(kz - 1)(z - 1)h\Omega -\right. \\
    \qquad \left. - (k - 1)[z^2 + (k - 3)z - k + 2] \right\},\\
    V = 2\left\{z\Omega^5 + 2z(k - 1)(z - 1)\Omega^3 - 4(z - 1)h\Omega^2 - \right. \\
    \qquad -[4(z - 1)^2h^2 + (2k -1)z^3 + 2(k^2 - 3k + 1)z^2 - (k^2 + 2k - 4)z + \\ \qquad \left.+4(k - 1)]\Omega + (k - 1)h(z - 1)^2 \right\}, \\
    W = z\Omega^6 + 2z(z - 1)h\Omega^5 - [z^2 + (k - 6)z + 4]\Omega^4 - \\
    \qquad -4(z - 1)h[z^2 + (k - 2)z + 4]\Omega^3 - \\
    \qquad - [20(z - 1)h^2 + z^3 - 2(5k - 2)z^2 + (k - 2)^2z + 4(k - 2)]\Omega^2 - \\
    \qquad - 2(z - 1)h[4(z - 1)^2h^2 - z^3 -2(3k -2)z^2 -  \\
    \qquad - (k^2 - 4k -4)z + z(k - 2)]\Omega - (z + k - 2)[4(z - 1)h^2 - z^3 + 2kz^2 - k^2z].
  \end{array}
\end{equation*}
Условие совместности системы \eqref{feq3_29}
\begin{equation}\label{feq3_30}
  UW - V^2 = 0
\end{equation}
преобразуем к виду $Z(z, \Omega)P(z, \Omega) = 0$, где
\begin{equation*}
  \begin{array}{l}
  Z(z, \Omega) = z\Omega^4 - 2[z^2 + (k - 2)z + 2]\Omega^2 - \\
  \qquad -8(z - 1)h\Omega + z(z + k - 2)^2 - 4(z - 1)^2 h^2, \\
  P(z, \Omega) = \Omega^6 - 2h\Omega^5 + [2(2k - 1)z - k -1]\Omega^4 - 4[(2k -1)z - k]\Omega^3 + \\
  \qquad +[(2k -1)^2z^2 - 2(2k - 1)z - k(k - 2) - 4kz(z - 1)h^2]\Omega^2 -\\
  \qquad - 2(k - z)^2h\Omega + (k - 1)(k - z)^2.
  \end{array}
\end{equation*}
На основании равенства $Z(z, \Omega) = 0$ имеем
\begin{equation*}
  W - zU = [\Omega^2 + 2h(z - 1)\Omega + z + k -2]Z(z, \Omega) = 0,
\end{equation*}
откуда в силу \eqref{feq3_29}, \eqref{feq3_30} $\nu_1^2 = z$ и, следовательно, $\nu_2 = 0$, что противоречит \eqref{feq3_20}. Итак, в области \eqref{feq3_20} величины $z, \Omega$ связаны уравнением
\begin{equation}\label{feq3_31}
  P(z, \Omega) = 0.
\end{equation}
Меняя $z$ в пределах $0 \ls z \ls 1$, определяем, вообще говоря, многозначную зависимость ${\Omega = \Omega(z)}$. Величину $\nu_1$ находим из системы \eqref{feq3_29}:
\begin{equation}\label{feq3_32}
  \ds \nu_1 = \frac{V(z, \Omega(z))}{U(z, \Omega(z))} \equiv \frac{W(z, \Omega(z))}{V(z, \Omega(z))},
\end{equation}
после чего
\begin{equation}\label{feq3_33}
  \nu_2 = \pm\sqrt{z - \nu_1^2}, \quad \nu_3 = \pm\sqrt{1 - z}.
\end{equation}


\begin{figure}[ht]
\centering
\includegraphics[width=1\linewidth]{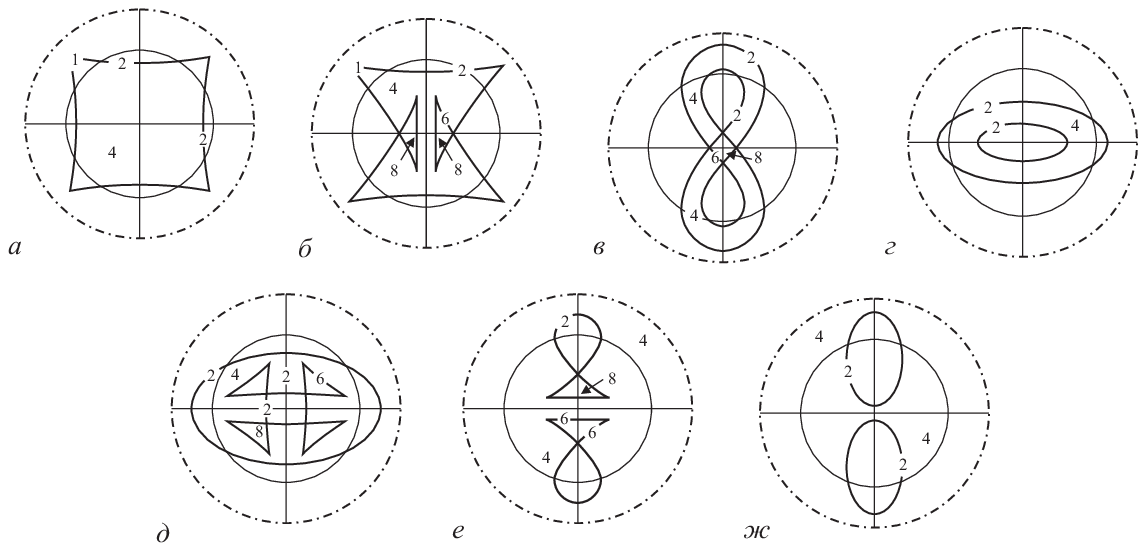}
\caption{}\label{fig_6_4}
\end{figure}

Реализуя вычисления по формулам \eqref{feq3_31} -- \eqref{feq3_33}, получаем виды ОВД, показанные на рис.~\ref{fig_6_4} в стереографической проекции полюса $(1, 0, 0)$. Сам полюс (<<бесконечно удаленная точка>>) изображен штрихпунктирной линией, тонкая сплошная окружность есть сечение сферы $\nu_1 = 0$, образы окружностей \eqref{feq3_11} и \eqref{feq3_12} -- вертикальный и горизонтальный отрезки. Соответствие рис.~\ref{fig_6_4},\,{\it а\,--\,ж} и областей на рис.~\ref{fig_6_3},\,{\it а} таково: {\it a -- \ts{I}а}, {\it б -- \ts{I}б}, {\it в -- \ts{II}}; {\it г -- \ts{IV}б}, {\it д -- \ts{IV}а}, {\it ж -- \ts{III}б}.

Число допустимых скоростей в связных компонентах дополнения обобщенной границы определяется из табл.~\ref{tab61}. Лишь в случаях {\it д} и {\it е} возникают компоненты, не имеющие выхода на сечение \eqref{feq3_12}. Здесь можно прибегнуть к анализу системы \eqref{feq3_11}, \eqref{feq3_13}. Впрочем, результат ясен из общих соображений.

Случай {\it а} хорошо известен: тор <<складывается>> в цилиндр, который затем проектируется <<ортогонально образующей>>. В случае {\it б} цилиндр предварительно деформирован поворотом одного из основания на угол меньше $\pi$. В случаях \textit{в}~-- \textit{ж} имеется кольцевидная область с самопересечением или особенностями типа <<сборки>> на границах ({\it д, е}). На эту область проектируются оба интегральных тора, порождая противоположные движения по кольцу (при $g = 0$ приведенная система на $S^2$ обратима). Неустойчивые случаи, отвечающие разделяющим кривым \eqref{feq3_16}, \eqref{feq3_19}, \eqref{feq3_28}, очевидно, отражают моменты исчезновения <<хвостов>>.

Изучим теперь интегральные поверхности в точках бифуркационного множества. Собственно критические движения, фазовые траектории которых состоят из критических точек интегрального отображения, достаточно хорошо изучены \cite{bib03,bib21}. Перечислим их согласно обозначениям на  рис.~\ref{fig_6_3},\,{\it а}.

Точке $N_1$ отвечает устойчивое равновесие тела. Область возможности движения сводится к единственной точке -- полюсу \eqref{feq3_6}.

На участках 1, 2 осуществляются колебания физического маятника вокруг второй и третьей осей инерции с амплитудой меньше $\pi$. ОВД сводится к отрезкам сечений \eqref{feq3_11}, \eqref{feq3_12}. В точках $N_2$, $N_3$ амплитуда колебаний достигает $\pi$. Во всех этих случаях интегральная поверхность диффеоморфна окружности.

Движения физического маятника происходят также на участках 5 (колебания около третьей оси с амплитудой больше $\pi$), 6 (вращения вокруг третьей оси в двух направлениях), 7 (колебания около второй оси с амплитудой больше $\pi$), 8 (вращения вокруг второй оси в двух направлениях). Однако в этих  случаях они уже не исчерпывают всех движений на соответствующей интегральной поверхности (в терминологии \cite{bib03} существуют простейшие движения, отличные от особо замечательных). Поэтому разберем указанные случаи более подробно.

Рассмотрим участок 4 (случай Делоне). Здесь возможны лишь собственно критические (особо замечательные) движения. Интегральное многообразие состоит из двух периодических траекторий. Действительно, из \eqref{feq3_2} -- \eqref{feq3_4} получаем
\begin{equation*}
  \nu_1 = \sqrt{1-\nu_3^2} - h \nu_3^2, \quad \nu_2 = \pm \sqrt{h\nu_3^2\left(2\sqrt{1 - \nu_3^2} - h\nu_3^2 \right)},
\end{equation*}
поэтому область возможности движения имеет вид <<восьмерки>> с центром в точке $\nu_1$, петли которой охватывают ось $\nu_3$. Движения по <<восьмерке>> происходят в обоих направлениях (в центре четыре допустимых скорости, в остальных точках~-- две). Следовательно, $J_{h, 0} = 2S^1$, как и утверждалось.

Рассмотрим участок 3. При исследовании кривой
\begin{equation}\label{feq3_34}
  k = h^2 + 1
\end{equation}
в работе \cite{bib21} допущены неточности. А именно в ней утверждается, что при $h = -\sqrt{k - 1}$ существуют особо замечательные движения, а при $h = \sqrt{k - 1}$, кроме названных, присутствуют и другие, простейшие движения. Из рис.~\ref{fig_6_1} (непрерывный переход от {\it б} к {\it а}) явствует, что это не так. Действительно, при условии \eqref{feq3_34} уравнение \eqref{feq3_31} приобретает вид
\begin{equation*}
  (\Omega - h)^2[(\Omega^2 + z - 1 - h^2)^2 + 4h^2z\Omega^2]=0,
\end{equation*}
откуда следует, что на границе ОВД
\begin{equation}\label{feq3_35}
  \Omega = h.
\end{equation}
Таким образом, в силу \eqref{feq3_21} $h > 0$ (случай $h = 0$, $k = 1$ уже отмечался -- это точка $N_3$). Подставляя \eqref{feq3_35}, в \eqref{feq3_32}, \eqref{feq3_33}, после исключения $z$ находим
\begin{equation}\label{feq3_36}
  \nu_2^2 = 1 - 2h\nu_1 - \nu_1^2, \quad \nu_3^2 = 2h\nu_1.
\end{equation}
Замкнутая кривая \eqref{feq3_36} лежит в полусфере $\nu_1 \gs 0$, содержит точки $\nu_2 = \pm1$ и охватывает ось $\nu_1$. На самой кривой две допустимых скорости, вне ее скоростей нет. Следовательно, ОВД диффеоморфна окружности. По ней происходят движения в двух направлениях. Интегральная поверхность $J_{h, k} = 2S^1$. Возвращаясь к  рис.~\ref{fig_6_3},\,{\it а} подведем итог: участкам бифуркационного множества, изображенного двойной линией, соответствуют две критические окружности, простой линией -- одна критическая окружность.

Обратимся к более подробному исследованию участков 5 -- 8. Здесь, кроме критических движений, имеются гладкие листы асимптотических поверхностей. Особые точки проектирования этих поверхностей на сферу Пуассона не принадлежат к числу критических точек интегрального отображения. Поэтому к ним применимы все рассуждения, использовавшиеся при построении ОВД в некритических случаях. В частности, при численном построении обобщенных границ могут с тем же успехом быть использованы уравнения \eqref{feq3_31} -- \eqref{feq3_33}. Никаких дополнительных особенностей здесь не возникает.

Рассмотрим точку на участке 5. Здесь $k = 0$, $0 < h < 1$. Соответствующая ОВД приведена на рис.~\ref{fig_6_5},\,{\it а}. Она разделяет случаи, изображенные на рис.~\ref{fig_6_4},\,{\it б, д}. С точки зрения трансформации ОВД переход из области {\it Iб} в область {\it IVа} (рис.~\ref{fig_6_3}) выглядит так: стороны внешней границы фигуры на рис.~\ref{fig_6_4},\,{\it б} складываются вдвое, порождая рис.~\ref{fig_6_5},\,{\it а}; затем горизонтальная средняя линия критической ОВД <<раздваивается>>, в результате чего приходим к рис.~\ref{fig_6_4},\,{\it д}. В фазовом же пространстве происходит следующее. На торе $J_{h, 1 - \varepsilon}$ (область {\it Iб}) выбираются два гомотопных между собой, но не гомотопных нулю цикла, которые при $\varepsilon \to +0$ склеиваются. В результате $J_{h, 1} = S^1 \times (S^1 \dot{\cup} S^1)$ (см. рис.~\ref{fig_3_2}). Средняя линия этой поверхности (след центра <<восьмерки>>) как раз и проектируется в горизонтальный отрезок на рис.~\ref{fig_6_5},\,{\it а}. Она является периодической траекторией движения физического маятника и служит $\alpha-$ и $\omega-$предельным циклом для остальных траекторий, лежащих на $J_{h, 1}$. Разрежем теперь $J_{h, 1}$ по средней линии так, чтобы каждая <<восьмерка>> распалась на на две окружности. Получим два тора, входящие в состав $J_{h, 1 + \varepsilon}$ (область {\it IVа}). Они проектируются на одну и ту же область, изображенную на рис.~\ref{fig_6_4},\,{\it д} (что связано с уже упоминавшейся обратимостью системы на $S^2$ при $g = 0$), поэтому число допустимых скоростей в этой области всюду четное.

\begin{figure}[ht]
\centering
\includegraphics[width=0.75\linewidth]{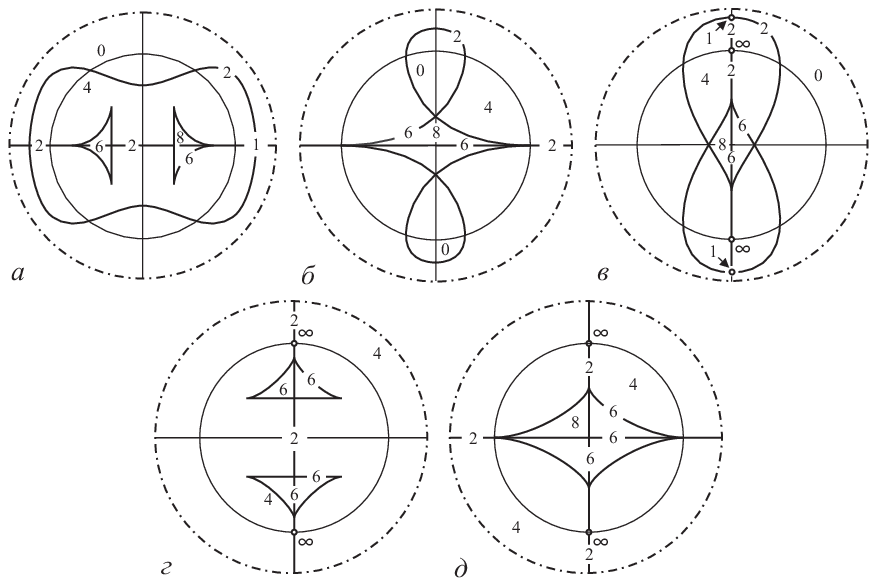}
\caption{}\label{fig_6_5}
\end{figure}


Область возможности движения, соответствующая точке $(1, h)$, $h > 1$ (участок~6 на рис.~\ref{fig_6_3},\,{\it а}) приведена на рис.~\ref{fig_6_5},\,{\it б}. ОВД для точек $(1 \mp \varepsilon, h)$ имеют вид, указанный на рис.~\ref{fig_6_4},\,{\it в, е}. Перестройка происходит с каждым интегральным тором в отдельности (число допустимых скоростей всюду четное). Качественно ее можно описать следующим образом. Рассмотрим на торе $(\varphi_1 \mathop{\rm mod}\nolimits 2, \varphi_2 \mathop{\rm mod}\nolimits 1)$, входящем в состав $J_{h, 1 - \varepsilon}$, цикл $\varphi_2 = 0$ и отождествим на нем точки, удаленные друг от друга на единицу. Результирующая поверхность гомеоморфна произведению <<восьмерки>> на отрезок, у которого края отождествлены по отображению $S^1 \dot{\cup} S^1 \to S^1 \dot{\cup} S^1$, сохраняющему ориентацию, но меняющему местами петли <<восьмерки>> (центральная симметрия). Эта поверхность, уже встречавшаяся в $\S$~\ref{ssec54}, обозначена $S^1 * (S^1 \dot{\cup} S^1)$. Она ориентируема, но представляет собой пространство нетривиального расслоения над окружностью со слоем $S^1 \dot{\cup} S^1$ (циклы, пространственно близкие к средней линии, должны иметь удвоенную длину (см. рис.~\ref{fig_5_7})). Раздвоим теперь среднюю линию этой поверхности так, чтобы каждая превратилась в  одну окружность. Получим тор $(\varphi_1 \mathop{\rm mod}\nolimits  1, \varphi_2 \mathop{\rm mod}\nolimits  2)$, а след средней линии на нем -- цикл $\varphi_1 = \varphi_2$ типа \eqref{feq2_1}. Итак, в данном случае $J_{h, 1} = 2S^1 * (S^1 \dot{\cup} S^1)$. Средние линии образовавшихся поверхностей -- это вращения маятника. Они являются предельными движениями для остальных траекторий связной компоненты при $t \to \pm \infty$.

Рассмотрим участок 7 (рис.~\ref{fig_6_3},\,{\it а}). На рис.~\ref{fig_6_5},\,{\it в} приведена ОВД для точки $(h^2, h)$, $h < 1$. В полюсах \eqref{feq3_9} сферы Пуассона имеется целая окружность допустимых скоростей. Они характеризуют возможные состояния тела, в которых ось динамической симметрии вертикальна. Каждая траектория поочередно пересекает полюсы \eqref{feq3_9} и при $t \to \pm \infty$ приближается к следу маятникового движения. Интегральная поверхность $S^1 \times (S^1 \dot{\cup} S^1)$ получается из близких интегральных многообразий точно так же, как на участке 5. Нетрудно проследить и трансформацию областей возможности движения при переходе из области $\ts{II}$ (рис.~\ref{fig_6_4},\,{\it в}) в область {\it Iб} (рис.~\ref{fig_6_4},\,{\it б}). Интересно отметить, что непрерывно меняя отображение проектирования $J_{h, h^2} \to S^2$, можно ликвидировать <<скрученность>> проекций, т.е. получить переход от области на рис.~\ref{fig_6_4},\,{\it а}  к области на рис.~\ref{fig_6_4},\,{\it г}. Такой подход реализуется, например, в задаче Клебша при нулевой постоянной интеграла площадей. Процесс непосредственного построения критической интегральной поверхности по соответствующей ОВД подробно описан в работе~\cite{bib60}.

\begin{figure}[ht]
\centering
\includegraphics[width=\linewidth]{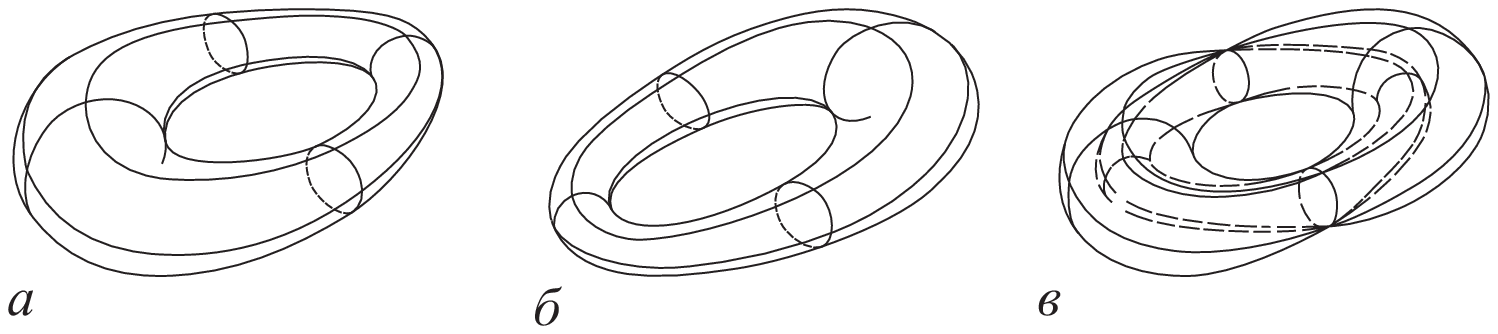}
\caption{}\label{fig_6_6}
\end{figure}


При переходе через участок 8 сверху вниз происходит изменение ОВД, которое можно проследить на рис.~\ref{fig_6_4},\,{\it е}, \ref{fig_6_5},\,{\it г}, \ref{fig_6_4},\,{\it д}. Здесь на каждом интегральном торе возникают два предельных цикла, отвечающих вращениям маятника и склеивающихся с соответствующими циклами второго тора. Результат есть поверхность $(S^1 \ddot{\cup} S^1) {\times} S^1$  (рис.~\ref{fig_6_6}). Попутно выясняется наглядный смысл глубокой теоремы Аппельрота \cite{bib03} о том, что в рассматриваемых движениях (третья группа 3-го класса) ось динамической симметрии периодически становится вертикальной, причем за исключением особо замечательных (в данном случае -- маятниковых) движений направление этой оси в вертикальном положении всегда одно и то же. В самом деле, прообразом полюса $(0, 0, 1)$ является окружность на одном из торов, трансверсальная предельным циклам (например, самая внутренняя окружность на  рис.~\ref{fig_6_6},\,{\it а}). Прообраз полюса $(0, 0, -1)$ есть аналогичная окружность на другом торе (например, самая внешняя окружность на рис.~\ref{fig_6_6},\,{\it б}). В результате все асимптотические движения периодически пересекают лишь одну из этих окружностей, а предельные поочередно пересекают обе. Кроме того, чем ближе асимптотическая траектория, проходящая через заданный полюс, к предельному циклу, тем, очевидно, ближе она периодически подходит к другому полюсу, никогда его не достигая. Этот факт тоже отмечен в одной из теорем Аппельрота. Можно дать следующее качественное описание поведения траекторий на сфере: пучок, выходящий из полюса $(0, 0, 1)$, <<улавливается>> нижним <<карманом>> (см.  рис.~\ref{fig_6_5},\,{\it г}), отражаясь от соответствующих складок, он преломляется и через определенный промежуток времени вновь концентрируется в исходном полюсе. Аналогичное явление происходит и с траекториями, пересекающими полюс $(0, 0, -1)$.

\begin{figure}[ht]
\centering
\includegraphics[width=0.4\linewidth]{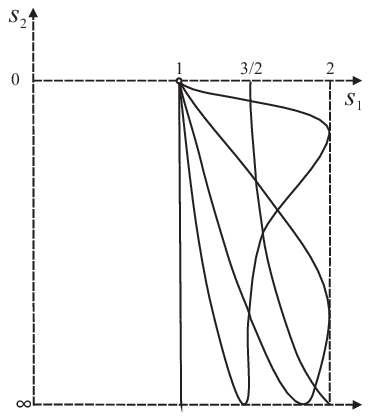}
\caption{}\label{fig_6_7}
\end{figure}

Изучим, наконец, точку $N_4$. Здесь $h = 1$, $k = 1$. Соответствующая ОВД указана на рис.~\ref{fig_6_5},\,{\it д}. Интегральная поверхность получается из рассмотренных ранее с помощью некоторых отождествлений, которые можно проследить на рисунках. Она слишком сложна для наглядного изображения. Относительно характера траекторий скажем следующее. Все они гомоклинические для неподвижной точки $(-1, 0, 0)$, что видно непосредственно из уравнений Ковалевской. Имеются две петли, не пересекающие полюсы $(0, 0, \pm1)$, -- движения по сечению $\nu_3 = 0$. В процессе  всех остальных движений ось динамической симметрии дважды становится вертикальной. Действительно, на плоскости переменных \eqref{feq1_6} траектории заключены в полуполосе $[1, 2] \times [-\infty, 0]$ и устроены, как показано на  рис.~\ref{fig_6_7}. Образ множества состояний с вертикальной третьей осью -- кривая, соединяющая точки $s_1 = 3/2, s_2 = 0$ и $s_1 = 2, s_2 = -\infty$. Эта кривая пересекает каждую траекторию в двух точках. Более детальный анализ здесь не проводится.

В заключение рассмотрим некоторые замечания к вопросу об эволюции критических поверхностей при изменении постоянной интеграла площадей. До тех пор пока некоторый участок бифуркационного множества меняется непрерывно (в очевидном смысле), сохраняется и топологическая структура соответствующей интегральной поверхности. Рассмотрим, например, переход к небольшим ненулевым значениям $g$ (рис.~\ref{fig_6_3},\,{\it б}). Произошло раздвоение участка~3. От него отщепились участки~$3'$ и~9, <<захватив>> одну из критических окружностей. Участок $8$ стал конечной длины. В остальных точках бифуркационного множества ничего не изменилось. Переходы, указанные пунктирными стрелками, таковы:

А) $\varnothing \to S^1 \to {\bf T}^2 \to S^1 \cup {\bf T}^2 \to 2{\bf T}^2 \to (S^1 \ddot{\cup} S^1) \times S^1 \to 2{\bf T}^2$;

Б) $\varnothing \to S^1 \to {\bf T}^2 \to (S^1 \dot{\cup} S^1) \times S^1 \to 2{\bf T}^2$.

Интегральная поверхность в $N_3$ -- средняя между $(S^1 \dot{\cup} S^1) \times S^1$ и $S^1 \cup {\bf T}^2$. Очевидно, это $S_\wedge^1 \times S^1$ ($S_\wedge^1$ -- окружность с угловой точкой). Угловая точка дает периодическое решение -- собственно критическое многообразие. Остальные траектории к нему асимптотически приближаются.

Анализируя эволюцию сечения бифуркационного множества на рис.~\ref{fig_6_1}, замечаем, что дополнительного исследования требуют лишь критические интегральные поверхности, отвечающие точкам границы области $V$.

Переход из области $V$ в область $\ts{III}$ совершенно очевиден, если обратиться к рис.~\ref{fig_6_2}: ситуация, промежуточная между рис.~\ref{fig_6_2},\,{\it в} и рис.~\ref{fig_6_2},\,{\it д} -- это суперпозиция проекций двух поверхностей вида $S^1 \times (S^1 \dot{\cup} S^1)$. Поэтому здесь бифуркация такова: $4{\bf T}^2 \to 2S^1 \times (S^1 \dot{\cup} S^1) \to 2{\bf T}^2$.

Перестройки между областями $\ts{II}$  и $V$ менее ясны. Наиболее простой способ их выяснения~-- рассмотреть <<след>> этих перестроек на плоскости $k = 0$ (см. рис.~\ref{fig_6_1},\,{\it г}). При этом из системы первых интегралов получаем
\begin{equation}\label{feq3_37}
  4\omega_1^2 + \omega_3^2 = 2h, \quad 2h(\omega_1^2 + \omega_2^2) + 4g\omega_1 = \pm \omega_3 \sqrt{2(h - 2g^2)}.
\end{equation}
\begin{equation}\label{feq3_38}
  \nu_1 = \omega_2^2 - \omega_1^2, \quad \nu_2 = -2\omega_1\omega_2, \quad h\nu_3 = \pm \omega_1 \sqrt{2(h - 2g^2)} + g \omega_3.
\end{equation}

Поскольку величины $\nu_i$ однозначно выражены через $\omega_1, \omega_2, \omega_3$ равенствами \eqref{feq3_38} (знаки в \eqref{feq3_37} и \eqref{feq3_38} согласованы), то интегральное многообразие, в данном случае одномерное, диффеоморфно кривой $\Gamma$ -- линии пересечения цилиндра и параболоида \eqref{feq3_37}. Естественно, что бифуркации $\Gamma$ происходят при значениях \eqref{feq2_4}. Соответствующий анализ весьма прост \cite{bib25,bib62}: на ограниченной ветви кривой \eqref{feq2_4} $\Gamma = 2(S^1 \dot{\cup} S^1)$, на неограниченной -- множество $\Gamma$ состоит из двух окружностей и двух изолированных точек. Возвращаясь к значениям $k > 0$, заключаем, что переход из области $V$ в область $\ts{II}$ через нижнюю границу (см. рис.~\ref{fig_6_1},\,{\it в, г}) сопровождается бифуркацией $4{\bf T}^2 \to 2S^1 \times (S^1 \dot{\cup} S^1) \to 2{\bf T}^2$, а через верхнюю -- бифуркацией $4{\bf T}^2 \to 2{\bf T}^2 \cup 2S^1 \to 2{\bf T}^2$. В последнем случае на двух непрерывно изменяющихся торах нет критических точек интегрального отображения.

\subsection{Комментарий к главе 6}\label{ssec64}

Среди точных решений задачи о движении твердого тела с закрепленной точкой наиболее сложным в плане аналитического и качественного исследования является, несомненно, решение, указанное С.В.\,Ковалевской \cite{bib24}. Оригинальность идей, приведших к его открытию, виртуозность приемов, сводящих к квадратурам уравнения Эйлера\,--\,Пуассона, широта чисто математических проблем, -- все это предопределило интерес ученых к решению Ковалевской, не ослабевающий на протяжении столетия. Отметим лишь некоторые из результатов, находящихся в непосредственной связи с настоящей работой.

Г.Г.\,Аппельрот \cite{bib03} первым выделил четыре класса особых движений гироскопа Ковалевской. Эти классы определяются наличием кратного корня многочлена пятой степени, находящегося в правой части уравнений, которые описывают зависимость от времени вспомогательных переменных Ковалевской. Более подробно особые движения изучил А.Ф.\,Ипатов \cite{bib21}. Он же дал геометрическую трактовку классов Аппельрота как отвечающих частям поверхности кратных корней упомянутого многочлена в пространстве констант первых интегралов. В работах \cite{bib61} -- \cite{bib63} методом годографов изучены движения, относящиеся к первому и второму классам, и так называемые особо замечательные движения третьего класса, в которых подвижный годограф угловой скорости тела -- замкнутая кривая.

В действительности классам Аппельрота в пространстве констант первых интегралов соответствует множество меры нуль. Остальные движения не были подробно изучены. В существенном предположении независимости интегралов уравнений Эйлера\,--\,Пуассона поведение углов прецессии и собственного вращения исследовалось В.В.\,Козловым \cite{bib26}. До последнего времени, однако, не было известно, при каких именно значениях постоянных первые интегралы являются независимыми. В данной главе доказано, в частности, что случаям зависимости интегралов отвечают именно классы Аппельрота. Конечно, результат этот подсказан всем опытом работы с разделяющимися переменными. Но все же он требует доказательства, так как в отличие, например, от лиувиллевых систем или случая Чаплыгина\,--\,Сретенского здесь не удается указать разделяющего преобразования, не содержащего постоянных интегралов. Более того, известные формулы, выражающие переменные Эйлера\,--\,Пуассона через переменные Ковалевской, выводятся по существу в предположении независимости интегралов.

В настоящей главе получено также полное описание интегральных многообразий и их бифуркаций, выполнен геометрический анализ при нулевой постоянной площадей. Отметим, что здесь критические интегральные поверхности строились на основе областей возможности движения на сфере Пуассона без использования вспомогательных переменных, так как связь последних с переменными Эйлера\,--\,Пуассона существенно сложнее, чем в случае Чаплыгина\,--\,Сретенского.

\end{document}